\documentclass[10pt]{article}
\usepackage[T1]{fontenc}
\usepackage[utf8]{inputenc}
\usepackage[a4paper,width=150mm,top=25mm,bottom=25mm]{geometry}

\usepackage{macros}
\usepackage{thmtools,thm-restate}

\bibliography{./bib}

\newcommand{\betasecond}{\beta_{\mathsf{2nd}}}
\newcommand{\betaunique}{\beta_{\mathsf{unique}}}
\newcommand{\betaKS}{\beta_{\mathsf{KS}}}

\newcommand{\isingHamiltonian}{\mathcal{H}}
\newcommand{\typeProportions}{\vec \alpha}
\newcommand{\configurationModel}{\mathbb{G}}
\newcommand{\rectangularStrip}{\mathsf{rect}}
\newcommand{\lowerBoundTag}{\mathsf{lb}}
\newcommand{\upperBoundTag}{\mathsf{ub}}
\newcommand{\plantedTag}{\mathsf{pl}}
\newcommand{\nullModelTag}{\mathsf{null}}
\newcommand{\broadcastTag}{\mathsf{bd}}
\newcommand{\zeroMultiset}{\mathsf{zeros}}
\newcommand{\differential}{\mathsf{d}}
\newcommand{\observedData}{\mathsf{DATA}}

\newcommand{\momentCounts}{\vec m}
\newcommand{\contourVariables}{\vec x}
\newcommand{\contourCoefficients}{\vec u}

\title{On Zeros and Algorithms for Disordered Systems II: The Antiferromagnetic Ising Model up to Kesten--Stigum}
\author{Ferenc Bencs\thanks{Email: \texttt{ferenc.bencs@gmail.com},
Centrum Wiskunde \& Informatica. Supported by the Netherlands Organisation of Scientific Research (NWO): VI.Veni.222.303.} \and Brice Huang\thanks{Email: \texttt{bmhuang@stanford.edu},
Stanford University. Supported by a Stanford Science Fellowship and NSF Mathematical Sciences Postdoctoral Research Fellowship.} \and Daniel Z. Lee\thanks{Email: \texttt{lee\_d@mit.edu},
Massachusetts Institute of Technology. Supported by the NSF CAREER grant CCF-2443045, and the Reed Fund at MIT.} \and Kuikui Liu\thanks{Email: \texttt{liukui@mit.edu},
Massachusetts Institute of Technology. Supported by the NSF CAREER grant CCF-2443045, and the Reed Fund at MIT.} \and Guus Regts\thanks{Email: \texttt{guusregts@gmail.com},
Korteweg de Vries Institute for Mathematics, University of Amsterdam. Supported by the Netherlands Organisation of Scientific Research (NWO): VI.Vidi.193.068.}}
\date{}

\begin{document}
\maketitle

\begin{abstract}
We give the first polynomial-time algorithm to estimate the partition function of the antiferromagnetic Ising model on the uniformly random $d$-regular graph to arbitrarily high accuracy for all inverse temperatures up to the Kesten--Stigum threshold. This is the entire regime in which the model is ``replica symmetric'', and is conjectured to be optimal for polynomial-time counting and sampling. To achieve our algorithmic result, we establish the absence of Fisher zeros for the partition function in a constant-width strip around the interval from $0$ up to the Kesten--Stigum threshold, with high probability over the randomness of the input graph.

\paragraph{AI Disclosure} The overarching vision and proof architecture were developed solely by the authors without AI input. GPT 5.5 Pro was used, with extensive guidance from the authors, for only a few minor technical aspects of the proof. The authors independently verified correctness and wrote everything by hand. We take full responsibility for its contents. A more detailed statement can be found in \cref{subsec:ack}.
\end{abstract}

{\small
\tableofcontents
}

\newpage

\section{Introduction}
The \emph{Ising model} is among the most important and well-studied models in statistical physics. Originally invented to mathematically model phase transition phenomena in real-world magnets \cite{Len20, Isi25}, it has since become a foundational model in the complexity theory of high-dimensional statistical tasks such as approximate counting and sampling. Given a $n$-vertex graph $G=(V,E)$, the \emph{Hamiltonian} of the Ising model is defined by the function
\begin{align*}
    \mathcal{H}_{G}(\sigma) \defeq \sum_{uv \in E} \frac{1 - \sigma_{u}\sigma_{v}}{2}, \qquad \forall \sigma \in \{\pm1\}^{V},
\end{align*}
which counts the number of edges cut by the bipartition formed between $\{v \in V : \sigma(v) = +1\}$ and $\{v \in V : \sigma(v) = -1\}$.\footnote{Note that the convention we use for the Hamiltonian differs from the one used in \cite{FL21} by an additive constant independent of $\sigma$.} For an \emph{inverse temperature} $\beta \in \R$, this Hamiltonian induces the \emph{Gibbs measure}
\begin{align*}
    \mu_{G,\beta}(\sigma) \propto \exp\wrapp{\beta \cdot \mathcal{H}_{G}(\sigma)}
\end{align*}
over $\{\pm1\}^{V}$, where the constant of proportionality normalizing $\mu_{G,\beta}$ into a probability measure is the \emph{partition function}
\begin{align*}
    Z_{G}(\beta) \defeq \sum_{\sigma \in \{\pm1\}^{V}} \exp\wrapp{\beta \cdot \mathcal{H}_{G}(\sigma)}.
\end{align*}
When $\beta \geq 0$, the model is \emph{antiferromagnetic} since $\mu_{G,\beta}$ favors configurations with a large number of disagreements across edges. When $\beta \leq 0$, the model is \emph{ferromagnetic} since $\mu_{G,\beta}$ favors configurations with a high amount of agreement. In this paper, we consider the antiferromagnetic setting. Our goal is to design efficient algorithms solving the following task.
\begin{problem}[Approximate Counting]\label{prob:approx-count}
Given a graph $G$, an inverse temperature $\beta \geq 0$, an error tolerance $\eta > 0$ and a failure probability $0 < \delta < 1$, compute an estimate $\widehat{Z}_{G}(\beta)$ such that
\begin{align*}
    \Pr\wrapb{e^{-\eta} \cdot Z_{G}(\beta) \leq \widehat{Z}_{G}(\beta) \leq e^{\eta} \cdot Z_{G}(\beta)} \geq 1 - \delta
\end{align*}
over the randomness of the algorithm (and of $G$, if applicable).
\end{problem}
Approximate counting is a fundamental algorithmic primitive in high-dimensional statistics because the partition function $Z_{G}(\beta)$ encodes many meaningful observables of the Gibbs measure $\mu_{G,\beta}$. These include, for instance, the expected energy $\E_{\sigma \sim \mu_{G,\beta}}\wrapb{\mathcal{H}_{G}(\sigma)}$, which can be obtained by differentiating $\log Z_{G}(\beta)$. Approximate counting is also intimately related to the task of \emph{(approximately) sampling} from the Gibbs measure $\mu_{G,\beta}$ \cite{JVV86}.

Over the past several decades, it was gradually realized that the computational complexity of \cref{prob:approx-count} is deeply influenced by \emph{phase transitions} the antiferromagnetic Ising model undergoes in the large-$n$ limit as the inverse temperature $\beta$ is varied. In particular, four disparate phase transitions have received extensive attention in the statistical physics and theoretical computer science communities:
\begin{itemize}
    \item the transition between computational tractability and intractability for approximate counting (and approximate sampling),
    \item the dynamical transition between rapid and torpid mixing of local Markov chains for sampling from $\mu_{G,\beta}$ (e.g. \emph{Glauber dynamics}),
    \item the static transition between absence and presence of long-range correlations in the Gibbs measure $\mu_{G,\beta}$, and
    \item the transition between non-clustering and clustering of the zeros of the partition function $Z_{G}(\beta)$ in a neighborhood of $\beta$.
\end{itemize}
A striking body of work suggests that all four apparently different transitions occur at the same point. This remarkable (and perhaps surprising) coincidence has since been made mathematically rigorous in the worst-case setting, where the input is an arbitrary graph $G$ of maximum degree $d$, and the relevant phase transition point is known as the \emph{uniqueness threshold} $\betaunique(d) \defeq 2\arctanh\wrapp{\frac{1}{d-1}} = \Theta\wrapp{\frac{1}{d-1}}$; see \cref{subsec:prior-work} for a more precise discussion. A longstanding conjecture in the field posits that an analogous coincidence holds in the average-case setting, with $G$ drawn uniformly at random from the set of all $d$-regular graphs. In this context, the relevant phase transition point is known as the \emph{Kesten--Stigum threshold}
\begin{align}\label{eq:betaKS}
    \betaKS(d) \defeq 2\arctanh\wrapp{\frac{1}{\sqrt{d-1}}} = \log\wrapp{\frac{\sqrt{d-1} + 1}{\sqrt{d-1} - 1}} = \wrapp{1 + o_{d}(1)} \cdot \frac{2}{\sqrt{d-2}}.
\end{align} 
This threshold originated from a statistical inference problem derived from the famous \emph{broadcast process} on the $(d-1)$-ary tree; see \cite[Section~1]{EKPS00} for a formal description. Hence, the Kesten--Stigum threshold is also sometimes referred to as the \emph{reconstruction threshold}, and delineates between the two regimes of $\beta$ in which certain long-range correlations are absent/present in the Gibbs measure on the $(d-1)$-ary tree. Recent work \cite{CLMS22, FL21} has also pinpointed $\betaKS(d)$ as the transition point delineating between the ``replica symmetric'' versus ``replica symmetry breaking'' phases for the Gibbs measure $\mu_{G,\beta}$ on the uniformly random $d$-regular graph, which may again be interpreted as a transition between absence/presence of a certain type of long-range correlation. Given this, the aforementioned conjecture then suggests that with high probability over a uniformly random $d$-regular graph, for all $\beta$ up to $\betaKS(d)$, the approximate counting task can be solved efficiently, local Markov chains mix rapidly, and there are no zeros of the partition function in a constant-radius neighborhood of $\beta$.

\subsection{Main Results}
In this paper, we make progress on the high-temperature (small $\beta$) side of the aforementioned conjecture by giving the first polynomial-time approximate counting algorithm for all $\beta$ up to $\betaKS(d)$. Moreover, we achieve this by showing that zeros of the partition function do not cluster near any $0 \leq \beta < \betaKS(d)$.

\begin{theorem}[Algorithms]\label{thm:alg}
For any $d \in \N_{\geq 3}$ and any $0 < \delta< 1$ and $0 \leq \beta \leq (1 - \delta) \cdot \betaKS(d)$, there exists $C = C(\delta,d) > 0$
such that the following holds: There is a deterministic algorithm running in time $(n/\varepsilon)^{C}$ which, given a $d$-regular $n$-vertex graph $G$ and $\varepsilon>0$, outputs a real number $\widehat{Z}_{G}(\beta)$ satisfying
\begin{align*}
    \Pr_{G}\wrapb{e^{-\varepsilon} \cdot Z_{G}(\beta) \leq \widehat{Z}_{G}(\beta) \leq e^{\varepsilon} \cdot Z_{G}(\beta)} \geq 1 -o_n(1).
\end{align*}
where the randomness is only over a uniformly random $d$-regular graph $G$.
\end{theorem}

Our algorithm is based on Barvinok's interpolation method \cite{Bar15, Bar16, Bar16book, Bar17perm}. 
More specifically, for a simple and explicit analytic function $\ell : \D \to \C$, we approximate $\log Z(\ell(z))$ up to additive $\pm\varepsilon$ error by its Taylor polynomial $\widehat{P}_{G}(z)$ of degree $O(\log(n/\varepsilon))$. 
To do this, we leverage \cite{PR17} and compute the $k$th coefficient of $\widehat{P}_{G}$ in $d^{O(k)}$-time for each $0 \leq k \leq O(\log(n/\varepsilon))$.
The key to guaranteeing correctness of our algorithm is the absence of zeros in a constant-width strip around the interval $[0,(1-\delta) \cdot \betaKS(d)]$ for any $\delta > 0$. 
This is our main technical result. 

\begin{theorem}[Zeros]\label{thm:zeros}
For any $d \in \N_{\geq 3}$, any $0 < \delta< 1$ and $0 \leq \beta \leq (1 - \delta) \cdot \betaKS(d)$, there exists $\eta = \eta(d,\delta) > 0$
such that 
with probability at least $1 -o_n(1)$ over the uniformly random $d$-regular graph $G$, the partition function $Z_{G}(\beta)$ has no zeros in the strip
\begin{align*}
    \wrapc{\beta \in \C : \dist\wrapp{\beta, [0,(1-\delta) \cdot \betaKS(d)]} < \eta}.
\end{align*}
\end{theorem}
In~\cref{subsec:alg-proof} we give the details on how this theorem implies the algorithmic statement~\cref{thm:alg}.
To achieve these results, we adapt several of the core techniques developed previously by the authors in the context of mean-field spin glasses \cite{BHLLR26}.
A few remarks are in order.

\begin{remark}[Sampling]\label{rmk:nosampler}
Our approximate counting algorithm does not immediately imply an approximate sampling algorithm for the Gibbs measure $\mu_{G,\beta}$. This is because we only work with the antiferromagnetic Ising model \emph{without external field}, which is not self-reducible in the sense of \cite{JVV86}. We note that the existence of a sampling algorithm for $\mu_{G,\beta}$ with arbitrarily small total variation error \emph{does} yield an approximate counting algorithm in a blackbox fashion. Indeed, by drawing sufficiently many samples from $\mu_{G,\beta}$, one can well-approximate $(\log Z_{G})'(\beta) = \E_{\sigma \sim \mu_{G,\beta}}\wrapb{\mathcal{H}_{G}(\sigma)}$, and then integrate this over $\beta$.
\end{remark}
\begin{remark}[Certifiability]\label{rmk:certifiable}
A natural desideratum of algorithms operating on random input instances is \emph{certifiability}. In our context, this means there is an \emph{efficiently testable} property $\mathscr{P}$ of the input such that
\begin{itemize}
    \item $\mathscr{P}$ holds with high probability for a uniformly random $d$-regular graph, and
    \item the output $\widehat{Z}_{G}(\beta)$ of the algorithm \emph{must} satisfy the required approximation guarantee on the event $\mathscr{P}$.
\end{itemize}
Unfortunately, we do not know how to show that our algorithm is certifiable in this sense. 
In particular, we do not know how to certify the requisite zero-freeness of $Z_{G}(\beta)$ given $G$. We leave designing a certifiable approximate counting algorithm as an interesting open problem.
\end{remark}

\begin{remark}
Unlike \cite{BHLLR26}, where the zero-free region was a disk in the complex plane centered at $0$, our zero-free region is a thin rectangle around the real interval $[0,(1-\epsilon) \cdot \betaKS(d)]$.
This suffices just as well for the algorithm in \cref{thm:alg}.
We expect that $Z_{G}(\beta)$ is actually not zero-free on a disk of radius $(1-\epsilon) \cdot \betaKS(d)$ centered at $0$, as the \emph{ferromagnetic} regime, corresponding to negative $\beta$ in our formulation, has a genuinely different behavior. For instance, even when the input graph is random $d$-regular, the Gibbs measure becomes bimodal with clusters corresponding to configurations dominated by a single spin as soon as $-\beta$ exceeds the \emph{uniqueness threshold} $\betaunique(d) = 2\arctanh\wrapp{\frac{1}{d-1}} < \betaKS(d)$.
\end{remark}

\subsection{Acknowledgements}\label{subsec:ack}
This paper is a continuation of a line of work initiated during the \emph{Workshop on the Combinatorial, Algorithmic and Probabilistic aspects of Partition Functions}, organized by F.B. and G.R., at the Centrum Wiskunde \& Informatica (CWI) Amsterdam in 2025. The authors are grateful to CWI for its hospitality. K.L. thanks Ronen Eldan for providing access to GPT Pro. B.H. thanks Mark Sellke for access to GPT Pro.

\paragraph{AI Disclosure} The overarching vision and proof architecture were developed solely by the authors without AI input. Our starting point was a strategy conceived and executed by the authors for mean-field spin glasses in previous work which did not use AI whatsoever \cite{BHLLR26}. For the present paper, GPT 5.5 Pro was used, with extensive guidance from the authors, for only a few minor technical aspects of the proof, namely routine but cumbersome moment calculations via Laplace's method, standard error/concentration estimates, etc. The authors independently verified the correctness of all arguments in the paper, and take full responsibility for its contents. The authors initially attempted to have GPT 5.6 Codex help merge several handwritten/hand-typed notes containing various pieces of the proof into a single cohesive document.  The result was very far from satisfactory, so the entirety of the paper was written and checked manually.

\paragraph{Independent Work} Recently, Chen--Wang--Zhang independently proved \cref{thm:alg,thm:zeros} and extended them to the hardcore model and the $q$-state antiferromagnetic Potts models on random $d$-regular graphs. Their proof also uses the second moment method and small subgraph conditioning, applied to a suitable truncation of the partition function. However, they are able to avoid the use of Jensen's Formula and analytic reweighting by using the method of \emph{tightness}, a strategy first implemented (to the best of our knowledge) by \cite{BCG22} in the context of random matrix theory. We thank Sidhanth Mohanty for pointing out this reference.

\subsection{Prior Work}\label{subsec:prior-work}

\paragraph{Worst-Case Approximate Counting Complexity} A long line of work has established that there is a critical threshold
\begin{align*}
    \betaunique(d) = 2\arctanh\wrapp{\frac{1}{d-1}} = \log\wrapp{\frac{d}{d-2}},
\end{align*}
known as the \emph{uniqueness threshold}, such that the following hold:
\begin{itemize}
    \item If $\beta < \betaunique(d)$, the approximate counting problem can be solved in polynomial-time for any graph of maximum degree $d$ \cite{SST14}. If $\beta > \betaunique(d)$, then the approximate counting problem for arbitrary graphs of maximum degree $d$ is $\mathsf{NP}$-hard \cite{SS14}.
    \item If $\beta < \betaunique(d)$, then \emph{Glauber dynamics} mixes in $O(n\log n)$-steps for every $n$-vertex graph $G$ of maximum degree $d$ \cite{CLV20, CLV21a, CFYZ21a, AJKPV22}. If $\beta > \betaunique(d)$, then there is an infinite sequence of $n$-vertex graphs of maximum degree $d$ (in fact, random $d$-regular bipartite graphs) such that Glauber dynamics requires $\exp(\Omega(n))$ steps to mix (see e.g. \cite{SS14}). Moreover, very recently, polynomial mixing of the Glauber dynamics was established \emph{at} the uniqueness threshold~\cite{CCYZ25}.
    \item If $\beta < \betaunique(d)$, then $\mu_{G,\beta}$ exhibits a strong form of correlation decay known as \emph{strong spatial mixing} for any graph $G$ of maximum degree $d$. This was established in \cite{SST14}, building on seminal work of \cite{Wei06}. On the other hand, when $\beta > \betaunique(d)$, then it is a folklore result that on the infinite $(d-1)$-ary tree rooted at a vertex $r$, there is a choice of boundary conditions $\tau,\tau'$ (namely all $+$ and all $-$) on the set $L_{r}(\ell)$ of vertices at depth-$\ell$ away from $r$ such that $\abs{\Pr[\sigma_{r} = +1 \mid \sigma_{L_{r}(\ell)} = \tau] - \Pr[\sigma_{r} = +1 \mid \sigma_{L_{r}(\ell)} = \tau']}$ remains bounded below by a constant independent of the distance $\ell$.
    \item For every $\epsilon > 0$, there exists a $\xi = \xi(d,\epsilon) > 0$ such that $Z_{G}(\beta)$ has no zeros in the strip
    \begin{align*}
        \wrapc{\beta \in \C : \dist\wrapp{\beta, [0,(1-\epsilon) \cdot \betaunique(d)]} < \xi}
    \end{align*}
    for all graphs of maximum degree $d$ \cite{LSS19b}. (See~\cite{BarBarIsing,GGHPcomplexityIsing,PRSdiskIsing} for related results on the zeros of $Z_G(\beta)$ on bounded degree graphs $G$.)
    Moreover, \cite{Sly10} implies that for every $\beta \geq \betaunique(d)$, there exists a graph $G$ of maximum degree $d$ such that $Z_{G}(\beta)$ has a zero within distance $\epsilon$ of $\beta$ \footnote{To be precise, if there would be an open $U$ set containing $[0,\betaunique+\delta]$ for some $\delta>0$ that is zero-free for all $d$-regular graphs, then along any locally convergent graph sequence $G_n$ the free energy $\tfrac{1}{v(G_n)}\log Z_{G_n}(\beta)$ would also converge \cite{CF16}. In particular, if $G_n$ is a random regular or bipartite random regular graph sequence, then with probability $1$ they should converge to the same value, but by \cite{Sly10} there exists a $\Phi_{nb}<K_-<K_+<\Phi_{b}$, such that with high probability the free energy at $\betaunique+\delta$ for a bipartite random regular graph is at least $K_+$ and for a random regular is at most $K_-$.}. 
\end{itemize}
Together, these works solidify the uniqueness threshold as the key phase transition point in the worst-case setting. The purpose of this paper is to make progress on the average-case picture.

\paragraph{Free Energy Formulas} It is well-known that by differentiating the logarithmic partition function $\log Z_{G}(\beta)$, one can retrieve key observables like the expected energy of the Gibbs measure. When $G$ is random, standard techniques can be used to show that $\log Z_{G}(\beta)$ concentrates well around its expectation. Hence, there is enormous interest in obtaining accurate approximations to the \emph{quenched free energy}
\begin{align*}
    \beta \mapsto \lim_{n \to \infty} \frac{1}{n} \E_{G}\wrapb{\log Z_{G}(\beta)}.
\end{align*}
Here, the normalization by $1/n$ is natural, since $\log Z_{G}(\beta)$ has magnitude $O(n)$. Existence of this limit was established for all $d \geq 3$ and $\beta > 0$ in \cite[Theorem~4]{BGT13}. Jensen's Inequality tells us that this is always upper bounded by the \emph{annealed free energy}
\begin{align*}
    \lim_{n \to \infty} \frac{1}{n} \log \E_{G}\wrapb{Z_{G}(\beta)} = \log 2 + \frac{d}{2}\log \frac{1 + e^{-\beta}}{2}.
\end{align*}
In \cite{CLMS22}, it was proved that when $\beta < \betaKS(d)$, $Z_{G}(\beta)$ itself concentrates well around its mean, $\E\wrapb{\frac{\langle \sigma,\sigma' \rangle}{n}} \leq o(1)$ for two independent samples $\sigma,\sigma' \sim \mu_{G,\beta}$ with high probability over $G$, and ``quenched equals annealed''. On the other hand, when $\beta > \betaKS(d)$, the annealed free energy is a strict upper bound on the quenched free energy. Thus, a phase transition occurs at $\betaKS(d)$. This transition is known as the \emph{replica symmetry breaking threshold} in the mathematical physics literature. Soon after, \cite{FL21} characterized the fluctuations of $Z_{G}(\beta)$ for $\beta < \betaKS(d)$. These types of approximations yield estimates to $\log Z_{G}(\beta)$ in expectation up to $\pm o(n)$ additive error. In contrast, our goal is to obtain estimates to $\log Z_{G}(\beta)$ with high probability up to $\pm \eta$ additive error for any arbitrarily small $\eta > 0$ given as input.


\paragraph{Markov Chains} It is now well-known that for the Ising model $\mu(\sigma) \propto \exp\wrapp{\frac{1}{2}\sigma^{\top}J\sigma}$ with symmetric interaction matrix $J \in \R^{n \times n}$, if $\lambda_{\max}(J) - \lambda_{\min}(J) < 1$, then discrete-time heat-bath Glauber dynamics admits a $\Omega(1/n)$ spectral gap and hence, mixes in $O(n^{2})$-steps \cite{BB19, EKZ22}. The mixing time was subsequently improved to $O(n\log n)$ \cite{AJKPV22}. A direct application of this \emph{spectral condition} to the antiferromagnetic Ising model on the random $d$-regular graph $G$, in which $J = -\frac{\beta}{2} \cdot A_{G}$ and $\beta \cdot \mathcal{H}_{G}(\sigma) = \frac{1}{2}\sigma^{\top}J\sigma$ where $A_{G}$ is the adjacency matrix of $G$, only yields rapid mixing when $\beta \leq O(1/d)$ due to the presence of the trivial eigenvalue $d$. However, by suitably modifying the techniques in \cite{BB19, EKZ22} (which are based on \emph{stochastic localization} \cite{Eld13, CE22}), and leveraging the fact that the nontrivial spectrum of $A_{G}$ is confined to the interval $\wrapb{-2\sqrt{d-1} - o(1), 2\sqrt{d-1} + o(1)}$ with high probability, \cite{KLR22} gave a hybrid sampling algorithm that runs in polynomial-time with high probability over $G$ as long as $\beta < \frac{1}{2\sqrt{d-1}}$. Subsequently, it was proved that Glauber dynamics mixes in $O(n\log n)$-time with high probability as long as $\beta < \frac{1}{4\sqrt{d-1}}$ \cite{AKV24}. We also note that in the same range of inverse temperatures, \cite{BBD25} proved $O(n \cdot \mathsf{polylog}(n))$-mixing of \emph{Kawasaki dynamics} in both the ferromagnetic and antiferromagnetic settings when the Gibbs measure is restricted to a fixed-magnetization slice. Via standard reductions \cite{JVV86}, these sampling algorithms can all be converted into polynomial-time approximate counting algorithms. They have the advantage of certifiability (see \cref{rmk:certifiable}), but are unable to reach the true Kesten--Stigum threshold that we achieve.

\paragraph{Approximate Counting via Zero-Freeness} As we previously mentioned, our algorithm is based on Barvinok's interpolation method, which at a high-level operates by truncating the Taylor series of the logarithmic partition function. This approach to designing approximate counting algorithms was originally devised for the permanent of complex matrices \cite{Bar16, Bar17perm}, but has since found myriad applications to models originating in statistical and quantum many-body physics; we refer interested readers to \cite{Bar16book, PR22} for more in-depth surveys of the literature.

Our work draws inspiration from recent results establishing zero-freeness via a combination of Jensen's Formula and the second moment method. To the best of our knowledge, this was first done for permanents of random complex matrices in \cite{EM18} (see also \cite{KL26}). The authors then used these techniques to design the first quasipolynomial-time algorithm for estimating the partition function of the Sherrington--Kirkpatrick model up to the replica symmetry breaking threshold \cite{BHLLR26}.

\section{Technical Overview}
As previously mentioned, \cref{thm:alg} is an immediate consequence of \cref{thm:zeros} via an application of Barvinok's interpolation method \cite{Bar16book,PR17}; the details of the algorithm are provided in \cref{subsec:alg-proof} below. In the remainder of this section, we provide a high-level overview of the key components of the proof of \cref{thm:zeros}.

\subsection{Zeros and Jensen's Formula}
For convenience, denote the strip in which we will establish zero-freeness by
\begin{align*}
    S(\epsilon,\xi) &\defeq \wrapc{\beta \in \C : \dist\wrapp{\beta, [0, (1 - \epsilon) \cdot \betaKS(d)]} < \xi}.
\end{align*} 
Our goal will be to prove that for sufficiently small constant $\epsilon,\xi$,
\begin{align}\label{eq:expected-zeros-goal}
    \E_{G}\wrapb{\#\{\omega \in S(\epsilon,\xi) : Z_{G}(\omega) = 0\}} = o_{n}(1),
\end{align}
where zeros are counted with multiplicity.
Markov's Inequality then yields \cref{thm:zeros}. Following \cite{BHLLR26}, the first step towards proving \cref{eq:expected-zeros-goal} is to relate the left-hand side to the moments of the partition function $Z_{G}(\beta)$ evaluated at complex values of $\beta$. This is achieved by \emph{Jensen's Formula}.
\begin{theorem}[Jensen's Formula]\label{thm:jensen-zeros}
Let $\Omega\subset\C$ be an open set that contains the disk $\overline{\D(0,R)}$ for some $R>0$. Let $f : \Omega \to \C$ be an analytic function such that $f(0)\neq 0$. Then
\begin{align*}
    \sum_{\omega \in \D(0,R) : f(\omega) = 0} \log\wrapp{\frac{R}{\abs{\omega}}} = \E_{\theta}\wrapb{\log \abs{\frac{f(Re^{i\theta})}{f(0)}}},
\end{align*}
where the expectation is taken with respect to $\theta \sim \mathsf{Unif}[-\pi,\pi]$.
\end{theorem}

\begin{corollary}\label{cor:expectedzero-ub}
Let $\Omega$ be as in \cref{thm:jensen-zeros}, and $f_{G} : \Omega \to \C$ be a random $G$-measurable function, which is analytic with $f_{G}(0) \neq 0$ deterministically. Then for any $0<r<R$ we have
\begin{align*}
    \E_{G}\wrapb{\#\{\omega \in \D(0,r) : f_{G}(\omega) = 0\}}
    \le \frac{1}{2\log(R/r)} \E_{\theta} \log \E_{G} \wrapb{
        \abs{\frac{f_{G}(Re^{i\theta})}{f_{G}(0)}}^{2}
    }
\end{align*}
\end{corollary}

A natural choice for the function $f_{G}$ is to let $f_{G}(z) = Z_{G}(\ell_{\epsilon,\xi}(z))$, where $\ell_{\epsilon,\xi} : \D \to S(\epsilon,\xi)$ is a suitable analytic map satisfying $\ell_{\epsilon,\xi}(0) = 0$. Such a $\ell_{\epsilon,\xi}$ exists by the Riemann Mapping Theorem, and an explicit construction was given in \cite{Bar16book}. However, what appears in the right-hand side then becomes the \emph{unnormalized} second moment of $2^{-n} Z_{G}(\beta)$, which unfortunately remains exponentially large.

A better choice is to normalize by the first moment and set
\begin{align*}
    f_{G}(z) = \frac{Z_{G}(\ell_{\epsilon,\xi}(z))}{\E_{G}\wrapb{Z_{G}(\ell_{\epsilon,\xi}(z))}}.
\end{align*}
Assuming one can additionally establish zero-freeness of the first moment $\E_{G}[Z_{G}(\beta)]$ in $S(\epsilon,\xi)$, $f_{G}$ is analytic on $\D$ as required. The right-hand side of \cref{cor:expectedzero-ub} can then be upper bounded by
\begin{align*}
    \frac{1}{2\epsilon} \sup_{\beta \in S(\epsilon,\xi)} \frac{\E_{G}\wrapb{\abs{Z_{G}(\beta)}^{2}} - \abs{\E_{G}[Z_{G}(\beta)]}^{2}}{\abs{\E_{G}[Z_{G}(\beta)]}^{2}}
\end{align*}
where we invoked the numerical inequality $\log(1+x) \leq x$. The numerator can be interpreted as the unnormalized variance of the partition function. In the replica-symmetric phase, the partition function itself concentrates well and one can hope that the variance is small relative to the squared expectation. However, typically, this normalized second moment is at best $O(1)$, not $1 + o_{n}(1)$. Moreover, for the antiferromagnetic Ising model on the random $d$-regular graph, boundedness of this normalized second moment only holds up to a strictly smaller threshold
\begin{align*}
    \betasecond(d) \defeq \log\wrapp{\frac{\sqrt{d} + 1}{\sqrt{d} - 1}} = \betaKS(d+1) < \betaKS(d),
\end{align*}
a phenomenon which was already observed in \cite{CLMS22}. These two issues are respectively remedied via two devices: reweighting, in the spirit of the \emph{small subgraph conditioning method}, and truncation. Both devices were crucial to obtain zero-freeness for mean-field spin glass models up to the replica symmetry breaking threshold \cite{BHLLR26}. We now describe each at a high-level.

\paragraph{Reweighting} Up to $\betasecond(d)$, the partition function $Z_{G}(\beta)$ genuinely has $O(1)$-multiplicative fluctuations. However, we can construct a random analytic function $A_{G}(\beta)$ such that the reweighted partition function
\begin{align*}
    X_{G}(\beta) \defeq \frac{Z_{G}(\beta) \cdot A_{G}(\beta)}{\E_{G}\wrapb{Z_{G}(\beta)}}
\end{align*}
has normalized variance $1 + o_{n}(1)$. Since multiplication by an analytic function can only introduce more zeros, zero-freeness for $X_{G}(\beta)$ in $S(\epsilon,\xi)$ implies zero-freeness for $Z_{G}(\beta)$ in $S(\epsilon,\xi)$. The idea is to then execute the strategy outlined above for the random function $X_{G}(\beta)$ in lieu of $Z_{G}(\beta)$.

How do we construct $A_{G}(\beta)$? In the case of real $\beta$, \cite{FL21} (building on prior work of \cite{CLMS22}) used the method of small subgraph conditioning to construct an explicit random variable $W(\beta)$ such that $\frac{Z_{G}(\beta)}{\E_{G}[Z_{G}(\beta)]} \overset{\mathsf{d}}{\to} W(\beta)$ as $n \to \infty$. In \cite{FL21}, this $W(\beta)$ is an infinite weighted sum of independent Poisson random variables. However, at its core, it is capturing a weighted sum of counts of cycles of each fixed length in the uniformly random $d$-regular graph $G$, and one can indeed write down an analogous random variable $W_{G}(\beta)$ based on this interpretation. Roughly speaking, we then distill $A_{G}(\beta)$ from $W_{G}(\beta)$ by considering $1/W_{G}(\beta)$ and applying suitable truncations (e.g. to the cycle lengths) to make it more analytically tractable. Our $A_{G}(\beta)$ bears a strong resemblance to the cluster expansion reweighting factor used in \cite{BHLLR26}.

\paragraph{Truncation} For $\betasecond(d) < \beta < \betaKS(d)$, the normalized second moment of $Z_{G}(\beta)$ is exponentially large in $n$. The source of this failure stems from the fact that in this regime,
\begin{align*}
    \Pr_{G}\wrapb{Z_{G}(\beta) > e^{cn} \cdot \E_{G}\wrapb{Z_{G}(\beta)}} \approx e^{-bn}
\end{align*}
for constants $b,c > 0$ satisfying $2c > b$. However, this is not a true barrier to concentration of $Z_{G}(\beta)$ around its expectation since the probability of such a large multiplicative deviation is still exponentially small in $n$. We remedy the situation by identifying a ``good'' high probability event $\mathcal{E}$ on the space of $d$-regular graphs and decomposing $Z_{G}(\beta)$ as $Z_{G}(\beta) \cdot \allone[G \in \mathcal{E}] + Z_{G}(\beta) \cdot \allone[G \notin \mathcal{E}]$ inside the right-hand side of Jensen's Formula. We handle the first term using the second moment method, and bound the second term just by using the fact that $\Pr[G \notin \mathcal{E}]$ is exponentially small in $n$. Following \cite{CLMS22, FL21}, we take $\mathcal{E}$ to be the set of all $G$ for which two independent samples $\sigma,\sigma'$ drawn from the Gibbs measure $\mu_{G}$ have overlap $R(\sigma,\sigma') \defeq \frac{\langle \sigma,\sigma' \rangle}{n}$ concentrated around $0$. For technical reasons, we must significantly sharpen the quantitative estimates given in \cite{CLMS22}.

\subsection{Preliminaries}
Let $\Sigma_n \defeq \{\pm1\}^{n}$. For $\tau,\sigma \in \Sigma_n$, their overlap is $R(\tau,\sigma) \defeq \frac{\langle \tau,\sigma \rangle}{n}$. For small $\varepsilon > 0$, let
\begin{align*}
    \Sigma_n(\varepsilon) \defeq \wrapc{\sigma \in \Sigma_n : \abs{R(\sigma,\allone)} < \varepsilon}.
\end{align*}

\paragraph{The Configuration Model} Throughout this paper, we work with the \emph{configuration model} of random $d$-regular graphs. Fix $d,n \in \N$ such that $dn$ is even, and write $N \defeq dn$; we view $d$ as a fixed constant (e.g. $d=3$), and $n$ as growing to $\infty$. We form $n$ ``clouds'', each uniquely labelled by a vertex $v$ in the eventual $n$-vertex (multi)graph $G$ that will be generated, and consisting of $d$ ``sockets'' (or ``half-edges'') labelled $(v,1),\dots,(v,d)$. We sample a uniformly random perfect matching $M$ on the combined set of $dn$ sockets, and contract each of the clouds to form $G$. This yields a distribution over $n$-vertex $d$-regular (multi)graphs.

Conditioned on the event that $G$ is simple, $G$ is uniformly distributed over $n$-vertex $d$-regular simple graphs. Moreover, this event occurs with probability $\Omega_{d}(1)$ and hence, events that occur with probability $1 - o_{n}(1)$ in the configuration model (e.g. zero-freeness and high-accuracy approximation of the partition function) also occur with probability $1 - o_{n}(1)$ under a uniformly random $d$-regular simple graph (i.e. we have \emph{contiguity}). Because of this and the simplification in calculations it affords, we often work directly with the random perfect matching $M$ throughout this paper, rather than directly with a truly uniformly random $d$-regular simple graph. For standard facts on the configuration model and random $d$-regular graphs, we refer the reader to e.g. \cite{Bol01}.

\paragraph{Cycles} For convenience, we call a given collection of pairwise disjoint pairs of sockets $\{(u_{t},i_{t}), (v_{t},j_{t})\}_{t=1}^{k}$ a $k$-cycle if they form a length-$k$ cycle in the resulting graph after contracting the clouds, i.e. $v_{t} = u_{t+1}$ for all $t$, and every cloud is visited exactly twice. We write $\mathscr{C}_{k}$ for the collection of all possible $k$-cycles. Note that
\begin{align}\label{eq:num-kcycle}
    \abs{\mathscr{C}_{k}} = \frac{(n)_{k} d^{k}(d-1)^{k}}{2k} = (1 \pm o_{n}(1)) \cdot \wrapp{dn}^{k}\lambda_{k},
\end{align}
where
\begin{align*}
    \lambda_{k} \defeq \frac{(d-1)^{k}}{2k},
\end{align*}
and for $x \in \R_{\geq0}$ and $k \in \N$,
\begin{align*}
    (x)_{k} \defeq x \cdot (x-1) \dotsb (x-k+1)
\end{align*}
is the \emph{falling factorial}. For instance, in this notation, we may write $\binom{n}{k} = \frac{(n)_{k}}{k!}$. The following lemma is well-known; a proof is provided in \cref{sec:technical} to keep the paper self-contained.
\begin{lemma}[Short Cycle Approximate Independence]\label{lem:shortcyc-approxindep}
Let $\Gamma$ be a collection of cycles, and let $m(\Gamma)$ denote the number of distinct edges appearing in $\Gamma$. Then
\begin{align*}
    \Pr[\Gamma \subseteq M] &= \prod_{j=0}^{m(\Gamma) - 1} \frac{1}{dn - (2j+1)}.
\end{align*}
Furthermore, if no edge appears more than once in $\Gamma$, we have
\begin{align*}
    \Pr[\Gamma \subseteq M] = \exp\wrapp{\pm O\wrapp{\frac{m^{2}}{n}}} \cdot \prod_{C \in \Gamma} \Pr[C \subseteq M] = \exp\wrapp{\pm O\wrapp{\frac{m^{2}}{n}}} \cdot \wrapp{\frac{1}{dn}}^{\abs{\Gamma}}.
\end{align*}
\end{lemma}

\paragraph{Partition Functions}
Given that we will often work directly with the perfect matching $M$ on $dn$ sockets, for $\sigma \in \{\pm1\}^{n}$, we abuse notation and also write $\sigma$ for its extension to a configuration on sockets in the obvious way: $\sigma(v,j) = \sigma(v)$ for all $v \in V, j \in [d]$. For $F \subseteq \binom{V \times [d]}{2}$ (e.g. $F = M$), define
\begin{align*}
    H_{F}(\sigma) &\defeq \sum_{\{(u,i),(v,j)\} \in F} \frac{1 - \sigma_{u}\sigma_{v}}{2} \\
    Z_{F}(\beta) &\defeq \sum_{\sigma \in \{\pm1\}^{V}} \exp\wrapp{\beta H_{F}(\sigma)}.
\end{align*}
It will also be convenient to define
\begin{align*}
    \widetilde{Z}_{F}(\beta) \defeq 2^{-\abs{V(F)}}\wrapp{\frac{1 + e^{\beta}}{2}}^{-|F|}Z_{F}(\beta),
\end{align*}
where $V(F)$ denotes the number of clouds $F$ intersects. Finally, for $\beta \in \C$, define
\begin{align*}
     \delta_{k} = \delta_{k}(\beta) \defeq \wrapp{\frac{e^{-\beta} - 1}{e^{-\beta} + 1}}^{k} \in \C.
\end{align*}
The following lemma computes a closed-form formula for the partition function of a cycle. To keep the paper self-contained, a proof is provided in \cref{sec:technical}.
\begin{lemma}\label{lem:pfn-disj-cycles}
For a $k$-cycle, its partition function is given by
\begin{align*}
    Z_{k}(\beta) \defeq 2^{k}\wrapp{\frac{1 + e^{\beta}}{2}}^{k} \cdot (1 + \delta_{k}(\beta)).
\end{align*}
In particular, if $F$ is a union of $m$ vertex-disjoint cycles of lengths $k_{1},\dots,k_{m}$, then we have
\begin{align*}
    \widetilde{Z}_{F}(\beta) = \prod_{i=1}^{m} (1 + \delta_{k_{i}}(\beta)).
\end{align*}
\end{lemma}


\paragraph{Miscellaneous}
\begin{lemma}[Gaussian Correction Series]\label{lem:gaussian-correction-series}
For every $0 < \delta < 1$ and every $\beta$ in a sufficiently small constant (depending on $d,\delta$) neighborhood of $[0,(1 - \delta)\beta_{\mathsf{KS}}]$, we have
\begin{align*}
    \sqrt{\frac{1}{1 - (d - 1) \cdot \abs{\delta_{2}}}} = \exp\wrapp{\sum_{k=1}^{\infty} \frac{\wrapp{(d-1) \cdot \abs{\delta_{2}}}^{k}}{2k}}.
\end{align*}
\end{lemma}
\begin{proof}
This is immediate from the definition of $\beta_{\mathsf{KS}}$ and the Taylor series $\log(1-x) = -\sum_{k=1}^{\infty} \frac{x^{k}}{k}$.
\end{proof}

\subsection{Proof of \cref{thm:alg}}\label{subsec:alg-proof}
Let 
us fix $\beta \in [0,(1-\delta)\betaKS]$ and $\eta=\eta(d,\delta)>0$ as guaranteed by \cref{thm:zeros}. 
Denote the $\eta$-neighborhood of $[0,(1-\delta)\betaKS]$ by $U$. 

Below we show how to run the algorithm from~\cite{PR17} on input of $G$ and $\varepsilon>0$.
By \cref{thm:zeros}, we know that with probability at least $1-o_n(1)$, $Z_G$ is non-vanishing on $U$. 
Thus, it remains to show that in in this case the output of the algorithm gives us the desired approximation of $Z_G(\beta)$.

Thus, for the rest of the proof, we assume that $Z_G$ is non-vanishing on $U$. 
Let us introduce the variable $z=e^\beta-1$. Then

\[
    q_G(z)=\frac{1}{2^{|V|}}\sum_{\sigma\in\{\pm1\}^V}(1+z)^{\mathcal{H}_G(\sigma)}
\]
is a graph polynomial, with $q_G(e^\beta-1)=\tfrac{1}{2^{|V|}}Z_G(\beta)$. Note that the coefficient of $z^k$ can be computed as a linear combination of counts of induced subgraphs on at most $2k$ vertices; that is, $q_G$ is a bounded induced subgraph counting polynomial in the sense of \cite{PR17}.

Now let us fix a disk $D$ of radius larger than $1$ around $0$ and a polynomial $p(z)$ such that $p(0)=0$, $p(1)=\exp(\beta)-1$, and $p$ maps $D$ into $U'=\exp(U)-1$ (see \cite[Lemma~2.2.3]{Bar16book}). Since $q_G$ is non-vanishing on $U'$, it follows that $q_G(p(z))$ does not vanish on $D$. Hence $\log(q_G(p(z)))$ is analytic on $D$ and has a convergent Taylor series around $0$. To approximate
\[
    q_G(p(1))=\tfrac{1}{2^{|V|}}Z_G(\beta)
\]
up to a multiplicative factor $e^\varepsilon$, it suffices to truncate the Taylor series of $\log(q_G(p(z)))$ after $O(\log(n/\varepsilon))$ terms (see \cite[Lemma~2.2.1]{Bar16book}).

The first $m$ Taylor coefficients of $\log(q_G(p(z)))$ can be computed from the first $m$ coefficients of $q_G(p(z))$. Since $p(0)=0$, these in turn depend only on the first $m$ coefficients of $q_G$. The coefficients of $\log(q_G(p(z)))$ can then be recovered using Newton identities (see \cite[Section~2]{PR17}). By \cite[Theorem~3.2]{PR17}, the first $O(\log(n/\varepsilon))$ coefficients of $q_G$ can be computed in time polynomial in $n$ and $1/\varepsilon$, since $q_G$ is a bounded induced subgraph counting polynomial. 
This gives the desired approximation in time polynomial in $n$ and $1/\varepsilon$.

\section{Reduction from Zero Counting to Moment Bounds}\label{sec:reduction-2ndmoment}

We will set out to prove our main technical result, \cref{thm:zeros}. 

For $\delta, \eta_0 \ge 0$, define
\begin{equation}\label{e:rect}
  \rectangularStrip(\delta,\eta_0)
  = \left\{
    \beta = E + i \eta \in \C :
    E \in [\beta_\lowerBoundTag, \betaKS - \delta], \eta \in [-\eta_0,\eta_0]
  \right\}\,,
\end{equation}
here $\beta_\lowerBoundTag\defeq=-d^{-10}$ is fixed throughout.

\begin{theorem}\label{t:zero-free}
  For every $d\in \mathbb{N}_{\geq 3}$ and  $\delta > 0$, there exists $\eta_0 = \eta_0(d,\delta) > 0$ such that for $G \sim \configurationModel(n,d)$,
  \[
    \Pr\left\{
      Z_{G}(\beta)\,\,\text{is zero-free on}\,\,\rectangularStrip(\delta,\eta_0)
    \right\} = 1 - o_n(1)\,.
  \]
\end{theorem}
Note that $G \sim \configurationModel(n,d)$ is simple with uniformly positive probability, and conditional on being simple is a uniformly random $d$-regular graph.
Thus the conclusion of ~\cref{t:zero-free} also holds if $G$ is a random $d$-regular graph and therefore ~\cref{t:zero-free} directly implies our main technical result \cref{thm:zeros}.

\paragraph{Asymptotic conventions}

Throughout this paper we treat $d$ as fixed once and for all.
We then set several parameters in the following order.
\begin{itemize}
  \item $\delta > 0$ is a small constant.
  \item $\varepsilon_0 = \varepsilon_0(d,\delta) \in (0,1)$ is set in ~\cref{l:basic-mt} below.
  \item $\varepsilon \in (0, \varepsilon_0)$ is small depending on $d,\delta$. 
  \item $c > 0$ is small depending on $d,\delta,\varepsilon$.
  \item $\eta_0 > 0$ is small depending on $d,\delta,\varepsilon$ and sometimes $c$.
\end{itemize}
We will sometimes use the notation $c = c(d,\delta,\varepsilon)$ to make the dependence between parameters explicit.
We will send $n\to\infty$ with all other parameters held fixed, and all assertions in this note implicitly assume $n$ is sufficiently large. 

\paragraph{Section outline}

We will prove ~\cref{t:zero-free} through two key inputs, outlined in the following subsections.
\begin{itemize}
  \item In \cref{ss:truncation}, we introduce a ``balanced'' version $Z_{G,< \varepsilon}(\beta)$ of $Z_{G}(\beta)$, where we only sum over $\sigma$ with similar numbers of $+$ and $-$ spins.
  For $\beta$ in our regime of interest, the dominant contribution to $\E[Z_{G}(\beta)]$ will come from such $\sigma$; see ~\cref{l:1mt-balanced}~\cref{i:1mt-balanced-deviation} and~\cref{c:1mt-balanced-complex}. 
  We then define an event ${\mathcal{E}}_{\beta,\varepsilon,c}$ on which the non-orthogonal contribution to the second moment $\E [|Z_{G,< \varepsilon}(\beta)|^2]$ is small. The main input from this subsection is ~\cref{t:truncation}, which states that for \textbf{real} $\beta \in [2\beta_\lowerBoundTag, \betaKS - \delta]$, 
  \[
    \E [\allone\{G\not\in {\mathcal{E}}_{\beta,\varepsilon,c}\} Z_{G,< \varepsilon}(\beta)] \le e^{-cn} \E[Z_{G,< \varepsilon}(\beta)]\,, 
  \]
  i.e. the complement of ${\mathcal{E}}_{\beta,\varepsilon,c}$ accounts for an exponentially small fraction of the first moment.
  \item In \cref{ss:moments} we state ~\cref{l:1mt} and ~\cref{t:reweight}, our main estimates on the first and (truncated and reweighted) second moments.
  Note that in contrast to the mixed $p$-spin case, \textbf{even the first moment $\E[Z_{G,< \varepsilon}(\beta)]$ is nontrivial to estimate}, as its evaluation at complex $\beta$ involves a combinatorial sum whose cancellations need to be controlled precisely.
  \item In \cref{ss:zero-free} we complete the proof of ~\cref{t:zero-free} using the above inputs, via the Jensen's formula strategy.
\end{itemize}
The inputs from \cref{ss:truncation}--\cref{ss:moments} are proved in the remainder of this paper.

\subsection{Input 1: second moments dominated by near-orthogonal contribution}
\label{ss:truncation}

For $\sigma^1,\sigma^2 \in \Sigma_n$ define the overlap $R(\sigma^1,\sigma^2) = \langle \sigma^1,\sigma^2\rangle / n$.
Let $\allone = (1,\ldots,1) \in \R^n$, and for $\varepsilon > 0$ let
\begin{align}
  \label{e:Sigma-N-eps}
  \Sigma_n(\varepsilon) &= \left\{
    \sigma \in \Sigma_n : |R(\sigma,\allone)| < \varepsilon
  \right\}\,, & 
  \Sigma_n(\varepsilon)^c &= \left\{
    \sigma \in \Sigma_n : |R(\sigma,\allone)| \ge \varepsilon
  \right\}\,.
\end{align}
Then, let
\begin{align*}
  Z_{G,< \varepsilon}(\beta)
  &= \sum_{\sigma \in \Sigma_n(\varepsilon)} \exp(\beta \isingHamiltonian_{G}(\sigma))\,, & 
  Z_{G,\ge \varepsilon}(\beta)
  &= \sum_{\sigma \in \Sigma_n(\varepsilon)^c} \exp(\beta \isingHamiltonian_{G}(\sigma))
\end{align*}
be the contributions to $Z_{G}(\beta)$ from $\Sigma_n(\varepsilon)$ and $\Sigma_n(\varepsilon)^c$.
The following simple estimate shows that the set $\Sigma_n(\varepsilon)$ of ``balanced'' $\sigma$ dominates $\E [Z_{G}]$.
\begin{lemma}[Proved in Appendix~\cref{a:basic-mt}]
  \label{l:1mt-balanced}
  For any $\varepsilon > 0$, there exists $c = c(d,\varepsilon) > 0$ such that for all (real) $\beta \in [2\beta_\lowerBoundTag, +\infty)$, the following holds.
  \begin{enumerate}[label=(\alph*)]
    \item \label{i:1mt-balanced-deviation}
    We have 
    \[
      \E[Z_{G,\ge \varepsilon}(\beta)] \le e^{-cn} \E[Z_{G,< \varepsilon}(\beta)]\,.
    \]
    \item \label{i:1mt-balanced-subgaussian}
    If we sample $\sigma \in \Sigma_n(\varepsilon)$ with probability $p(\sigma) \propto \E[\exp(\beta \isingHamiltonian_{G}(\sigma))]$, then for all $s \in (0,\varepsilon]$ independent of $n$,
    \[
      \Pr_\sigma(|R(\sigma,\allone)| \ge s) \le 2\exp\left(-cns^2\right)\,.
    \]
  \end{enumerate}
\end{lemma}
Let
\begin{align*}
  Z_{G,< \varepsilon}^{(2)}(\beta)
  &= \sum_{\sigma^1,\sigma^2 \in \Sigma_n(\varepsilon)} 
  \allone\{|R(\sigma^1,\sigma^2)| < \varepsilon\} 
  \exp\left(\beta \isingHamiltonian_{G}(\sigma^1) + \bar \beta \isingHamiltonian_{G}(\sigma^2)\right)\,, \\
  Z_{G,\ge \varepsilon}^{(2)}(\beta)
  &= \sum_{\sigma^1,\sigma^2 \in \Sigma_n(\varepsilon)} 
  \allone\{|R(\sigma^1,\sigma^2)| \ge \varepsilon\} 
  \exp\left(\beta \isingHamiltonian_{G}(\sigma^1) + \bar \beta \isingHamiltonian_{G}(\sigma^2)\right)\,.
\end{align*}
be the contributions to $|Z_{G,< \varepsilon}(\beta)|^2$ from $(\sigma^1,\sigma^2)$ with $|R(\sigma^1,\sigma^2)| < \varepsilon$ and $|R(\sigma^1,\sigma^2)| \ge \varepsilon$.
For $\beta \in \C$, $\varepsilon,c > 0$, define the event
\begin{equation}\label{e:overlap-event}
  {\mathcal{E}}_{\beta,\varepsilon,c} = \left\{
    G : |Z_{G,\ge \varepsilon}^{(2)}(\beta)| 
    \le e^{-cn} \left|\E [Z_{G,< \varepsilon}(\beta)]\right|^2
  \right\}\,.
\end{equation}
Recall that here and throughout, the function $\varepsilon_0(d,\delta)$ is given by ~\cref{l:basic-mt} below.
The main input of this subsection is the following.
\begin{theorem}[Proved in \cref{s:truncation}]\label{t:truncation}
  Let $\delta > 0$ and $\varepsilon_0 = \varepsilon_0(d,\delta)$.
  For all $\varepsilon \in (0, \varepsilon_0)$, there exists $c = c(d,\delta,\varepsilon) > 0$ such that for all (real) $\beta \in [2\beta_\lowerBoundTag, \betaKS - \delta]$,
  \begin{equation}\label{e:truncation}
    \E\left[
      \allone\{G\not\in {\mathcal{E}}_{\beta,\varepsilon,c}\}
      Z_{G,< \varepsilon}(\beta)
    \right]
    \le
    e^{-cn} \E[Z_{G,< \varepsilon}(\beta)]\,,
  \end{equation}
  where ${\mathcal{E}}_{\beta,\varepsilon,c}$ is defined in \eqref{e:overlap-event} and the expectations are over $G \sim \configurationModel(n,d)$.
\end{theorem}
\begin{remark}
  In ~\cref{t:truncation} and Corollary~\cref{c:truncation} below, $\beta$ is real and several definitions above simplify: $Z_{G,< \varepsilon}^{(2)}(\beta)$ and $Z_{G,\ge \varepsilon}^{(2)}(\beta)$ can be equivalently defined without the complex conjugates, and ${\mathcal{E}}_{\beta,\varepsilon,c}$ can be equivalently defined without the absolute values.
  We make these definitions in this way in anticipation of the proof of ~\cref{t:zero-free} in \cref{ss:zero-free} below, where we will allow $\beta$ to be complex.
\end{remark}
\begin{remark}
  \cite[Lemma 4.7]{CLMS22} proves a weaker version of ~\cref{t:truncation} where the $e^{-cn}$ terms in both \eqref{e:overlap-event} and \eqref{e:truncation} are replaced by $o_n(1)$ (and $Z_{G,<\varepsilon}$ replaced by $Z_{G}$).
  In our application it will be important that $G \not\in {\mathcal{E}}_{\beta,\varepsilon,c}$ contributes only a $e^{-cn}$ fraction of the first moment. 
  This is crucial for Corollary~\cref{c:truncation-complex} below, where it implies a suitable extension of ~\cref{t:truncation} (and Corollary~\cref{c:truncation}) for $\beta$ in a constant-width complex strip around $[2\beta_\lowerBoundTag, \betaKS - \delta]$.
\end{remark}
We will also need the following variant of ~\cref{t:truncation} for the near-orthogonal second moment.
\begin{corollary}[Proved in \cref{s:truncation}]\label{c:truncation}
  Let $\delta > 0$ and $\varepsilon_0 = \varepsilon_0(d,\delta)$.
  For all $\varepsilon \in (0, \varepsilon_0)$, there exists $c = c(d,\delta,\varepsilon) > 0$ such that for all (real) $\beta \in [2\beta_\lowerBoundTag, \betaKS - \delta]$,
  \[
    \E\left[
      \allone\{G\not\in {\mathcal{E}}_{\beta,\varepsilon,c}\}
      Z_{G,< \varepsilon}^{(2)}(\beta)
    \right] \le e^{-cn} \E[Z_{G,< \varepsilon}(\beta)]^2\,.
  \]
\end{corollary}

\subsection{Input 2: moment estimates of complex partition function}
\label{ss:moments}

We next introduce the estimates on the first and (truncated and reweighted) second moments we will need, which will hold for complex $\beta$ in the following rectangular strip in the complex plane.
For $\delta, \eta_0 \ge 0$, define
\begin{equation}\label{e:rect-plus}
  \rectangularStrip_+(\delta,\eta_0)
  = \left\{
    \beta = E + i \eta \in \C :
    E \in [2\beta_\lowerBoundTag, \betaKS - \delta], \eta \in [-\eta_0,\eta_0]
  \right\}\,.
\end{equation}
This differs from $\rectangularStrip(\delta,\eta_0)$ defined in \eqref{e:rect} only in that we replaced $\beta_\lowerBoundTag$ with $2\beta_\lowerBoundTag$.
\begin{lemma}[First moment estimate, proved in \cref{s:moments}]\label{l:1mt}
  Let $\delta > 0$ and $\varepsilon_0 = \varepsilon_0(d,\delta)$.
  For $\varepsilon \in (0, \varepsilon_0)$ there exists $\eta_0 = \eta_0(d,\delta,\varepsilon) > 0$ such that the following holds.
  For all $\beta = E + i\eta \in \rectangularStrip_+(\delta,\eta_0)$,
  \[
    \E \left[ Z_{G,< \varepsilon}(\beta) \right]
    = (1 + o_n(1)) \sqrt{\frac{1 + e^\beta}{d(e^\beta-1) + 2}} \cdot 2^n \left(\frac{1+e^\beta}{2}\right)^{dn/2}\,,
  \]
  where we use the branch of the square root that is complex analytic on the complement of the negative real line.
  In particular, for any $c>0$, there exists $\eta_0 = \eta_0(d,\delta,\varepsilon,c) > 0$ such that for all $\beta = E + i\eta \in \rectangularStrip_+(\delta,\eta_0)$,
  \[
    \left|\E \left[ Z_{G,< \varepsilon}(\beta) \right]\right|
    \ge e^{-cn} \E \left[ Z_{G,< \varepsilon}(E) \right]\,.
  \]
\end{lemma}
Note that this implies $\beta \mapsto \E [Z_{G,< \varepsilon}(\beta)]$ is zero-free on $\rectangularStrip_+(\delta,\eta_0)$.
Together with Corollary~\cref{c:1mt-balanced-complex} below, this implies the same is true for $\beta \mapsto \E [Z_{G}(\beta)]$ (after possibly shrinking $\eta_0$).
\begin{theorem}[Reweighted second moment estimate, proved in \cref{s:moments}]\label{t:reweight}
  There exists a function $A_{G}(\beta)$ such that the following holds.
  \begin{enumerate}[label=(R\arabic*)]
    \item \label{i:AN-complex-analytic} For all realizations of $G$, $\beta \mapsto A_{G}(\beta)$ is complex analytic, with $A_{G}(0) = 1$.
  \end{enumerate}
  Moreover, let $\delta > 0$ and $\varepsilon_0 = \varepsilon_0(d,\delta)$.
  For all $\varepsilon \in (0, \varepsilon_0)$ there exists $\eta_0 = \eta_0(d,\delta,\varepsilon) > 0$ such that the following holds.
  \begin{enumerate}[label=(R\arabic*)]
    \setcounter{enumi}{1}
    \item \label{i:AN-moments-bounded} For all $K>0$, there exists $C = C(d,\delta,\varepsilon,K) > 0$ such that 
    \[
      \sup_{\beta \in \rectangularStrip_+(\delta,\eta_0)} 
      \E[|A_{G}(\beta)|^K] \le C\,.
    \]
    \item \label{i:AN-second-moment} We have
    \[
      \sup_{\beta \in \rectangularStrip_+(\delta,\eta_0)} \left|
        \frac{\E[Z_{G,< \varepsilon}^{(2)}(\beta) |A_{G}(\beta)|^2]}{|\E[Z_{G,< \varepsilon}(\beta)]|^2} - 1
      \right|
      = o_n(1)\,.
    \]
  \end{enumerate}
\end{theorem}
In the proof of ~\cref{t:zero-free}, $A_{G}(\beta)$ will be our reweighting factor.
Property \cref{i:AN-complex-analytic} is of course needed to run Jensen's formula on the reweighted partition function.
Property \cref{i:AN-moments-bounded} says that $O(1)$-degree moments of $A_{G}(\beta)$ are $O(1)$; this can be expected because $A_{G}(\beta)$ should be an $O(1)$-scale random variable --- in the mixed $p$-spin case it was lognormal, and here it is the exponential of a compound Poisson.
Property \cref{i:AN-second-moment} says that reweighting makes the near-orthogonal contribution to the second moment match the first moment squared.

\subsection{Proof of zero-freeness}
\label{ss:zero-free}

We now derive ~\cref{t:zero-free}, our main zero-freeness result, using the above inputs.
We first derive the following extension of ~\cref{t:truncation} and Corollary~\cref{c:truncation} for $\beta \in \rectangularStrip_+(\delta,\eta_0)$.
\begin{corollary}\label{c:truncation-complex}
  Let $\delta > 0$ and $\varepsilon_0 = \varepsilon_0(d,\delta)$.
  For all $\varepsilon \in (0, \varepsilon_0)$ 
  there exist $c = c(d,\delta,\varepsilon) > 0$ and $\eta_0 = \eta_0(d,\delta,\varepsilon) > 0 $ such that for all $\beta \in \rectangularStrip_+(\delta,\eta_0)$,
  \[
    \E\left[
      \allone\{G\not\in {\mathcal{E}}_{\beta,\varepsilon,c}\}
      |Z_{G,< \varepsilon}(\beta)|
    \right]
    \le e^{-cn} \left|\E\left[
      Z_{G,< \varepsilon}(\beta)
    \right]\right|
  \]
  and 
  \[
    \E\left[
      \allone\{G\not\in {\mathcal{E}}_{\beta,\varepsilon,c}\}
      |Z_{G,< \varepsilon}^{(2)}(\beta)|
    \right] \le e^{-cn} \left|\E\left[
      Z_{G,< \varepsilon}(\beta)
    \right]\right|^2
  \]
\begin{proof}
  Let $c^{\cref{t:truncation}}(d,\delta,\varepsilon)$ and $c^{\cref{c:truncation}}(d,\delta,\varepsilon)$ be the functions given by ~\cref{t:truncation} and Corollary~\cref{c:truncation}, and $\eta_0^{\cref{l:1mt}}(d,\delta,\varepsilon,c)$ be the function given by ~\cref{l:1mt}.
  Given $\varepsilon \in (0, \varepsilon_0)$ we set
  \begin{align*}
    c &= \frac12 \min\left(c^{\cref{t:truncation}}(d,\delta,\varepsilon),c^{\cref{c:truncation}}(d,\delta,\varepsilon)\right)\,, & 
    \eta_0 = \eta_0^{\cref{l:1mt}}(d,\delta,\varepsilon,c/2)\,.
  \end{align*}
  Consider any $\beta = E + i \eta \in \rectangularStrip_+(\delta,\eta_0)$.
  For any realization of $G$, we have 
  \[
    |Z_{G,< \varepsilon}(\beta)| \le \sum_{\sigma \in \Sigma_n(\varepsilon)} |\exp(\beta \isingHamiltonian_{G}(\sigma))|
    = \sum_{\sigma \in \Sigma_n(\varepsilon)} \exp(E \isingHamiltonian_{G}(\sigma))
    = Z_{G,< \varepsilon}(E)\,,
  \]
  and similarly
  \begin{align*}
    |Z_{G,< \varepsilon}^{(2)}(\beta)| &\le Z_{G,< \varepsilon}^{(2)}(E)\,, & 
    |Z_{G,\ge \varepsilon}^{(2)}(\beta)| &\le Z_{G,\ge \varepsilon}^{(2)}(E)\,.
  \end{align*}
  Furthermore by ~\cref{l:1mt},
  \[
    \left|\E[Z_{G,< \varepsilon}(\beta)]\right|
    \ge e^{-cn/2} \E[Z_{G,< \varepsilon}(E)]\,.
  \]
  We next argue that ${\mathcal{E}}_{E,\varepsilon,2c} \subseteq {\mathcal{E}}_{\beta,\varepsilon,c}$, where these events are defined in \eqref{e:overlap-event}.
  Indeed, if $G \in {\mathcal{E}}_{E,\varepsilon,2c}$,
  \[
    Z_{G,\ge \varepsilon}^{(2)}(E) \le e^{-2cn} \E [Z_{G,< \varepsilon}(E)]^2\,.
  \]
  Then,
  \[
    |Z_{G,\ge \varepsilon}^{(2)}(\beta)|
    \le Z_{G,\ge \varepsilon}^{(2)}(E)
    \le e^{-2cn} \E [Z_{G,< \varepsilon}(E)]^2
    \le e^{-cn} |\E[Z_{G,< \varepsilon}(\beta)]|^2\,,
  \]
  so $G \in {\mathcal{E}}_{\beta,\varepsilon,c}$.
  Since $\beta \in \rectangularStrip_+(\delta,\eta_0)$, we have $E \in [2\beta_\lowerBoundTag, \betaKS - \delta]$.
  It follows that
  \begin{align*}
    \E\left[
      \allone\{G\not\in {\mathcal{E}}_{\beta,\varepsilon,c}\}
      |Z_{G,< \varepsilon}(\beta)|
    \right]
    &\le \E\left[
      \allone\{G\not\in {\mathcal{E}}_{E,\varepsilon,2c}\}
      Z_{G,< \varepsilon}(E)
    \right] \\
    &\!\stackrel{(*)}{\le} 
    e^{-2cn} \E[Z_{G,< \varepsilon}(E)]
    \le e^{-cn} |\E[Z_{G,< \varepsilon}(\beta)]|\,,
  \end{align*}
  where the step marked $(*)$ is by ~\cref{t:truncation}.
  Similarly,
  \begin{align*}
    \E\left[
      \allone\{G\not\in {\mathcal{E}}_{\beta,\varepsilon,c}\}
      |Z_{G,< \varepsilon}^{(2)}(\beta)|
    \right]
    &\le \E\left[
      \allone\{G\not\in {\mathcal{E}}_{E,\varepsilon,2c}\}
      Z_{G,< \varepsilon}^{(2)}(E)
    \right] \\
    &\!\stackrel{(*)}{\le} 
    e^{-2cn} \E[Z_{G,< \varepsilon}(E)]^2
    \le e^{-cn} |\E[Z_{G,< \varepsilon}(\beta)]|^2\,,
  \end{align*}
  where the step marked $(*)$ is by Corollary~\cref{c:truncation}.
\end{proof}
\end{corollary}
The following corollary similarly extends ~\cref{l:1mt-balanced}~\cref{i:1mt-balanced-deviation} to $\beta \in \rectangularStrip_+(\delta,\eta_0)$.
\begin{corollary}\label{c:1mt-balanced-complex}
  Let $\delta > 0$ and $\varepsilon_0 = \varepsilon_0(d,\delta)$.
  For all $\varepsilon \in (0, \varepsilon_0)$, there exists $c = c(d,\delta,\varepsilon) > 0$ and $\eta_0 = \eta_0(d,\delta,\varepsilon) > 0$ such that for all $\beta \in \rectangularStrip_+(\delta,\eta_0)$, 
  \[
    \E\left[ |Z_{G,\ge \varepsilon}(\beta)| \right]
    \le e^{-cn} \left|\E\left[ Z_{G,< \varepsilon}(\beta) \right]\right|\,.
  \]
\begin{proof}
  Let $c^{\cref{l:1mt-balanced}}(d,\varepsilon)$ be the function given by ~\cref{l:1mt-balanced} and $\eta_0^{\cref{l:1mt}}(d,\delta,\varepsilon,c)$ be the function given by ~\cref{l:1mt}.
  Consider any $\varepsilon \in (0, \varepsilon_0)$.
  We set
  \begin{align*}
    c &= \frac12 c^{\cref{l:1mt-balanced}}(d,\varepsilon)\,, &
    \eta_0 &= \eta_0^{\cref{l:1mt}}(d,\delta,\varepsilon,c)\,.
  \end{align*}
  Consider any $\beta = E + i\eta \in \rectangularStrip_+(\delta,\eta_0)$.
  For any realization of $G$, we have (similarly to above)
  \[
    |Z_{G,\ge \varepsilon}(\beta)| \le Z_{G,\ge \varepsilon}(E)\,.
  \]
  Furthermore, by ~\cref{l:1mt},
  \[
    \left|\E \left[Z_{G,< \varepsilon}(\beta)\right]\right| 
    \ge e^{-cn} 
    \E \left[Z_{G,< \varepsilon}(E)\right]\,.
  \]
  Then,
  \[
    \E\left[ |Z_{G,\ge \varepsilon}(\beta)| \right]
    \le \E\left[ Z_{G,\ge \varepsilon}(E) \right]
    \le e^{-2cn} \E\left[ Z_{G,< \varepsilon}(E) \right]
    \le e^{-cn} \left|\E \left[Z_{G,< \varepsilon}(\beta)\right]\right|\,. \qedhere
  \]
\end{proof}
\end{corollary}

\begin{proof}[Proof of ~\cref{t:zero-free}]
  Let $\varepsilon_0 = \varepsilon_0(d,\delta/2)$ be given by ~\cref{l:basic-mt}.
  Let $c^{\cref{c:truncation-complex}}(d,\delta,\varepsilon)$ and $c^{\cref{c:1mt-balanced-complex}}(d,\delta,\varepsilon)$ be the functions given by Corollaries~\cref{c:truncation-complex} and \cref{c:1mt-balanced-complex}.
  Let $\eta_0^{\cref{l:1mt}}(d,\delta,\varepsilon,c)$, $\eta_0^{\cref{t:reweight}}(d,\delta,\varepsilon)$, $\eta_0^{\cref{c:truncation-complex}}(d,\delta,\varepsilon)$, and $\eta_0^{\cref{c:1mt-balanced-complex}}(d,\delta,\varepsilon)$ be the functions $\eta_0$ from ~\cref{l:1mt}, ~\cref{t:reweight}, Corollary~\cref{c:truncation-complex}, and Corollary~\cref{c:1mt-balanced-complex}.
  Consider any $\varepsilon \in (0, \varepsilon_0)$.
  We set
  \begin{align}
    \nonumber
    c &= \min\left(
      c^{\cref{c:truncation-complex}}(d,\delta/2,\varepsilon), 
      c^{\cref{c:1mt-balanced-complex}}(d,\delta/2,\varepsilon)
    \right)\,, \\
    \label{e:zero-free-eta}
    \eta_0 &= \frac12 \min\left(
      \eta_0^{\cref{l:1mt}}(d,\delta/2,\varepsilon,c/2),
      \eta_0^{\cref{t:reweight}}(d,\delta/2,\varepsilon),
      \eta_0^{\cref{c:truncation-complex}}(d,\delta/2,\varepsilon),
      \eta_0^{\cref{c:1mt-balanced-complex}}(d,\delta/2,\varepsilon)
    \right)\,.
  \end{align}
  Note that
  \[
    {\mathsf{int}}(\rectangularStrip(\delta,\eta_0)) \subseteq \rectangularStrip_+(\delta/2,2\eta_0)\,.
  \]
  For $R>0$, let $\D(R) = \{z\in \C : |z| \le R\}$. 
  By \cite[Lemma 8.1]{Bar17perm} there exist constants $0 < r < R$ and a complex analytic 
  $\ell : \D(R) \rightarrow \C$ such that $\ell(0) = 0$ and
  \[
    \rectangularStrip(\delta,\eta_0)
    \subseteq \ell(\D(r))
    \subseteq \ell(\D(R))
    \subseteq \rectangularStrip_+(\delta/2,2\eta_0)\,.
  \]
  Define the reweighted partition function
  \[
    X_{G}(\beta) = \frac{Z_{G}(\beta) A_{G}(\beta)}{\E[Z_{G}(\beta)]}\,.
  \]
  Note that this is complex analytic on $\rectangularStrip_+(\delta/2,2\eta_0)$, as $A_{G}$ is complex analytic by property \cref{i:AN-complex-analytic}, and ~\cref{l:1mt} and Corollary~\cref{c:1mt-balanced-complex} together ensure the denominator is zero-free in this domain.
  Define
  \[
    f_{G}(z) = X_{G}(\ell(z))\,,
  \]
  which is complex analytic on $\D(R) \subseteq \ell^{-1}(\rectangularStrip_+(\delta/2,2\eta_0))$.
  Finally for $S \subseteq  \D(R)$ and complex analytic $f : \D(R) \rightarrow \C$ let
  \[
    \zeroMultiset(f,S) = \left\{
      \omega \in S : f(\omega) = 0
    \right\}
  \]
  where we count multiplicity.
  We obtain the following chain of inequalities where the step marked $(*)$ is by Jensen's formula:
  \begin{align*}
    |\zeroMultiset(Z_{G}(\cdot), \rectangularStrip(\delta,\eta_0))|
    &\le |\zeroMultiset(X_{G}(\cdot), \ell(\D(r)))|
    \le |\zeroMultiset(f_{G}(\cdot), \D(r))| \\
    &\le \left(\log\frac{R}{r}\right)^{-1} \sum_{\omega \in \zeroMultiset(f_{G}(\cdot), \D(r))} \log \left(\frac{R}{|\omega|}\right) \\
    &\le \left(\log\frac{R}{r}\right)^{-1} \sum_{\omega \in \zeroMultiset(f_{G}(\cdot), \D(R))} \log \left(\frac{R}{|\omega|}\right) \\
    &\!\stackrel{(*)}{=} \left(\log\frac{R}{r}\right)^{-1} \E_{\theta \sim \unif([0,2\pi])} \log \left|\frac{f_{G}(Re^{i\theta})}{f_{G}(0)}\right|\,.
  \end{align*}
  Note that $f_{G}(Re^{i\theta}) = X_{G}(\ell(Re^{i\theta}))$ by definition, while
  \[
    f_{G}(0) = X_{G}(0) = \frac{Z_{G}(0) A_{G}(0)}{\E[Z_{G}(0)]} = 1\,,
  \]
  using that $\ell(0) = 0$, $Z_{G}(0) = 2^n$ deterministically, and $A_{G}(0) = 1$ by property \cref{i:AN-complex-analytic}.
  Thus
  \begin{equation}\label{e:zero-free-step1}
    |\zeroMultiset(Z_{G}(\cdot), \rectangularStrip(\delta,\eta_0))|
    \le \left(\log\frac{R}{r}\right)^{-1} \E_{\theta \sim \unif([0,2\pi])} \log |X_{G}(\ell(Re^{i\theta}))|\,.
  \end{equation}
  We will now argue that
  \begin{equation}\label{e:zero-free-goal}
    \sup_{\beta \in \rectangularStrip_+(\delta/2,2\eta_0)}
    \E_{G\sim \configurationModel(n,d)} \log |X_{G}(\beta)|
    \le o_n(1)\,.
  \end{equation}
  Together with \eqref{e:zero-free-step1} and the fact that $\ell(\D(R)) \subseteq \rectangularStrip_+(\delta/2,2\eta_0)$, this implies
  \begin{align*}
    \E_{G\sim \configurationModel(n,d)} [|\zeroMultiset(Z_{G}(\cdot), \rectangularStrip(\delta,\eta_0))|]
    &\le \left(\log\frac{R}{r}\right)^{-1} \E_{G\sim \configurationModel(n,d)} \E_{\theta \sim \unif([0,2\pi])} \log |X_{G}(\ell(Re^{i\theta}))| \\
    &\le \left(\log\frac{R}{r}\right)^{-1} \sup_{\beta \in \rectangularStrip_+(\delta/2,2\eta_0)} \E_{G\sim \configurationModel(n,d)} \log |X_{G}(\beta)|
    = o_n(1)\,,
  \end{align*}
  which implies the desired conclusion by Markov's inequality.
  We now turn to the proof of \eqref{e:zero-free-goal}; all expectations that follow will be over $G \sim \configurationModel(n,d)$.
  Consider any $\beta \in \rectangularStrip_+(\delta/2,2\eta_0)$, and let ${\mathcal{E}}_{\beta,\varepsilon,c}$ be defined in \eqref{e:overlap-event}.
  Define
  \[
    {\widetilde X}_{G}(\beta) = \frac{Z_{G,< \varepsilon}(\beta) A_{G}(\beta)}{\E[Z_{G,< \varepsilon}(\beta)]}\,.
  \]
  By Jensen's inequality,
  \begin{align}
    \nonumber
    \E \log |X_{G}(\beta)|
    &\le \log \E [|X_{G}(\beta)|]
    \le \log \left(
        \E [|X_{G}(\beta) - {\widetilde X}_{G}(\beta)|] 
        + \E [|{\widetilde X}_{G}(\beta)|] 
    \right) \\
    \nonumber
    &\le \log \Big(
      \E [|X_{G}(\beta) - {\widetilde X}_{G}(\beta)|] 
      + \E [\allone\{G \not\in {\mathcal{E}}_{\beta,\varepsilon,c}\} |{\widetilde X}_{G}(\beta)|] \\
      \nonumber
      &\qquad + \E [\allone\{G \in {\mathcal{E}}_{\beta,\varepsilon,c}\} |{\widetilde X}_{G}(\beta)|]
    \Big) \\
    \nonumber
    &\le \log \Big(
      \E [|X_{G}(\beta) - {\widetilde X}_{G}(\beta)|] 
      + \E [\allone\{G \not\in {\mathcal{E}}_{\beta,\varepsilon,c}\} |{\widetilde X}_{G}(\beta)|] \\
      \label{e:zero-free-jensen}
      &\qquad + \E [\allone\{G \in {\mathcal{E}}_{\beta,\varepsilon,c}\} |{\widetilde X}_{G}(\beta)|^2]^{1/2}
    \Big)\,.
  \end{align}
  We will bound the moments appearing in \eqref{e:zero-free-jensen}.
  Let $\beta = E + i\eta$.
  Note that deterministically,
  \[
    \max\left(
      |Z_{G,< \varepsilon}(\beta)|,
      |Z_{G,\ge \varepsilon}(\beta)|,
      |Z_{G}(\beta)|
    \right)
    \le 2^n e^{\max(E,0) dn/2} 
    \le 2^n e^{\betaKS dn/2} 
    \equiv Z_{n,\upperBoundTag}\,,
  \]
  and (by Corollary~\cref{c:1mt-balanced-complex} and ~\cref{l:1mt})

  \begin{align*}
    \min\left(
      |\E[Z_{G}(\beta)]|,
      |\E[Z_{G,< \varepsilon}(\beta)]|
    \right)
    &\ge (1-e^{-cn}) |\E[Z_{G,< \varepsilon}(\beta)]| \\
    &\ge (1-e^{-cn}) e^{-cn/2} \E[Z_{G,< \varepsilon}(E)] \\
    &\ge (1-e^{-cn})^2 e^{-cn/2} \E[Z_{G}(E)] \\
    &\ge e^{-cn} \cdot 2^n e^{2\beta_\lowerBoundTag \cdot dn / 2}
    \equiv Z_{n,\lowerBoundTag}\,.
  \end{align*}
  Recall that $\beta \in \rectangularStrip_+(\delta/2,2\eta_0)$.
  By ~\cref{l:1mt},
  \begin{equation}\label{e:abs-inside-outside-comparison}
    \E[|Z_{G,< \varepsilon}(\beta)|]
    \le \E[Z_{G,< \varepsilon}(E)]
    \le e^{cn/2} |\E[Z_{G,< \varepsilon}(\beta)]|\,.
  \end{equation}
  Then, for a small constant $\iota > 0$ we will set below, we obtain the following chain of inequalities, where the step marked $(*)$ is by Corollary~\cref{c:1mt-balanced-complex}, the step marked $(**)$ is by H\"older's inequality, and the step marked $(***)$ is by Corollary~\cref{c:1mt-balanced-complex} and \eqref{e:abs-inside-outside-comparison}:
  \begin{align*}
    &\E [|X_{G}(\beta) - {\widetilde X}_{G}(\beta)|] \\
    &\le \E\left[
      \left|
        \frac{Z_{G}(\beta) A_{G}(\beta)}{\E[Z_{G}(\beta)]}
        - \frac{Z_{G,< \varepsilon}(\beta) A_{G}(\beta)}{\E[Z_{G}(\beta)]}
      \right|
    \right] \\
    &\qquad\qquad + 
    \E \left[
      \left|
        \frac{Z_{G,< \varepsilon}(\beta) A_{G}(\beta)}{\E[Z_{G}(\beta)]}
        - \frac{Z_{G,< \varepsilon}(\beta) A_{G}(\beta)}{\E[Z_{G,< \varepsilon}(\beta)]}
      \right|
    \right] \\
    &\stackrel{(*)}{\le} 
    \frac{\E[|Z_{G,\ge \varepsilon}(\beta)| |A_{G}(\beta)|]}{|\E[Z_{G}(\beta)]|}
    + \frac{e^{-cn}}{1-e^{-cn}}  
    \frac{\E[|Z_{G,< \varepsilon}(\beta)| |A_{G}(\beta)|]}{|\E[Z_{G,< \varepsilon}(\beta)]|} \\
    &\le
    \frac{(Z_{n,\upperBoundTag})^\iota \E[|Z_{G,\ge \varepsilon}(\beta)|^{1-\iota} |A_{G}(\beta)|]}{(Z_{n,\lowerBoundTag})^\iota |\E[Z_{G}(\beta)]|^{1-\iota}}
    + 2e^{-cn}
    \frac{(Z_{n,\upperBoundTag})^\iota \E[|Z_{G,< \varepsilon}(\beta)|^{1-\iota} |A_{G}(\beta)|]}{(Z_{n,\lowerBoundTag})^\iota |\E[Z_{G,< \varepsilon}(\beta)]|^{1-\iota}} \\
    &\stackrel{(**)}{\le}
    \frac{(Z_{n,\upperBoundTag})^\iota \E[|Z_{G,\ge \varepsilon}(\beta)|]^{1-\iota} \E[|A_{G}(\beta)|^{1/\iota}]^\iota}{(Z_{n,\lowerBoundTag})^\iota |\E[Z_{G}(\beta)]|^{1-\iota}} \\
    &\qquad\qquad + 2e^{-cn}
    \frac{(Z_{n,\upperBoundTag})^\iota \E[|Z_{G,< \varepsilon}(\beta)|]^{1-\iota} \E[|A_{G}(\beta)|^{1/\iota}]^{\iota}}{(Z_{n,\lowerBoundTag})^\iota |\E[Z_{G,< \varepsilon}(\beta)]|^{1-\iota}} \\
    &\stackrel{(***)}{\le}
    \left(\frac{e^{-cn}}{1-e^{-cn}}\right)^{1-\iota} \frac{(Z_{n,\upperBoundTag})^\iota \E[|A_{G}(\beta)|^{1/\iota}]^\iota}{(Z_{n,\lowerBoundTag})^\iota}
    + 2e^{-cn}
    \frac{(Z_{n,\upperBoundTag})^\iota \E[|A_{G}(\beta)|^{1/\iota}]^{\iota}}{(Z_{n,\lowerBoundTag})^\iota} \cdot e^{c(1-\iota) n / 2} \\
    &\le 2(e^{-c(1-\iota) n} + e^{-c(1+\iota) n / 2}) \cdot \left(
      \frac{Z_{n,\upperBoundTag}}{Z_{n,\lowerBoundTag}} \cdot \E[|A_{G}(\beta)|^{1/\iota}]
    \right)^{\iota}
  \end{align*}
  By property \cref{i:AN-moments-bounded}, for $C = C(d,\delta/2,\varepsilon,1/\iota)$ defined therein,
  \begin{equation}\label{e:AN-1-iota}
    \E[|A_{G}(\beta)|^{1/\iota}] \le C\,.
  \end{equation}
  Thus, if we set $\iota$ to be a suitably small constant,
  \begin{equation}\label{e:zero-free-unbalanced}
    \E [|X_{G}(\beta) - {\widetilde X}_{G}(\beta)|]
    \le 2(e^{-c(1-\iota) n} + e^{-c(1+\iota) n / 2}) \cdot
    \left(
      \frac{CZ_{n,\upperBoundTag}}{Z_{n,\lowerBoundTag}}
    \right)^\iota
    \le e^{-cn/3}\,.
  \end{equation}
  We similarly obtain the following chain of inequalities, where the step marked $(*)$ is by H\"older's inequality and the step marked $(**)$ is by Corollary~\cref{c:truncation-complex}:
  \begin{align*}
    \E [\allone\{G \not\in {\mathcal{E}}_{\beta,\varepsilon,c}\} |{\widetilde X}_{G}(\beta)|]
    &= \frac{
      \E[\allone\{G \not\in {\mathcal{E}}_{\beta,\varepsilon,c}\}
      |Z_{G,< \varepsilon}(\beta)| |A_{G}(\beta)|]
    }{|\E[Z_{G,< \varepsilon}(\beta)]|} \\
    &\le \frac{
      (Z_{n,\upperBoundTag})^\iota \cdot \E[
        \allone\{G \not\in {\mathcal{E}}_{\beta,\varepsilon,c}\}
        |Z_{G,< \varepsilon}(\beta)|^{1-\iota} |A_{G}(\beta)|
      ]
    }{
      (Z_{n,\lowerBoundTag})^\iota
      |\E[Z_{G,< \varepsilon}(\beta)]|^{1-\iota}
    } \\
    &\!\stackrel{(*)}{\le} \frac{
      (Z_{n,\upperBoundTag})^\iota
      \cdot \E[
        \allone\{G \not\in {\mathcal{E}}_{\beta,\varepsilon,c}\}
        |Z_{G,< \varepsilon}(\beta)|
      ]^{1-\iota}
      \cdot \E[
        |A_{G}(\beta)|^{1/\iota}
      ]^\iota
    }{
      (Z_{n,\lowerBoundTag})^\iota
      |\E[Z_{G,< \varepsilon}(\beta)]|^{1-\iota}
    } \\
    &\!\!\stackrel{(**)}{\le} \frac{
      (Z_{n,\upperBoundTag})^\iota 
      \cdot e^{-c(1-\iota) n}
      \cdot \E[
        |A_{G}(\beta)|^{1/\iota}
      ]^\iota
    }{
      (Z_{n,\lowerBoundTag})^\iota
    } \\
    &= e^{-c(1-\iota) n} \cdot \left(
      \frac{Z_{n,\upperBoundTag}}{Z_{n,\lowerBoundTag}}
      \cdot \E[|A_{G}(\beta)|^{1/\iota}]
    \right)^\iota\,.
  \end{align*}
  By the bound \eqref{e:AN-1-iota}, if we set $\iota$ to be a suitably small constant,
  \begin{equation}\label{e:zero-free-1mt}
    \E [\allone\{G \not\in {\mathcal{E}}_{\beta,\varepsilon,c}\} |{\widetilde X}_{G}(\beta)|]
    \le e^{-c(1-\iota) n} \cdot \left(
      \frac{CZ_{n,\upperBoundTag}}{Z_{n,\lowerBoundTag}}
    \right)^\iota
    \le e^{-cn/2}\,.
  \end{equation}
  Turning to the remaining term in \eqref{e:zero-free-jensen}, we write
  \[
    \E [\allone\{G \in {\mathcal{E}}_{\beta,\varepsilon,c}\} |{\widetilde X}_{G}(\beta)|^2]
    = \frac{\E [
      \allone\{G \in {\mathcal{E}}_{\beta,\varepsilon,c}\}
      |Z_{G,< \varepsilon}(\beta)|^2
      |A_{G}(\beta)|^2]
    }{|\E[Z_{G,< \varepsilon}(\beta)]|^2}
    = \frac{{\textup{(i)}} + {\textup{(ii)}} + {\textup{(iii)}}}{|\E[Z_{G,< \varepsilon}(\beta)]|^2}\,,
  \]
  where
  \begin{align*}
    {\textup{(i)}} &= \E \left[
      \allone\{G \in {\mathcal{E}}_{\beta,\varepsilon,c}\} 
      \left(|Z_{G,< \varepsilon}(\beta)|^2 - Z_{G,< \varepsilon}^{(2)}(\beta)\right)
      |A_{G}(\beta)|^2
    \right]\,, \\
    {\textup{(ii)}} &= -\E \left[
      \allone\{G \not\in {\mathcal{E}}_{\beta,\varepsilon,c}\} 
      Z_{G,< \varepsilon}^{(2)}(\beta)
      |A_{G}(\beta)|^2
    \right]\,, \\
    {\textup{(iii)}} &= \E \left[
      Z_{G,< \varepsilon}^{(2)}(\beta)
      |A_{G}(\beta)|^2
    \right]\,.
  \end{align*}
  By definition of the event ${\mathcal{E}}_{\beta,\varepsilon,c}$,
  \begin{align*}
    \frac{|{\textup{(i)}}|}{|\E[Z_{G,< \varepsilon}(\beta)]|^2}
    &\le \frac{
      \E [
        \allone\{G \in {\mathcal{E}}_{\beta,\varepsilon,c}\} 
        |Z_{G,\ge \varepsilon}^{(2)}(\beta)|
        |A_{G}(\beta)|^2
      ]
    }{
      |\E[Z_{G,< \varepsilon}(\beta)]|^2
    } \\
    &\le e^{-cn} \cdot \E[|A_{G}(\beta)|^2]
    \le e^{-cn} \cdot C(d,\delta/2,\varepsilon,2)
    \le e^{-cn/2}\,.
  \end{align*}
  We have the deterministic bound
  \[
    |Z_{G,< \varepsilon}^{(2)}(\beta)| \le (Z_{n,\upperBoundTag})^2\,.
  \]
  So, similarly to above, we obtain the following chain of inequalities, where the step $(*)$ is by H\"older's inequality and the step $(**)$ is by Corollary~\cref{c:truncation-complex}:
  \begin{align*}
    \frac{|{\textup{(ii)}}|}{|\E[Z_{G,< \varepsilon}(\beta)]|^2}
    &\le \frac{
      \E [
        \allone\{G \not\in {\mathcal{E}}_{\beta,\varepsilon,c}\} 
        |Z_{G,< \varepsilon}^{(2)}(\beta)|
        |A_{G}(\beta)|^2
      ]
    }{|\E[Z_{G,< \varepsilon}(\beta)]|^2} \\
    &\le \frac{
      (Z_{n,\upperBoundTag})^{2\iota} \cdot \E [
        \allone\{G \not\in {\mathcal{E}}_{\beta,\varepsilon,c}\} 
        |Z_{G,< \varepsilon}^{(2)}(\beta)|^{1-\iota}
        |A_{G}(\beta)|^2
      ]
    }{
      (Z_{n,\lowerBoundTag})^{2\iota} |\E[Z_{G,< \varepsilon}(\beta)]|^{2(1-\iota)}
    } \\
    &\!\stackrel{(*)}{\le} \frac{
      (Z_{n,\upperBoundTag})^{2\iota} \cdot \E [
        \allone\{G \not\in {\mathcal{E}}_{\beta,\varepsilon,c}\} 
        |Z_{G,< \varepsilon}^{(2)}(\beta)|
      ]^{1-\iota} \cdot \E[
        |A_{G}(\beta)|^{2/\iota}
      ]^\iota
    }{
      (Z_{n,\lowerBoundTag})^{2\iota} |\E[Z_{G,< \varepsilon}(\beta)]|^{2(1-\iota)}
    } \\
    &\!\!\stackrel{(**)}{\le} \frac{
      (Z_{n,\upperBoundTag})^{2\iota} \cdot e^{-c(1-\iota)n} \cdot \E[
        |A_{G}(\beta)|^{2/\iota}
      ]^\iota
    }{
      (Z_{n,\lowerBoundTag})^{2\iota}
    } \\
    &\le e^{-c(1-\iota)n} \left(
      \frac{(Z_{n,\upperBoundTag})^2}{(Z_{n,\lowerBoundTag})^2}
      \cdot \E[|A_{G}(\beta)|^{2/\iota}]
    \right)^\iota\,.
  \end{align*}
  By property \cref{i:AN-moments-bounded}, $\E[|A_{G}(\beta)|^{2/\iota}] \le C(d,\delta/2,\varepsilon,2/\iota)$.
  Setting $\iota$ suitably small yields
  \[
    \frac{|{\textup{(ii)}}|}{|\E[Z_{G,< \varepsilon}(\beta)]|^2} \le e^{-cn/2}\,.
  \]
  Finally property \cref{i:AN-second-moment} yields
  \[
    \left|\frac{{\textup{(iii)}}}{|\E[Z_{G,< \varepsilon}(\beta)]|^2} - 1\right| = o_n(1)\,.
  \]
  Combining the above yields (for an error $o_n(1)$ uniform in $\beta \in \rectangularStrip_+(\delta/2,2\eta_0)$)
  \[
    \E [\allone\{G \in {\mathcal{E}}_{\beta,\varepsilon,c}\} |{\widetilde X}_{G}(\beta)|^2] \le 1 + o_n(1)\,.
  \]
  Together with \eqref{e:zero-free-jensen}, \eqref{e:zero-free-unbalanced}, and \eqref{e:zero-free-1mt} this implies \eqref{e:zero-free-goal}.
\end{proof}

\subsection{Basic moment estimates of real partition function}

In addition to ~\cref{l:1mt-balanced}, we additionally record some simple estimates on the expectations of replicated versions of $Z_{G,< \varepsilon}(\beta)$.
These estimates are simpler than ~\cref{l:1mt} and ~\cref{t:reweight} in two ways: they hold for \textbf{real} $\beta \in [2\beta_\lowerBoundTag, \betaKS - \delta]$ and will be accurate to leading \textbf{exponential} order, whereas ~\cref{l:1mt} and ~\cref{t:reweight} hold for \textbf{complex} $\beta \in \rectangularStrip_+(\delta,\eta_0)$ and are accurate up to a multiplicative $1+o_n(1)$ error.
The proofs are routine applications of Laplace's method and deferred to Appendix~\cref{a:basic-mt}.

Recall the set $\Sigma_n(\varepsilon)$ defined in \eqref{e:Sigma-N-eps}.
For $\beta \in \R$, define
\begin{align}\nonumber
  Z_{G,< \varepsilon}^{(3)}(\beta) &= 
  \sum_{\sigma^1,\sigma^2,\sigma^3 \in \Sigma_n(\varepsilon)} 
  \allone\{\max(|R(\sigma^1,\sigma^2)|, |R(\sigma^1,\sigma^3)|, |R(\sigma^2,\sigma^3)|) < \varepsilon\} \\
  \label{e:Z3}
  &\qquad\qquad\qquad \exp\left(\beta \isingHamiltonian_{G}(\sigma^1) + \beta \isingHamiltonian_{G}(\sigma^2) + \beta \isingHamiltonian_{G}(\sigma^3)\right)\,.
\end{align}
\begin{lemma}\label{l:basic-mt}
  There exists $\varepsilon_0 = \varepsilon_0(d,\delta) \in (0,1)$ such that for all $\varepsilon \in (0, \varepsilon_0)$ and $\upsilon > 0$, the following holds for all real $\beta \in [2\beta_\lowerBoundTag, \betaKS - \delta]$.
  \begin{enumerate}[label=(\alph*)]
    \item \label{i:basic-mt-annealed-2} We have
    \[
      \E[Z_{G,< \varepsilon}^{(2)}(\beta)] \le e^{\upsilon n} \E[Z_{G,< \varepsilon}(\beta)]^2\,.
    \]
    \item \label{i:basic-mt-annealed-2-subgaussian} Let
    \begin{equation}\label{e:SigmaN2}
      \Sigma_{n,2}(\varepsilon) = \{(\sigma,{\sigma'}) \in \Sigma_n(\varepsilon)^2 : |R(\sigma,{\sigma'})| < \varepsilon \}\,.
    \end{equation}
    There exists $c = c(d,\delta,\varepsilon) > 0$ such that if we sample $(\sigma,{\sigma'}) \in \Sigma_{n,2}(\varepsilon)$ with probability 
    \[
      p(\sigma,{\sigma'}) \propto \E[\exp(\beta \isingHamiltonian_{G}(\sigma) + \beta \isingHamiltonian_{G}({\sigma'}))]\,,
    \]
    then for all $s\in (0,\varepsilon]$ independent of $n$,
    \[
      \Pr_{\sigma,{\sigma'}}(\max(|R(\sigma,\allone)|, |R({\sigma'},\allone)|, |R(\sigma,{\sigma'})|) \ge s) \le 6\exp\left(-ns^2/(2c)\right)\,.
    \]
    \item \label{i:basic-mt-annealed-3} We have
    \[
      \E[Z_{G,< \varepsilon}^{(3)}(\beta)] \le e^{\upsilon n} \E[Z_{G,< \varepsilon}(\beta)]^3\,.
    \]
  \end{enumerate}
\end{lemma}

\section{Estimate of truncation error}
\label{s:truncation}

In this section we prove ~\cref{t:truncation} and Corollary~\cref{c:truncation}.
We therefore assume the setting of these results and take $d,\delta$, $\varepsilon_0 = \varepsilon_0(d,\delta) > 0$ (given by ~\cref{l:basic-mt}), and $\varepsilon \in (0, \varepsilon_0)$ as fixed.
We will in general suppress dependence on these parameters.
Throughout this section $\beta \in [2\beta_\lowerBoundTag, \betaKS - \delta]$ will be real.

Most of the section will be devoted to ~\cref{p:truncation} below, which is a variant of ~\cref{t:truncation} and Corollary~\cref{c:truncation} that we will formulate next.
Recall $\Sigma_n(\varepsilon)$ defined in \eqref{e:Sigma-N-eps}.
For $\beta \in \R$, define the Gibbs measure $\mu_{G,\beta,<\varepsilon}$ as the probability measure on $\Sigma_n(\varepsilon)$ with
\begin{equation}\label{e:gibbs-msr}
  \mu_{G,\beta,<\varepsilon}(\sigma) = \frac{1}{Z_{G,< \varepsilon}(\beta)} \exp(\beta \isingHamiltonian_{G}(\sigma))\,.
\end{equation}
For $A \subseteq (\Sigma_n(\varepsilon))^k$ a predicate in $k$ replicas $\sigma^1,\ldots,\sigma^k \in \Sigma_n(\varepsilon)$, define
\[
  \mu_{G,\beta,<\varepsilon}(A) = \Pr\left(
    (\sigma^1,\ldots,\sigma^k) \in A : \sigma^i \stackrel{IID}{\sim} \mu_{G,\beta,<\varepsilon}
  \right)\,.
\]
This is a random variable depending on $G$.
When $A$ is a predicate in one replica $\sigma^1$, we will sometimes write $\sigma$ for $\sigma^1$.

Similarly, for ${\sigma'} \in \Sigma_n(\varepsilon)$ and $\varepsilon > 0$, define the probability measure on $\Sigma_n(\varepsilon)$
\begin{equation}\label{e:gibbs-msr-band}
  \mu_{G,\beta,<\varepsilon}^{2,{\sigma'}}(\sigma) 
  = \frac{
    \allone\{|R(\sigma,{\sigma'})| < \varepsilon\} \exp(\beta \isingHamiltonian_{G}(\sigma))
  }{
    {\widetilde Z}^{(2)}_{G,< \varepsilon}(\beta,{\sigma'})
  }\,,
\end{equation}
where
\begin{equation}\label{e:tZ2}
  {\widetilde Z}^{(2)}_{G,< \varepsilon}(\beta,{\sigma'}) 
  = \sum_{\sigma \in \Sigma_n(\varepsilon)} \allone\{|R(\sigma,{\sigma'})| < \varepsilon\} \exp(\beta \isingHamiltonian_{G}(\sigma))\,.
\end{equation}
Thus $\mu_{G,\beta,<\varepsilon}^{2,{\sigma'}}$ is the measure $\mu_{G,\beta,<\varepsilon}$ conditioned on the orthogonal band $\{\sigma \in \Sigma_n(\varepsilon) : |R(\sigma,{\sigma'})| < \varepsilon\}$.
We define the Gibbs probability $\mu_{G,\beta,<\varepsilon}^{2,{\sigma'}}(A)$ analogously to above.
\begin{proposition}\label{p:truncation}
  There exists $c = c(d,\delta,\varepsilon) > 0$ such that the following holds for all $\beta \in [2\beta_\lowerBoundTag, \betaKS - \delta]$.
  \begin{enumerate}[label=(\alph*)]
    \item \label{i:truncation-one-replica}
    Define the event
    \begin{equation}\label{e:overlap-event-O}
      {\mathcal{O}}_{\beta,c} = \left\{
        G : \mu_{G,\beta,<\varepsilon}(|R(\sigma^1,\sigma^2)| \ge \varepsilon) \le e^{-cn}
      \right\}\,.
    \end{equation}
    Then,
    \begin{equation}\label{e:truncation-O}
      \E\left[
        \allone\{G\not\in {\mathcal{O}}_{\beta,c}\}
        Z_{G,< \varepsilon}(\beta)
      \right]
      \le
      e^{-cn} \E[Z_{G,< \varepsilon}(\beta)]\,.
    \end{equation}
    \item \label{i:truncation-two-replica}
    Define
    \begin{equation}\label{e:overlap-event-O-two-replica}
      {\mathcal{O}}^{(2)}_{\beta,c} = \left\{
        ({\sigma'},G) : \mu_{G,\beta,<\varepsilon}^{2,{\sigma'}}(|R(\sigma^1,\sigma^2)| \ge \varepsilon) \le e^{-cn}
      \right\}\,.
    \end{equation}
    Then,
    \begin{align}\nonumber
      \sum_{\sigma,{\sigma'} \in \Sigma_n(\varepsilon)}
      \allone\{|R(\sigma,{\sigma'})| < \varepsilon\}
      \E\Big[
        \allone\{({\sigma'},G) \not\in {\mathcal{O}}^{(2)}_{\beta,c}\}
        &\exp\left(\beta \isingHamiltonian_{G}(\sigma) + \beta \isingHamiltonian_{G}({\sigma'})\right)
      \Big] \\
      \label{e:truncation-O-two-replica}
      &\le
      e^{-cn} \E[Z_{G,< \varepsilon}^{(2)}(\beta)]\,.
    \end{align}
  \end{enumerate}
\end{proposition}
We will derive ~\cref{t:truncation} and Corollary~\cref{c:truncation} from ~\cref{p:truncation} in \cref{ss:truncation-completion}.
\begin{remark}
  Parts \cref{i:truncation-one-replica} and \cref{i:truncation-two-replica} of ~\cref{p:truncation} will be respectively used to prove ~\cref{t:truncation} and Corollary~\cref{c:truncation}.
  The proofs of these two parts are similar and will be carried out ``in parallel'' in \cref{ss:truncation-overview}--\cref{ss:planted-boosting-rigorous}.
  No intuition is lost by focusing on part \cref{i:truncation-one-replica}, which is somewhat simpler.
\end{remark}

\subsection{Reduction to planted model and section overview}
\label{ss:truncation-overview}

The first step of the proof of ~\cref{p:truncation}~\cref{i:truncation-one-replica} follows the strategy of \cite{CLMS22}, where we reduce to a statement about the following planted model.
\begin{definition}[Planted model]\label{d:planted-model}
  Let $\nu^\plantedTag_{\beta}$ be the probability measure on $(\sigma^*,G)$, where $\sigma^* \in \Sigma_n(\varepsilon)$ and $G$ is a $d$-regular (multi-)graph, with
  \[
    \nu^\plantedTag_{\beta}(\sigma^*,G)
    \propto
    \exp(\beta \isingHamiltonian_{G}(\sigma^*))\,.
  \]
  Similarly to the configuration model, we identify each possible $G$ with a perfect matching on a set of $dn$ labeled half-edges, which contains $d$ half-edges attached to each of $n$ labeled vertices.
  We will use $\Pr^\plantedTag_{\beta}$, $\E^\plantedTag_{\beta}$ to denote probability and expectation with respect to $(\sigma^*,G) \sim \nu^\plantedTag_{\beta}$.
\end{definition}
\begin{remark}\label{r:planted-interpretation}
  The planted model $\nu^\plantedTag_{\beta}$ can be thought of as the following Bayesian inference problem.
  We first sample a signal $\sigma^* \in \Sigma_n(\varepsilon)$ from
  \begin{equation}\label{e:planted-model-x-dist}
    \nu^\plantedTag_{\beta}(\sigma^*) \propto \E_{G\sim \configurationModel(n,d)} \exp(\beta \isingHamiltonian_{G}(\sigma^*))\,.
  \end{equation}
  This is exactly the distribution described in ~\cref{l:1mt-balanced}~\cref{i:1mt-balanced-subgaussian}.
  Then we draw $G$ from a (disassortative if $\beta > 0$, assortative if $\beta < 0$) stochastic block model with latent communities $\sigma^*$, i.e.
  \begin{equation}\label{e:planted-model-G-dist-given-x}
    \nu^\plantedTag_{\beta}(G | \sigma^*) \propto \exp(\beta \isingHamiltonian_{G}(\sigma^*))\,.
  \end{equation}
  The question of whether we can (weakly) recover $\sigma^*$ given $G$ will play an important role in the analysis of this model.
  In particular, note that $\mu_{G,\beta,<\varepsilon}$ defined in \eqref{e:gibbs-msr} is the posterior of $\sigma^*$ given $G$.
\end{remark}
\begin{remark}\label{r:non-uniform}
  The distribution \eqref{e:planted-model-x-dist} is non-uniform, as the right-hand side expectation depends on the number of $+$ coordinates of $\sigma^*$.
  However, like the uniform distribution, it will concentrate on $\sigma^*$ with about equally many $+$ and $-$ coordinates, as can be seen from \eqref{e:planted-model-x-dist} and ~\cref{l:1mt-balanced}.
\end{remark}
For the proof of ~\cref{p:truncation}~\cref{i:truncation-two-replica}, we reduce to a statement about a slightly more complex planted model, as follows.
\begin{definition}[Two-replica planted model]\label{d:two-replica-planted-model}
  Let $\nu^{\plantedTag,2}_{\beta}$ be the probability measure on $(\sigma^*,{\sigma'}^*,G)$, where $\sigma^*,{\sigma'}^* \in \Sigma_n(\varepsilon)$ and $G$ is a $d$-regular (multi-)graph, with
  \[
    \nu^{\plantedTag,2}_{\beta}(\sigma^*,{\sigma'}^*,G)
    \propto
    \allone\{|R(\sigma^*,{\sigma'}^*)| < \varepsilon\}
    \exp(\beta \isingHamiltonian_{G}(\sigma^*) + \beta \isingHamiltonian_{G}({\sigma'}^*))\,.
  \]
  We will use $\Pr^{\plantedTag,2}_{\beta}$, $\E^{\plantedTag,2}_{\beta}$ to denote probability and expectation with respect to $(\sigma^*,{\sigma'}^*,G) \sim \nu^{\plantedTag,2}_{\beta}$.
\end{definition}
\begin{remark}
  Similarly to Remark~\cref{r:planted-interpretation}, we will think of ~\cref{d:two-replica-planted-model} as a Bayesian inference problem where $\sigma^*$ is the ``signal'' and $({\sigma'}^*,G)$ is the ``observation.''
  Note that $(\sigma^*,{\sigma'}^*)$ follows the law described in ~\cref{l:basic-mt}~\cref{i:basic-mt-annealed-2-subgaussian}, and $\mu_{G,\beta,<\varepsilon}^{2,{\sigma'}^*}$ defined in \eqref{e:gibbs-msr-band} is the posterior of $\sigma^*$ given $({\sigma'}^*,G)$.
\end{remark}
We will reduce ~\cref{p:truncation} to the following.
\begin{proposition}\label{p:truncation-planted}
  There exists $c = c(d,\delta,\varepsilon) > 0$ such that for all $\beta \in [2\beta_\lowerBoundTag, \betaKS - \delta]$, the following holds.
  \begin{enumerate}[label=(\alph*)]
    \item \label{i:truncation-planted-one-replica} We have
    \[
      \E^\plantedTag_{\beta} \mu_{G,\beta,<\varepsilon}(|R(\sigma,\sigma^*)| \ge \varepsilon) \le e^{-cn}\,.
    \]
    \item \label{i:truncation-planted-two-replica} We have
    \[
      \E^{\plantedTag,2}_{\beta} \mu_{G,\beta,<\varepsilon}^{2,{\sigma'}^*}(|R(\sigma,\sigma^*)| \ge \varepsilon) \le e^{-cn}\,.
    \]
  \end{enumerate}
\end{proposition}
\begin{proof}[Proof of ~\cref{p:truncation}, assuming ~\cref{p:truncation-planted}]
  We will prove that ~\cref{p:truncation-planted}~\cref{i:truncation-planted-one-replica} implies ~\cref{p:truncation}~\cref{i:truncation-one-replica}.
  ~\cref{p:truncation-planted}~\cref{i:truncation-planted-two-replica} implies ~\cref{p:truncation}~\cref{i:truncation-two-replica} by an analogous argument.

  Let $c^{\cref{p:truncation-planted}}(d,\delta,\varepsilon)$ be the function from ~\cref{p:truncation-planted}.
  Let $c = c^{\cref{p:truncation-planted}}(d,\delta,\varepsilon)/2$.
  For any $\beta \in [2\beta_\lowerBoundTag, \betaKS - \delta]$, ~\cref{p:truncation-planted}~\cref{i:truncation-planted-one-replica} implies
  \[
    \E^\plantedTag_{\beta} \mu_{G,\beta,<\varepsilon}(|R(\sigma,\sigma^*)| \ge \varepsilon) \le e^{-2cn}\,.
  \]
  Since $\mu_{G,\beta,<\varepsilon}$ is also the posterior of $\sigma^*$ given $G$, $\sigma^*$ is exchangeable with an independent Gibbs sample $\sigma^2$ in the averaged distribution $\E^\plantedTag_{\beta} \mu_{G,\beta,<\varepsilon}$ (this is known as Nishimori's property).
  Thus,
  \begin{equation}\label{e:truncation-reduction}
    \E^\plantedTag_{\beta} \mu_{G,\beta,<\varepsilon}(|R(\sigma^1,\sigma^2)| \ge \varepsilon) \le e^{-2cn}\,.
  \end{equation}
  So,
  \begin{align*}
    \frac{
      \E[\allone\{G \not\in {\mathcal{O}}_{\beta,c}\} Z_{G,< \varepsilon}(\beta)]
    }{
      \E[Z_{G,< \varepsilon}(\beta)]
    }
    &= \frac{
      \sum_{\sigma^* \in \Sigma_n(\varepsilon)}
      \E_{G\sim \configurationModel(n,d)}[
        \allone\{G \not\in {\mathcal{O}}_{\beta,c}\}
        \exp(\beta \isingHamiltonian_{G}(\sigma^*))
      ]
    }{
      \sum_{\sigma^* \in \Sigma_n(\varepsilon)}
      \E_{G\sim \configurationModel(n,d)}[
        \exp(\beta \isingHamiltonian_{G}(\sigma^*))
      ]
    } \\
    &= \Pr^\plantedTag_{\beta}(G \not\in {\mathcal{O}}_{\beta,c})
    = \Pr^\plantedTag_{\beta}\left(
      \mu_{G,\beta,<\varepsilon}(|R(\sigma^1,\sigma^2)| \ge \varepsilon) > e^{-cn}
    \right)\,.
  \end{align*}
  This is bounded by $e^{-cn}$ by \eqref{e:truncation-reduction} and Markov's inequality.
\end{proof}
We now outline the proof of ~\cref{p:truncation-planted}, deferring formal proofs to \cref{ss:planted-boosting-concentration}--\cref{ss:planted-boosting-rigorous}.

\cite[Subsection 4.3]{CLMS22} proves a weaker form of ~\cref{p:truncation-planted}~\cref{i:truncation-planted-one-replica}, showing a bound of $o_n(1)$ in place of $e^{-cn}$.
Their proof is based on the following reduction to a broadcast problem.
The task of recovering a given spin $\sigma^*_v$ from $G$ can only be easier if we additionally reveal the spins $(\sigma^*_w : w \in G \setminus B_R(v))$, where $B_R(v)$ is the radius $R$ neighborhood of $v$ in $G$.
For a large constant $R$, the neighborhood $B_R(v)$ is typically a tree, and the resulting inference problem amounts to a broadcast process on a tree.
In this formulation, $\beta$ plays the role of the signal strength, and reconstruction is possible precisely above $\betaKS$.
Consequently, for $\beta < \betaKS$, a typical spin $\sigma^*_v$ is not recoverable, and the averaged posterior $\E^\plantedTag_{\beta} \mu_{G,\beta,<\varepsilon}$ concentrates on configurations $\sigma$ orthogonal to $\sigma^*$.

By optimizing this proof, one can improve the $o_n(1)$ bound to an explicit inverse polynomial in $n$.
However, it seems difficult to show a bound of $e^{-cn}$ by directly adapting this method, since neighborhoods of $G$ are typically treelike only up to a small logarithmic depth.

The first ingredient in the proof of ~\cref{p:truncation-planted}~\cref{i:truncation-planted-one-replica} is to adapt the method of \cite{CLMS22} in a different direction.
Specifically, we show (see ~\cref{l:planted-augmented} below) that even given $G$ \textbf{and a small constant amount of side information about $\sigma^*$}, it remains impossible to recover $\sigma^*$ with probability more than $o_n(1)$.
Formally, we define the following augmented planted model.
\begin{definition}[Augmented planted model]\label{d:augmented-planted-model}
  For $t\ge 0$, let $\nu^\plantedTag_{\beta,t}$ be the probability measure on $(\sigma^*,G,{\boldsymbol{v}})$, where $\sigma^* \in \Sigma_n(\varepsilon)$, $G$ is a $d$-regular (multi-)graph, and ${\boldsymbol{v}} \in \R^n$, with
  \begin{equation}\label{e:augmented-planted-model}
    \nu^\plantedTag_{\beta,t}(\sigma^*,G,\differential{\boldsymbol{v}})
    \propto
    \exp(\isingHamiltonian_{G,\beta,t}(\sigma^*,{\boldsymbol{v}})) \differential{\boldsymbol{v}}
    \equiv 
    \exp\left(
      \beta \isingHamiltonian_{G}(\sigma^*)
      - \frac12 \|{\boldsymbol{v}} - \sqrt{t} \sigma^*\|^2
    \right) \differential{\boldsymbol{v}}\,.
  \end{equation}
  We will use $\Pr^\plantedTag_{\beta,t}$, $\E^\plantedTag_{\beta,t}$ to denote probability and expectation with respect to $(\sigma^*,G,{\boldsymbol{v}}) \sim \nu^\plantedTag_{\beta,t}$.
\end{definition}
\begin{remark}\label{r:augmented-planted-interpretation}
  As $\int_{\R^n} \exp(-\frac12 \|{\boldsymbol{v}} - \sqrt{t}\sigma^*\|^2) \differential {\boldsymbol{v}}$ does not depend on $(\sigma^*,G)$, the marginal law of $(\sigma^*,G)$ under $\nu^\plantedTag_{\beta,t}$ is $\nu^\plantedTag_{\beta}$; contrast this with the non-uniformity of $\sigma^*$ discussed in Remark~\cref{r:non-uniform}.
  Consequently, $\nu^\plantedTag_{\beta,t}$ can be thought of as an extension of the Bayesian inference problem described in Remark~\cref{r:planted-interpretation}: we draw $(\sigma^*,G) \sim \nu^\plantedTag_{\beta}$, and then ${\boldsymbol{v}}$ conditionally independent of $G$ from
  \[
    \nu^\plantedTag_{\beta,t}({\boldsymbol{v}} | \sigma^*) \propto \exp\left(
      - \frac12 \|{\boldsymbol{v}} - \sqrt{t} \sigma^*\|^2
    \right)\,.
  \]
  Equivalently, ${\boldsymbol{v}}$ can be generated by ${\boldsymbol{v}} = \sqrt{t} \sigma^* + {\boldsymbol{g}}$ for an independent ${\boldsymbol{g}} \sim {\mathcal{N}}(0,I_n)$.
  We write
  \[
    \mu_{G,\beta,t,<\varepsilon}^{{\boldsymbol{v}}}(\sigma) 
    \propto
    \exp(\isingHamiltonian_{G,\beta,t}(\sigma,{\boldsymbol{v}}))\,, \qquad\qquad
    \sigma \in \Sigma_n(\varepsilon)
  \]
  for the posterior of $\sigma^*$ given $(G,{\boldsymbol{v}})$.
  Note that $\mu_{G,\beta,0,<\varepsilon}^{{\boldsymbol{v}}} = \mu_{G,\beta,<\varepsilon}$.
\end{remark}
For ~\cref{p:truncation-planted}~\cref{i:truncation-planted-two-replica}, we define an analogous augmented model.
\begin{definition}[Two-replica augmented planted model]\label{d:two-replica-augmented-planted-model}
  For $t\ge 0$, let $\nu^{\plantedTag,2}_{\beta,t}$ be the probability measure on $(\sigma^*,{\sigma'}^*,G,{\boldsymbol{v}})$, where $\sigma^*,{\sigma'}^* \in \Sigma_n(\varepsilon)$, $G$ is a $d$-regular (multi-)graph, and ${\boldsymbol{v}} \in \R^n$, with
  \begin{align*}
    \nu^{\plantedTag,2}_{\beta,t}(\sigma^*,{\sigma'}^*,G,\differential{\boldsymbol{v}})
    &\propto
    \allone\{|R(\sigma^*,{\sigma'}^*)| < \varepsilon\}
    \exp\left(
      \beta \isingHamiltonian_{G}(\sigma^*)
      + \beta \isingHamiltonian_{G}({\sigma'}^*)
      - \frac12 \|{\boldsymbol{v}} - \sqrt{t} \sigma^*\|^2
    \right) \differential{\boldsymbol{v}} \\
    &= \allone\{|R(\sigma^*,{\sigma'}^*)| < \varepsilon\}
    \exp\left(
      \isingHamiltonian_{G,\beta,t}(\sigma^*,{\boldsymbol{v}})
      + \beta \isingHamiltonian_{G}({\sigma'}^*)
    \right) \differential{\boldsymbol{v}} 
  \end{align*}
  We will use $\Pr^{\plantedTag,2}_{\beta,t}$, $\E^{\plantedTag,2}_{\beta,t}$ to denote probability and expectation with respect to $(\sigma^*,{\sigma'}^*,G,{\boldsymbol{v}}) \sim \nu^{\plantedTag,2}_{\beta,t}$.
  We also write
  \[
    \mu_{G,\beta,t,<\varepsilon}^{2,{\sigma'}^*,{\boldsymbol{v}}}(\sigma) 
    \propto
    \allone\{|R(\sigma,{\sigma'}^*)| < \varepsilon\}
    \exp(\isingHamiltonian_{G,\beta,t}(\sigma,{\boldsymbol{v}}))\,,
    \qquad\qquad
    \sigma \in \Sigma_n(\varepsilon)
  \]
  for the posterior of $\sigma^*$ given $({\sigma'}^*,G,{\boldsymbol{v}})$.
\end{definition}
We will prove the following lemma in Appendix~\cref{s:planted-augmented} by a direct adaptation of the method of \cite{CLMS22}.
\begin{lemma}\label{l:planted-augmented}
  For all $\varepsilon_1 > 0$, there exists $t_0 = t_0(d,\delta,\varepsilon,\varepsilon_1) > 0$ such that the following holds.
  \begin{enumerate}[label=(\alph*)]
    \item \label{i:planted-augmented-one-replica} We have
    \[
      \sup_{\beta \in [2\beta_\lowerBoundTag, \betaKS - \delta]}
      \sup_{t\in [0,t_0]}
      \E^\plantedTag_{\beta,t} \mu_{G,\beta,t,<\varepsilon}^{{\boldsymbol{v}}}(|R(\sigma,\sigma^*)| \ge \varepsilon_1) \le \varepsilon_1\,.
    \]
    \item \label{i:planted-augmented-two-replica} We have
    \[
      \sup_{\beta \in [2\beta_\lowerBoundTag, \betaKS - \delta]}
      \sup_{t\in [0,t_0]}
      \E^{\plantedTag,2}_{\beta,t} \mu_{G,\beta,t,<\varepsilon}^{2,{\sigma'}^*,{\boldsymbol{v}}}(|R(\sigma,\sigma^*)| \ge \varepsilon_1) \le \varepsilon_1\,.
    \]
  \end{enumerate}
\end{lemma}
The second ingredient in the proof of ~\cref{p:truncation-planted} is ~\cref{p:planted-boosting} below.
This states that if $\E^\plantedTag_{\beta} \mu_{G,\beta,<\varepsilon}$ places more than $e^{-cn}$ probability on overlaps away from $0$, for a small $c > 0$, then with a small amount of side information $t$, $\E^\plantedTag_{\beta,t} \mu_{G,\beta,t,<\varepsilon}^{{\boldsymbol{v}}}$ concentrates on overlaps away from $0$ (and similarly for $\E^{\plantedTag,2}_{\beta} \mu_{G,\beta,<\varepsilon}^{2,{\sigma'}^*}$ and $\E^{\plantedTag,2}_{\beta,t} \mu_{G,\beta,t,<\varepsilon}^{2,{\sigma'}^*,{\boldsymbol{v}}}$).
\begin{proposition}\label{p:planted-boosting}
  For any $\upsilon > 0$, there exists $t_\upperBoundTag = t_\upperBoundTag(d,\varepsilon,\upsilon) > 0$ such that the following holds for all $\beta \in [2\beta_\lowerBoundTag, \betaKS - \delta]$.
  \begin{enumerate}[label=(\alph*)]
    \item\label{i:planted-boosting-one-replica} If 
    \begin{equation}\label{e:planted-boosting-hypothesis}
      \E^\plantedTag_{\beta} \mu_{G,\beta,<\varepsilon}(|R(\sigma,\sigma^*)| \ge \varepsilon) \ge e^{-\upsilon n}\,,
    \end{equation}
    then there exists $t \in [0, t_\upperBoundTag]$ such that
    \begin{equation}\label{e:planted-boosting-conclusion}
      \E^\plantedTag_{\beta,t} \mu_{G,\beta,t,<\varepsilon}^{{\boldsymbol{v}}}(|R(\sigma,\sigma^*)| \ge \varepsilon/10) = 1 - o_n(1)\,.
    \end{equation}
    \item\label{i:planted-boosting-two-replica} If
    \begin{equation}\label{e:planted-boosting-hypothesis-two-replica}
      \E^{\plantedTag,2}_{\beta} \mu_{G,\beta,<\varepsilon}^{2,{\sigma'}^*}(|R(\sigma,\sigma^*)| \ge \varepsilon) \ge e^{-\upsilon n}\,,
    \end{equation}
    then there exists $t \in [0, t_\upperBoundTag]$ such that
    \begin{equation}\label{e:planted-boosting-conclusion-two-replica}
      \E^{\plantedTag,2}_{\beta,t} \mu_{G,\beta,t,<\varepsilon}^{2,{\sigma'}^*,{\boldsymbol{v}}}(|R(\sigma,\sigma^*)| \ge \varepsilon/10) = 1 - o_n(1)\,.
    \end{equation}
  \end{enumerate}
  Finally, the function $t_\upperBoundTag$ satisfies, for any fixed $d,\varepsilon$,
  \[
    \lim_{\upsilon\downarrow 0} t_\upperBoundTag(d,\varepsilon,\upsilon) = 0\,.
  \]
\end{proposition}
~\cref{p:planted-boosting} is the main technical contribution of this section.
We provide an overview of its proof in \cref{ss:planted-boosting-overview}, followed by a rigorous proof in \cref{ss:planted-boosting-concentration}--\cref{ss:planted-boosting-rigorous}.
From the above, ~\cref{p:truncation-planted} follows easily.
\begin{proof}[Proof of ~\cref{p:truncation-planted}, assuming ~\cref{l:planted-augmented} and ~\cref{p:planted-boosting}]
  We will prove that ~\cref{l:planted-augmented}~\cref{i:planted-augmented-one-replica} and ~\cref{p:planted-boosting}~\cref{i:planted-boosting-one-replica} imply ~\cref{p:truncation-planted}~\cref{i:truncation-planted-one-replica}.
  ~\cref{l:planted-augmented}~\cref{i:planted-augmented-two-replica} and ~\cref{p:planted-boosting}~\cref{i:planted-boosting-two-replica} imply ~\cref{p:truncation-planted}~\cref{i:truncation-planted-two-replica} by an analogous argument.

  Let $t_0 = t_0(d,\delta,\varepsilon,\varepsilon/10)$ be given by ~\cref{l:planted-augmented}. 
  Let $t_\upperBoundTag$ be the function given by ~\cref{p:planted-boosting}, and let $\upsilon > 0$ be small enough that $t_\upperBoundTag(d,\varepsilon,\upsilon) \le t_0$.
  We will show the conclusion holds with $c = \upsilon$.
  Suppose for contradiction that there exists $\beta \in [2\beta_\lowerBoundTag, \betaKS - \delta]$ such that
  \[
    \E^\plantedTag_{\beta} \mu_{G,\beta,<\varepsilon} \left(|R(\sigma,\sigma^*)| \ge \varepsilon\right) > e^{-\upsilon n}\,.
  \]
  Applying ~\cref{p:planted-boosting}~\cref{i:planted-boosting-one-replica} with this $\upsilon$ yields that for some $t \le t_\upperBoundTag(d,\varepsilon,\upsilon) \le t_0$, we have
  \[
    \E^\plantedTag_{\beta,t} \mu_{G,\beta,t,<\varepsilon}^{{\boldsymbol{v}}}(|R(\sigma,\sigma^*)| \ge \varepsilon/10) = 1 - o_n(1)\,.
  \]
  However, ~\cref{l:planted-augmented}~\cref{i:planted-augmented-one-replica} implies that
  \[
    \E^\plantedTag_{\beta,t} \mu_{G,\beta,t,<\varepsilon}^{{\boldsymbol{v}}}(|R(\sigma,\sigma^*)| \ge \varepsilon/10) \le \varepsilon/10\,.
  \]
  Since ~\cref{l:basic-mt} provides $\varepsilon_0 < 1$, this bound is $\le 1/10$, which contradicts the previous display.
\end{proof}

\subsection{Strong overlap bound from side information: overview}
\label{ss:planted-boosting-overview}

In this subsection we outline the main ideas of the proof of ~\cref{p:planted-boosting}.
To simplify the presentation we focus on part \cref{i:planted-boosting-one-replica}.
For an interval $I \subseteq [-1,1]$, define
\[
  \Sigma_n(\sigma^*,I) = \left\{
    \sigma \in \Sigma_n(\varepsilon) : R(\sigma,\sigma^*) \in I
  \right\}\,,
\]
and the restricted partition function and free energy
\begin{align}\label{e:restricted-Z-and-F-to-I}
  Z_{G,\beta,t}^{\sigma^*,{\boldsymbol{v}}}(I)
  &= \sum_{\sigma \in \Sigma_n(\sigma^*,I)}
  \exp(\isingHamiltonian_{G,\beta,t}(\sigma,{\boldsymbol{v}}))\,, &
  F_{n,\beta,t}^{\sigma^*,G,{\boldsymbol{v}}}(I)
  &= \frac{1}{n} \log Z_{G,\beta,t}^{\sigma^*,{\boldsymbol{v}}}(I)\,,
\end{align}
where $\isingHamiltonian_{G,\beta,t}$ is defined in \eqref{e:augmented-planted-model}.
All the intervals $I$ we consider will include a feasible value of $R(\sigma,\sigma^*)$, so that these are well-defined.
Roughly, the hypothesis \eqref{e:planted-boosting-hypothesis} of ~\cref{p:planted-boosting}~\cref{i:planted-boosting-one-replica} states that the planted model with parameters $(\beta,0)$ puts nearly as much free energy outside $(-\varepsilon,\varepsilon)$ as within it:
\begin{equation}\label{e:fe-time-zero}
  \E^\plantedTag_{\beta,0} F_{n,\beta,0}^{\sigma^*,G,{\boldsymbol{v}}}([-1,1] \setminus (-\varepsilon,\varepsilon))
  \ge \E^\plantedTag_{\beta,0} F_{n,\beta,0}^{\sigma^*,G,{\boldsymbol{v}}}((-\varepsilon,\varepsilon)) - \upsilon\,.
\end{equation}
While this is not exactly equivalent to \eqref{e:planted-boosting-hypothesis}, which is formulated in terms of the Gibbs measure, we will see that it is equivalent up to some adjustment in $\upsilon$ because the free energies $F_{n,\beta,t}^{\sigma^*,G,{\boldsymbol{v}}}(I)$ concentrate well (see Lemmas~\cref{l:planted-fe-conc} and \cref{l:boosting-init}).
To show the conclusion \eqref{e:planted-boosting-conclusion} of ~\cref{p:planted-boosting}~\cref{i:planted-boosting-one-replica}, it similarly suffices to show that in the planted model with parameters $(\beta,t)$, the free energy restricted to $(-\varepsilon/10,\varepsilon/10)$ is smaller than the overall free energy:
\begin{equation}\label{e:fe-overtake}
  \E^\plantedTag_{\beta,t} F_{n,\beta,t}^{\sigma^*,G,{\boldsymbol{v}}}([-1,1])
  \ge \E^\plantedTag_{\beta,t} F_{n,\beta,t}^{\sigma^*,G,{\boldsymbol{v}}}((-\varepsilon/10,\varepsilon/10)) + \upsilon\,.
\end{equation}
In words, we will show that adding a small amount of side information causes the free energy on overlaps in $[\varepsilon,1]$ to grow faster than in $(-\varepsilon/10,\varepsilon/10)$.
Formally, let us define the Gibbs measure associated to $Z_{G,\beta,t}^{\sigma^*,{\boldsymbol{v}}}(I)$ by
\[
  \mu_{G,\beta,t,<\varepsilon}^{{\boldsymbol{v}}}(\sigma \,|\, I)
  = \frac{1}{Z_{G,\beta,t}^{\sigma^*,{\boldsymbol{v}}}(I)} \allone\{\sigma \in \Sigma_n(\sigma^*,I)\}
  \exp(\isingHamiltonian_{G,\beta,t}(\sigma,{\boldsymbol{v}}))\,.
\]
As the notation suggests, $\sigma \sim \mu_{G,\beta,t,<\varepsilon}^{{\boldsymbol{v}}}(\cdot \,|\, I)$ is distributed as $\sigma \sim \mu_{G,\beta,t,<\varepsilon}^{{\boldsymbol{v}}}$ conditioned on $\sigma \in \Sigma_n(\sigma^*,I)$.
By slight abuse of notation, for a measurable function $f$ of $k$ replicas $\sigma^1,\ldots,\sigma^k$, we write 
\begin{equation}\label{e:gibbs-avg-of-function}
  \mu_{G,\beta,t,<\varepsilon}^{{\boldsymbol{v}}}(f \,|\, I) = \E\left(
    f(\sigma^1,\ldots,\sigma^k) : \sigma^i \stackrel{IID}{\sim} \mu_{G,\beta,t,<\varepsilon}^{{\boldsymbol{v}}}(\cdot \,|\, I)
  \right)\,,
\end{equation}
and write $\sigma$ for $\sigma^1$ if $k=1$.
The following lemma is proved by a standard gaussian integration by parts calculation, and controls the derivatives of the restricted free energies.
We provide its proof at the end of this subsection.
\begin{lemma}\label{l:fe-deriv}
  For any $\beta \in \R$, $t > 0$, and interval $I \subseteq [-1,1]$, we have
  \[
    \frac{\differential}{\differential t} \E^\plantedTag_{\beta,t} F_{n,\beta,t}^{\sigma^*,G,{\boldsymbol{v}}}(I)
    = \frac12 \left(
      -1 + 2\E^\plantedTag_{\beta,t} \mu_{G,\beta,t,<\varepsilon}^{{\boldsymbol{v}}}(R(\sigma,\sigma^*) \,|\, I)
      - \E^\plantedTag_{\beta,t} \mu_{G,\beta,t,<\varepsilon}^{{\boldsymbol{v}}}(R(\sigma^1,\sigma^2) \,|\, I)
    \right)\,.
  \]
\end{lemma}
When $I = (-\varepsilon/10,\varepsilon/10)$, this derivative can be upper bounded by
\[
  \frac{\differential}{\differential t} \E^\plantedTag_{\beta,t} F_{n,\beta,t}^{\sigma^*,G,{\boldsymbol{v}}}((-\varepsilon/10,\varepsilon/10))
  \le \frac12 (-1 + \varepsilon / 5)\,,
\]
as clearly
\[
  \mu_{G,\beta,t,<\varepsilon}^{{\boldsymbol{v}}}(R(\sigma,\sigma^*) \,|\, (-\varepsilon/10,\varepsilon/10)) \le \varepsilon/10\,,
\]
while the term involving $R(\sigma^1,\sigma^2)$ is nonnegative (see Fact~\cref{f:replica-overlap-nonnegative}).
For $I \subseteq [\varepsilon,1]$, the term
\[
  2\mu_{G,\beta,t,<\varepsilon}^{{\boldsymbol{v}}}(R(\sigma,\sigma^*) \,|\, I)
  \ge 2\min(I)
\]
from ~\cref{l:fe-deriv} provides a significant positive contribution.
So, $\E^\plantedTag_{\beta,t} F_{n,\beta,t}^{\sigma^*,G,{\boldsymbol{v}}}(I)$ will grow faster than $\E^\plantedTag_{\beta,t} F_{n,\beta,t}^{\sigma^*,G,{\boldsymbol{v}}}((-\varepsilon/10,\varepsilon/10))$, \textbf{unless} the term involving $R(\sigma^1,\sigma^2)$ in ~\cref{l:fe-deriv} contributes enough to cancel the contribution from the $R(\sigma,\sigma^*)$ term.
For this to occur we must have
\[
  \E^\plantedTag_{\beta,t} \mu_{G,\beta,t,<\varepsilon}^{{\boldsymbol{v}}}\left(
    R(\sigma^1,\sigma^2) \ge \frac32 \min(I) \,\Big|\, I
  \right) = \Omega(1)\,.
\]
(The number $\frac32$ is arbitrary here, and any constant in $(1,2)$ will suffice --- the main point is that the $R(\sigma,\sigma^*)$ term in ~\cref{l:fe-deriv} has a coefficient of $2$ while the $R(\sigma^1,\sigma^2)$ term doesn't, so the $R(\sigma^1,\sigma^2)$ term needs to be bigger by a constant factor to cancel the $R(\sigma,\sigma^*)$ term.) However,
\begin{equation}
  \label{e:planted-boosting-conditional-probability}
  \E^\plantedTag_{\beta,t} \mu_{G,\beta,t,<\varepsilon}^{{\boldsymbol{v}}}\left(
    R(\sigma^1,\sigma^2) \ge \frac32 \min(I) \,\Big|\, I
  \right)
  \le \E^\plantedTag_{\beta,t} \frac{
    \mu_{G,\beta,t,<\varepsilon}^{{\boldsymbol{v}}}(R(\sigma^1,\sigma^2) \ge \frac32 \min(I))
  }{
    \mu_{G,\beta,t,<\varepsilon}^{{\boldsymbol{v}}}(R(\sigma,\sigma^*) \in I)^2
  }\,,
\end{equation}
and the numerator can be controlled in expectation by Nishimori's property:
\begin{equation}\label{e:planted-boosting-nishimori}
  \E^\plantedTag_{\beta,t}
  \mu_{G,\beta,t,<\varepsilon}^{{\boldsymbol{v}}}\left(
    R(\sigma^1,\sigma^2) \ge \frac32 \min(I)
  \right)
  = \E^\plantedTag_{\beta,t}
  \mu_{G,\beta,t,<\varepsilon}^{{\boldsymbol{v}}}\left(
    R(\sigma,\sigma^*) \ge \frac32 \min(I)
  \right)\,.
\end{equation}
By translating from Gibbs probabilities to free energies, the following claim can be made precise.
\begin{claim}[Informal, see ~\cref{p:boosting-core}]\label{cl:boosting-strategy}
  Suppose that for $t > 0$, small $\upsilon' > 0$, and $I$ of the form $I = [x,1]$ where $x\ge \varepsilon$,
  \begin{equation}\label{e:fe-nearly-overtake}
    \E^\plantedTag_{\beta,t} F_{n,\beta,t}^{\sigma^*,G,{\boldsymbol{v}}}((-\varepsilon/10,\varepsilon/10))
    - \E^\plantedTag_{\beta,t} F_{n,\beta,t}^{\sigma^*,G,{\boldsymbol{v}}}(I)
    \le \upsilon'\,,
  \end{equation}
  and that \eqref{e:fe-overtake} does not already hold for this $t$.
  Then, either:
  \begin{enumerate}[label=(\roman*)]
    \item \label{i:informal-fe-make-progress} We have
    \[
      \frac{\differential}{\differential t} \left(
        \E^\plantedTag_{\beta,t} F_{n,\beta,t}^{\sigma^*,G,{\boldsymbol{v}}}((-\varepsilon/10,\varepsilon/10))
        - \E^\plantedTag_{\beta,t} F_{n,\beta,t}^{\sigma^*,G,{\boldsymbol{v}}}(I)
      \right) \le -\varepsilon/10\,,
    \]
    in which case we make progress towards \eqref{e:fe-overtake}, or
    \item \label{i:informal-fe-increment-k} \eqref{e:fe-nearly-overtake} holds with $I' = [\frac32 x, 1]$ in place of $I$ and some $\upsilon''$ in place of $\upsilon'$, which satisfies $\upsilon'' \downarrow 0$ as $\upsilon' \downarrow 0$.
  \end{enumerate}
\end{claim}
Indeed, if \eqref{e:fe-nearly-overtake} holds and \eqref{e:fe-overtake} does not hold, then
\[
  \E^\plantedTag_{\beta,t} F_{n,\beta,t}^{\sigma^*,G,{\boldsymbol{v}}}(I)
  \ge \E^\plantedTag_{\beta,t} F_{n,\beta,t}^{\sigma^*,G,{\boldsymbol{v}}}([-1,1]) - \upsilon - \upsilon'\,.
\]
By free energy concentration, this implies that with high probability, the denominator of \eqref{e:planted-boosting-conditional-probability} is not too small.
If the upper bound \eqref{e:planted-boosting-conditional-probability} is $o_n(1)$, the above arguments show that outcome \cref{i:informal-fe-make-progress} holds.
Otherwise, the numerator of \eqref{e:planted-boosting-conditional-probability} must also be not too small.
In light of \eqref{e:planted-boosting-nishimori} and free energy concentration, this implies outcome \cref{i:informal-fe-increment-k}.
Due to the various translations between Gibbs probabilities and free energies, the $\upsilon''$ we obtain will be significantly larger than $\upsilon'$, but this will be immaterial to the argument below.

Claim~\cref{cl:boosting-strategy} leads to the following dynamical argument.
For $0\le k \le k_\upperBoundTag \equiv \lfloor (\log \frac{1}{\varepsilon}) / (\log \frac32) \rfloor$, let
\begin{align}\label{e:xk}
  x_k &= \min((3/2)^k \varepsilon, 0.9)\,, &
  I_k = [x_k, 1]\,.
\end{align}
At each time $t\ge 0$, there is an ``active'' $k = k(t)$, which is monotonically increasing with initialization $k(0) = 0$.
Let
\[
  \phi(t) = 
  \E^\plantedTag_{\beta,t} F_{n,\beta,t}^{\sigma^*,G,{\boldsymbol{v}}}((-\varepsilon/10,\varepsilon/10))
  - \E^\plantedTag_{\beta,t} F_{n,\beta,t}^{\sigma^*,G,{\boldsymbol{v}}}(I_{k(t)})\,.
\]
At $t=0$, the hypothesis \eqref{e:fe-time-zero} (and symmetry of $\isingHamiltonian_{G}(\sigma)$ under negation of $\sigma$) implies $\phi(0) \le \upsilon$.
At any $t$ such that \eqref{e:planted-boosting-conditional-probability} is $o_n(1)$, Claim~\cref{cl:boosting-strategy} shows that (either \eqref{e:fe-overtake} holds, in which case we are done, or) $\phi'(t) \le -\varepsilon/10$.
If we ever reach a $t$ such that \eqref{e:planted-boosting-conditional-probability} is not $o_n(1)$, we increment $k(t)$.
This can cause $\phi(t)$ to increase discontinuously, by an amount controlled by \cref{i:informal-fe-increment-k} above.

Aside from at most $k_\upperBoundTag$ such discontinuities, $\phi$ is decreasing with derivative $\le -\varepsilon/10$.
We can therefore ensure $\phi(t) \le -\upsilon$ at some $t \in [0,t_\upperBoundTag]$, for suitable $t_\upperBoundTag$. 
This implies \eqref{e:fe-overtake}.
\begin{remark}
  In \eqref{e:xk}, we have $x_k = (3/2)^k \varepsilon$ for all $k \le k_{\upperBoundTag} - 1$; the minimum with $0.9$ is relevant only for $x_{k_\upperBoundTag}$. 
  The minimum with $0.9$ is for technical convenience, as our proof of the concentration of $F_{n,\beta,t}^{\sigma^*,G,{\boldsymbol{v}}}(I)$ (~\cref{l:planted-fe-conc}, see also ~\cref{d:cI}) requires the interval $I$ to not have an endpoint in $\pm (0.9,1)$.
  Note that we always have $[x_{k_\upperBoundTag},1] \supseteq [\frac32 x_{k_\upperBoundTag - 1}, 1]$, so that \cref{i:informal-fe-increment-k} still controls the discontinuous increase of $\phi(t)$ in this step.
\end{remark}
\begin{remark}
  Due to the blowup in $\upsilon'$ each time \cref{i:informal-fe-increment-k} is invoked, the $t_\upperBoundTag$ we ultimately obtain tends to $0$ very slowly as $\upsilon \downarrow 0$.
  Tracing through the proof below shows that it scales as
  \[
    t_\upperBoundTag \approx \upsilon^{(1/2)^{k_\upperBoundTag}} \approx \upsilon^{\varepsilon^{(\log 2) / \log (3/2)}}\,.
  \]
\end{remark}
We conclude this subsection with the proof of the key free energy derivative formula.
\begin{proof}[Proof of ~\cref{l:fe-deriv}]
  Recall from Remark~\cref{r:augmented-planted-interpretation} that $(\sigma^*,G,{\boldsymbol{v}}) \sim \nu^\plantedTag_{\beta,t}$ can be sampled by $(\sigma^*,G) \sim \nu^\plantedTag_{\beta}$ and ${\boldsymbol{v}} = \sqrt{t}\sigma^* + {\boldsymbol{g}}$ for an independent ${\boldsymbol{g}} \sim {\mathcal{N}}(0,I_n)$.
  Thus
  \begin{align*}
    &\E^\plantedTag_{\beta,t} F_{n,\beta,t}^{\sigma^*,G,{\boldsymbol{v}}}(I)
    = \E^\plantedTag_{\beta} \E_{\boldsymbol{g}} F_{n,\beta,t}^{\sigma^*,G,\sqrt{t}\sigma^*+{\boldsymbol{g}}}(I) \\
    &= \frac{1}{n} \E^\plantedTag_{\beta} \E_{\boldsymbol{g}} 
    \log \sum_{\sigma \in \Sigma_n(\sigma^*,I)} \exp\left(
      \beta \isingHamiltonian_{G}(\sigma) 
      - \frac12 \|\sqrt{t} \sigma^* + {\boldsymbol{g}} - \sqrt{t} \sigma\|^2
    \right) \\
    &= \frac{1}{n} \E^\plantedTag_{\beta} \E_{\boldsymbol{g}} 
    \log \sum_{\sigma \in \Sigma_n(\sigma^*,I)} \exp\left(
      \beta \isingHamiltonian_{G}(\sigma) 
      - \frac{t}{2} \|\sigma^* - \sigma\|^2
      - \sqrt{t} \langle \sigma^* - \sigma, {\boldsymbol{g}}\rangle
      - \frac12 \|{\boldsymbol{g}}\|^2
    \right) \\
    &= - t - \frac12 + \frac{1}{n} \E^\plantedTag_{\beta} \E_{\boldsymbol{g}} 
    \log \sum_{\sigma \in \Sigma_n(\sigma^*,I)} \exp\left(
      \beta \isingHamiltonian_{G}(\sigma) 
      + t \langle \sigma,\sigma^*\rangle
      + \sqrt{t} \langle \sigma - \sigma^*, {\boldsymbol{g}}\rangle
    \right)\,,
  \end{align*}
  where the last step uses that $\|\sigma\|^2 = \|\sigma^*\|^2 = n$ deterministically and $\E[\|{\boldsymbol{g}}\|^2] = n$.
  Hence,
  \[
    \frac{\differential}{\differential t} \E^\plantedTag_{\beta,t} F_{n,\beta,t}^{\sigma^*,G,{\boldsymbol{v}}}(I)
    = -1 + {\textup{(i)}} + {\textup{(ii)}}\,,
  \]
  where
  \begin{align*}
    {\textup{(i)}} &= \frac{1}{n} \E^\plantedTag_{\beta} \E_{\boldsymbol{g}} \frac{
      \sum_{\sigma \in \Sigma_n(\sigma^*,I)}
      \langle \sigma,\sigma^*\rangle
      \exp(
        \beta \isingHamiltonian_{G}(\sigma) 
        + t \langle \sigma,\sigma^*\rangle
        + \sqrt{t} \langle \sigma - \sigma^*, {\boldsymbol{g}}\rangle
      )
    }{
      \sum_{\sigma \in \Sigma_n(\sigma^*,I)}
      \exp(
        \beta \isingHamiltonian_{G}(\sigma) 
        + t \langle \sigma,\sigma^*\rangle
        + \sqrt{t} \langle \sigma - \sigma^*, {\boldsymbol{g}}\rangle
      )
    }\,, \\
    {\textup{(ii)}} &= \frac{1}{n} \cdot \frac{1}{2\sqrt{t}} \E^\plantedTag_{\beta} \E_{\boldsymbol{g}} \frac{
      \sum_{\sigma \in \Sigma_n(\sigma^*,I)}
      \langle \sigma - \sigma^*, {\boldsymbol{g}}\rangle
      \exp(
        \beta \isingHamiltonian_{G}(\sigma) 
        + t \langle \sigma,\sigma^*\rangle
        + \sqrt{t} \langle \sigma - \sigma^*, {\boldsymbol{g}}\rangle
      )
    }{
      \sum_{\sigma \in \Sigma_n(\sigma^*,I)}
      \exp(
        \beta \isingHamiltonian_{G}(\sigma) 
        + t \langle \sigma,\sigma^*\rangle
        + \sqrt{t} \langle \sigma - \sigma^*, {\boldsymbol{g}}\rangle
      )
    }\,.
  \end{align*}
  Inspecting the relevant definitions reveals
  \begin{align*}
     {\textup{(i)}} &= \E^\plantedTag_{\beta} \E_{\boldsymbol{g}} \frac{
      \sum_{\sigma \in \Sigma_n(\sigma^*,I)}
      R(\sigma,\sigma^*)
      \exp(
        \beta \isingHamiltonian_{G}(\sigma) 
        - \frac12 \|\sqrt{t} (\sigma^* - \sigma) + {\boldsymbol{g}}\|^2
      )
    }{
      \sum_{\sigma \in \Sigma_n(\sigma^*,I)}
      \exp(
        \beta \isingHamiltonian_{G}(\sigma) 
        - \frac12 \|\sqrt{t} (\sigma^* - \sigma) + {\boldsymbol{g}}\|^2
      )
    } \\
    &= \E^\plantedTag_{\beta} \E_{\boldsymbol{g}} \mu_{G,\beta,t,<\varepsilon}^{\sqrt{t}\sigma^* + {\boldsymbol{g}}}(R(\sigma,\sigma^*) \,|\, I)
    = \E^\plantedTag_{\beta,t} \mu_{G,\beta,t,<\varepsilon}^{{\boldsymbol{v}}}(R(\sigma,\sigma^*) \,|\, I)\,.
  \end{align*}
  We can evaluate ${\textup{(ii)}}$ by a routine gaussian integration by parts calculation, see e.g. \cite[Lemma 1.4]{Pan13b}.
  Abbreviating
  \begin{align*}
    \widetilde{\isingHamiltonian}_{G,\beta,t}(\sigma;\sigma^*,{\boldsymbol{g}}) 
    &= \beta \isingHamiltonian_{G}(\sigma) - \frac12 \|\sqrt{t} \sigma^* + {\boldsymbol{g}} - \sqrt{t} \sigma\|^2
    + tn + \frac12 \|{\boldsymbol{g}}\|^2 \\
    &= \beta \isingHamiltonian_{G}(\sigma) + t \langle \sigma,\sigma^*\rangle + \sqrt{t} \langle \sigma - \sigma^*, {\boldsymbol{g}}\rangle\,,
  \end{align*}
  we have
  \begin{align*}
    {\textup{(ii)}} &= \frac{1}{2n} \E^\plantedTag_{\beta} \E_{\boldsymbol{g}}\frac{
      \sum_{\sigma \in \Sigma_n(\sigma^*,I)}
      \langle \sigma - \sigma^*, \sigma - \sigma^*\rangle
      \exp(\widetilde{\isingHamiltonian}_{G,\beta,t}(\sigma;\sigma^*,{\boldsymbol{g}}))
    }{
      \sum_{\sigma \in \Sigma_n(\sigma^*,I)}
      \exp(\widetilde{\isingHamiltonian}_{G,\beta,t}(\sigma;\sigma^*,{\boldsymbol{g}}))
    } \\
    &\quad - \frac{1}{2n} \E^\plantedTag_{\beta} \E_{\boldsymbol{g}} \frac{
      \sum_{\sigma^1,\sigma^2 \in \Sigma_n(\sigma^*,I)}
      \langle \sigma^1 - \sigma^*, \sigma^2 - \sigma^*\rangle
      \exp(\widetilde{\isingHamiltonian}_{G,\beta,t}(\sigma^1;\sigma^*,{\boldsymbol{g}}) + \widetilde{\isingHamiltonian}_{G,\beta,t}(\sigma^2;\sigma^*,{\boldsymbol{g}}))
    }{
      \sum_{\sigma^1,\sigma^2 \in \Sigma_n(\sigma^*,I)}
      \exp(\widetilde{\isingHamiltonian}_{G,\beta,t}(\sigma^1;\sigma^*,{\boldsymbol{g}}) + \widetilde{\isingHamiltonian}_{G,\beta,t}(\sigma^2;\sigma^*,{\boldsymbol{g}}))
    } \\
    &= \frac12 \E^\plantedTag_{\beta} \E_{\boldsymbol{g}} \mu_{G,\beta,t,<\varepsilon}^{\sqrt{t}\sigma^* + {\boldsymbol{g}}}\left(
      R(\sigma - \sigma^*, \sigma - \sigma^*)
      - R(\sigma^1 - \sigma^*, \sigma^2 - \sigma^*) \,|\, I
    \right) \\
    &= \frac12 \left(
      1 - \E^\plantedTag_{\beta} \E_{\boldsymbol{g}} 
      \mu_{G,\beta,t,<\varepsilon}^{\sqrt{t}\sigma^* + {\boldsymbol{g}}}(R(\sigma^1,\sigma^2) \,|\, I)
    \right) = \frac12 \left(
      1 - \E^\plantedTag_{\beta,t} 
      \mu_{G,\beta,t,<\varepsilon}^{{\boldsymbol{v}}}(R(\sigma^1,\sigma^2) \,|\, I)
    \right)\,.
  \end{align*}
  Combining the above identities yields the result.
\end{proof}

\subsection{Concentration of restricted partition functions}
\label{ss:planted-boosting-concentration}

Recall $Z_{G,\beta,t}^{\sigma^*,{\boldsymbol{v}}}(I)$ and $F_{n,\beta,t}^{\sigma^*,G,{\boldsymbol{v}}}(I)$ defined in \eqref{e:restricted-Z-and-F-to-I}.
Analogously define the two-replica versions
\begin{align}\label{e:restricted-Z-and-F-to-I-two-replica}
  Z_{G,\beta,t}^{2,\sigma^*,{\sigma'}^*,{\boldsymbol{v}}}(I)
  &= \sum_{\sigma \in \Sigma_n(\sigma^*,I)}
  \allone\{|R(\sigma,{\sigma'}^*)| < \varepsilon\}
  \exp(\isingHamiltonian_{G,\beta,t}(\sigma,{\boldsymbol{v}}))\,, \\
  F_{n,\beta,t}^{2,\sigma^*,{\sigma'}^*,G,{\boldsymbol{v}}}(I)
  &= \frac{1}{n} \log Z_{G,\beta,t}^{2,\sigma^*,{\sigma'}^*,{\boldsymbol{v}}}(I)\,.
\end{align}
Similarly to \eqref{e:gibbs-avg-of-function}, for a measurable function $f$ of $k$ replicas $\sigma^1,\ldots,\sigma^k$, write
\begin{equation}\label{e:gibs-avg-of-function-two-replica}
  \mu_{G,\beta,t,<\varepsilon}^{2,{\sigma'}^*,{\boldsymbol{v}}}(f \,|\, I)
  =
  \E\left(
    f(\sigma^1,\ldots,\sigma^k) : \sigma^i \stackrel{IID}{\sim}
    \mu_{G,\beta,t,<\varepsilon}^{2,{\sigma'}^*,{\boldsymbol{v}}}(\cdot \,|\, I)
  \right)\,,
\end{equation}
where $\mu_{G,\beta,t,<\varepsilon}^{2,{\sigma'}^*,{\boldsymbol{v}}}(\cdot \,|\, I)$ is the Gibbs measure corresponding to $Z_{G,\beta,t}^{2,\sigma^*,{\sigma'}^*,{\boldsymbol{v}}}(I)$.
\textbf{In this subsection, we prove ~\cref{l:planted-fe-conc} below, which gives concentration of these restricted free energies.}
These are an important input to the dynamical proof of ~\cref{p:planted-boosting}, carried out in \cref{ss:planted-boosting-rigorous}.
\begin{definition}\label{d:cI}
  Let ${\mathcal{I}}$ be the set of intervals $I \subseteq [-1,1]$ of the form $(a,b)$ or $[a,b]$, where $a<b$, such that $I$ either either contains or is disjoint with each of $(0.9,1]$ and $[-1,-0.9)$.
\end{definition}
All the intervals $I$ we will consider in our formal argument will be of this form.
We will always treat $I$ as fixed as $n\to\infty$ so that for sufficiently large $n$, $I$ contains a feasible value of $R(\sigma,\sigma^*)$ and the above restricted free energies are well-defined.
\begin{lemma}\label{l:planted-fe-conc}
  There exists $C = C(d,\delta,\varepsilon) > 0$ such that the following holds.
  For any $t \ge 0$, $\upsilon > 0$, let
  \[
    L = L(d,\delta,\varepsilon,t,\upsilon) = C (1+t) \upsilon^{1/2}\,.
  \]
  Then, for all $\beta \in [2\beta_\lowerBoundTag, \betaKS - \delta]$ and intervals $I,I' \in {\mathcal{I}}$:
  \begin{enumerate}[label=(\alph*)]
    \item\label{i:planted-fe-conc-one-replica} We have
    \[
      \Pr^\plantedTag_{\beta,t} \left(
        |
          (F_{n,\beta,t}^{\sigma^*,G,{\boldsymbol{v}}}(I) - F_{n,\beta,t}^{\sigma^*,G,{\boldsymbol{v}}}(I')) 
          - \E^\plantedTag_{\beta,t} (F_{n,\beta,t}^{\sigma^*,G,{\boldsymbol{v}}}(I) - F_{n,\beta,t}^{\sigma^*,G,{\boldsymbol{v}}}(I'))
        | \le L
      \right) \ge 1-e^{-\upsilon n}\,.
    \]
    \item\label{i:planted-fe-conc-two-replica} We have
    \begin{align*}
      \Pr^{\plantedTag,2}_{\beta,t} \Big(
        |
          (F_{n,\beta,t}^{2,\sigma^*,{\sigma'}^*,G,{\boldsymbol{v}}}(I) &- F_{n,\beta,t}^{2,\sigma^*,{\sigma'}^*,G,{\boldsymbol{v}}}(I')) \\
          &- \E^{\plantedTag,2}_{\beta,t} (F_{n,\beta,t}^{2,\sigma^*,{\sigma'}^*,G,{\boldsymbol{v}}}(I) - F_{n,\beta,t}^{2,\sigma^*,{\sigma'}^*,G,{\boldsymbol{v}}}(I'))
        | \le L
      \Big) \ge 1-e^{-\upsilon n}\,.
    \end{align*}
  \end{enumerate}
\end{lemma}
As a first step, we recall from Remark~\cref{r:augmented-planted-interpretation} that for $(\sigma^*,G,{\boldsymbol{v}}) \sim \nu^\plantedTag_{\beta,t}$, ${\boldsymbol{v}}$ can be generated by 
\[
  {\boldsymbol{v}} = \sqrt{t} \sigma^* + {\boldsymbol{g}}
\]
for ${\boldsymbol{g}} \sim {\mathcal{N}}(0,{\boldsymbol{I}}_n)$ independent of $(\sigma^*,G)$.
Throughout this subsection we write ${\boldsymbol{g}}$ in this form.
Define
\begin{align*}
  {\widetilde F}_{n,\beta,t}^{\sigma^*,G,{\boldsymbol{v}}}(I)
  &= \frac1n \log \sum_{\sigma \in \Sigma_n(\sigma^*,I)}
  \exp\left(
    \beta \isingHamiltonian_{G}(\sigma) 
    + t \langle \sigma,\sigma^*\rangle 
    + \sqrt{t} \langle \sigma - \sigma^*, {\boldsymbol{g}}\rangle
  \right)\,, \\ 
  {\widetilde F}_{n,\beta,t}^{2,\sigma^*,{\sigma'}^*,G,{\boldsymbol{v}}}(I)
  &= \frac1n \log \sum_{\sigma \in \Sigma_n(\sigma^*,I)}
  \allone\{|R(\sigma,{\sigma'}^*)| < \varepsilon\}
  \exp\left(
    \beta \isingHamiltonian_{G}(\sigma) 
    + t \langle \sigma,\sigma^*\rangle 
    + \sqrt{t} \langle \sigma - \sigma^*, {\boldsymbol{g}}\rangle
  \right)\,.
\end{align*}
A direct computation shows
\begin{equation}\label{e:F-to-tF}
  \begin{aligned}
  F_{n,\beta,t}^{\sigma^*,G,{\boldsymbol{v}}}(I) - F_{n,\beta,t}^{\sigma^*,G,{\boldsymbol{v}}}(I') 
  &= {\widetilde F}_{n,\beta,t}^{\sigma^*,G,{\boldsymbol{v}}}(I) - {\widetilde F}_{n,\beta,t}^{\sigma^*,G,{\boldsymbol{v}}}(I')\,, \\
  F_{n,\beta,t}^{2,\sigma^*,{\sigma'}^*,G,{\boldsymbol{v}}}(I) - F_{n,\beta,t}^{2,\sigma^*,{\sigma'}^*,G,{\boldsymbol{v}}}(I')
  &= {\widetilde F}_{n,\beta,t}^{2,\sigma^*,{\sigma'}^*,G,{\boldsymbol{v}}}(I) - {\widetilde F}_{n,\beta,t}^{2,\sigma^*,{\sigma'}^*,G,{\boldsymbol{v}}}(I')\,.
  \end{aligned}
\end{equation}
We will show concentration of ${\widetilde F}_{n,\beta,t}^{\sigma^*,G,{\boldsymbol{v}}}(I)$ and ${\widetilde F}_{n,\beta,t}^{2,\sigma^*,{\sigma'}^*,G,{\boldsymbol{v}}}(I)$ individually.
Define
\begin{align*}
  {\widehat F}_{n,\beta,t}(\sigma^*,I) &= \E^\plantedTag_{\beta,t}[{\widetilde F}_{n,\beta,t}^{\sigma^*,G,{\boldsymbol{v}}}(I) | \sigma^*]\,, & 
  {\widehat F}_{n,\beta,t}^2(\sigma^*,{\sigma'}^*,I) &= \E^{\plantedTag,2}_{\beta,t} [{\widetilde F}_{n,\beta,t}^{2,\sigma^*,{\sigma'}^*,G,{\boldsymbol{v}}}(I) | \sigma^*, {\sigma'}^*]\,.
\end{align*}
\begin{lemma}\label{l:fe-conc-around-conditional-mean}
  There exists $C = C(d) > 0$ such that for any $\beta \in [2\beta_\lowerBoundTag, \betaKS - \delta]$, $t\ge 0$, $\upsilon > 0$, and interval $I \in {\mathcal{I}}$, we have
  \begin{align}\label{e:fe-conc-around-conditional-mean-one-replica}
    \Pr^\plantedTag_{\beta,t}\left(
      |{\widetilde F}_{n,\beta,t}^{\sigma^*,G,{\boldsymbol{v}}}(I) - {\widehat F}_{n,\beta,t}(\sigma^*,I)| \le C(1+t)^{1/2}\upsilon^{1/2}
    \right) \ge 1 - e^{-\upsilon n}\,, \\
    \label{e:fe-conc-around-conditional-mean-two-replica}
    \Pr^{\plantedTag,2}_{\beta,t}\left(
      |{\widetilde F}_{n,\beta,t}^{2,\sigma^*,{\sigma'}^*,G,{\boldsymbol{v}}}(I) - {\widehat F}_{n,\beta,t}^2(\sigma^*,{\sigma'}^*,I)| \le C(1+t)^{1/2}\upsilon^{1/2}
    \right) \ge 1 - e^{-\upsilon n}\,.
  \end{align}
\begin{proof}
  We only prove the first conclusion, as the second is analogous.
  Fix any $\sigma^* \in \Sigma_n(\varepsilon)$ and consider the function
  \[
    f(G,{\boldsymbol{g}}) = \frac1n \log \sum_{\sigma \in \Sigma_n(\sigma^*,I)}
    \exp\left(
      \beta \isingHamiltonian_{G}(\sigma) 
      + t \langle \sigma,\sigma^*\rangle 
      + \sqrt{t} \langle \sigma - \sigma^*, {\boldsymbol{g}}\rangle
    \right)\,.
  \]
  Note that if we change $G$ to $G'$ by swapping a pair of half-edges $(e_1,f_1), (e_2,f_2)$ to $(e_1,f_2), (e_2,f_1)$, then $|\isingHamiltonian_{G}(\sigma) - \isingHamiltonian_{G'}(\sigma)| \le 2$ for all $\sigma \in \Sigma_n(\sigma^*,I)$.
  Hence
  \[
    |f(G,{\boldsymbol{g}}) - f(G',{\boldsymbol{g}})| \le 2|\beta| / n\,.
  \]
  Also,
  \[
    \|\nabla_{\boldsymbol{g}} f(G,{\boldsymbol{g}})\|
    = \frac{\sqrt{t}}{n} \|\mu_{G,\beta,t,<\varepsilon}^{{\boldsymbol{v}}}(\sigma - \sigma^* \,|\, I) \|
    \le 2\sqrt{\frac{t}{n}}\,.
  \]
  By the Azuma--Hoeffding inequality (applied to $G$ sampled from the law \eqref{e:planted-model-G-dist-given-x} conditional on $\sigma^*$) and gaussian concentration of Lipschitz functions (applied to ${\boldsymbol{g}}$, which is independent of $G$), $f(G,{\boldsymbol{g}})$ is subgaussian with variance proxy $O(d\beta^2 + t)/n$.
  The conclusion follows.
\end{proof}
\end{lemma}
It remains to establish concentration of the conditional means ${\widehat F}_{n,\beta,t}(\sigma^*,I)$ and ${\widehat F}_{n,\beta,t}^2(\sigma^*,{\sigma'}^*,I)$.
We achieve this through the following Lipschitz estimate, whose proof we defer to Appendix~\cref{a:restricted-fe-lipschitz}.
\begin{proposition}\label{p:fe-conc-conditional-mean-lipschitz}
  There exists $C = C(d) > 0$ such that for any $\beta \in [2\beta_\lowerBoundTag, \betaKS - \delta]$, $t\ge 0$, and interval $I \in {\mathcal{I}}$, the following holds.
  \begin{enumerate}[label=(\alph*)]
    \item \label{i:fe-conc-conditional-mean-lipschitz-one-replica} If $\sigma^*,\widetilde{\sigma}^* \in \Sigma_n(\varepsilon)$ have Hamming distance $1$, then
    \[
      |{\widehat F}_{n,\beta,t}(\sigma^*,I) - {\widehat F}_{n,\beta,t}(\widetilde{\sigma}^*,I)| \le C(1+t) / n\,.
    \]
    \item \label{i:fe-conc-conditional-mean-lipschitz-two-replica} If $(\sigma^*,{\sigma'}^*), (\widetilde{\sigma}^*,\widetilde{\sigma'}^*) \in \Sigma_{n,2}(\varepsilon)$ (recall \eqref{e:SigmaN2}) have Hamming distance $1$, then
    \[
      |{\widehat F}_{n,\beta,t}^2(\sigma^*,{\sigma'}^*,I) - {\widehat F}_{n,\beta,t}^2(\widetilde{\sigma}^*,\widetilde{\sigma'}^*,I)| \le C(1+t) / n\,.
    \]
  \end{enumerate}
\end{proposition}
\noindent From this, ~\cref{l:planted-fe-conc} now follows easily.
\begin{proof}[Proof of ~\cref{l:planted-fe-conc}]
  Note that for fixed $I$, ${\widehat F}_{n,\beta,t}(\sigma^*,I)$ depends on $\sigma^*$ only through $R(\sigma^*,\allone)$, and ${\widehat F}_{n,\beta,t}^2(\sigma^*,{\sigma'}^*,I)$ depends on $\sigma^*,{\sigma'}^*$ only through $R(\sigma^*,\allone)$, $R({\sigma'}^*,\allone)$, $R(\sigma^*,{\sigma'}^*)$.
  By ~\cref{p:fe-conc-conditional-mean-lipschitz}, the functions
  \begin{align*}
    R(\sigma^*,\allone) &\mapsto {\widehat F}_{n,\beta,t}(\sigma^*,I)\,, &
    (R(\sigma^*,\allone), R({\sigma'}^*,\allone), R(\sigma^*,{\sigma'}^*)) &\mapsto {\widehat F}_{n,\beta,t}^2(\sigma^*,{\sigma'}^*,I)
  \end{align*}
  are $C(1+t)$-Lipschitz, after possibly adjusting $C = C(d)$.
  By \eqref{e:fe-conc-around-conditional-mean-one-replica} from ~\cref{l:fe-conc-around-conditional-mean} and a union bound, with probability at least $1-2e^{-\upsilon n}$ under $\Pr^\plantedTag_{\beta,t}$,
  \[
    \left|
      ({\widetilde F}_{n,\beta,t}^{\sigma^*,G,{\boldsymbol{v}}}(I) - {\widetilde F}_{n,\beta,t}^{\sigma^*,G,{\boldsymbol{v}}}(I'))
      -({\widehat F}_{n,\beta,t}(\sigma^*,I)-{\widehat F}_{n,\beta,t}(\sigma^*,I'))
    \right|
    \le C(1+t)^{1/2}\upsilon^{1/2}\,.
  \]
  It remains to control the fluctuation of the conditional mean
  \[
    \Delta(\sigma^*) \equiv {\widehat F}_{n,\beta,t}(\sigma^*,I)-{\widehat F}_{n,\beta,t}(\sigma^*,I')\,.
  \]
  The above Lipschitz estimate shows that $\Delta(\sigma^*)$ is a $C(1+t)$-Lipschitz function of $R(\sigma^*,\allone)$.
  As noted in Remarks~\cref{r:planted-interpretation} and \cref{r:augmented-planted-interpretation}, the marginal laws of $\sigma^*$ under $\Pr^\plantedTag_{\beta}$ and $\Pr^\plantedTag_{\beta,t}$ coincide, and further equal the distribution described in ~\cref{l:1mt-balanced}~\cref{i:1mt-balanced-subgaussian}.
  Consequently, after increasing $C = C(d,\varepsilon)$, ~\cref{l:1mt-balanced}~\cref{i:1mt-balanced-subgaussian} implies
  \[
    \Pr^\plantedTag_{\beta,t}\left(
      |\Delta(\sigma^*)-\E^\plantedTag_{\beta,t}\Delta(\sigma^*)|
      \le C(1+t)\upsilon^{1/2}
    \right)
    \ge 1-e^{-\upsilon n}\,.
  \]
  Adjusting constants gives the desired concentration for
  ${\widetilde F}_{n,\beta,t}^{\sigma^*,G,{\boldsymbol{v}}}(I)-{\widetilde F}_{n,\beta,t}^{\sigma^*,G,{\boldsymbol{v}}}(I')$.
  Finally, \eqref{e:F-to-tF} converts this statement to the same statement for
  $F_{n,\beta,t}^{\sigma^*,G,{\boldsymbol{v}}}(I)-F_{n,\beta,t}^{\sigma^*,G,{\boldsymbol{v}}}(I')$.

  For part \cref{i:planted-fe-conc-two-replica}, apply \eqref{e:fe-conc-around-conditional-mean-two-replica} from ~\cref{l:fe-conc-around-conditional-mean} and repeat the same argument with
  \[
    \Delta_2(\sigma^*,{\sigma'}^*) \equiv
    {\widehat F}_{n,\beta,t}^2(\sigma^*,{\sigma'}^*,I)-{\widehat F}_{n,\beta,t}^2(\sigma^*,{\sigma'}^*,I').
  \]
  ~\cref{p:fe-conc-conditional-mean-lipschitz}~\cref{i:fe-conc-conditional-mean-lipschitz-two-replica} ensures $\Delta_2$ is a $C(1+t)$-Lipschitz function of
  $(R(\sigma^*,\allone),R({\sigma'}^*,\allone),R(\sigma^*,{\sigma'}^*))$, where $C = C(d)$. 
  Then, ~\cref{l:basic-mt}~\cref{i:basic-mt-annealed-2-subgaussian} gives the required concentration of this vector under $\Pr^{\plantedTag,2}_{\beta,t}$.
  Note that the $c$ given by ~\cref{l:basic-mt}~\cref{i:basic-mt-annealed-2-subgaussian} depends on $d,\delta,\varepsilon$, so that in the resulting concentration bound
  \[
    \Pr^{\plantedTag,2}_{\beta,t}\left(
      |\Delta_2(\sigma^*,{\sigma'}^*)-\E^{\plantedTag,2}_{\beta,t}\Delta_2(\sigma^*,{\sigma'}^*)|
      \le C(1+t)\upsilon^{1/2}
    \right)
    \ge 1-e^{-\upsilon n}\,,
  \]
  the parameter $C$ now depends on $d,\delta,\varepsilon$.
  Using again \eqref{e:F-to-tF} completes the proof.
\end{proof}

\subsection{Strong overlap bound from side information: rigorous argument}
\label{ss:planted-boosting-rigorous}

In this subsection we prove ~\cref{p:planted-boosting}.
We will first prove part \cref{i:planted-boosting-one-replica}, and then explain the changes needed for part \cref{i:planted-boosting-two-replica}.

The $R(\sigma^1,\sigma^2)$ term in ~\cref{l:fe-deriv} is nonnegative thanks to the following general fact.
\begin{fact}\label{f:replica-overlap-nonnegative}
  For all $\beta\in \R$, $t \ge 0$, all realizations of $G,{\boldsymbol{v}}$, and interval $I \in {\mathcal{I}}$,
  \[
    \mu_{G,\beta,t,<\varepsilon}^{{\boldsymbol{v}}}(R(\sigma^1,\sigma^2) \,|\, I) \ge 0
  \]
\begin{proof}
  Write $\bar \sigma = \mu_{G,\beta,t,<\varepsilon}^{{\boldsymbol{v}}}(\sigma \,|\, I)$ for the mean of $\mu_{G,\beta,t,<\varepsilon}^{{\boldsymbol{v}}}(\cdot \,|\, I)$.
  Then,
  \[
    \mu_{G,\beta,t,<\varepsilon}^{{\boldsymbol{v}}}(R(\sigma^1,\sigma^2) \,|\, I)
    = \frac1n \mu_{G,\beta,t,<\varepsilon}^{{\boldsymbol{v}}}(\langle \bar \sigma + (\sigma^1 - \bar \sigma),\bar \sigma + (\sigma^2 - \bar \sigma)\rangle \,|\, I)
    = \|\bar \sigma\|^2 / n \ge 0\,. \qedhere
  \]
\end{proof}
\end{fact}
We now assume the setting of ~\cref{p:planted-boosting}.
Let $\beta \in [2\beta_\lowerBoundTag, \betaKS - \delta]$, $\varepsilon > 0$, and $\upsilon>0$.
Let 
\[
  I_* = (-\varepsilon/10,\varepsilon/10)\,.
\]
Recall that for $0\le k \le k_\upperBoundTag \equiv \lfloor (\log \frac{1}{\varepsilon}) / (\log \frac32) \rfloor$, we defined $x_k = \min((3/2)^k \varepsilon, 0.9)$ and $I_k = [x_k,1]$.
Let
\begin{equation}\label{e:phik}
  \phi_k(t)
  = \E^\plantedTag_{\beta,t} F_{n,\beta,t}^{\sigma^*,G,{\boldsymbol{v}}}(I_*)
  - \E^\plantedTag_{\beta,t} F_{n,\beta,t}^{\sigma^*,G,{\boldsymbol{v}}}(I_k)\,.
\end{equation}
Let $L$ be the function from ~\cref{l:planted-fe-conc}; we will from now on suppress its dependence on $d,\delta,\varepsilon$, writing $L(t,\upsilon)$ for $L(d,\delta,\varepsilon,t,\upsilon)$.
Further define
\[
  {\widetilde L} (t,\upsilon) = \upsilon + L(t,\upsilon)\,.
\]
\begin{lemma}\label{l:boosting-init}
  If the hypothesis \eqref{e:planted-boosting-hypothesis} of ~\cref{p:planted-boosting}~\cref{i:planted-boosting-one-replica} holds, then $\phi_0(0) \le {\widetilde L}(0,2\upsilon)$.
\begin{proof}
  Assume for contradiction that $\phi_0(0) > {\widetilde L}(0,2\upsilon)$.
  Recall that $I_0 = [\varepsilon,1]$.
  By symmetry of $\mu_{G,\beta,<\varepsilon}$ over negation of $\sigma$,
  \[
    \E^\plantedTag_{\beta} \mu_{G,\beta,<\varepsilon}(|R(\sigma,\sigma^*)| \ge \varepsilon)
    = 2 \E^\plantedTag_{\beta} \mu_{G,\beta,<\varepsilon}(R(\sigma,\sigma^*) \ge \varepsilon)
    = 2 \E^\plantedTag_{\beta} \mu_{G,\beta,<\varepsilon}(R(\sigma,\sigma^*) \in I_0)\,.
  \]
  Then,
  \begin{align*}
    \E^\plantedTag_{\beta} \mu_{G,\beta,<\varepsilon}(R(\sigma,\sigma^*) \in I_0)
    &= \E^\plantedTag_{\beta,0} \left[\exp\left(
      n(F_{n,\beta,0}^{\sigma^*,G,{\boldsymbol{v}}}(I_0) - F_{n,\beta,0}^{\sigma^*,G,{\boldsymbol{v}}}([-1,1]) )
    \right)\right] \\
    &\le \E^\plantedTag_{\beta,0} \left[\exp\left(
      n(F_{n,\beta,0}^{\sigma^*,G,{\boldsymbol{v}}}(I_0) - F_{n,\beta,0}^{\sigma^*,G,{\boldsymbol{v}}}(I_*) )
    \right) \wedge 1 \right]\,.
  \end{align*}
  Define the event
  \[
    {\mathcal{E}} = \left\{
      |(F_{n,\beta,0}^{\sigma^*,G,{\boldsymbol{v}}}(I_0) - F_{n,\beta,0}^{\sigma^*,G,{\boldsymbol{v}}}(I_*))
      - \E^\plantedTag_{\beta,0}(F_{n,\beta,0}^{\sigma^*,G,{\boldsymbol{v}}}(I_0) - F_{n,\beta,0}^{\sigma^*,G,{\boldsymbol{v}}}(I_*))| \le L(0,2\upsilon)
    \right\}\,.
  \]
  By ~\cref{l:planted-fe-conc}~\cref{i:planted-fe-conc-one-replica}, we have $\Pr^\plantedTag_{\beta,0}({\mathcal{E}}) \ge 1-e^{-2\upsilon n}$.
  So,
  \begin{align*}
    \E^\plantedTag_{\beta} \mu_{G,\beta,<\varepsilon}(R(\sigma,\sigma^*) \in I_0)
    &\le e^{-2\upsilon n} + \E^\plantedTag_{\beta,0} \left[\allone\{{\mathcal{E}}\}\exp\left(
      n(F_{n,\beta,0}^{\sigma^*,G,{\boldsymbol{v}}}(I_0) - F_{n,\beta,0}^{\sigma^*,G,{\boldsymbol{v}}}(I_*) )
    \right) \right] \\
    &\le e^{-2\upsilon n} + \exp\left(
      n \E^\plantedTag_{\beta,0} (F_{n,\beta,0}^{\sigma^*,G,{\boldsymbol{v}}}(I_0) - F_{n,\beta,0}^{\sigma^*,G,{\boldsymbol{v}}}(I_*) + L(0,2\upsilon)
    \right) \\
    &\le e^{-2\upsilon n} + \exp\left(
      n (-\phi_0(0) + L(0,2\upsilon))
    \right)
    \le 2e^{-2\upsilon n}\,,
  \end{align*}
  where the last step uses the assumption $\phi_0(0) > {\widetilde L}(0,2\upsilon) = 2\upsilon + L(0,2\upsilon)$.
  Thus
  \[
    \E^\plantedTag_{\beta} \mu_{G,\beta,<\varepsilon}(|R(\sigma,\sigma^*)| \ge \varepsilon) \le 4e^{-2\upsilon n}\,,
  \]
  contradicting the hypothesis \eqref{e:planted-boosting-hypothesis} that
  \[
    \E^\plantedTag_{\beta} \mu_{G,\beta,<\varepsilon}(|R(\sigma,\sigma^*)| \ge \varepsilon)
    \ge e^{-\upsilon n}\,. \qedhere
  \]
\end{proof}
\end{lemma}
The following proposition formalizes Claim~\cref{cl:boosting-strategy} and is the core of our dynamical argument for ~\cref{p:planted-boosting}~\cref{i:planted-boosting-one-replica}.
\begin{proposition}\label{p:boosting-core}
  Let $t>0$ and suppose that
  \begin{equation}\label{e:fe-not-yet-overtake}
    \E^\plantedTag_{\beta,t} F_{n,\beta,t}^{\sigma^*,G,{\boldsymbol{v}}}([-1,1])
    < \E^\plantedTag_{\beta,t} F_{n,\beta,t}^{\sigma^*,G,{\boldsymbol{v}}}(I_*) + \upsilon\,.
  \end{equation}
  Then, for each $0\le k\le k_\upperBoundTag - 1$, either:
  \begin{enumerate}[label=(\roman*)]
    \item \label{i:formal-fe-make-progress} $\phi'_k(t) \le -\varepsilon/10$, or
    \item \label{i:formal-fe-increment-k} $\phi_{k+1}(t) \le {\widetilde L}(t, 2(\phi_k(t) + 3\upsilon))$.
  \end{enumerate}
  If $k = k_\upperBoundTag$, then outcome \cref{i:formal-fe-make-progress} holds.
\end{proposition}
The proof of this proposition appears below, after the following preparatory lemma.
\begin{lemma}\label{l:boosting-core-prep}
  Let $t>0$ and $0\le k\le k_\upperBoundTag - 1$.
  If \eqref{e:fe-not-yet-overtake} holds and outcome \cref{i:formal-fe-increment-k} of ~\cref{p:boosting-core} does not hold, then
  \begin{equation}\label{e:boosting-core-prep}
    \E^\plantedTag_{\beta,t} \mu_{G,\beta,t,<\varepsilon}^{{\boldsymbol{v}}}\left(R(\sigma^1,\sigma^2) \ge \frac32 x_k \,|\, I_k\right) \le e^{-\iota n}
  \end{equation}
  for some $\iota = \iota(d,t,\upsilon) > 0$.
\begin{proof}
  Recall that $x_{k+1} \le \frac32 x_k$ for all $0\le k\le k_{\upperBoundTag} - 1$.
  We have
  \begin{align}
    \nonumber
    \E^\plantedTag_{\beta,t} \mu_{G,\beta,t,<\varepsilon}^{{\boldsymbol{v}}}\left(R(\sigma^1,\sigma^2) \ge \frac32 x_k \,|\, I_k\right) 
    &\le \E^\plantedTag_{\beta,t} \mu_{G,\beta,t,<\varepsilon}^{{\boldsymbol{v}}}\left(R(\sigma^1,\sigma^2) \ge x_{k+1} \,|\, I_k\right) \\
    \label{e:boosting-core-prep-step1}
    &\le \E^\plantedTag_{\beta,t} \left[
      \frac{
        \mu_{G,\beta,t,<\varepsilon}^{{\boldsymbol{v}}}\left(R(\sigma^1,\sigma^2) \in I_{k+1}\right)
      }{
        \mu_{G,\beta,t,<\varepsilon}^{{\boldsymbol{v}}}\left(R(\sigma,\sigma^*) \in I_k\right)^2
      } \wedge 1
    \right]\,.
  \end{align}
  Define the event
  \[
    {\mathcal{E}}_1 = \left\{
      |(F_{n,\beta,t}^{\sigma^*,G,{\boldsymbol{v}}}(I_k) - F_{n,\beta,t}^{\sigma^*,G,{\boldsymbol{v}}}([-1,1]))
      - \E^\plantedTag_{\beta,t}(F_{n,\beta,t}^{\sigma^*,G,{\boldsymbol{v}}}(I_k) - F_{n,\beta,t}^{\sigma^*,G,{\boldsymbol{v}}}([-1,1]))|
      \le \upsilon
    \right\}\,.
  \]
  If we set $\iota$ small enough that $L(t,\iota) \le \upsilon$, ~\cref{l:planted-fe-conc}~\cref{i:planted-fe-conc-one-replica} yields $\Pr^\plantedTag_{\beta,t}({\mathcal{E}}_1) \ge 1-e^{-\iota n}$.
  On event ${\mathcal{E}}_1$,
  \begin{align*}
    \mu_{G,\beta,t,<\varepsilon}^{{\boldsymbol{v}}}\left(R(\sigma,\sigma^*) \in I_k\right)
    &= \exp\left(
      n(F_{n,\beta,t}^{\sigma^*,G,{\boldsymbol{v}}}(I_k) - F_{n,\beta,t}^{\sigma^*,G,{\boldsymbol{v}}}([-1,1]) )
    \right) \\
    &\ge \exp\left(
      n\E^\plantedTag_{\beta,t} (F_{n,\beta,t}^{\sigma^*,G,{\boldsymbol{v}}}(I_k) - F_{n,\beta,t}^{\sigma^*,G,{\boldsymbol{v}}}([-1,1]) - \upsilon)
    \right) \\
    &\ge \exp\left(
      n\E^\plantedTag_{\beta,t} (F_{n,\beta,t}^{\sigma^*,G,{\boldsymbol{v}}}(I_k) - F_{n,\beta,t}^{\sigma^*,G,{\boldsymbol{v}}}(I_*) - 2\upsilon)
    \right)
    = \exp\left(
      -n(\phi_k(t) + 2\upsilon)
    \right)\,.
  \end{align*}
  Plugging into \eqref{e:boosting-core-prep-step1} shows
  \begin{align}
    \nonumber
    \E^\plantedTag_{\beta,t} \mu_{G,\beta,t,<\varepsilon}^{{\boldsymbol{v}}}\left(R(\sigma^1,\sigma^2) \ge \frac32 x_k \,|\, I_k\right)
    &\le e^{-\iota n} + \E^\plantedTag_{\beta,t} \left[
      \allone\{{\mathcal{E}}_1\}
      \frac{
        \mu_{G,\beta,t,<\varepsilon}^{{\boldsymbol{v}}}\left(R(\sigma^1,\sigma^2) \in I_{k+1}\right)
      }{
        \mu_{G,\beta,t,<\varepsilon}^{{\boldsymbol{v}}}\left(R(\sigma,\sigma^*) \in I_k\right)^2
      }
    \right] \\
    \label{e:boosting-core-prep-step2}
    &\le e^{-\iota n} + e^{2(\phi_k(t) + 2\upsilon)n} 
    \E^\plantedTag_{\beta,t} \mu_{G,\beta,t,<\varepsilon}^{{\boldsymbol{v}}}\left(R(\sigma^1,\sigma^2) \in I_{k+1}\right)\,.
  \end{align}
  The remaining expectation can be bounded by Nishimori's property and calculations similar to ~\cref{l:boosting-init}:
  \begin{align*}
    \E^\plantedTag_{\beta,t} \mu_{G,\beta,t,<\varepsilon}^{{\boldsymbol{v}}}\left(R(\sigma^1,\sigma^2) \in I_{k+1}\right)
    &= \E^\plantedTag_{\beta,t} \mu_{G,\beta,t,<\varepsilon}^{{\boldsymbol{v}}}\left(R(\sigma,\sigma^*) \in I_{k+1}\right) \\
    &= \E^\plantedTag_{\beta,t} \left[\exp\left(
      n(F_{n,\beta,t}^{\sigma^*,G,{\boldsymbol{v}}}(I_{k+1}) - F_{n,\beta,t}^{\sigma^*,G,{\boldsymbol{v}}}([-1,1]) )
    \right)\right] \\
    &\le \E^\plantedTag_{\beta,t} \left[\exp\left(
      n(F_{n,\beta,t}^{\sigma^*,G,{\boldsymbol{v}}}(I_{k+1}) - F_{n,\beta,t}^{\sigma^*,G,{\boldsymbol{v}}}(I_*) )
    \right) \wedge 1 \right]\,.
  \end{align*}
  Define the event
  \[
    {\mathcal{E}}_2 = \left\{
      |(F_{n,\beta,t}^{\sigma^*,G,{\boldsymbol{v}}}(I_{k+1}) - F_{n,\beta,t}^{\sigma^*,G,{\boldsymbol{v}}}(I_*))
      - \E^\plantedTag_{\beta,t}(F_{n,\beta,t}^{\sigma^*,G,{\boldsymbol{v}}}(I_{k+1}) - F_{n,\beta,t}^{\sigma^*,G,{\boldsymbol{v}}}(I_*))| \le L(t,2(\phi_k(t) + 3\upsilon))
    \right\}\,.
  \]
  By definition of $L$, we have $\Pr^\plantedTag_{\beta,t}({\mathcal{E}}_2) \ge 1-e^{-2(\phi_k(t) + 3\upsilon) n}$.
  So,
  \begin{align*}
    &\E^\plantedTag_{\beta,t} \mu_{G,\beta,t,<\varepsilon}^{{\boldsymbol{v}}}\left(R(\sigma^1,\sigma^2) \in I_{k+1}\right) \\
    &\le e^{-2(\phi_k(t) + 3\upsilon) n} + \E^\plantedTag_{\beta,t} \left[\allone\{{\mathcal{E}}_2\}\exp\left(
      n(F_{n,\beta,t}^{\sigma^*,G,{\boldsymbol{v}}}(I_{k+1}) - F_{n,\beta,t}^{\sigma^*,G,{\boldsymbol{v}}}(I_*) )
    \right) \right] \\
    &\le e^{-2(\phi_k(t) + 3\upsilon) n} + \exp\left(
      n \E^\plantedTag_{\beta,t} (F_{n,\beta,t}^{\sigma^*,G,{\boldsymbol{v}}}(I_{k+1}) - F_{n,\beta,t}^{\sigma^*,G,{\boldsymbol{v}}}(I_*) + L(t,2(\phi_k(t) + 3\upsilon))
    \right) \\
    &\le e^{-2(\phi_k(t) + 3\upsilon) n} + \exp\left(
      n (-\phi_{k+1}(t) + L(t,2(\phi_k(t) + 3\upsilon)))
    \right)
    \le 2e^{-2(\phi_k(t) + 3\upsilon) n}\,,
  \end{align*}
  where the last step uses the assumption 
  \[
    \phi_{k+1}(t) > {\widetilde L}(t,2(\phi_k(t) + 3\upsilon)) = 2(\phi_k(t) + 3\upsilon) + L(t,2(\phi_k(t) + 3\upsilon))\,.
  \]
  Plugging into \eqref{e:boosting-core-prep-step2} shows
  \[
    \E^\plantedTag_{\beta,t} \mu_{G,\beta,t,<\varepsilon}^{{\boldsymbol{v}}}\left(R(\sigma^1,\sigma^2) \ge \frac32 x_k \,|\, I_k\right)
    \le e^{-\iota n} + e^{2(\phi_k(t) + 2\upsilon)n} \cdot 2e^{-2(\phi_k(t) + 3\upsilon) n}
    = e^{-\iota n} + 2e^{-2\upsilon n}\,.
  \]
  This implies the conclusion after adjusting $\iota$.
\end{proof}
\end{lemma}
\begin{proof}[Proof of ~\cref{p:boosting-core}]
  If $0\le k\le k_\upperBoundTag - 1$ and outcome \cref{i:formal-fe-increment-k} of this proposition holds, there is nothing to argue.
  If $0\le k\le k_\upperBoundTag - 1$ and outcome \cref{i:formal-fe-increment-k} does not hold, ~\cref{l:boosting-core-prep} implies that \eqref{e:boosting-core-prep} holds.
  If $k = k_\upperBoundTag$, we have $\frac32 x_k > 1$, so \eqref{e:boosting-core-prep} holds trivially.
  So, in all cases we may assume \eqref{e:boosting-core-prep} holds.
  We will show this implies outcome \cref{i:formal-fe-make-progress} of this proposition.
  By ~\cref{l:fe-deriv},
  \begin{align*}
    \phi'_k(t) &= \frac12 \Big(
      2\E^\plantedTag_{\beta,t} \mu_{G,\beta,t,<\varepsilon}^{{\boldsymbol{v}}}(R(\sigma,\sigma^*) \,|\, I_*)
      - 2\E^\plantedTag_{\beta,t} \mu_{G,\beta,t,<\varepsilon}^{{\boldsymbol{v}}}(R(\sigma,\sigma^*) \,|\, I_k) \\
      &\qquad - \E^\plantedTag_{\beta,t} \mu_{G,\beta,t,<\varepsilon}^{{\boldsymbol{v}}}(R(\sigma^1,\sigma^2) \,|\, I_*)
      + \E^\plantedTag_{\beta,t} \mu_{G,\beta,t,<\varepsilon}^{{\boldsymbol{v}}}(R(\sigma^1,\sigma^2) \,|\, I_k)
    \Big) \\
    &\le \frac12 \left(
      2 \cdot \frac{\varepsilon}{10}
      - 2 x_k
      + \E^\plantedTag_{\beta,t} \mu_{G,\beta,t,<\varepsilon}^{{\boldsymbol{v}}}(R(\sigma^1,\sigma^2) \,|\, I_k)
    \right)\,,
  \end{align*}
  where we used Fact~\cref{f:replica-overlap-nonnegative} to bound the negative term containing $R(\sigma^1,\sigma^2)$.
  This last expectation can be bounded by
  \[
    \E^\plantedTag_{\beta,t} \mu_{G,\beta,t,<\varepsilon}^{{\boldsymbol{v}}}(R(\sigma^1,\sigma^2) \,|\, I_k)
    \le \frac32 x_k + \E^\plantedTag_{\beta,t} \mu_{G,\beta,t,<\varepsilon}^{{\boldsymbol{v}}}\left(R(\sigma^1,\sigma^2) \ge \frac32 x_k \,|\, I_k\right)
    \le \frac32 x_k + e^{-\iota n}\,,
  \]
  for $\iota$ given by ~\cref{l:boosting-core-prep}.
  This implies (recalling $x_k \ge x_0 = \varepsilon$)
  \[
    \phi'_k(t) \le \frac12 \left( 2\cdot \frac{\varepsilon}{10} - \frac12 x_k + e^{-\iota n} \right)
    \le \frac12 \left( \frac{\varepsilon}{5} - \frac{\varepsilon}{2} + e^{-\iota n} \right) \le -\frac{\varepsilon}{10}\,,
  \]
  so outcome \cref{i:formal-fe-make-progress} holds.
\end{proof}
We now define the main dynamical process in the proof of ~\cref{p:planted-boosting}~\cref{i:planted-boosting-one-replica}.
\begin{definition}\label{d:dynamical-process}
  Let $\tau \ge 0$ denote the minimum $t \ge 0$ such that 
  \begin{equation}\label{e:formal-fe-overtake}
    \E^\plantedTag_{\beta,t} F_{n,\beta,t}^{\sigma^*,G,{\boldsymbol{v}}}([-1,1])
    \ge \E^\plantedTag_{\beta,t} F_{n,\beta,t}^{\sigma^*,G,{\boldsymbol{v}}}(I_*) + \upsilon\,,
  \end{equation}
  i.e. \eqref{e:fe-not-yet-overtake} does \textbf{not} hold.
  We set $\tau = \infty$ if no such $t$ exists.
  Note that both sides of \eqref{e:formal-fe-overtake} are continuous in $t$, so the set of $t$ satisfying it has a well-defined minimum if it is nonempty.

  We next define an increasing sequence $0 = \tau_0 \le \cdots \le \tau_{k_\upperBoundTag+1} = \tau$, as follows.
  For $0\le k \le k_\upperBoundTag-1$, define $\tau_{k+1}$ recursively as the smallest $t \in [\tau_k, \tau]$ such that
  \begin{equation}\label{e:formal-fe-increment-k}
    \phi_{k+1}(t) \le {\widetilde L}(t, 2(\phi_k(t) + 3\upsilon))\,.
  \end{equation}
  i.e. outcome \cref{i:formal-fe-increment-k} in ~\cref{p:boosting-core} holds.
  (It is possible that $\tau_{k+1} = \tau_k$.)
  We set $\tau_{k+1} = \tau$ (and $\tau_{k'} = \tau$ for all $k' > k+1$) if no such $t$ exists.
  Similarly to above, the set of $t \in [\tau_k, \tau]$ satisfying \eqref{e:formal-fe-increment-k} has a minimum if it is nonempty.
  The intervals $[\tau_k,\tau_{k+1})$ for $0\le k\le k_\upperBoundTag$ partition $[0,\tau)$.
\end{definition}
The following lemma upper bounds the $\tau_k$ and $\tau$.
\begin{lemma}\label{l:boosting-recursive-ubs}
  Recursively define the (increasing) sequences $(t_k)_{0\le k\le k_\upperBoundTag + 1}$ and $(b_k)_{0\le k\le k_\upperBoundTag}$ by $t_0 = 0$, $b_0 = {\widetilde L}(0, 2\upsilon)$ and for $k\ge 1$,
  \begin{align*}
    t_k &= t_{k-1} + \frac{10}{\varepsilon} (b_{k-1} + \upsilon)\,, & 
    b_k &= {\widetilde L}(t_k, 2(b_{k-1} + 3\upsilon))\,.
  \end{align*}
  Let $k_* \in \{0,\ldots,k_\upperBoundTag+1\}$ be the smallest number for which $\tau_{k_*} = \tau$.
  Then, for all $1\le k\le k_*$ we have $\tau_k \le t_k$, and for all $0\le k\le k_*-1$ we have $\phi_k(\tau_k) \le b_k$.
  In particular $\tau = \tau_{k_*} \le t_{k_\upperBoundTag+1}$.
\begin{proof}
  We proceed by induction on $k$.
  The base case $k=0$ holds because $\tau_0 = 0 = t_0$ by definition, and ~\cref{l:boosting-init} ensures 
  \[
    \phi_0(\tau_0 = 0) \le {\widetilde L}(0,2\upsilon) = b_0\,.
  \]
  Now, assume for induction that we have proven the lemma up to $k \le k_* - 1$.

  Suppose further that $k \le k_{\upperBoundTag} - 1$.
  For $t\in [\tau_k, \tau_{k+1})$, \eqref{e:formal-fe-increment-k} does not hold by definition of $\tau_{k+1}$, and \eqref{e:formal-fe-overtake} does not hold because $t < \tau_{k+1} \le \tau_{k_*} = \tau$.
  It follows that the hypotheses of ~\cref{p:boosting-core} hold and its outcome \cref{i:formal-fe-increment-k} does not hold.
  So, its outcome \cref{i:formal-fe-make-progress} holds, i.e. $\phi'_k(t) \le -\varepsilon/10$.

  Otherwise, if $k \ge k_{\upperBoundTag}$, then the bounds $k \le k_* - 1$ and $k_* \le k_\upperBoundTag + 1$ imply $k = k_{\upperBoundTag}$.
  In this case, ~\cref{p:boosting-core} directly implies outcome \cref{i:formal-fe-make-progress} holds, i.e. $\phi'_k(t) \le -\varepsilon/10$.

  So, in all cases $\phi'_k(t) \le -\varepsilon/10$.
  Then, for all $t\in [\tau_k, \tau_{k+1})$,
  \begin{equation}\label{e:boosting-recursion-b}
    \phi_k(t)
    \le \phi_k(\tau_k) - \frac{\varepsilon}{10} (t - \tau_k)
    \le b_k - \frac{\varepsilon}{10} (t - \tau_k)
  \end{equation}
  by the inductive hypothesis.
  Furthermore,
  \begin{align*}
    \phi_k(t)
    &= \E^\plantedTag_{\beta,t} F_{n,\beta,t}^{\sigma^*,G,{\boldsymbol{v}}}(I_*)
    - \E^\plantedTag_{\beta,t} F_{n,\beta,t}^{\sigma^*,G,{\boldsymbol{v}}}(I_k) \\
    &\ge \E^\plantedTag_{\beta,t} F_{n,\beta,t}^{\sigma^*,G,{\boldsymbol{v}}}(I_*)
    - \E^\plantedTag_{\beta,t} F_{n,\beta,t}^{\sigma^*,G,{\boldsymbol{v}}}([-1,1])
    \ge -\upsilon
  \end{align*}
  because \eqref{e:formal-fe-overtake} does not hold.
  Together with \eqref{e:boosting-recursion-b} and the inductive hypothesis this implies
  \[
    t
    \le \tau_k + \frac{10}{\varepsilon} (b_k + \upsilon)
    \le t_k + \frac{10}{\varepsilon} (b_k + \upsilon)
    = t_{k+1}\,.
  \]
  Since this holds for all $t\in [\tau_k, \tau_{k+1})$, we conclude $\tau_{k+1} \le t_{k+1}$.
  Finally, if $k+1 \le k_* - 1$, we have
  \begin{align*}
    \phi_{k+1}(\tau_{k+1})
    &\le {\widetilde L}(\tau_{k+1}, 2(\phi_k(\tau_{k+1}) + 3\upsilon)) \\
    &\le {\widetilde L}(\tau_{k+1}, 2(\phi_k(\tau_k) + 3\upsilon))
    \le {\widetilde L}(t_{k+1}, 2(b_k + 3\upsilon))
    = b_{k+1}\,,
  \end{align*}
  where the three inequalities are by \eqref{e:formal-fe-increment-k}, decreasingness of $\phi_k$ on $[\tau_k,\tau_{k+1}]$ (which we just showed), and increasingness of ${\widetilde L}$ in both inputs.
  This completes the induction.
\end{proof}
\end{lemma}

\begin{proof}[Proof of ~\cref{p:planted-boosting}~\cref{i:planted-boosting-one-replica}]
  Set $t_\upperBoundTag = t_\upperBoundTag(d,\varepsilon,\upsilon) = t_{k_\upperBoundTag+1}$.
  Inspecting the recursions in the statement of ~\cref{l:boosting-recursive-ubs} reveals that for fixed $d,\varepsilon$, we have $\lim_{\upsilon \downarrow 0} t_\upperBoundTag(d,\varepsilon,\upsilon) = 0$ as desired.

  By ~\cref{l:boosting-recursive-ubs}, there exists $\tau \le t_\upperBoundTag$ such that \eqref{e:formal-fe-overtake} holds.
  Let $t = \tau$, and consider the event
  \[
    {\mathcal{E}} = \left\{
      |(F_{n,\beta,t}^{\sigma^*,G,{\boldsymbol{v}}}(I_*) - F_{n,\beta,t}^{\sigma^*,G,{\boldsymbol{v}}}([-1,1]))
      - \E^\plantedTag_{\beta,t}(F_{n,\beta,t}^{\sigma^*,G,{\boldsymbol{v}}}(I_*) - F_{n,\beta,t}^{\sigma^*,G,{\boldsymbol{v}}}([-1,1]))| \le \upsilon/2
    \right\}\,.
  \]
  By ~\cref{l:planted-fe-conc}~\cref{i:planted-fe-conc-one-replica}, $\Pr^\plantedTag_{\beta,t}({\mathcal{E}}) = 1-o_n(1)$ (and in fact this probability is $1-e^{-\iota n}$ for some $\iota > 0$ chosen so that $L(t,\iota) \le \upsilon/2$).
  Then, recalling $I_* = (-\varepsilon/10,\varepsilon/10)$,
  \begin{align*}
    \E^\plantedTag_{\beta,t} \mu_{G,\beta,t,<\varepsilon}^{{\boldsymbol{v}}}(|R(\sigma,\sigma^*)| < \varepsilon/10)
    &= o_n(1) + \E^\plantedTag_{\beta,t} \left[
      \allone\{{\mathcal{E}}\} \mu_{G,\beta,t,<\varepsilon}^{{\boldsymbol{v}}}(R(\sigma,\sigma^*) \in I_*)
    \right] \\
    &= o_n(1) + \E^\plantedTag_{\beta,t} \left[
      \allone\{{\mathcal{E}}\} \exp\left(
        n(F_{n,\beta,t}^{\sigma^*,G,{\boldsymbol{v}}}(I_*) - F_{n,\beta,t}^{\sigma^*,G,{\boldsymbol{v}}}([-1,1]))
      \right)
    \right] \\
    &\le o_n(1) + \exp\left(
      n \E^\plantedTag_{\beta,t}(F_{n,\beta,t}^{\sigma^*,G,{\boldsymbol{v}}}(I_*) - F_{n,\beta,t}^{\sigma^*,G,{\boldsymbol{v}}}([-1,1]) + \upsilon/2)
    \right) \\
    &\le o_n(1) + e^{-\upsilon n /2}\,,
  \end{align*}
  where the last estimate uses \eqref{e:formal-fe-overtake}.
  This implies the desired estimate \eqref{e:planted-boosting-conclusion}.
\end{proof}
The proof of ~\cref{p:planted-boosting}~\cref{i:planted-boosting-two-replica} is essentially analogous, and argues through the below ``two-replica'' versions of the preceding lemmas and propositions.
We thus only give an outline of the proof.
Recall $Z_{G,\beta,t}^{2,\sigma^*,{\sigma'}^*,{\boldsymbol{v}}}(I)$ and $F_{n,\beta,t}^{2,\sigma^*,{\sigma'}^*,G,{\boldsymbol{v}}}(I)$ defined in \eqref{e:restricted-Z-and-F-to-I-two-replica}, and $\mu_{G,\beta,t,<\varepsilon}^{2,{\sigma'}^*,{\boldsymbol{v}}}(f \,|\, I)$ defined in \eqref{e:gibs-avg-of-function-two-replica}.
\begin{lemma}\label{l:fe-deriv-two-replica}
  For any $\beta \in \R$, $t > 0$, and interval $I \in {\mathcal{I}}$, we have
  \[
    \frac{\differential}{\differential t} \E^{\plantedTag,2}_{\beta,t} F_{n,\beta,t}^{2,\sigma^*,{\sigma'}^*,G,{\boldsymbol{v}}}(I)
    = \frac12 \left(
      -1 + 2\E^{\plantedTag,2}_{\beta,t} \mu_{G,\beta,t,<\varepsilon}^{2,{\sigma'}^*,{\boldsymbol{v}}}(R(\sigma,\sigma^*) \,|\, I)
      - \E^{\plantedTag,2}_{\beta,t} \mu_{G,\beta,t,<\varepsilon}^{2,{\sigma'}^*,{\boldsymbol{v}}}(R(\sigma^1,\sigma^2) \,|\, I)
    \right)\,.
  \]
\begin{proof}
  Analogous to ~\cref{l:fe-deriv}.
\end{proof}
\end{lemma}
\begin{fact}\label{f:replica-overlap-nonnegative-two-replica}
  For all $\beta\in \R$, $t \ge 0$, all realizations of ${\sigma'}^*,G,{\boldsymbol{v}}$, and interval $I \in {\mathcal{I}}$,
  \[
    \mu_{G,\beta,t,<\varepsilon}^{2,{\sigma'}^*,{\boldsymbol{v}}}(R(\sigma^1,\sigma^2) \,|\, I) \ge 0\,.
  \]
\begin{proof}
  Analogous to Fact~\cref{f:replica-overlap-nonnegative}.
\end{proof}
\end{fact}
Similarly to $\phi_k$ defined in \eqref{e:phik}, let
\[
  \psi_k(t)
  = \E^{\plantedTag,2}_{\beta,t} F_{n,\beta,t}^{2,\sigma^*,{\sigma'}^*,G,{\boldsymbol{v}}}(I_*)
  - \E^{\plantedTag,2}_{\beta,t} F_{n,\beta,t}^{2,\sigma^*,{\sigma'}^*,G,{\boldsymbol{v}}}(I_k)\,.
\]
\begin{lemma}\label{l:boosting-init-two-replica}
  If the hypothesis \eqref{e:planted-boosting-hypothesis-two-replica} of ~\cref{p:planted-boosting}~\cref{i:planted-boosting-two-replica} holds, then $\psi_0(0) \le {\widetilde L}(0,2\upsilon)$.
\begin{proof}
  Analogous to ~\cref{l:boosting-init}, using ~\cref{l:planted-fe-conc}~\cref{i:planted-fe-conc-two-replica} in place of ~\cref{l:planted-fe-conc}~\cref{i:planted-fe-conc-one-replica}.
\end{proof}
\end{lemma}
\begin{proposition}\label{p:boosting-core-two-replica}
  Let $t>0$ and suppose that
  \[
    \E^{\plantedTag,2}_{\beta,t} F_{n,\beta,t}^{2,\sigma^*,{\sigma'}^*,G,{\boldsymbol{v}}}([-1,1])
    < \E^{\plantedTag,2}_{\beta,t} F_{n,\beta,t}^{2,\sigma^*,{\sigma'}^*,G,{\boldsymbol{v}}}(I_*) + \upsilon\,.
  \]
  Then, for each $0\le k\le k_\upperBoundTag - 1$, either:
  \begin{enumerate}[label=(\roman*)]
    \item \label{i:formal-fe-make-progress-two-replica} $\psi'_k(t) \le -\varepsilon/10$, or
    \item \label{i:formal-fe-increment-k-two-replica} $\psi_{k+1}(t) \le {\widetilde L}(t, 2(\psi_k(t) + 3\upsilon))$.
  \end{enumerate}
  If $k = k_\upperBoundTag$, then outcome \cref{i:formal-fe-make-progress-two-replica} holds.
\begin{proof}
  Analogous to ~\cref{p:boosting-core}, using ~\cref{l:fe-deriv-two-replica}, Fact~\cref{f:replica-overlap-nonnegative-two-replica}, and ~\cref{l:planted-fe-conc}~\cref{i:planted-fe-conc-two-replica} in place of ~\cref{l:fe-deriv}, Fact~\cref{f:replica-overlap-nonnegative}, and ~\cref{l:planted-fe-conc}~\cref{i:planted-fe-conc-one-replica}.
\end{proof}
\end{proposition}
\begin{proof}[Proof of ~\cref{p:planted-boosting}~\cref{i:planted-boosting-two-replica}]
  This is analogous to the proof of part \cref{i:planted-boosting-one-replica} above, using ~\cref{l:planted-fe-conc}~\cref{i:planted-fe-conc-two-replica},\linebreak[4] ~\cref{l:boosting-init-two-replica} and ~\cref{p:boosting-core-two-replica} in place of ~\cref{l:planted-fe-conc}~\cref{i:planted-fe-conc-one-replica}, ~\cref{l:boosting-init} and ~\cref{p:boosting-core}.
  In particular, we can define the sequence $\tau_k$ like in ~\cref{d:dynamical-process}, but with $\E^{\plantedTag,2}_{\beta,t} F_{n,\beta,t}^{2,\sigma^*,{\sigma'}^*,G,{\boldsymbol{v}}}$ in place of $\E^\plantedTag_{\beta,t} F_{n,\beta,t}^{\sigma^*,G,{\boldsymbol{v}}}$ and $\psi_k$ in place of $\phi_k$.
  Then the bound from ~\cref{l:boosting-recursive-ubs} holds verbatim.
  The rest of the proof is analogous.
\end{proof}

\subsection{Completion of truncation estimates}
\label{ss:truncation-completion}

We now complete the proofs of ~\cref{t:truncation} and Corollary~\cref{c:truncation}.

\begin{proof}[Proof of ~\cref{t:truncation}]
  Throughout this proof, we will write $\E^\nullModelTag$ for expectation with respect to $G \sim \configurationModel(n,d)$, and $\Pr^\plantedTag_{\beta}$, $\E^\plantedTag_{\beta}$ for probability and expectation with respect to the planted distribution $\nu^\plantedTag_{\beta}$ from ~\cref{d:planted-model}.

  Recall that $\varepsilon_0 = \varepsilon_0(d,\delta)$ is given by ~\cref{l:basic-mt}.
  Let $c(d,\delta,\varepsilon)$ be the function given by ~\cref{p:truncation}, and let 
  \[
    c = c(d,\delta,\varepsilon) / 7\,.
  \]
  Consider any $\beta \in [2\beta_\lowerBoundTag, \betaKS - \delta]$.
  The conclusion \eqref{e:truncation-O} of ~\cref{p:truncation}~\cref{i:truncation-one-replica} is equivalent to
  \begin{equation}\label{e:truncation-input}
    \Pr^\plantedTag_{\beta}(G\not\in {\mathcal{O}}_{\beta,7c}) \le e^{-7cn}\,.
  \end{equation}
  Likewise, the desired conclusion \eqref{e:truncation} of the current theorem is equivalent to
  \begin{equation}\label{e:truncation-goal}
    \Pr^\plantedTag_{\beta}(G\not\in {\mathcal{E}}_{\beta,\varepsilon,c}) \le e^{-cn}\,.
  \end{equation}
  Recall for convenience that
  \begin{align*}
    {\mathcal{O}}_{\beta,7c} &= \left\{
      \mu_{G,\beta,<\varepsilon}(|R(\sigma^1,\sigma^2)| \ge \varepsilon) \le e^{-7cn}
    \right\}
    = \left\{
      Z_{G,\ge \varepsilon}^{(2)}(\beta) \le e^{-7cn} Z_{G,< \varepsilon}(\beta)^2
    \right\}\,, \\
    {\mathcal{E}}_{\beta,\varepsilon,c} &= \left\{
      Z_{G,\ge \varepsilon}^{(2)}(\beta) \le e^{-cn} \E^\nullModelTag[Z_{G,< \varepsilon}(\beta)]^2
    \right\}\,.
  \end{align*}
  It follows that
  \[
    {\mathcal{O}}_{\beta,7c} \cap \left\{
      Z_{G,< \varepsilon}(\beta) \le e^{3cn} \E^\nullModelTag[Z_{G,< \varepsilon}(\beta)]
    \right\} \subseteq {\mathcal{E}}_{\beta,\varepsilon,c}\,,
  \]
  and thus
  \begin{equation}\label{e:truncation-step1}
    \Pr^\plantedTag_{\beta}(G\not\in {\mathcal{E}}_{\beta,\varepsilon,c})
    \le \Pr^\plantedTag_{\beta}(G\not\in {\mathcal{O}}_{\beta,7c})
    + \Pr^\plantedTag_{\beta}\left(Z_{G,< \varepsilon}(\beta) > e^{3cn} \E^\nullModelTag[Z_{G,< \varepsilon}(\beta)]\right)\,.
  \end{equation}
  We now bound the last probability.
  Note that on the event ${\mathcal{O}}_{\beta,7c}$,
  \[
    Z_{G,< \varepsilon}^{(2)}(\beta)
    = Z_{G,< \varepsilon}(\beta)^2 - Z_{G,\ge \varepsilon}^{(2)}(\beta)
    \ge (1-e^{-7cn}) Z_{G,< \varepsilon}(\beta)^2\,.
  \]
  Thus, by ~\cref{l:basic-mt}~\cref{i:basic-mt-annealed-2},
  \[
    \E^\nullModelTag\left[
      \allone\{G \in {\mathcal{O}}_{\beta,7c}\} Z_{G,< \varepsilon}(\beta)^2
    \right]
    \le (1-e^{-7cn})^{-1} \E^\nullModelTag\left[Z_{G,< \varepsilon}^{(2)}(\beta)\right]
    \le e^{cn} \E^\nullModelTag[Z_{G,< \varepsilon}(\beta)]^2\,.
  \]
  This implies
  \[
    \E^\plantedTag_{\beta}\left[
      \allone\{G \in {\mathcal{O}}_{\beta,7c}\} Z_{G,< \varepsilon}(\beta)
    \right]
    \le e^{cn} \E^\nullModelTag[Z_{G,< \varepsilon}(\beta)]\,.
  \]
  By Markov's inequality,
  \begin{align*}
    &\Pr^\plantedTag_{\beta}\left(
      Z_{G,< \varepsilon}(\beta) > e^{3cn} \E^\nullModelTag[Z_{G,< \varepsilon}(\beta)]
    \right) \\
    &\le \Pr^\plantedTag_{\beta}(G \not\in {\mathcal{O}}_{\beta,7c})
    + \Pr^\plantedTag_{\beta}\left(
      \allone\{G\in {\mathcal{O}}_{\beta,7c}\} Z_{G,< \varepsilon}(\beta) 
      > e^{3cn} \E^\nullModelTag[Z_{G,< \varepsilon}(\beta)]
    \right) \\
    &\le \Pr^\plantedTag_{\beta}(G \not\in {\mathcal{O}}_{\beta,7c}) + e^{-2cn}\,.
  \end{align*}
  Plugging into \eqref{e:truncation-step1} shows
  \[
    \Pr^\plantedTag_{\beta}(G\not\in {\mathcal{E}}_{\beta,\varepsilon,c})
    \le 2\Pr^\plantedTag_{\beta}(G \not\in {\mathcal{O}}_{\beta,7c}) + e^{-2cn}\,.
  \]
  Together with \eqref{e:truncation-input} this implies \eqref{e:truncation-goal}.
\end{proof}
We turn to the proof of Corollary~\cref{c:truncation}. 
\begin{lemma}\label{l:plant-one-of-many}
  Let $\E^\plantedTag_{\beta}$ denote expectation with respect to $({\sigma'}^*,G) \sim \nu^\plantedTag_{\beta}$ from ~\cref{d:planted-model} (and not $(\sigma^*,G)$, which was the convention earlier in the paper.)
  \begin{enumerate}[label=(\alph*)]
    \item \label{i:plant-one-of-two} Recall ${\widetilde Z}^{(2)}_{G,< \varepsilon}(\beta,{\sigma'})$ defined in \eqref{e:tZ2}.
    Then,
    \[
      \E^\plantedTag_{\beta} [{\widetilde Z}^{(2)}_{G,< \varepsilon}(\beta,{\sigma'}^*)] 
      = \frac{\E[Z^{(2)}_{G,< \varepsilon}(\beta)]}{\E[Z_{G,< \varepsilon}(\beta)]}\,.
    \]
    \item \label{i:plant-one-of-three} Define
    \begin{align*}
      {\widetilde Z}^{(2,2)}_{G,< \varepsilon}(\beta,{\sigma'}^*) &= 
      \sum_{\sigma^1,\sigma^2 \in \Sigma_n(\varepsilon)}
      \allone\{\max(|R(\sigma^1,{\sigma'}^*)|, |R(\sigma^2,{\sigma'}^*)|, |R(\sigma^1,\sigma^2)|) < \varepsilon\} \\
      &\qquad\qquad\qquad \exp(\beta \isingHamiltonian_{G}(\sigma^1) + \beta \isingHamiltonian_{G}(\sigma^2))\,.
    \end{align*}
    Then, for $Z_{G,< \varepsilon}^{(3)}(\beta)$ defined in \eqref{e:Z3},
    \[
      \E^\plantedTag_{\beta}[{\widetilde Z}^{(2,2)}_{G,< \varepsilon}(\beta,{\sigma'}^*)] 
      = \frac{\E[Z_{G,< \varepsilon}^{(3)}(\beta)]}{\E[Z_{G,< \varepsilon}(\beta)]}\,.
    \]
  \end{enumerate}
\begin{proof}
  We write
  \[
    \frac{\E[Z_{G,< \varepsilon}^{(2)}(\beta)]}{\E[Z_{G,< \varepsilon}(\beta)]}
    = \frac{
      \E[\sum_{\sigma^2 \in \Sigma_n(\varepsilon)} \exp(\beta \isingHamiltonian_{G}(\sigma^2)) {\widetilde Z}^{(2)}_{G,< \varepsilon}(\beta,\sigma^2)]
    }{
      \E[\sum_{\sigma^2 \in \Sigma_n(\varepsilon)} \exp(\beta \isingHamiltonian_{G}(\sigma^2))]
    }\,.
  \]
  We identify $\sigma^2$ with ${\sigma'}^*$ and recognize this as an expectation over $\E^\plantedTag_{\beta}$.
  Similarly,
  \[
    \frac{\E[Z_{G,< \varepsilon}^{(3)}(\beta)]}{\E[Z_{G,< \varepsilon}(\beta)]}
    = \frac{
      \E[\sum_{\sigma^3 \in \Sigma_n(\varepsilon)} \exp(\beta \isingHamiltonian_{G}(\sigma^3)) {\widetilde Z}^{(2,2)}_{G,< \varepsilon}(\beta,\sigma^3)]
    }{
      \E[\sum_{\sigma^3 \in \Sigma_n(\varepsilon)} \exp(\beta \isingHamiltonian_{G}(\sigma^3))]
    }\,.
  \]
  We identify $\sigma^3$ with ${\sigma'}^*$ and recognize this as an expectation over $\E^\plantedTag_{\beta}$.
\end{proof}
\end{lemma}

\begin{proof}[Proof of Corollary~\cref{c:truncation}]
  Recall that $\varepsilon_0 = \varepsilon_0(d,\delta)$ is given by ~\cref{l:basic-mt}.
  Let $c^{\cref{t:truncation}}(d,\delta,\varepsilon)$, $c^{\cref{p:truncation}}(d,\delta,\varepsilon)$ be the functions given by ~\cref{t:truncation} and ~\cref{p:truncation}.
  Let
  \[
    c = \min\left(
      \frac14 c^{\cref{t:truncation}}(d,\delta,\varepsilon), 
      \frac12 c^{\cref{p:truncation}}(d,\delta,\varepsilon)
    \right)\,.
  \]
  Then,
  \begin{equation}\label{e:cor-truncation-step1}
    \frac{
      \E\left[
        \allone\{G\not\in {\mathcal{E}}_{\beta,\varepsilon,c}\}
        Z_{G,< \varepsilon}^{(2)}(\beta)
      \right]
    }{
      \E[Z_{G,< \varepsilon}^{(2)}(\beta)]
    }
    = \Pr^{\plantedTag,2}_{\beta}(G \not\in {\mathcal{E}}_{\beta,\varepsilon,c})
    \le {\textup{(i)}} + {\textup{(ii)}}\,,
  \end{equation}
  for
  \begin{align*}
    {\textup{(i)}} &= \Pr^{\plantedTag,2}_{\beta}(({\sigma'}^*,G) \not\in {\mathcal{O}}^{(2)}_{\beta,2c})\,, & 
    {\textup{(ii)}} &= \Pr^{\plantedTag,2}_{\beta}\left(
      ({\sigma'}^*,G) \in {\mathcal{O}}^{(2)}_{\beta,2c}
      \,\,\text{and}\,\,
      G \not\in {\mathcal{E}}_{\beta,\varepsilon,c}
    \right)
  \end{align*}
  The left-hand side of \eqref{e:truncation-O-two-replica} is precisely 
  \[
    \Pr^{\plantedTag,2}_{\beta}\left(({\sigma'}^*,G) \not\in {\mathcal{O}}^{(2)}_{\beta,c^{\cref{p:truncation}}(d,\delta,\varepsilon)}\right)
    \cdot \E[Z_{G,< \varepsilon}^{(2)}(\beta)]\,.
  \]
  So, ~\cref{p:truncation}~\cref{i:truncation-two-replica} implies ${\textup{(i)}} \le e^{-2cn}$.
  We can bound ${\textup{(ii)}}$ by the following comparison to $\nu^\plantedTag_{\beta}$ from ~\cref{d:planted-model}.
  Let $\Pr^\plantedTag_{\beta}$, $\E^\plantedTag_{\beta}$ denote probability and expectation with respect to $({\sigma'}^*,G) \sim \nu^\plantedTag_{\beta}$ (and not $(\sigma^*,G)$, which was the convention earlier in the paper).
  Observe that the likelihood ratio between the marginal law of $({\sigma'}^*,G) \sim \nu^{\plantedTag,2}_{\beta}$ and $({\sigma'}^*,G) \sim \nu^\plantedTag_{\beta}$ is precisely
  \[
    \frac{\differential \nu^{\plantedTag,2}_{\beta}}{\differential \nu^\plantedTag_{\beta}}({\sigma'}^*,G)
    = \frac{
      {\widetilde Z}^{(2)}_{G,< \varepsilon}(\beta,{\sigma'}^*)
    }{
      \E^\plantedTag_{\beta} [{\widetilde Z}^{(2)}_{G,< \varepsilon}(\beta,{\sigma'}^*)]
    }\,,
  \]
  where ${\widetilde Z}^{(2)}_{G,< \varepsilon}(\beta,{\sigma'}^*)$ is defined in \eqref{e:tZ2}.
  Thus, 
  \[
    {\textup{(ii)}}
    = \E^\plantedTag_{\beta}\left[
      \frac{
        {\widetilde Z}^{(2)}_{G,< \varepsilon}(\beta,{\sigma'}^*)
      }{
        \E^\plantedTag_{\beta} [{\widetilde Z}^{(2)}_{G,< \varepsilon}(\beta,{\sigma'}^*)]
      } \allone\left\{
        ({\sigma'}^*,G) \in {\mathcal{O}}^{(2)}_{\beta,2c}
        \,\,\text{and}\,\,
        G \not\in {\mathcal{E}}_{\beta,\varepsilon,c}
      \right\}
    \right]\,.
  \]
  By the Cauchy--Schwarz inequality, we further have
  \begin{equation}\label{e:cor-truncation-step2}
    {\textup{(ii)}} \le \frac{
      {\textup{(iii)}}^{1/2} \cdot {\textup{(iv)}}^{1/2}
    }{
      \E^\plantedTag_{\beta} [{\widetilde Z}^{(2)}_{G,< \varepsilon}(\beta,{\sigma'}^*)]
    }\,,
  \end{equation}
  for
  \begin{align*}
    {\textup{(iii)}} &= \E^\plantedTag_{\beta}\left[
      \allone\{({\sigma'}^*,G) \in {\mathcal{O}}^{(2)}_{\beta,2c}\} 
      ({\widetilde Z}^{(2)}_{G,< \varepsilon}(\beta,{\sigma'}^*))^2
    \right]\,, & 
    {\textup{(iv)}} &= \Pr^\plantedTag_{\beta}\left(G \not\in {\mathcal{E}}_{\beta,\varepsilon,c}\right)\,.
  \end{align*}
  Recall ${\widetilde Z}^{(2,2)}_{G,< \varepsilon}(\beta,{\sigma'}^*)$ from ~\cref{l:plant-one-of-many}.
  On the event ${\mathcal{O}}^{(2)}_{\beta,2c}$, we have
  \[
    {\widetilde Z}^{(2,2)}_{G,< \varepsilon}(\beta,{\sigma'}^*) \ge (1-e^{-2cn}) ({\widetilde Z}^{(2)}_{G,< \varepsilon}(\beta,{\sigma'}^*))^2\,.
  \]
  Thus, by Lemmas~\cref{l:plant-one-of-many}~\cref{i:plant-one-of-three} and \cref{l:basic-mt}~\cref{i:basic-mt-annealed-3},
  \begin{align*}
    {\textup{(iii)}}
    &\le (1-e^{-2cn})^{-1} \cdot
    \E^\plantedTag_{\beta}\left[
      {\widetilde Z}^{(2,2)}_{G,< \varepsilon}(\beta,{\sigma'}^*)
    \right] \\
    &= (1-e^{-2cn})^{-1} \cdot
    \frac{\E[Z_{G,< \varepsilon}^{(3)}(\beta)]}{\E[Z_{G,< \varepsilon}(\beta)]}
    \le e^{cn} \E[Z_{G,< \varepsilon}(\beta)]^2\,.
  \end{align*}
  ~\cref{t:truncation} implies
  \[
    {\textup{(iv)}} \le \Pr^\plantedTag_{\beta}\left(G \not\in {\mathcal{E}}_{\beta,4c}\right) \le e^{-4cn}\,.
  \]
  Using ~\cref{l:plant-one-of-many}~\cref{i:plant-one-of-two} to evaluate the denominator of \eqref{e:cor-truncation-step2}, we deduce
  \[
    {\textup{(ii)}} \le 
    {\textup{(iii)}}^{1/2} \cdot {\textup{(iv)}}^{1/2} \cdot \frac{\E[Z_{G,< \varepsilon}(\beta)]}{\E[Z^{(2)}_{G,< \varepsilon}(\beta)]}
    \le e^{-3cn/2} \cdot \frac{\E[Z_{G,< \varepsilon}(\beta)]^2}{\E[Z_{G,< \varepsilon}^{(2)}(\beta)]}\,.
  \]
  Plugging back into \eqref{e:cor-truncation-step1}, and multiplying through by $\E[Z_{G,< \varepsilon}^{(2)}(\beta)]$, yields
  \[
    \E\left[
      \allone\{G\not\in {\mathcal{E}}_{\beta,\varepsilon,c}\}
      Z_{G,< \varepsilon}^{(2)}(\beta)
    \right]
    \le e^{-2cn} \E[Z_{G,< \varepsilon}^{(2)}(\beta)]
    + e^{-3cn/2} \E[Z_{G,< \varepsilon}(\beta)]^2\,.
  \]
  Finally ~\cref{l:basic-mt}~\cref{i:basic-mt-annealed-2} implies
  \[
    \E[Z_{G,< \varepsilon}^{(2)}(\beta)] 
    \le e^{cn/2} \E[Z_{G,< \varepsilon}(\beta)]^2\,,
  \]
  which concludes the proof.
\end{proof}

\section{Moment estimates for complex \texorpdfstring{$\beta$}{beta}}
\label{s:moments}


In this section, we prove \cref{l:1mt,t:reweight}. We first define the reweighting factor $A_{G}(\beta)$ required in \cref{t:reweight}, after which the first and (reweighted) second moments can be treated in a conceptually similar manner. Throughout this section, we work directly with the uniformly random perfect matching $M$ on $dn$ sockets, rather than the random $d$-regular (multi)graph $G$ it contracts to.

\paragraph{The Reweighting Factor} We begin by motivating the reweighting factor; the impatient reader may skip directly to \cref{eq:reweight-factor} for the precise definition. Let $C_{k}$ denote the number of $k$-cycles in $M$. \cite{FL21} proved that
\begin{align*}
    \frac{Z_{M}(\beta)}{\E[Z_{M}(\beta)]} \overset{\mathsf{d}}{\to} \prod_{k=1}^{\infty} (1 + \delta_{k})^{\Lambda_{k}} \exp(-\lambda_{k}\delta_{k}) \qquad \text{as} \qquad n \to \infty,
\end{align*}
where $\Lambda_{k} \sim \mathsf{Poi}(\lambda_{k})$ are independent random variables asymptotically describing the joint distribution of $\{C_{k}\}_{k=1}^{\infty}$ for a uniformly random perfect matching $M$ on $V \times [d]$. This suggests the following ``ideal'' reweighting:
\begin{align*}
    A_{L}(\beta) \defeq \exp\wrapp{\sum_{k=1}^{L} \lambda_{k} \cdot \delta_{k}(\beta)} \cdot \prod_{k=1}^{L} (1 + \delta_{k}(\beta))^{-C_{k}}.
\end{align*}
For convenience, we write
\begin{align*}
    b_{k} \defeq \wrapp{1 + \delta_{k}}^{-1} - 1, \qquad \forall k \in \N.
\end{align*}
This is a complex-valued function of $\beta$ which may be negative or positive; indeed, for $\beta \geq 0$, the $b_{k}$ alternate in sign depending on $k$. Then, the ``ideal'' reweighting can be written as
\begin{align*}
    A_{L}(\beta) &= \exp\wrapp{\sum_{k=1}^{L} \lambda_{k}\delta_{k}} \cdot \prod_{k=1}^{L} (1 + b_{k})^{C_{k}} \\
    &= \exp\wrapp{\sum_{k=1}^{L} \lambda_{k}\delta_{k}} \cdot \prod_{k=1}^{L} \prod_{C \in \mathscr{C}_{k}} (1 + b_{k})^{\allone[C \subseteq M]} \\
    &= \exp\wrapp{\sum_{k=1}^{L} \lambda_{k}\delta_{k}} \cdot \prod_{k=1}^{L} \prod_{C \in \mathscr{C}_{k}} \wrapp{1 + b_{k} \cdot \allone[C \subseteq M]} \\
    &= \exp\wrapp{\sum_{k=1}^{L} \lambda_{k}\delta_{k}} \cdot \sum_{\Gamma} \allone[\Gamma \subseteq M] \prod_{C \in \Gamma} b_{\abs{C}},
\end{align*}
where $\Gamma$ denotes a collection of cycles having length at most $L$, and we define the event that all edges of $\Gamma$ occur in $M$ via
\begin{align*}
    \wrapc{\Gamma \subseteq M} \defeq \bigwedge_{C \in \Gamma} \wrapc{C \subseteq M}.
\end{align*}
We can think of $\Gamma$ as a graph since the union of edges of all $C \in \Gamma$ must still form a valid matching on the $dn$ sockets. Note, however, that a given socket pair can appear multiple times in $\Gamma$, but any given socket can have only at most one other socket matched to it.

For technical reasons related to our inability to adequately control the contribution of large $\Gamma$, we will need to impose an additional truncation. To do this, for a tuple $\bm{c} \in \N^{L}$, we write
\begin{align*}
    \abs{\bm{c}} \defeq \sum_{k=1}^{L} k \cdot \bm{c}_{k}.
\end{align*}
For a given $\Gamma$, we write $\bm{c}(\Gamma) \in \N^{L}$ for the vector where $\bm{c}_{k}(\Gamma)$ gives the number of $k$-cycles in $\Gamma$. We then write
\begin{align*}
    \abs{\Gamma} = \abs{\bm{c}(\Gamma)} = \sum_{C \in \Gamma} \abs{C}.
\end{align*}
We view $\abs{\Gamma}$ as the total number of edges in $\Gamma$ with multiplicity, i.e. if a given socket pair occurs in multiple $C \in \Gamma$, it will be counted multiple times in $\abs{\Gamma}$. We will stratify everything by $\bm{c}(\Gamma)$. For convenience, for $\bm{c} \in \N^{L}$, let
\begin{align*}
    \mathscr{A}(\bm{c}) \defeq \wrapc{\Gamma \subseteq \bigcup_{k=1}^{L} \mathscr{C}_{k} : \bm{c}(\Gamma) = \bm{c}}, \qquad \mathscr{A}(m) \defeq \bigcup_{\bm{c} \in \N^{L} : \abs{\bm{c}} = m} \mathscr{A}(\bm{c}), \qquad \mathscr{A} \defeq \bigcup_{m \in \N} \mathscr{A}(m).
\end{align*}
We also define ``good'' and ``bad'' $\Gamma$ via
\begin{align}\label{eq:good-bad-graphs}
\begin{split}
    \mathscr{G}(\bm{c}) &\defeq \wrapc{\Gamma \in \mathscr{A}(\bm{c}) : \Gamma \text{ consists of vertex-disjoint cycles}} \subseteq \mathscr{A}(\bm{c}) \\
    \mathscr{G}(m) &\defeq \bigcup_{\bm{c} \in \N^{L} : \abs{\bm{c}} = m} \mathscr{G}(\bm{c}) \subseteq \mathscr{A}(m), \qquad \mathscr{G} \defeq \bigcup_{m \in \N} \mathscr{G}(m) \\
    \mathscr{B}(\bm{c}) &\defeq \mathscr{A}(\bm{c}) \setminus \mathscr{G}(\bm{c}), \qquad \mathscr{B}(m) \defeq \mathscr{A}(m) \setminus \mathscr{G}(m), \qquad \mathscr{B} \defeq \bigcup_{m \in \N} \mathscr{B}(m).
\end{split}
\end{align}
Note that for any $\Gamma \in \mathscr{G}$, any cycle can appear at most once. With all of the above notation in hand, let us define the ``true'' reweighting we will use by
\begin{align}\label{eq:reweight-factor}
    A_{L,m}(\beta) \defeq \exp\wrapp{\sum_{k=1}^{L} \lambda_{k}\delta_{k}} \cdot \sum_{\Gamma \in \mathscr{G}(\leq m)} \allone[\Gamma \subseteq M] \prod_{C \in \Gamma} b_{\abs{C}},
\end{align}
for appropriately chosen $m$; eventually, we will see that $L,m\in\N$ are chosen to be sufficiently slow-growing functions of $n$.

\begin{lemma}[Restatement of \cref{i:AN-complex-analytic} from \cref{t:reweight}]\label{lem:reweighting-analytic}
For all $L,m \in \N$ and all possible outcomes of the perfect matching $M$, the functions $A_{L}(\beta)$ and $A_{L,m}(\beta)$ are analytic on the strip $\wrapc{\beta \in \C : \abs{\imag \beta} < \frac{\pi}{2}}$, and $A_{L}(0) = A_{L,m}(0) = 1$.
\end{lemma}
\begin{remark}
Note that this region contains any sufficiently small constant neighborhood of $[0,\beta_{\mathsf{KS}}]$.
\end{remark}
To minimize distractions, the proof of this lemma is deferred to \cref{sec:technical}. The following statements are the main results in this section.

\begin{theorem}[Reweighted Second Moment; Restatement of \cref{i:AN-second-moment} from \cref{t:reweight}]\label{thm:reweighted-2ndmoment}
Let $L = \frac{\log\log N}{\log\log\log N}$ and $m = \frac{\log N}{(\log\log N)^2}$. For all $0 < \eta < 1$, all $\beta$ within a sufficiently small constant (depending on $d,\eta$) neighborhood of $[0,(1-\eta) \cdot \beta_{\mathsf{KS}}]$, and all sufficiently small $\epsilon > 0$
\begin{align*}
    \frac{\E_{M}\wrapb{Z_{M,<\epsilon}^{(2)}(\beta) \cdot \abs{A_{L,m}(\beta)}^{2}}}{\abs{\E_{M}\wrapb{Z_{M,<\epsilon}(\beta)}}^{2}} = 1 \pm o_{n}(1),
\end{align*}
where the constants in $o_{n}(1)$ may depend on $d,\epsilon,\eta$.
\end{theorem}
\begin{theorem}[Bounded Moments; Restatement of \cref{i:AN-moments-bounded} from \cref{t:reweight}]\label{thm:bounded-moments}
Let $L = \frac{\log\log N}{\log\log\log N}$ and $m = \frac{\log N}{(\log\log N)^2}$. For all $0 < \eta < 1$ and all fixed $K > 0$ (independent of $n$), there exists a constant $C = C(K,d,\eta) > 0$ such that for all $\beta$ within a sufficiently small constant (depending on $d,\eta$) neighborhood of $[0,(1-\eta) \cdot \beta_{\mathsf{KS}}]$,
\begin{align*}
    \E\wrapb{\abs{A_{L,m}(\beta)}^{K}} \leq C.
\end{align*}
\end{theorem}

The remainder of this section is organized as follows. In \cref{subsec:1st-moment-complex}, we prove \cref{l:1mt} using contour integration and the method of steepest descent. In \cref{subsec:2nd-moment-complex}, we prove \cref{thm:reweighted-2ndmoment} assuming a certain conditional second moment estimate stated as \cref{prop:2nd-moment}. The proof of the latter is deferred to \cref{s:deferred-complex-moment}, as it is conceptually identical to the proof of \cref{l:1mt}, but complicated by the presence of more variables. The proof (sketch) of \cref{thm:bounded-moments} is deferred to \cref{sec:reweight-bounded-moments}, as it is conceptually identical to the proof (contained in \cref{subsec:2nd-moment-complex}) that \cref{prop:2nd-moment} implies \cref{thm:reweighted-2ndmoment}.

\subsection{First moment: Proof of \texorpdfstring{\cref{l:1mt}}{First Moment}}\label{subsec:1st-moment-complex}
Let us recall that we work with the configuration model on $N=nd$ sockets and $n$ clouds. 

We first give an alternative expression for $ Z_{G,< \varepsilon}(\beta)$ in terms of a contour integral.
Let for $\alpha\in [0,1]$ and $N\in \mathbb{N}$,
\[
  \Phi_{z,\alpha}^{(N)}(x,y)
  := \frac{x^2}{2} + \frac{y^2}{2} + z\cdot xy
  - \left(\alpha + \frac{1}{N}\right) \log x
  - \left(1 - \alpha + \frac{1}{N}\right) \log y\,.
\]
\begin{lemma}\label{lem:1st moment contour representation}
For any $\varepsilon\in (0,1/2]$ we have   
\begin{equation}\label{e:1mt-restate}
  \E[Z_{M,< \varepsilon}(\beta)]
  = \sum_\alpha \binom{n}{\alpha n} 
  \cdot \frac{(\alpha N)!((1-\alpha) N)!}{(N-1)!! (N)^{N/2}} 
  \cdot \frac{1}{(2\pi i)^2}
  \oint \exp\left(N \cdot \Phi^{(N)}_{e^\beta, \alpha}(x,y)\right)\,\differential (x,y)\,,
\end{equation}
where the sum is over $\alpha \in \{0,1/n,\ldots,1\}$ such that $|\alpha-1/2|<\varepsilon$.
\end{lemma}
\begin{proof}
By the exponential generating function for matchings (See e.g.~\cite[Chapter 5]{Stanleybookvol2}) we have 
\begin{equation}\label{e:1mt}
  \E[Z_{M}(\beta)]
  = \sum_{m=0}^n \binom{n}{m} \cdot \frac{(md)!((n-m)d)!}{(nd-1)!!} \cdot [u^{md} v^{(n-m)d}] \exp\left(\frac{u^2}{2} + \frac{v^2}{2} + e^\beta uv\right)\,.
\end{equation}
We can extract the coefficient from the power series using the residue theorem:
\[
  [u^{md} v^{(n-m)d}] \exp\left(\frac{u^2}{2} + \frac{v^2}{2} + e^\beta uv\right)
  = \frac{1}{(2\pi i)^2} \oint \frac{1}{u^{md+1} v^{(n-m)d+1}} \exp\left(\frac{u^2}{2} + \frac{v^2}{2} + e^\beta uv\right) \, \differential (u,v)
\]
Parameterizing $m = n\alpha$, $N=nd$, and $(u,v) = \sqrt{N} (x,y)$, the last display equals
\begin{align*}
  \frac{1}{(2\pi i)^2 (N)^{N/2}} 
  \oint \exp\left(N \cdot \Phi^{(N)}_{e^\beta, \alpha}(x,y)\right)\,\differential (x,y)\,.
\end{align*}
Restricting the summation to the relevant $\alpha$ gives the required expression.
\end{proof}

To evaluate the truncated first moment we will first evaluate the contour integral by steepest descend, after which we combine this with some calculations involving Stirling approximations to arrive at a proof of ~\cref{l:1mt}.

\subsubsection{The contour integral}

We start with gathering some definitions.
Consider positive real $z\ge 1$ and $\alpha = 1/2$.
  Then the stationarity condition of $\Phi^{(N)}_{z,\alpha}(x,y)$ yields
  \begin{align*}
    x + zy &= \left(\frac{1}{2} + \frac{1}{N} \right) \frac{1}{x} \\
    y + zx &= \left(\frac{1}{2} + \frac{1}{N} \right) \frac{1}{y}\,.
  \end{align*}
  By symmetry, this has solution 
  \[
    (x_z(\alpha = 1/2), y_z(\alpha = 1/2)) = (\lambda_\star,\lambda_\star) \,\,\text{or}\,\,(-\lambda_\star,-\lambda_\star)\,,
  \]
  where
  \[
    \lambda_\star = \lambda_\star(z)
    \equiv \sqrt{\frac{1}{1+z} \left(\frac{1}{2} + \frac{1}{N} \right)}\,.
  \]

\begin{proposition}\label{p:contour-integral} There exists $\varepsilon=\varepsilon(d)>0$ such that the following holds uniformly for 
$\alpha\in[1/2-\varepsilon,1/2+\varepsilon]$ and 
\[ 
    z=E+i\eta, \qquad E\in[1-\varepsilon,e^{\betaKS}+\varepsilon], \qquad |\eta|\leq\varepsilon. 
\] 
There is a neighborhood of $(E,1/2)$ on which there exists a unique stationary point 
$   (x_z(\alpha),y_z(\alpha)) $ of $\Phi^{(N)}_{z,\alpha}$. In particular, $(x_E(1/2),y_E(1/2))=(\lambda_\star(E),\lambda_star(E)).$

Moreover, 
\[ 
    \begin{aligned} 
        &\oint \exp\left(N\Phi^{(N)}_{z,\alpha}(x,y)\right) \,\differential(x,y) \\ &\qquad= -\left(1+O\left(\frac1n\right)\right) \frac{4\pi}{N} \frac{ \exp\left( N\Phi^{(N)}_{z,\alpha} \left(x_z(\alpha),y_z(\alpha)\right) \right) }{ \sqrt{ \det\left( -\nabla^2_{x,y}\Phi^{(N)}_{z,\alpha} \left(x_z(\alpha),y_z(\alpha)\right) \right) } }. 
    \end{aligned} 
\] 
The branch of the square root is chosen by continuous continuation from positive real $z$ and $\alpha=1/2$. \end{proposition}
\begin{proof}
\textbf{Existence of stationary point.}
For $\alpha\in [0,1]$ the stationary condition for $\Phi^{(N)}_{z,\alpha}$ can be written as
\begin{equation}\label{eq:stationary}
\left\{\begin{array}{ll} x^2 + zxy - \left(\alpha + \frac{1}{N} \right) &=0 \\
    y^2 + zxy - \left((1-\alpha) + \frac{1}{N} \right) &=0.
    \end{array}\right.
\end{equation}

We view this as the equation $F(x,y,z,\alpha)=0$ where $F:\C^4\to \C^2$ is the holomorphic function defined by the left-hand side of~\eqref{eq:stationary}.
The Jacobian of $F$ is given by
\begin{align*}
  J(F)=\left( \begin {array}{cccc} 2x+zy &zx& xy &-1
                                \\  zy&2y+zx& xy &1
                                \end{array}\right).
\end{align*}
The determinant of the matrix obtained from $J(F)$ by taking the first two columns is equal to $4xy+2x^2z+2y^2z$ and hence is non-zero for $\alpha=1/2$, $z\geq 1$ and $x=y=\pm \lambda_\star$.
It follows from the implicit function theorem there exists a open set $U$ containing $(1/2,z)$ and holomorphic functions $x$ and $y$ defined on $U$ such that on $U$, $F(x_z(\alpha),y_z(\alpha),z,\alpha)=0$ and $x_z(1/2)=y_z(1/2)=\lambda_\star$.
We may moreover assume that $U$ is chosen small enough so that both $x$ and $y$ are non-zero on $U$.
It thus follows that the required stationary point of $\Phi^{(N)}_{z,\alpha}$ exists.

\textbf{Computation of the integral.}
We now focus on the computation of the integral
  \begin{align*}
    \oint \exp\left(N \cdot \Phi^{(N)}_{z,\alpha}(x,y)\right)\,\differential (x,y).
    \end{align*}
Let us write
\[
    x_z(\alpha)=r_1e^{i\theta^*},\qquad y_z(\alpha)=r_2e^{i\theta^*}
\]
and set
\[
    \vartheta^*=(\theta^*,\phi^*)
\]
By Cauchy's theorem we may contour integrate the integral on the torus $\T=\{(x,y)\mid |x|=r_1, |y|=r_2\}$.
For the contour integral, we restrict to
\[
    \alpha\in n^{-1}\Z.
\]
Then $\alpha N$ and $(1-\alpha)N$ are integers, and $\exp(N\widetilde{\Phi}^{(N)}_{z,\alpha})$ is single-valued on the  torus.
When viewing $\Phi^{(N)}_{z,\alpha}$ as function of $(\theta,\phi)$ on $\T$ we will denote it as $\tilde\Phi^{(N)}_{z,\alpha}$.
We then have
\begin{align}
    &\oint
    \exp\left(
        N\Phi^{(N)}_{z,\alpha}(x,y)
    \right)
    \,\differential(x,y)
    \nonumber\\
    &\qquad=
    -r_1r_2
    \int_{[-\pi/2,3\pi/2]^2}
    \exp\left(
        N\tilde{\Phi}^{(N)}_{z,\alpha}(\theta,\phi)
    \right)
    e^{i(\theta+\phi)}
    \,\differential(\theta,\phi).
\label{eq:fm-complex-angular-integral}
\end{align}

It turns out that that it is natural to split the contour integral on the right-hand side of~\cref{eq:fm-complex-angular-integral} into six parts.
We shall see that the dominant contributions come from small neighboourhoods around the two stationary points.
We first look at the case $\alpha=1/2$ and $z=E\geq 1-\varepsilon$ (for $\varepsilon$ to be determined later) to obtain information about the Hessian of $\tilde \Phi^{(N)}_{z,\alpha}$.


First, the stationary points 
$(x_0,y_0)=\pm (\lambda_\star,\lambda_\star)$ correspond to $(0,0)$ and $(\pi,\pi)$ in the new coordinates. 
We focus on the point $(0,0)$.

Let us make a few observations. At $(0,0)$ the Hessian of $\tilde \Phi_{z,1/2}$ is given by
\begin{equation}
\label{eq:fm-real-angular-hessian}
    \nabla^2
    \tilde{\Phi}^{(N)}_{E,1/2}(0,0)
    =
    -(\lambda^\star(E))^2
    \begin{pmatrix}
        2+E & E\\
        E & 2+E
    \end{pmatrix}.
\end{equation}
Its eigenvalues are
\[
    -2(\lambda^\star(E))^2
    \qquad\text{and}\qquad
    -2(1+E)(\lambda^\star(E))^2
    =
    -2\left(
        \frac12+\frac1N
    \right).
\]
They are uniformly bounded away from zero when $E$ ranges over a fixed compact neighborhood of $[1,e^{\betaKS}]$.
Thus $(0,0)$ and $(\pi,\pi)$ are both local maxima of $\re(\tilde\Phi^{(N)}_{E,1/2})$, a direct computation gives that they are also the only global maxima on the torus $\T$.
It thus follows from \eqref{eq:fm-real-angular-hessian} and compactness, that there exists constants $\gamma,\delta,c>0$ such that, uniformly in $E$ and all sufficiently large $N$ we have
\begin{itemize}
    \item the eigenvalues of the real part of the Hessian of $\Tilde \Phi^{(N)}_{E,1/2}$ are bounded above by $-2c$ in the $\gamma$-neighborhoods of $(0,0)$ (and $(\pi,\pi)$ by symmetry);
    \item outside these two neighborhoods, 
    \[
    \re\tilde\Phi_{E,1/2}^{(N)}(\theta,\phi)\le \re\tilde\Phi_{E,1/2}^{(N)}(0,0)-2\delta.
    \]
\end{itemize}
We now perturb $(E,1/2)$ and extend these properties to the general setting.
To this end note that the higher order derivatives of $\tilde\Phi$ and the stationary point $(x_z(\alpha),y_z(\alpha)$ depend continuously on $(z,\alpha,1/N)$. 
Therefore, provided $\varepsilon$ is small enough, the proceeding uniform real-case estimates, continuity, symmetry and compactness imply that there exists a neighbourhood $U_1$ of $(0,0)$ containing $\vartheta^*$ such that the following holds:
\begin{enumerate}
    \item[P(1)] the eigenvalues of the real part of the $\tilde\Phi^{(N)}_{z,\alpha}$ are bounded above by $-c$
    on $U_1$ and on $U_2:=U_1+\pi(1,1)$;
    \item[P(2)] for all $(\theta,\phi)\in \T\setminus (U_1\cup U_2)$,
    \begin{equation}
    \label{eq:fm-perturbed-gap}
        \re\tilde{\Phi}^{(N)}_{z,\alpha}(\theta,\phi)
        \leq
        \re\tilde{\Phi}^{(N)}_{z,\alpha}(\vartheta^*)   - \delta;
    \end{equation}

    \item[P(3)] all derivatives $\tilde\Phi$ of order at most four are uniformly bounded on $U_1\cup U_2$.
\end{enumerate}

We now start computing the contributions to the contour integral~\cref{eq:fm-complex-angular-integral} from various parts of the torus, starting with the contribution of the stationary points.

\textbf{Contribution around the stationary point $\vartheta^*$} 
Let $z_0=x_z(\alpha)$ and $y_0=y_z(\alpha)$ and let
\[
    H=\nabla^2_{x,y}     \Phi^{(N)}_{z,\alpha}(x_0,y_0).
\]
Since $(x_0,y_0)$ is stationary, we have
\begin{equation} \label{eq:A relates to Hessian}
    A=\nabla^2
    \tilde{\Phi}^{(N)}_{z,\alpha}(\vartheta^*)=-DHD,
  \end{equation}
where $D=\diag(x_0,y_0)$. In particular, by observation P(1) we have
\begin{equation}
\label{eq:fm-negative-real-hessian}
    \re(t^{\matrixtranspose}At)\le -c\|t\|^2
\end{equation}
for every $t\in\R^2$.

Let $B_N=[-N^{-2/5},N^{-2/5}]^2\subseteq [-1,1]^2$. 
We want to understand the integral \eqref{eq:fm-complex-angular-integral} restricted to $\vartheta^*+B_N$. 
By Taylor's theorem we have
\begin{align}\label{eq:taylor for Phi}
\tilde{\Phi}^{(N)}_{z,\alpha}(\vartheta^*+s)=\tilde{\Phi}^{(N)}_{z,\alpha}(\vartheta^*)+\tfrac{1}{2} s^{\matrixtranspose}As+P_3(s)+\mathcal{E}(s)
\end{align}
Here $P_3$ is an odd homogeneous polynomials of degree three and the error term satisfies $|\mathcal{E}(s)|\le C\|s\|^4$ for some constant $C>0.$
By changing the variable $s=t/\sqrt{N}$ we obtain
\begin{align}
    &N\left[
        \tilde{\Phi}^{(N)}_{z,\alpha}
        \left(\vartheta^*+\frac{t}{\sqrt N} \right)
        -
        \tilde{\Phi}^{(N)}_{z,\alpha}(\vartheta^*)
    \right]
    +
    \frac{i(t_1+t_2)}{\sqrt N}
    \nonumber\\
    &\qquad=\frac12t^{\matrixtranspose}At
    +
    \frac1{\sqrt N}Q_3(t)
    +
    \tilde{\mathcal{E}}(t),
\label{eq:fm-rescaled-expansion}
\end{align}
where $Q_3(t):=P_3(t)+i(t_1+t_2)$,
and
\begin{equation}
\label{eq:fm-rescaled-remainder}
    |\tilde{\mathcal{E}}(t)|
    \leq
    Cn^{-1}\|t\|^4.
\end{equation}

Uniformly for $t\in [-N^{1/10},N^{1/10}]^2$ we have
\begin{equation}
\label{eq:fm-odd-cancellation}
    \exp\left(\frac{1}{\sqrt{N}}Q_3(t)+\tilde{\mathcal{E}}(t)\right)=1+\frac{1}{\sqrt{N}}Q_3(t)+O(
    n^{-2/5}).
\end{equation}
Since $Q_3$ is odd and $[-N^{1/10},N^{1/10}]^2$ is symmetric under $t\mapsto -t$ we have,
\[
    \int_{[-N^{1/10},N^{1/10}]^2}Q_3(t)\exp(\tfrac{1}{2}t^{\matrixtranspose}At) \differential t=0.
\]
This implies that 
\begin{align}
    \int_{B_N}\exp\left (\Tilde{\Phi}_{E,1/2}(\vartheta^*+s)-\Tilde{\Phi}_{E,1/2}(\vartheta^*)\right)\differential t&= \frac{1}{N}\int_{[-N^{1/10},N^{1/10}]^2}
    \exp\left(
        \frac12t^{\matrixtranspose}At
        +
        \frac1{\sqrt N}Q_3(t)
        +
        \tilde{\mathcal{E}}(t)
    \right)
    \,\differential t
    \nonumber\\
    &=
    \frac{1}{N}\int_{\R^2}
    \exp\left(
        \frac12t^{\matrixtranspose}At
    \right)
    \,\differential t
    +
    O\left(n^{-1}\right),
\label{eq:fm-local-gaussian-reduction}
\end{align}
where we use that $\exp(\tfrac{1}{2}t^T At)$ is bounded since $\re(A)$ is negative definite.

Since $A$ is complex symmetric and $\re(-A)$ is positive definite, the complex Gaussian identity gives
\begin{equation}
\label{eq:fm-complex-gaussian}
    \int_{\R^2}
    \exp\left(
        \frac12t^{\matrixtranspose}At
    \right)
    \,\differential t
    =
    \frac{2\pi}{
        \sqrt{\det(-A)}
    }.
\end{equation}
The square-root branch is obtained by continuous continuation from real negative-definite matrices. Identity \eqref{eq:fm-complex-gaussian} follows from the usual real Gaussian integral by analytic continuation; the estimate \eqref{eq:fm-negative-real-hessian} provides uniform integrable domination.

Combining \eqref{eq:fm-local-gaussian-reduction} and \eqref{eq:fm-complex-gaussian}, and recalling the Jacobian of the rescaling, gives
\begin{align}
    &\int_{\vartheta^*+B_N}
    \exp\left(
        N\tilde{\Phi}^{(N)}_{z,\alpha}(\theta,\phi)
    \right)
    e^{i(\theta+\phi)}
    \,\differential(\theta,\phi)
    \nonumber\\
    &\qquad=
    \left(
        1+O\left(\frac1n\right)
    \right)
    \frac{2\pi}{N}
    \frac{
        \exp\left(
            N\tilde{\Phi}^{(N)}_{z,\alpha}(\vartheta^*)
        \right)
        e^{i(\theta^*+\phi^*)}
    }{
        \sqrt{\det(-A)}
    }.
\label{eq:fm-one-angular-saddle}
\end{align}
The error in \eqref{eq:fm-one-angular-saddle} is relative because $\det(-A)$ is uniformly bounded above and bounded away from zero.

By~\cref{eq:A relates to Hessian} we have
\[
    \det(-A)=(x_0y_0)^2\det(-H) \Longleftrightarrow \sqrt{\det(-A)}= x_0y_0\sqrt{\det(-H)}.
\]
and
\[
    x_0y_0=r_1r_2e^{i(\theta^*+\phi^*)},
\]
thus the contribution of the neighborhood $B_N$ around $\vartheta^*$ to the contour integral~\cref{eq:fm-complex-angular-integral} is given by
\begin{equation}
\label{eq:fm-one-contour-saddle}
    -\left(
        1+O\left(\frac1n\right)
    \right)
    \frac{2\pi}{N}
    \frac{
        \exp\left(
            N\Phi^{(N)}_{z,\alpha}(x_0,y_0)
        \right)
    }{
        \sqrt{
            \det\left(
                -\nabla^2_{x,y}
                \Phi^{(N)}_{z,\alpha}(x_0,y_0)
            \right)
        }
    }.
\end{equation}

\noindent
\textbf{Contribution of the rest}
The other stationary point is at $\vartheta^*+(\pi,\pi)$, thus by symmetry the contribution around this point, i.e. $I'_N=\vartheta^*+(\pi,\pi)+B_N$is the same as \eqref{eq:fm-one-contour-saddle}.

Now, inside either of the two neighborhoods $U_1\cup U_2$ given by observation P(1), uniform strict concavity gives
\[ 
    \re\tilde\Phi_{z,\alpha}(\vartheta)\le \re\tilde\Phi_{z,\alpha}(\vartheta^*)-c\|\vartheta-\vartheta^*\|^2,
\]
Therefore, outside of $I_N\cup I'_N$ but inside $U_1\cup U_2$, we have
\[
    \re\tilde\Phi_{z,\alpha}(\vartheta)\le \re\tilde\Phi_{z,\alpha}(\vartheta^*)-cN^{-4/5},
\]
for some $c>0$, whereas outside $U_1\cup U_2$ we can replace $-cN^{-4/5}$ by $-\delta$ by P(2).
Together with  \eqref{eq:fm-perturbed-gap}, this yields
\begin{align}
    &\left|
        \int_{
            [-\pi/2,3\pi/2]^2
            \setminus
            \left(
                I_N
                \cup
                I'_N
            \right)
        }
        \exp\left(
            N\tilde{\Phi}^{(N)}_{z,\alpha}(\theta,\phi)
        \right)
        e^{i(\theta+\phi)}
        \,\differential(\theta,\phi)
    \right|
    \nonumber\\
    &\qquad\leq
    C
    \exp\left(
        N\re
        \tilde{\Phi}^{(N)}_{z,\alpha}(\vartheta^*)
    \right)
    e^{-cN^{1/5}}.
\label{eq:fm-tail-bound}
\end{align}

By compactness (and if necessary decreasing $\varepsilon$) $|x_z(\alpha)|,|y_z(\alpha)|$ are bounded away from $0$ and infinity as well as the determinant of the Hessian of $\Phi_{z,\alpha}$ at the stationary points. Therefore the magnitude of the sum of the two local terms~\cref{eq:fm-one-contour-saddle} is comparable to
\[
    \frac1N
    \exp\left(
        N\re
        \tilde{\Phi}^{(N)}_{z,\alpha}(\vartheta^*)
    \right).
\]
Thus \eqref{eq:fm-tail-bound} has relative size
\[
    O\left(
        Ne^{-cN^{1/5}}
    \right)
    =
    O\left(\frac1n\right).
\]
Adding the two identical saddle contributions and absorbing the remainder, we conclude that
\[
\begin{aligned}
    &\oint
    \exp\left(
        N\Phi^{(N)}_{z,\alpha}(x,y)
    \right)
    \,\differential(x,y)
    \\
    &\qquad=
    -\left(
        1+O\left(\frac1n\right)
    \right)
    \frac{4\pi}{N}
    \frac{
        \exp\left(
            N\Phi^{(N)}_{z,\alpha}
            \left(
                x_z(\alpha),y_z(\alpha)
            \right)
        \right)
    }{
        \sqrt{
            \det\left(
                -\nabla^2_{x,y}
                \Phi^{(N)}_{z,\alpha}
                \left(
                    x_z(\alpha),y_z(\alpha)
                \right)
            \right)
        }
    }.
\end{aligned}
\]
This proves the proposition.
\end{proof}

\subsubsection{Derivative calculations}
Here we calculate the first and second derivative of $\Phi^{(N)}_{z,\alpha}$ at the stationary point, when regarded as a function of $\alpha$.
\begin{lemma}
  \label{l:2deriv}
  Let $z$ be as in~\cref{p:contour-integral}.
  Then,
  \begin{align*}
    \frac{\differential}{\differential \alpha} \Phi^{(N)}_{z,\alpha}(x_z(\alpha),y_z(\alpha)) \Big|_{\alpha = 1/2} &= 0\,, \\
    \frac{\differential^2}{\differential^2 \alpha} \Phi^{(N)}_{z,\alpha}(x_z(\alpha),y_z(\alpha)) \Big|_{\alpha = 1/2} 
    &= - \lambda_\star^{-2} 
    = - (1+z) \left(\frac12 + \frac{1}{N}\right)^{-1}\,.
  \end{align*}
\begin{proof}
Throughout the proof we write $\Phi_{z,\alpha}$ instead of $\Phi_{z,\alpha}^{(N)}$.
  Let
  \[
    \partial_x\Phi_{z,\alpha}(x,y) = x + zy - \left(\alpha + \frac{1}{N}\right) \frac{1}{x}\,, \qquad 
    \partial_y\Phi_{z,\alpha}(x,y) = y + zx - \left(1 - \alpha + \frac{1}{N}\right) \frac{1}{y}
  \]
  denote the derivatives of $\Phi_{z,\alpha}(x,y)$ through the arguments $x,y$, and 
  \[
    \nabla \Phi_{z,\alpha} = (\partial_x \Phi_{z,\alpha}, \partial_y \Phi_{z,\alpha})^{\matrixtranspose}\,.
  \]
  Let
  \[
    \partial_\alpha \Phi_{z,\alpha}(x,y) = - \log x + \log y
  \]
  denote the derivative of $\Phi_{z,\alpha}(x,y)$ through the parameter $\alpha$.
  Then, by chain rule
  \begin{align*}
    \frac{\differential}{\differential \alpha} \Phi_{z,\alpha}(x_z(\alpha),y_z(\alpha))
    &= \partial_\alpha \Phi_{z,\alpha}(x_z(\alpha),y_z(\alpha)) + \nabla \Phi_{z,\alpha}(x_z(\alpha),y_z(\alpha)) \cdot (x'_z(\alpha),y'_z(\alpha)) \\
    &= \partial_\alpha \Phi_{z,\alpha}(x_z(\alpha),y_z(\alpha))\,,
  \end{align*}
  where the $\nabla \Phi_{z,\alpha}$ term vanishes because $\nabla \Phi_{z,\alpha}(x_z(\alpha),y_z(\alpha)) = 0$.
  Evaluating at $\alpha = \frac12$ yields
  \[
    \frac{\differential}{\differential \alpha} \Phi_{z,\alpha}(x_z(\alpha),y_z(\alpha)) \Big|_{\alpha = 1/2}
    = \partial_\alpha \Phi_{z,\alpha}(\lambda_\star,\lambda_\star)
    = 0\,.
  \]
  Define the second derivatives $\nabla^2$, $\partial_\alpha \nabla$, and $\partial_\alpha^2$ of $\Phi_{z,\alpha}$ similarly as above.
  Then,
  \begin{align*}
    \frac{\differential^2}{\differential^2 \alpha} \Phi_{z,\alpha}(x_z(\alpha),y_z(\alpha))
    &= \frac{\differential}{\differential \alpha} \left(\partial_\alpha \Phi_{z,\alpha}(x_z(\alpha),y_z(\alpha))\right) \\
    &= \partial_\alpha^2 \Phi_{z,\alpha}(x_z(\alpha),y_z(\alpha))
    + \partial_\alpha \nabla \Phi_{z,\alpha}(x_z(\alpha),y_z(\alpha)) \cdot (x'_z(\alpha),y'_z(\alpha))^{\matrixtranspose}\,.
  \end{align*}
  We can calculate $(x'_z(\alpha),y'_z(\alpha))$ via the implicit function theorem: since $(x_z(\alpha),y_z(\alpha))$ solves
  \[
    \nabla \Phi_{z,\alpha}(x_z(\alpha),y_z(\alpha)) = 0\,,
  \]
  differentiating in $\alpha$ yields
  \[
    \partial_\alpha \nabla \Phi_{z,\alpha}(x_z(\alpha),y_z(\alpha))
    + \nabla^2 \Phi_{z,\alpha}(x_z(\alpha),y_z(\alpha)) (x'_z(\alpha),y'_z(\alpha)) = 0\,,
  \]
  that is
  \[
    (x'_z(\alpha),y'_z(\alpha))
    = - \left(\nabla^2 \Phi_{z,\alpha}(x_z(\alpha),y_z(\alpha))\right)^{-1} \partial_\alpha \nabla \Phi_{z,\alpha}(x_z(\alpha),y_z(\alpha))\,.
  \]
  Hence
  \[
    \frac{\differential^2}{\differential^2 \alpha} \Phi_{z,\alpha}(x_z(\alpha),y_z(\alpha))
    = \partial_\alpha^2 \Phi_{z,\alpha}
    - \left\langle 
      \partial_\alpha \nabla \Phi_{z,\alpha},
      \left(\nabla^2 \Phi_{z,\alpha}\right)^{-1} 
      \partial_\alpha \nabla \Phi_{z,\alpha}
    \right\rangle
  \]
  (To conserve notation we have omitted that all the functions on the right-hand side are evaluated at $(x_z(\alpha),y_z(\alpha))$.)
  At $\alpha = 1/2$, we have $x_z(\alpha) = y_z(\alpha) = \lambda_\star$.
  At this input we evaluate
  \begin{align*}
    \partial_\alpha^2 \Phi_{z,\alpha} &= 0\,, & 
    \partial_\alpha \nabla \Phi_{z,\alpha} &= - \frac{1}{\lambda_\star} \begin{bmatrix}1 \\ -1\end{bmatrix}\,,
  \end{align*}
  and
  \[
    \nabla^2 \Phi_{z,\alpha} = \begin{bmatrix}
      1 + (\frac12 + \frac{1}{N}) \frac{1}{\lambda_\star^2} & z \\ 
      z & 1 + (\frac12 + \frac{1}{N}) \frac{1}{\lambda_\star^2}
    \end{bmatrix} = \begin{bmatrix}
      2+z & z \\ 
      z & 2+z
    \end{bmatrix}
  \]
  Thus
  \begin{align*}
    \frac{\differential^2}{\differential^2 \alpha} \Phi_{z,\alpha}(x_z(\alpha),y_z(\alpha)) \Big|_{\alpha = 1/2}
    &= -\frac{1}{\lambda_\star^2} \left\langle 
      \begin{bmatrix}1 \\ -1\end{bmatrix}, 
      \begin{bmatrix}
        2+z & z \\ 
        z & 2+z
      \end{bmatrix}^{-1}
      \begin{bmatrix}1 \\ -1\end{bmatrix}
    \right\rangle \\
    &= -\frac{1}{\lambda_\star^2}
    = - (1+z) \left(\frac12 + \frac{1}{N}\right)^{-1}\,.
  \end{align*}
\end{proof}
\end{lemma}
\subsubsection{Proof of~\cref{l:1mt}}

    Put $N=nd$ and $z=e^\beta$ and define
    \[
        \mathcal A_{n,\varepsilon}
        =
        \left\{
          \alpha\in\left\{0,\frac{1}{n},\ldots,1\right\}:
          \left|\alpha-\frac{1}{2}\right|<\varepsilon
        \right\}.
    \]
    By the coefficient representation for the truncated first moment (~\cref{lem:1st moment contour representation}) and ~\cref{p:contour-integral},
    \begin{align*}
        \E\left[Z_{M,<\varepsilon}(\beta)\right]
        &=
        \sum_{\alpha\in\mathcal A_{n,\varepsilon}}
        \binom{n}{\alpha n}
        \frac{(\alpha N)!((1-\alpha)N)!}{(N-1)!!\,N^{\frac{N}{2}}}
        \frac{1}{(2\pi i)^2}
        \\
        &\qquad\qquad\times
        \left[
          -\left(1+O\left(\frac{1}{n}\right)\right)
          \frac{4\pi}{N}
          \frac{
            \exp\left(
              N\Phi^{(N)}_{z,\alpha}(x_z(\alpha),y_z(\alpha))
            \right)
          }{
            \sqrt{
              \det\left(
                -\nabla^2\Phi^{(N)}_{z,\alpha}(x_z(\alpha),y_z(\alpha))
              \right)
            }
          }
        \right].
  \end{align*}
  
    By Stirling-approximation we have
    \begin{align*}
        &\binom{n}{\alpha n}
        \frac{(\alpha N)!((1-\alpha)N)!}{(N-1)!!\,N^{\frac{N}{2}}}
        \frac{1}{\pi N}
        \\
        &\qquad=
        \frac{1+O\left(\frac{1}{n}\right)}{\sqrt{\pi n}}
        \exp\left(
          -n(d-1)H(\alpha)-\frac{N}{2}
        \right),
    \end{align*}
    uniformly for $\alpha\in \mathcal{A}_{n,\varepsilon}$, where
    \[ 
    H(\alpha    )=-\alpha\log\alpha-(1-\alpha)\log(1-\alpha)
    \]
    denotes the entropy function (with natural logarithms).
    
    Now define
    \[
    f_{N    ,z}=(d-1)H(\alpha)-d(\Phi_{z,\alpha}(x_z(\alpha),y_z(\alpha))-\tfrac{1}{2})
    \]
    and
    \[
        a_{N,z}(\alpha)=\left(
            \det(-\nabla^2 \Phi_{z,\alpha}(x_z(\alpha),y_z(\alpha)) )
        \right)^{-1/2}.
    \]
Then    
    \begin{equation}\label{e:1mt-simp}
        \E\left[Z_{M,<\varepsilon}(\beta)\right]
        =
        \frac{1}{\sqrt{\pi n}}
        \sum_{\alpha\in\mathcal A_{n,\varepsilon}}
        a_{N,z}(\alpha)
        \exp\left(-nf_{N,z}(\alpha)\right)
        \left(1+O\left(\frac{1}{n}\right)\right),
    \end{equation}
    where the error is uniform over $\alpha\in\mathcal A_{n,\varepsilon}$ and over the parameter region under consideration.
    
    Interchanging $x$ and $y$ in the stationary-point equations and using uniqueness of the stationary points gives
    \[
        x_z(1-\alpha)=y_z(\alpha),
        \qquad
        y_z(1-\alpha)=x_z(\alpha).
    \]
    Consequently,
    \[
        f_{N,z}(1-\alpha)=f_{N,z}(\alpha),
        \qquad
        a_{N,z}(1-\alpha)=a_{N,z}(\alpha).
    \]
    In particular, both functions are even around $\alpha=\frac{1}{2}$.
    
    By~\cref{l:2deriv},
    \[
        f'_{N,z}\left(\frac{1}{2}\right)=0
    \]
    and
    \begin{equation} \label{eq:1mt-kappa}
        \kappa_{N,z}=f''_{N,z}\left(\frac{1}{2}\right)
        =
        d(1+z)\left(\frac{1}{2}+\frac{1}{N}\right)^{-1}
        -4(d-1)
        =
        2d(z-1)+4+O\left(\frac{1}{n}\right).
    \end{equation}

    After decreasing the neighborhood in~\cref{p:contour-integral} if necessary, there exists a constant $c>0$ such that
  \[
    \re\kappa_{N,z}\geq c
  \]
  uniformly over all sufficiently large $n$.

  Symmetry implies that uniformly for sufficiently small $u$
  \[
    f_{N,z}\left(\frac{1}{2}+u\right)=f_{N,z}\left(\frac{1}{2}\right)+\frac{\kappa_{N,z}}{2}u^2+O(u^4),
  \]
  and
  \[
    a_{N,z}\left(\frac{1}{2}+u\right)=a_{N,z}\left(\frac{1}{2}\right)+O(u^2).
  \]
  By further decreasing $\varepsilon$ if necessary, we have
  \[
    \re\left(f_{N,z}\left(\frac{1}{2}+u\right)-f_{N,z}\left(\frac{1}{2}\right)\right)\ge cu^2
  \]
  for all $|u|\le \varepsilon$ for some $c>0$. It follows that the contribution to \eqref{e:1mt-simp} from
  \[
    \left|\alpha-\frac{1}{2}\right|>n^{-\frac{2}{5}}
  \]
  is exponentially small relative to the main contribution.

  For $\alpha=\frac{m}{n}$ satisfying
  \[
    \left|\alpha-\frac{1}{2}\right|\leq n^{-\frac{2}{5}}, \text{ let }
    t_m=\frac{m-\frac{n}{2}}{\sqrt{n}}.
  \]

  The preceeding Taylor expansion gives
  \begin{align*}
    &a_{N,z}\left(\frac{m}{n}\right)
    \exp\left(
      -n\left[
        f_{N,z}\left(\frac{m}{n}\right)
        -
        f_{N,z}\left(\frac{1}{2}\right)
      \right]
    \right)
    \\
    &\qquad=
    a_{N,z}\left(\frac{1}{2}\right)
    \exp\left(
      -\frac{\kappa_{N,z}}{2}t_m^2
    \right)
    \left[
      1+
      O\left(
        \frac{1+t_m^2+t_m^4}{n}
      \right)
    \right].
  \end{align*}

The possible values of
\[
    t_m:=\sqrt n\left(\frac mn-\frac12\right)
        =\frac{m-n/2}{\sqrt n}
\]
form the affine lattice
\[
    \Lambda_n
    =
    \frac1{\sqrt n}\left(\mathbb Z-\frac n2\right).
\]
Since $\kappa_{N,z}$ ranges over a compact subset of
$\{w\in\mathbb C:\Re w>0\}$, Poisson summation gives
\[
    \sum_{m\in\mathbb Z}
    \exp\left(
        -\frac{\kappa_{N,z}}{2n}
        \left(m-\frac n2\right)^2
    \right)
    =
    \sqrt{\frac{2\pi n}{\kappa_{N,z}}}
    \left(1+O(e^{-cn})\right).
\]
The same argument, applied with the polynomial weights $t^2$
and $t^4$, gives
\[
    \sum_{m\in\mathbb Z}
    t_m^{2j}e^{-\kappa_{N,z}t_m^2/2}
    =
    O(\sqrt n),
    \qquad j=1,2.
\]
Moreover, the contribution from $|t_m|>n^{1/10}$ is
$O(\sqrt n\,e^{-c n^{1/5}})$. It follows that the accumulated
Taylor error is $O(n^{-1})$ relative to the leading Gaussian
sum.

Hence
  \begin{align}
    &\sum_{\alpha\in\mathcal A_{n,\varepsilon}}
    a_{N,z}(\alpha)
    \exp\left(-nf_{N,z}(\alpha)\right)
    \left(1+O\left(\frac{1}{n}\right)\right) \nonumber
    \\
    &\qquad=
    a_{N,z}\left(\frac{1}{2}\right)
    \exp\left(
      -nf_{N,z}\left(\frac{1}{2}\right)
    \right)
    \sqrt{\frac{2\pi n}{\kappa_{N,z}}}
    \left(1+O\left(\frac{1}{n}\right)\right).
    \label{eq:1mt-discrete-laplace}
  \end{align}

  Now we have to evaluate this at $\alpha=1/2$. We have,
  \[
    x_z\left(\frac{1}{2}\right)
    =
    y_z\left(\frac{1}{2}\right)
    =
    \lambda_\star,
    \qquad \text{ where }
    \lambda_\star^2
    =
    \frac{\frac{1}{2}+\frac{1}{N}}{1+z}.
  \]
  Moreover,
  \[
    \nabla^2\Phi^{(N)}_{z,1/2}(\lambda_\star,\lambda_\star)
    =
    \begin{bmatrix}
      2+z & z\\
      z & 2+z
    \end{bmatrix},
  \]
  and therefore
  \[
    a_{N,z}\left(\frac{1}{2}\right)
    =
    \frac{1}{2\sqrt{1+z}}.
  \]
  A direct calculation reveals that
  \[
    \exp\left(-nf_{N,z}(1/2)\right)=2(1+z)2^n\left(\frac{1+z}{2}\right)^{N/2}\left(1+O\left(n^{-1}\right)\right),
  \]
 and thus
  \[
     \E\left[Z_{M,<\varepsilon}(\beta)\right]
    =
    \sqrt{\frac{1+z}{d(z-1)+2}}\cdot
    2^n
    \left(\frac{1+z}{2}\right)^{\frac{N}{2}}
    \left(1+O\left(\frac{1}{n}\right)\right),
  \]
  where substituting $z=e^\beta$ and $N=nd$ proves the result.

\subsection{Reweighted second moment: Proof of \texorpdfstring{\cref{thm:reweighted-2ndmoment}}{Reweighted Second Moment}}\label{subsec:2nd-moment-complex}
The proof of \cref{thm:reweighted-2ndmoment} rests on the following key estimate on the conditional second moment. We call it ``decoupling'' because the first term in the right-hand side is equal (up to small error) to the \emph{unconditional} (normalized) second moment.
\begin{proposition}[Conditional Second Moment Decoupling]\label{prop:2nd-moment}
For all $0 < \eta < 1$, all $\beta$ within a sufficiently small constant (depending on $d,\eta$) neighborhood of $[0,(1-\eta)\beta_{\mathsf{KS}}]$, all sufficiently small $\epsilon > 0$, all $m \leq C\log n$ for some constant $C>0$ and all matchings $F \subseteq \binom{V \times [d]}{2}$ with $|F| \leq m$, we have
\begin{align*}
    \frac{\E_{M}\wrapb{Z_{M,<\epsilon}^{(2)}(\beta) \,\, \bigg| \,\, F \subseteq M}}{\abs{\E_{M}\wrapb{Z_{M,<\epsilon}(\beta)}}^{2}} &= \sqrt{\frac{1}{1 - (d-1) \cdot \abs{\delta_{2}}}} \cdot \abs{\widetilde{Z}_{F}(\beta)}^{2} + \frac{K^{m}}{n},
\end{align*}
where $K>1$ is a constant depending on $\eta$ and $d$.
\end{proposition}
In the remainder of this subsection, we use \cref{prop:2nd-moment} to complete the proof of \cref{thm:reweighted-2ndmoment}; the proof of \cref{prop:2nd-moment} is provided in \cref{s:deferred-complex-moment}. Observe that
\begin{align*}
    \abs{A_{L,m}(\beta)}^{2} &= \exp\wrapp{\sum_{k=1}^{L} \lambda_{k} \cdot 2\re \delta_{k}} \cdot \sum_{\Gamma_{1},\Gamma_{2} \in \mathscr{G}(\leq m)} \allone[\Gamma_{1} \cup \Gamma_{2} \subseteq M] \prod_{C \in \Gamma_{1}} b_{\abs{C}} \prod_{C \in \Gamma_{2}} \overline{b_{\abs{C}}}.
\end{align*}
The quantity we are interested in can then be written as
\begin{align}\label{eq:reweighted-2ndmoment}
\begin{split}
    &\frac{\E_{M}\wrapb{Z_{M,<\epsilon}^{(2)}(\beta) \cdot \abs{A_{L,m}(\beta)}^{2}}}{\abs{\E_{M}\wrapb{Z_{M,<\epsilon}(\beta)}}^{2}} \\
    &= \exp\wrapp{\sum_{k=1}^{L} \lambda_{k} \cdot 2\re \delta_{k}} \cdot \sum_{\Gamma_{1},\Gamma_{2} \in \mathscr{G}(\leq m)} \frac{\E_{M}\wrapb{Z_{M,<\epsilon}^{(2)}(\beta) \cdot \allone[\Gamma_{1} \cup \Gamma_{2} \subseteq M]}}{\abs{\E_{M}\wrapb{Z_{M,<\epsilon}(\beta)}}^{2}} \prod_{C \in \Gamma_{1}} b_{\abs{C}} \prod_{C \in \Gamma_{2}} \overline{b_{\abs{C}}}.
\end{split}
\end{align}
We will see that the dominant contributions to the above double sum stem from pairs $\Gamma_{1},\Gamma_{2}$ such that for all pairs of cycles $C,C' \in \Gamma_{1} \cup \Gamma_{2}$, either $C,C'$ are vertex-disjoint \emph{or $C = C'$}. Roughly speaking, the reason is that $\Pr[\Gamma \subseteq M] \approx \wrapp{\frac{1}{dn}}^{m(\Gamma)}$, where $m(\Gamma)$ denotes the number of \emph{distinct} edges in the graph $G_{\Gamma} = (V(\Gamma),E(\Gamma))$ formed from $\Gamma$ by contracting the clouds and throwing away the isolated vertices, and we have $\abs{E(\Gamma)} = \abs{V(\Gamma)}$ if and only if the cycles in $\Gamma$ are either pairwise disjoint or identical. Because of this, we will eventually instead stratify everything by \emph{three} tuples $\bm{c} = \bm{c}^{(1)} + \bm{c}^{(2)} + \bm{c}^{(3)}$ with $\bm{c}(\Gamma^{(i)}) = \bm{c}^{(i)}$ for $i=1,2,3$, where one should think of $\Gamma_{1} = \Gamma^{(1)} \cup \Gamma^{(3)}$, $\Gamma_{2} = \Gamma^{(2)} \cup \Gamma^{(3)}$, and $\Gamma^{(3)}$ collects the cycles shared between $\Gamma_{1},\Gamma_{2}$.

In anticipation of the above, for $\bm{c} \in \N^{L}$, define its ``idealized weight'' by
\begin{align*}
    w(\bm{c}) \defeq \prod_{k=1}^{L} \frac{\lambda_{k}^{\bm{c}_{k}}\abs{1 + \delta_{k}}^{2\bm{c}_{k}}}{\bm{c}_{k}!},
\end{align*}
and for $\bm{c}^{(1)},\bm{c}^{(2)},\bm{c}^{(3)} \in \N^{L}$,
\begin{align*}
    \psi\wrapp{\bm{c}^{(1)},\bm{c}^{(2)},\bm{c}^{(3)}} &\defeq \sqrt{\frac{1}{1 - (d-1) \cdot \abs{\delta_{2}}}} \cdot w\wrapp{\bm{c}^{(1)}}w\wrapp{\bm{c}^{(2)}}w\wrapp{\bm{c}^{(3)}}.
\end{align*}
Also, to shorten some of the expressions we will obtain, define
\begin{align*}
    \Xi \defeq \exp\wrapp{\sum_{k=1}^{L} \lambda_{k} \cdot 2\re \delta_{k}}.
\end{align*}
Note that $\Xi \leq \exp(\exp(O(L)))$, and so $\Xi$ can be rendered irrelevant by letting $L$ be a sufficiently slow-growing function of $N$. The idealized sum we wish to relate the left-hand side of \cref{thm:reweighted-2ndmoment} to is
\begin{align*}
    \Psi &\defeq \Xi \cdot \sum_{\bm{c}^{(1)},\bm{c}^{(2)},\bm{c}^{(3)} \in \N^{L}} \psi\wrapp{\bm{c}^{(1)},\bm{c}^{(2)},\bm{c}^{(3)}} \cdot \prod_{k=1}^{L} b_{k}^{\bm{c}_{k}^{(1)}} \cdot \overline{b_{k}}^{\bm{c}_{k}^{(2)}} \cdot \abs{b_{k}}^{2\bm{c}_{k}^{(3)}} \\
    &= \sqrt{\frac{1}{1 - (d-1) \cdot \abs{\delta_{2}}}} \cdot \exp\wrapp{\sum_{k=1}^{L} \lambda_{k} \cdot \wrapp{2 \re \delta_{k} + b_{k} \cdot \abs{1 + \delta_{k}}^{2} + \overline{b_{k}} \cdot \abs{1 + \delta_{k}}^{2} + \abs{b_{k}}^{2} \cdot \abs{1 + \delta_{k}}^{2}}} \tag{Definition of $\Xi$ and Taylor series of $e^{x}$} \\
    &= \sqrt{\frac{1}{1 - (d-1) \cdot \abs{\delta_{2}}}} \cdot \exp\wrapp{-\sum_{k=1}^{L} \lambda_{k}\abs{\delta_{k}}^{2}} \tag{Using $\wrapp{\abs{b_{k}}^{2} + 2\re b_{k}} \cdot \abs{1 + \delta_{k}}^{2} = -2\re \delta_{k} - \abs{\delta_{k}}^{2}$} \\
    &= \exp\wrapp{\sum_{k=L+1}^{\infty} \frac{\wrapp{(d-1) \cdot \abs{\delta_{2}}}^{k}}{2k}}. \tag{\cref{lem:gaussian-correction-series}}
\end{align*}
In the penultimate equality, we used the fact that
\begin{align*}
    b_{k} = (1 + \delta_{k})^{-1} - 1 = - \frac{\delta_{k}}{1 + \delta_{k}} = -\frac{\delta_{k} \cdot \wrapp{1 + \overline{\delta_{k}}}}{\abs{1 + \delta_{k}}^{2}}
\end{align*}
implies that
\begin{align*}
    \wrapp{\abs{b_{k}}^{2} + 2 \re b_{k}} \cdot \abs{1 + \delta_{k}}^{2} &= \abs{\delta_{k}}^{2} - 2 \re \wrapp{\delta_{k} \cdot \wrapp{1 + \overline{\delta_{k}}}} \\
    &= -2 \re \delta_{k} - \abs{\delta_{k}}^{2}.
\end{align*}
Clearly $\Psi \geq 1$, and we have
\begin{align*}
    \Psi \leq \exp\wrapp{\frac{1}{2(L+1)} \sum_{k=L+1}^{\infty} (1 - \eta/2)^{k}} = \exp\wrapp{\frac{(1 - \eta/2)^{L+1}}{\eta(L+1)}} \leq 1 + \exp(-O_{\eta}(L))
\end{align*}
for $\beta$ in a sufficiently small constant neighborhood of $[0,(1-\eta)\beta_{\mathsf{KS}}]$.

For $\bm{c}^{(1)},\bm{c}^{(2)},\bm{c}^{(3)} \in \N^{L}$, let us now define
\begin{align*}
    \varphi\wrapp{\bm{c}^{(1)},\bm{c}^{(2)},\bm{c}^{(3)}} &\defeq \sum_{\substack{\Gamma^{(i)} \in \mathscr{G}(\bm{c}^{(i)}), i=1,2,3 \\
    \Gamma^{(1)} \cup \Gamma^{(3)}, \Gamma^{(2)} \cup \Gamma^{(3)} \in \mathscr{G}(\leq m)}} \frac{\E_{M}\wrapb{Z_{M,<\epsilon}^{(2)}(\beta) \cdot \allone[\Gamma^{(1)} \cup \Gamma^{(2)} \cup \Gamma^{(3)} \subseteq M]}}{\abs{\E_{M}\wrapb{Z_{M,<\epsilon}(\beta)}}^{2}}
\end{align*}
and
\begin{align*}
    \Phi &\defeq \Xi \cdot \sum_{\bm{c}^{(1)},\bm{c}^{(2)},\bm{c}^{(3)} \in \N^{L}} \varphi\wrapp{\bm{c}^{(1)},\bm{c}^{(2)},\bm{c}^{(3)}} \cdot \prod_{k=1}^{L} b_{k}^{\bm{c}_{k}^{(1)}} \cdot \overline{b_{k}}^{\bm{c}_{k}^{(2)}} \cdot \abs{b_{k}}^{2\bm{c}_{k}^{(3)}}.
\end{align*}
Note that by \cref{eq:reweighted-2ndmoment} and decomposing $\Gamma_{1} = \Gamma^{(1)} \cup \Gamma^{(3)}, \Gamma_{2} = \Gamma^{(2)} \cup \Gamma^{(3)}$,
\begin{align*}
    \Phi = \frac{\E_{M}\wrapb{Z_{M,<\epsilon}^{(2)}(\beta) \cdot \abs{A_{L,m}(\beta)}^{2}}}{\abs{\E_{M}\wrapb{Z_{M,<\epsilon}(\beta)}}^{2}},
\end{align*}
which is precisely the left-hand side in \cref{thm:reweighted-2ndmoment}. Since, as established above, we have $1 \leq \Psi \leq 1 + o(1)$, \cref{thm:reweighted-2ndmoment} is proved if we can show that $\abs{\Phi - \Psi} \leq o(1)$.

To formally prove \cref{thm:reweighted-2ndmoment}, we rely on the following key estimate.
\begin{proposition}\label{prop:uniform-converge-small-Gamma}
Fix $L,m \in \N$. Then for all $\bm{c}^{(1)},\bm{c}^{(2)},\bm{c}^{(3)} \in \N^{L}$ satisfying $\abs{\bm{c}^{(1)}} + \abs{\bm{c}^{(3)}}, \abs{\bm{c}^{(2)}} + \abs{\bm{c}^{(3)}} \leq m$,
\begin{align*}
    \varphi\wrapp{\bm{c}^{(1)},\bm{c}^{(2)},\bm{c}^{(3)}} = \wrapp{1 \pm \frac{\widehat{\mathcal{E}}(\eta,L,m)}{n}} \cdot \psi\wrapp{\bm{c}^{(1)},\bm{c}^{(2)},\bm{c}^{(3)}},
\end{align*}
where
\begin{align*}
    \widehat{\mathcal{E}}(\eta,L,m) \defeq O\wrapp{\exp\wrapp{e^{O_{d}(L)} + O(m\log m)}} + \exp(O_{d,\eta}(m)).
\end{align*}
\end{proposition}
Note that once we have \cref{prop:uniform-converge-small-Gamma}, we can deduce \cref{thm:reweighted-2ndmoment} just by bounding the tails of a compound Poisson. We now spell out this argument.

First, note that the outer sum in the definition of $\Phi$ can be truncated to $\abs{\bm{c}^{(1)}} + \abs{\bm{c}^{(3)}}, \abs{\bm{c}^{(2)}} + \abs{\bm{c}^{(3)}} \leq m$ with zero loss, since $\varphi(\cdot)$ is zero for $\wrapp{\bm{c}^{(1)},\bm{c}^{(2)},\bm{c}^{(3)}}$ violating these constraints. Hence,
\begin{align*}
    \abs{\Phi - \Psi} &\leq \Xi \cdot \sum_{\substack{\bm{c}^{(1)},\bm{c}^{(2)},\bm{c}^{(3)} \in \N^{L} \\ \abs{\bm{c}^{(1)}} + \abs{\bm{c}^{(3)}}, \abs{\bm{c}^{(2)}} + \abs{\bm{c}^{(3)}} \leq m}} \abs{\varphi\wrapp{\bm{c}^{(1)},\bm{c}^{(2)},\bm{c}^{(3)}} - \psi\wrapp{\bm{c}^{(1)},\bm{c}^{(2)},\bm{c}^{(3)}}} \cdot \prod_{k=1}^{L} \abs{b_{k}}^{\bm{c}_{k}^{(1)} + \bm{c}_{k}^{(2)} + 2\bm{c}_{k}^{(3)}} \\
    &\quad + \Xi \cdot \sum_{\substack{\bm{c}^{(1)},\bm{c}^{(2)},\bm{c}^{(3)} \in \N^{L} \\ \abs{\bm{c}^{(1)}} + \abs{\bm{c}^{(3)}} + \abs{\bm{c}^{(3)}} > m}} \psi\wrapp{\bm{c}^{(1)},\bm{c}^{(2)},\bm{c}^{(3)}} \prod_{k=1}^{L} \abs{b_{k}}^{\bm{c}_{k}^{(1)} + \bm{c}_{k}^{(2)} + 2\bm{c}_{k}^{(3)}} \tag{Triangle Inequality} \\
    &\leq \frac{\widehat{\mathcal{E}}(\eta,L,m)}{n}  \cdot \Xi \cdot \sum_{\substack{\bm{c}^{(1)},\bm{c}^{(2)},\bm{c}^{(3)} \in \N^{L} \\ \abs{\bm{c}^{(1)}},\abs{\bm{c}^{(2)}},\abs{\bm{c}^{(3)}} \leq m}} \psi\wrapp{\bm{c}^{(1)},\bm{c}^{(2)},\bm{c}^{(3)}} \cdot \prod_{k=1}^{L} \abs{b_{k}}^{\bm{c}_{k}^{(1)} + \bm{c}_{k}^{(2)} + 2\bm{c}_{k}^{(3)}} \\
    &\quad + \Xi \cdot \sum_{\substack{\bm{c}^{(1)},\bm{c}^{(2)},\bm{c}^{(3)} \in \N^{L} \\ \abs{\bm{c}^{(1)}} + \abs{\bm{c}^{(2)}} + \abs{\bm{c}^{(3)}} > m}} \psi\wrapp{\bm{c}^{(1)},\bm{c}^{(2)},\bm{c}^{(3)}} \cdot \prod_{k=1}^{L} \abs{b_{k}}^{\bm{c}_{k}^{(1)} + \bm{c}_{k}^{(2)} + 2\bm{c}_{k}^{(3)}} \tag{\cref{prop:uniform-converge-small-Gamma}} \\
    &\lesssim \frac{\widehat{\mathcal{E}}(\eta,L,m)}{n}  \cdot \Xi \cdot \sum_{\substack{\bm{c}^{(1)},\bm{c}^{(2)},\bm{c}^{(3)} \in \N^{L} \\ \abs{\bm{c}^{(1)}},\abs{\bm{c}^{(2)}},\abs{\bm{c}^{(3)}} \leq m}} \prod_{k=1}^{L} \frac{\lambda_{k}^{\bm{c}_{k}^{(1)} + \bm{c}_{k}^{(2)} + \bm{c}_{k}^{(3)}}\abs{b_{k}}^{\bm{c}_{k}^{(1)} + \bm{c}_{k}^{(2)} + 2\bm{c}_{k}^{(3)}}\abs{1 + \delta_{k}}^{2\bm{c}_{k}^{(1)} + 2\bm{c}_{k}^{(2)} + 2\bm{c}_{k}^{(3)}}}{\bm{c}_{k}^{(1)}!\bm{c}_{k}^{(2)}!\bm{c}_{k}^{(3)}!} \\
    &\quad + \Xi \cdot \sum_{\substack{\bm{c}^{(1)},\bm{c}^{(2)},\bm{c}^{(3)} \in \N^{L} \\ \abs{\bm{c}^{(1)}} + \abs{\bm{c}^{(2)}} + \abs{\bm{c}^{(3)}} > m}} \prod_{k=1}^{L} \frac{\lambda_{k}^{\bm{c}_{k}^{(1)} + \bm{c}_{k}^{(2)} + \bm{c}_{k}^{(3)}}\abs{b_{k}}^{\bm{c}_{k}^{(1)} + \bm{c}_{k}^{(2)} + 2\bm{c}_{k}^{(3)}}\abs{1 + \delta_{k}}^{2\bm{c}_{k}^{(1)} + 2\bm{c}_{k}^{(2)} + 2\bm{c}_{k}^{(3)}}}{\bm{c}_{k}^{(1)}!\bm{c}_{k}^{(2)}!\bm{c}_{k}^{(3)}!} \tag{Treating $\sqrt{\frac{1}{1 - (d-1) \cdot \abs{\delta_{2}}}}$ as a constant} \\
    &\leq \frac{\widehat{\mathcal{E}}(\eta,L,m)}{n}  \cdot \Xi \cdot \sum_{\substack{\bm{c}^{(1)},\bm{c}^{(2)},\bm{c}^{(3)} \in \N^{L} \\ \abs{\bm{c}^{(1)}},\abs{\bm{c}^{(2)}},\abs{\bm{c}^{(3)}} \leq m}} \prod_{k=1}^{L} \frac{\lambda_{k}^{\bm{c}_{k}^{(1)} + \bm{c}_{k}^{(2)} + \bm{c}_{k}^{(3)}}B^{k\bm{c}_{k}^{(1)} + k\bm{c}_{k}^{(2)} + k\bm{c}_{k}^{(3)}}}{\bm{c}_{k}^{(1)}!\bm{c}_{k}^{(2)}!\bm{c}_{k}^{(3)}!} \\
    &\quad + \Xi \cdot \sum_{\substack{\bm{c}^{(1)},\bm{c}^{(2)},\bm{c}^{(3)} \in \N^{L} \\ \abs{\bm{c}^{(1)}} + \abs{\bm{c}^{(2)}} + \abs{\bm{c}^{(3)}} > m}} \prod_{k=1}^{L} \frac{\lambda_{k}^{\bm{c}_{k}^{(1)} + \bm{c}_{k}^{(2)} + \bm{c}_{k}^{(3)}}B^{k\bm{c}_{k}^{(1)} + k\bm{c}_{k}^{(2)} + k\bm{c}_{k}^{(3)}}}{\bm{c}_{k}^{(1)}!\bm{c}_{k}^{(2)}!\bm{c}_{k}^{(3)}!} \tag{There exists a constant $B = B(d) > 0$ such that $\max\wrapc{\abs{b_{k}(1 + \delta_{k})^{2}}, \abs{b_{k}^{2}(1+\delta_{k})^{2}}} \leq B^{k}$} \\
    &\leq \frac{\widehat{\mathcal{E}}(\eta,L,m)}{n}  \cdot \Xi \cdot \sum_{\bm{c}^{(1)},\bm{c}^{(2)},\bm{c}^{(3)} \in \N^{L}} \prod_{k=1}^{L} \frac{\lambda_{k}^{\bm{c}_{k}^{(1)} + \bm{c}_{k}^{(2)} + \bm{c}_{k}^{(3)}}B^{k\bm{c}_{k}^{(1)} + k\bm{c}_{k}^{(2)} + k\bm{c}_{k}^{(3)}}}{\bm{c}_{k}^{(1)}!\bm{c}_{k}^{(2)}!\bm{c}_{k}^{(3)}!} \\
    &\quad + \Xi \cdot e^{-sm} \cdot \sum_{\bm{c}^{(1)},\bm{c}^{(2)},\bm{c}^{(3)} \in \N^{L}} \prod_{k=1}^{L} \frac{\lambda_{k}^{\bm{c}_{k}^{(1)} + \bm{c}_{k}^{(2)} + \bm{c}_{k}^{(3)}}\wrapp{Be^{s}}^{k\bm{c}_{k}^{(1)} + k\bm{c}_{k}^{(2)} + k\bm{c}_{k}^{(3)}}}{\bm{c}_{k}^{(1)}!\bm{c}_{k}^{(2)}!\bm{c}_{k}^{(3)}!} \tag{Using the trick $\allone\wrapb{\abs{\bm{c}^{(1)}} + \abs{\bm{c}^{(2)}} + \abs{\bm{c}^{(3)}} > m} \leq \exp\wrapp{s \cdot \wrapp{\abs{\bm{c}^{(1)}} + \abs{\bm{c}^{(2)}} + \abs{\bm{c}^{(3)}} - m}}$} \\
    &= \frac{\widehat{\mathcal{E}}(\eta,L,m)}{n} \cdot \Xi \cdot \exp\wrapp{3\sum_{k=1}^{L} \lambda_{k}B^{k}} + \Xi \cdot e^{-sm} \cdot \exp\wrapp{3 \sum_{k=1}^{L} \lambda_{k}B^{k}e^{sk}} \\
    &\leq \wrapp{\frac{\widehat{\mathcal{E}}(\eta,L,m)}{n} + e^{-m}} \cdot \exp\wrapp{\exp(O_{d}(L))}. \tag{Setting $s = 1$}
\end{align*}
Putting everything together, we obtain
\begin{align*}
    \Phi = 1 + \exp\wrapp{-O_{\eta}(L)} + \wrapp{\frac{\widehat{\mathcal{E}}(\eta,L,m)}{n}  + e^{-m}} \cdot \exp\wrapp{\exp(O_{d}(L))},
\end{align*}
where recall
\begin{align*}
    \widehat{\mathcal{E}}(\eta,L,m) = O\wrapp{\exp\wrapp{e^{O_{d}(L)} + O(m\log m)}} + \exp(O_{d,\eta}(m)).
\end{align*}
Simplifying (and leaving out constant factors),
\begin{align*}
    \Phi = 1 + \exp\wrapp{-O_{\eta}(L)} + \frac{\exp\wrapp{e^{O_{d}(L)} + O_{d}(m\log m)}}
    {n} + e^{-m} \cdot \exp\wrapp{\exp(O_{d}(L))}.
\end{align*}
We have $\exp(-O_{\eta}(L)) \leq o(1)$ as long as $L \geq \omega(1)$. We also have $e^{-m} \cdot \exp(\exp(O_{d}(L))) \leq o(1)$ as long as $m \gtrsim \exp(\Omega_{d}(L))$ and $L \geq \omega(1)$. Finally, we have $\frac{\exp\wrapp{e^{O_{d}(L)} + O_{d}(m\log m)}}{n} \leq o(1)$ as long as $m \leq o\wrapp{\frac{\log N}{\log\log N}}$ and $L \leq o(\log\log N)$. Setting $L =\frac{\log\log N}{\log\log\log N}$ and $m =\frac{\log N}{(\log\log N)^2}$, we obtain $\Phi = 1 \pm o(1)$ as desired. 

\subsection{Proof of \cref{prop:uniform-converge-small-Gamma}}\label{subsec:uniform-converge-small-Gamma}
For $\bm{c}^{(1)},\bm{c}^{(2)},\bm{c}^{(3)} \in \N^{L}$, let us define the intermediate quantity
\begin{align*}
    \varphi^{\mathscr{G}}\wrapp{\bm{c}^{(1)},\bm{c}^{(2)},\bm{c}^{(3)}} &\defeq \sum_{\substack{\Gamma^{(i)} \in \mathscr{G}(\bm{c}^{(i)}), i=1,2,3 \\
    \Gamma^{(1)} \cup \Gamma^{(2)} \cup \Gamma^{(3)} \in \mathscr{G}}} \frac{\E_{M}\wrapb{Z_{M,<\epsilon}^{(2)}(\beta) \cdot \allone[\Gamma^{(1)} \cup \Gamma^{(2)} \cup \Gamma^{(3)} \subseteq M]}}{\abs{\E_{M}\wrapb{Z_{M,<\epsilon}(\beta)}}^{2}}.
\end{align*}
Note that when $\abs{\bm{c}^{(1)}} + \abs{\bm{c}^{(3)}},\abs{\bm{c}^{(2)}} + \abs{\bm{c}^{(3)}} \leq m$, the only difference between $\varphi\wrapp{\bm{c}^{(1)},\bm{c}^{(2)},\bm{c}^{(3)}}$ and $\varphi^{\mathscr{G}}\wrapp{\bm{c}^{(1)},\bm{c}^{(2)},\bm{c}^{(3)}}$ is that in the latter, we have introduced an additional truncation, where we additionally enforce that all pairs of cycles in $\Gamma^{(1)} \cup \Gamma^{(2)}$ are pairwise disjoint. In this case, the former upper bounds the latter.

\cref{prop:uniform-converge-small-Gamma} is an immediate consequence of the following pair of lemmas.
\begin{lemma}\label{lem:good}
Under the assumptions of \cref{prop:uniform-converge-small-Gamma},
\begin{align*}
    \varphi^{\mathscr{G}}\wrapp{\bm{c}^{(1)},\bm{c}^{(2)},\bm{c}^{(3)}} = \wrapp{1 \pm \frac{O(L^{3}m^{2}) + \exp(O_{d,\eta}(m))}{n}} \cdot \psi\wrapp{\bm{c}^{(1)},\bm{c}^{(2)},\bm{c}^{(3)}}.
\end{align*}
\end{lemma}
\begin{lemma}\label{lem:bad}
Under the assumptions of \cref{prop:uniform-converge-small-Gamma},
\begin{align*}
    \abs{\varphi\wrapp{\bm{c}^{(1)},\bm{c}^{(2)},\bm{c}^{(3)}} - \varphi^{\mathscr{G}}\wrapp{\bm{c}^{(1)},\bm{c}^{(2)},\bm{c}^{(3)}}} \leq \frac{O\wrapp{\exp\wrapp{e^{O_{d}(L)} + O(m\log m)}}}{n} \cdot \psi\wrapp{\bm{c}^{(1)},\bm{c}^{(2)},\bm{c}^{(3)}}.
\end{align*}
\end{lemma}
We prove each in turn.

\begin{proof}[Proof of \cref{lem:good}]
Let $\Gamma^{(i)} \in \mathscr{G}(\bm{c}^{(i)})$ for $i=1,2,3$ such that $\Gamma^{(1)} \cup \Gamma^{(2)} \cup \Gamma^{(3)} \in \mathscr{G}$. By \cref{prop:2nd-moment} and \cref{lem:pfn-disj-cycles}, for a constant $K > 0$ depending only on $d,\eta$,
\begin{align*}
    &\frac{\E_{M}\wrapb{Z_{M,<\epsilon}^{(2)}(\beta) \,\, \bigg| \,\, \Gamma^{(1)} \cup \Gamma^{(2)} \cup \Gamma^{(3)} \subseteq M}}{\abs{\E_{M}\wrapb{Z_{M,<\epsilon}(\beta)}}^{2}} \\
    &= \sqrt{\frac{1}{1 - (d-1) \cdot \abs{\delta_{2}}}} \cdot \prod_{k=1}^{L} \abs{1 + \delta_{k}}^{2\bm{c}_{k}^{(1)} + 2\bm{c}_{k}^{(2)} + 2\bm{c}_{k}^{(3)}} + \frac{K^{m}}{n}.
\end{align*}
Observing that $\sqrt{\frac{1}{1 - (d-1) \cdot \abs{\delta_{2}}}} \geq \Omega(1)$ and $\prod_{k=1}^{L} \abs{1 + \delta_{k}}^{2\bm{c}_{k}^{(1)} + 2\bm{c}_{k}^{(2)} + 2\bm{c}_{k}^{(3)}} \gtrsim \exp(-O_{d}(m))$, we can convert the above into relative error approximation at the expense of blowing up the error by a factor of $\exp(O_{d,\eta}(m))$, i.e. we have
\begin{align*}
    &\frac{\E_{M}\wrapb{Z_{M,<\epsilon}^{(2)}(\beta) \,\, \bigg| \,\, \Gamma^{(1)} \cup \Gamma^{(2)} \cup \Gamma^{(3)} \subseteq M}}{\abs{\E_{M}\wrapb{Z_{M,<\epsilon}(\beta)}}^{2}} \\
    &= \wrapp{1 \pm \frac{\exp(O_{d,\eta}(m))}{n}} \cdot \sqrt{\frac{1}{1 - (d-1) \cdot \abs{\delta_{2}}}} \cdot \prod_{k=1}^{L} \abs{1 + \delta_{k}}^{2\bm{c}_{k}^{(1)} + 2\bm{c}_{k}^{(2)} + 2\bm{c}_{k}^{(3)}}.
\end{align*}
Summing, we have
\begin{align*}
    \varphi^{\mathscr{G}}\wrapp{\bm{c}^{(1)},\bm{c}^{(2)},\bm{c}^{(3)}} &= \wrapp{1 \pm \frac{\exp(O_{d,\eta}(m))}{n}} \cdot \sqrt{\frac{1}{1 - (d-1) \cdot \abs{\delta_{2}}}} \cdot \prod_{k=1}^{L} \abs{1 + \delta_{k}}^{2\bm{c}_{k}^{(1)} + 2\bm{c}_{k}^{(2)} + 2\bm{c}_{k}^{(3)}} \\
    &\quad \cdot \sum_{\substack{\Gamma^{(i)} \in \mathscr{G}(\bm{c}^{(i)}), i=1,2,3 \\
    \Gamma^{(1)} \cup \Gamma^{(2)} \cup \Gamma^{(3)} \in \mathscr{G}}} \Pr[\Gamma^{(1)} \cup \Gamma^{(2)} \cup \Gamma^{(3)} \subseteq M] \\
    &= \wrapp{1 \pm \frac{\exp(O_{d,\eta}(m))}{n}} \cdot \sqrt{\frac{1}{1 - (d-1) \cdot \abs{\delta_{2}}}} \cdot \prod_{k=1}^{L} \abs{1 + \delta_{k}}^{2\bm{c}_{k}^{(1)} + 2\bm{c}_{k}^{(2)} + 2\bm{c}_{k}^{(3)}} \\
    &\quad\cdot \exp\wrapp{\pm O\wrapp{\frac{m^{2}}{n}}} \cdot \wrapp{\frac{1}{dn}}^{\abs{\bm{c}^{(1)}} + \abs{\bm{c}^{(2)}} + \abs{\bm{c}^{(3)}}} \cdot \abs{\widehat{\mathscr{G}}\wrapp{\bm{c}^{(1)},\bm{c}^{(2)},\bm{c}^{(3)}}}, \tag{\cref{lem:shortcyc-approxindep}}
\end{align*}
where
\begin{align*}
    \widehat{\mathscr{G}}\wrapp{\bm{c}^{(1)},\bm{c}^{(2)},\bm{c}^{(3)}} \defeq \wrapc{\wrapp{\Gamma^{(1)},\Gamma^{(2)},\Gamma^{(3)}} : \substack{\bm{c}(\Gamma^{(i)}) = \bm{c}^{(i)}, i=1,2,3 \\ \Gamma^{(1)} \cup \Gamma^{(2)} \cup \Gamma^{(3)} \in \mathscr{G}}}.
\end{align*}
Given this, all that remains is to approximate the cardinality of this collection, which is achieved in the following claim.
\begin{claim}\label{claim:goodcount-approx}
Under the assumptions of \cref{prop:uniform-converge-small-Gamma}, then
\begin{align*}
    1 \geq \frac{\abs{\widehat{\mathscr{G}}\wrapp{\bm{c}^{(1)},\bm{c}^{(2)},\bm{c}^{(3)}}}}{(dn)^{\abs{\bm{c}^{(1)}} + \abs{\bm{c}^{(2)}} + \abs{\bm{c}^{(3)}}} \cdot \prod_{k=1}^{L} \frac{\lambda_{k}^{\bm{c}_{k}^{(1)} + \bm{c}_{k}^{(2)} + \bm{c}_{k}^{(3)}}}{\bm{c}_{k}^{(1)}!\bm{c}_{k}^{(2)}!\bm{c}_{k}^{(3)}!}} \geq \exp\wrapp{-O\wrapp{\frac{L^{3}m^{2}}{n}}}.
\end{align*}
\end{claim}
\begin{proof}
Observe that by erasing the ordering of the $k$-cycles among the $\Gamma^{(1)},\Gamma^{(2)},\Gamma^{(3)}$ for each $k=1,\dots,L$ and taking their union, we can see that for $\bm{c} = \bm{c}^{(1)} + \bm{c}^{(2)} + \bm{c}^{(3)}$,
\begin{align*}
    \prod_{k=1}^{L} \frac{\bm{c}_{k}^{(1)}!\bm{c}_{k}^{(2)}!\bm{c}_{k}^{(3)}!}{\bm{c}_{k}!} \cdot \abs{\widehat{\mathscr{G}}\wrapp{\bm{c}^{(1)},\bm{c}^{(2)},\bm{c}^{(3)}}} = \abs{\mathscr{G}(\bm{c})}.
\end{align*}
Hence, it suffices to prove that
\begin{align}\label{eq:Gc-size-estimate}
    1 \geq \frac{\abs{\mathscr{G}(\bm{c})}}{(dn)^{\abs{\bm{c}}} \cdot \prod_{k=1}^{L} \frac{\lambda_{k}^{\bm{c}_{k}}}{\bm{c}_{k}!}} \geq \exp\wrapp{-O\wrapp{\frac{L^{3}m^{2}}{n}}}.
\end{align}
Observe that $\abs{\mathscr{G}(\bm{c})} + \abs{\mathscr{B}(\bm{c})} = \abs{\mathscr{A}(\bm{c})} = \prod_{k=1}^{L} \binom{\abs{\mathscr{C}_{k}}}{\bm{c}_{k}}$. Moreover, by first determining which pair of cycles intersect at a given cloud, we have the crude upper bound
\begin{align*}
    \abs{\mathscr{B}(\bm{c})} &\leq \sum_{1 \leq i < j \leq L} \frac{j}{n} \cdot \abs{\mathscr{C}_{i}} \cdot \abs{\mathscr{C}_{j}} \cdot \binom{\abs{\mathscr{C}_{i}} - 1}{\bm{c}_{i} - 1} \cdot \binom{\abs{\mathscr{C}_{j}} - 1}{\bm{c}_{j} - 1} \cdot \prod_{k \neq i,j} \binom{\abs{\mathscr{C}_{k}}}{\bm{c}_{k}} \\
    &\leq \frac{L}{n} \cdot \prod_{k=1}^{L} \binom{\abs{\mathscr{C}_{k}}}{\bm{c}_{k}} \cdot \sum_{1 \leq i < j \leq L} \bm{c}_{i}\bm{c}_{j} \\
    &\leq O\wrapp{\frac{L^{3}m^{2}}{n}} \cdot \wrapp{\abs{\mathscr{G}(\bm{c})} + \abs{\mathscr{B}(\bm{c})}}. \tag{Using $j \leq L$ and $\bm{c}_{i},\bm{c}_{j} \leq 2m$}
\end{align*}
Hence,
\begin{align*}
    \prod_{k=1}^{L} \binom{\abs{\mathscr{C}_{k}}}{\bm{c}_{k}} \geq \abs{\mathscr{G}(\bm{c})} \geq \exp\wrapp{-O\wrapp{\frac{L^{3}m^{2}}{n}}} \cdot \prod_{k=1}^{L} \binom{\abs{\mathscr{C}_{k}}}{\bm{c}_{k}}
\end{align*}
Finally, recall that $\abs{\mathscr{C}_{k}} = (n)_{k} \cdot d^{k} \cdot \lambda_{k}$. Clearly, we have $\wrapp{\abs{\mathscr{C}_{k}}}_{\bm{c}_{k}} \leq \lambda_{k}^{\bm{c}_{k}}(dn)^{k \cdot \bm{c}_{k}}$. In the other direction,
\begin{align*}
    \frac{\wrapp{\abs{\mathscr{C}_{k}}}_{\bm{c}_{k}}}{\lambda_{k}^{\bm{c}_{k}}(dn)^{k\bm{c}_{k}}} &\geq \frac{\wrapp{(n-k)^{k}d^{k}\lambda_{k} - \bm{c}_{k}}^{\bm{c}_{k}}}{\lambda_{k}^{\bm{c}_{k}}(dn)^{k \cdot \bm{c}_{k}}} \\
    &= \wrapp{\wrapp{1 - \frac{k}{n}}^{k} - \frac{\bm{c}_{k}}{\lambda_{k} (dn)^{k}}}^{\bm{c}_{k}} \\
    &\geq \wrapp{1 - \frac{k^{2}}{n} - \frac{\bm{c}_{k}}{\lambda_{k} (dn)^{k}}}^{\bm{c}_{k}} \tag{Bernoulli Inequality} \\
    &\geq \wrapp{1 - \frac{L^{2}}{n} - \frac{4m}{n^{k}d^{k}(d-1)^{k}}}^{2m} \tag{Using $k \leq L$, $k\bm{c}_{k} \leq 2m$, and $\lambda_{k} = \frac{(d-1)^{k}}{2k}$} \\
    &\geq \wrapp{1 - \frac{L^{2}}{n} - \frac{2m}{dn}}^{2m} \tag{Using $k \geq 1$} \\
    &\geq \exp\wrapp{-O\wrapp{\frac{L^{2}m + m^{2}}{n}}}.
\end{align*}
Combined with the above approximation to $\abs{\mathscr{G}(\bm{c})}$ yields the claim.
\end{proof}
Putting everything together completes the proof.
\end{proof}

\begin{proof}[Proof of \cref{lem:bad}]
Treating $\sqrt{\frac{1}{1 - (d-1) \cdot \abs{\delta_{2}}}}$ as a constant depending only on $d, \beta$, \cref{prop:2nd-moment} tells us that for any $\Gamma \in \mathscr{A}(\bm{c})$,
\begin{align*}
    \frac{\E_{M}\wrapb{Z_{M,<\epsilon}^{(2)}(\beta) \cdot \allone[\Gamma \subseteq M]}}{\abs{\E_{M}\wrapb{Z_{M,<\epsilon}(\beta)}}^{2}} \lesssim \Pr[\Gamma \subseteq M] \cdot \abs{Z_{\Gamma}(\beta)}^{2} \lesssim \Pr[\Gamma \subseteq M] \cdot \exp\wrapp{O\wrapp{\abs{\Gamma}}},
\end{align*}
where in the final step, we used a very crude upper bound on $\abs{Z_{\Gamma}(\beta)}$ combined with the fact that $\beta$ is bounded by some constant depending on $d$ but not on $N$. Here, note that $Z_{\Gamma}(\beta)$ only considers the number $m(\Gamma)$ of distinct edges occurring in $\Gamma$, which is always at most $\abs{\Gamma}$. Combining this with the fact that $b_{\abs{C}} \leq \exp\wrapp{O\wrapp{\abs{C}}}$ as above, the constraints on $\bm{c}^{(1)},\bm{c}^{(2)},\bm{c}^{(3)}$, and setting $\bm{c} = \bm{c}^{(1)} + \bm{c}^{(2)} + \bm{c}^{(3)}$, we have
\begin{align*}
    &\abs{\varphi\wrapp{\bm{c}^{(1)},\bm{c}^{(2)},\bm{c}^{(3)}} - \varphi^{\mathscr{G}}\wrapp{\bm{c}^{(1)},\bm{c}^{(2)},\bm{c}^{(3)}}} \\
    &\lesssim \exp\wrapp{O(m)} \cdot \sum_{\substack{\Gamma^{(i)} \in \mathscr{G}(\bm{c}^{(i)}),i=1,2,3 \\ \Gamma^{(1)} \cup \Gamma^{(3)}, \Gamma^{(2)} \cup \Gamma^{(3)} \in \mathscr{G}(m) \\ \Gamma^{(1)} \cup \Gamma^{(2)} \cup \Gamma^{(3)} \notin \mathscr{G}}} \Pr[\Gamma^{(1)} \cup \Gamma^{(2)} \cup \Gamma^{(3)} \subseteq M] \\
    &\leq \exp\wrapp{O(m)} \cdot \prod_{k=1}^{L} \frac{\bm{c}_{k}!}{\bm{c}_{k}^{(1)}!\bm{c}_{k}^{(2)}!\bm{c}_{k}^{(3)}!} \cdot \sum_{\Gamma \in \mathscr{B}(\bm{c})} \Pr[\Gamma \subseteq M] \tag{Setting $\Gamma = \Gamma^{(1)} \cup \Gamma^{(2)} \cup \Gamma^{(3)}$} \\
    &\lesssim \exp\wrapp{O(m\log m)} \cdot \sum_{\Gamma \in \mathscr{B}(\bm{c})} \Pr[\Gamma \subseteq M]. \tag{Using the constraints on $\bm{c}$}
\end{align*}
Given this, our goal will be to prove
\begin{align}\label{eq:sumbad-unif-bound}
    \sum_{\Gamma \in \mathscr{B}(\bm{c})} \Pr[\Gamma \subseteq M] \leq \frac{\exp(\exp(O_{d}(L)) + O(m\log m))}{n} \cdot \prod_{k=1}^{L} \frac{\lambda_{k}^{\bm{c}_{k}}}{\bm{c}_{k}!}.
\end{align}
We will again use very crude estimates. Note that $\prod_{k=1}^{L} \frac{\lambda_{k}^{\bm{c}_{k}}}{\bm{c}_{k}!} \gtrsim m^{-O(m)}$ since $\abs{\bm{c}} \leq 2m$ implies $\prod_{k=1}^{L} \frac{1}{\bm{c}_{k}!} \geq \frac{1}{(2m)!}$ and $\prod_{k=1}^{L} \lambda_{k}^{\bm{c}_{k}} \geq \prod_{k=1}^{L} \frac{1}{2k} \geq \frac{1}{(2m)!}$. While this may seem tiny, we're again thinking of $m,L$ as extremely slow-growing functions in $N$. Hence, even with such a weak lower bound, we will be done if we can show
\begin{align*}
    \sum_{\Gamma \in \mathscr{B}(\leq 2m)} \Pr[\Gamma \subseteq M] \leq \frac{\exp(\exp(O_{d}(L)) + O(m\log m))}{n}.
\end{align*}
Note that we're even summing over all $\bm{c} \in \N^{L}$ satisfying $\abs{\bm{c}} \leq 2m$, which makes the left-hand side only larger. This is primarily to help us cut down on the notation required for the proof.

Following the usual terminology of cluster expansions, let us say two cycles are compatible if and only if they do not share any clouds, i.e. they are vertex-disjoint. Otherwise, we say they are incompatible. In the language of abstract polymer models, each cycle is its own polymer. A \emph{cluster} is a connected component within the incompatibility graph on cycles. Observe that any collection $\Gamma$ of cycles can be decomposed into a union of vertex-disjoint \emph{clusters}. The collections $\Gamma \in \mathscr{B}(\leq 2m)$ are precisely the ones such that at least one cluster in $\Gamma$ has at least two cycles.

For a given collection of cycles $\Gamma$, let $m(\Gamma)$ denote the number of distinct edges in $\Gamma$. Note that if $\Gamma$ consists of clusters $\Gamma_{1},\dots,\Gamma_{t}$, then $m(\Gamma) = \sum_{s=1}^{t} m(\Gamma_{s})$, $\abs{\Gamma} = \sum_{s=1}^{t} \abs{\Gamma_{s}}$, and by \cref{lem:shortcyc-approxindep},
\begin{align*}
    \Pr[\Gamma \subseteq M] &= \exp\wrapp{\pm O(m^{2}/n)} \cdot \prod_{s=1}^{t} \Pr[\Gamma_{s} \subseteq M].
\end{align*}

Let $\mathsf{Clusters}$ denote the collection of all clusters. Let $\mathsf{Simple} \defeq \bigcup_{k=1}^{L} \mathscr{C}_{k} \subseteq \mathsf{Clusters}$ denote the collection of ``simple clusters'', i.e. clusters which consist of a single cycle. Let $\mathsf{Bad} = \mathsf{Clusters} \setminus \mathsf{Simple}$ denote the set of all ``bad clusters'', i.e. the clusters which consist of multiple cycles. Define
\begin{align*}
    \Theta^{\mathsf{simple}} &= \sum_{\Gamma \in \mathsf{Simple}} \Pr[\Gamma \subseteq M] \\
    \Theta^{\mathsf{bad}} &= \sum_{\Gamma \in \mathsf{Bad}} \Pr[\Gamma \subseteq M].
\end{align*}
Since every $\Gamma \in \mathscr{B}(\leq 2m)$ can be decomposed into a vertex-disjoint union of ``simple clusters'' and at least one ``bad cluster'', combined with the above upper bound on $\Pr[\Gamma \subseteq M]$ yields
\begin{align*}
    \sum_{\Gamma \in \mathscr{B}(\leq 2m)} \Pr[\Gamma \subseteq M] &\leq \exp(O(m^{2}/n)) \cdot \sum_{\substack{\mathcal{S} \subseteq \mathsf{Simple} \\ \mathcal{B} \subseteq \mathsf{Bad} : \mathcal{B} \neq \emptyset}} \prod_{\Gamma \in \mathcal{S}} \Pr[\Gamma \subseteq M] \prod_{\Gamma \in \mathcal{B}} \Pr[\Gamma \subseteq M] \\
    &\leq \exp(O(m^{2}/n)) \cdot \prod_{\Gamma \in \mathsf{Simple}} (1 + \Pr[\Gamma \subseteq M]) \cdot \wrapp{\prod_{\Gamma \in \mathsf{Bad}} (1 + \Pr[\Gamma \subseteq M]) - 1} \\
    &\leq \exp\wrapp{O\wrapp{\frac{m^{2}}{n}} + \Theta^{\mathsf{simple}}} \cdot \wrapp{\exp\wrapp{\Theta^{\mathsf{bad}}} - 1},
\end{align*}
Observe that
\begin{align*}
    \Theta^{\mathsf{simple}} &= \sum_{k=1}^{L} \sum_{C \in \mathscr{C}_{k}} \Pr[C \subseteq M] \\
    &\leq \sum_{k=1}^{L} \lambda_{k} \cdot (dn)^{k} \cdot \Pr[C \subseteq M] \tag{Using $\abs{\mathscr{C}_{k}} = \frac{(n)_{k} \cdot d^{k}(d-1)^{k}}{2k} \leq \lambda_{k} \cdot (dn)^{k}$} \\
    &\leq \exp(O(L^{2}/n)) \cdot \sum_{k=1}^{L} \lambda_{k} \tag{Again using $\Pr[C \subseteq M] \leq \exp(O(L^{2}/n)) \cdot \wrapp{\frac{1}{dn}}^{\abs{C}}$} \\
    &\leq \exp\wrapp{O\wrapp{\frac{L^{2}}{n}} + O_{d}(L)}. \tag{Using $\lambda_{k} = \frac{(d-1)^{k}}{2k}$}
\end{align*}
For $\Theta^{\mathsf{bad}}$, observe that
\begin{align*}
    \sum_{\Gamma \in \mathsf{Bad}} \Pr[\Gamma \subseteq M] \leq \exp\wrapp{O\wrapp{\frac{m^{2}}{n}}} \cdot \sum_{\Gamma \in \mathsf{Bad}} \frac{1}{(dn)^{m(\Gamma)}}.
\end{align*}
Now, every choice of $\Gamma$ gives rise to a connected (multi)graph $G_{\Gamma} = (V(\Gamma),E(\Gamma))$ with $|E(\Gamma)| = m(\Gamma)$ by contracting the clouds and removing the isolated vertices. For a given $G_{\Gamma}$, there are at most $n^{\abs{V(\Gamma)}}d^{2m(\Gamma)}$ possible ways of embedding $G_{\Gamma}$ into the configuration model, since there are $n(n-1)\dotsb\wrapp{n-\abs{V(\Gamma)}-1}$ choices for the clouds, and then for each edge in $G_{\Gamma}$, there are at most $d^{2}$ choices for the pair of sockets giving rise to that edge. Moreover, if $\Gamma \in \mathsf{Bad}$, then $G_{\Gamma}$ must satisfy $m(\Gamma) \geq \abs{V(\Gamma)} + 1$, since $\Gamma$ has at least two cycles which intersect at some vertex/cloud but which are not identical. It follows that
\begin{align*}
    \sum_{\Gamma \in \mathsf{Bad}} \frac{1}{(dn)^{m(\Gamma)}} \leq \sum_{\text{possible } G_{\Gamma}} \frac{n^{\abs{V(\Gamma)}}d^{2m(\Gamma)}}{(dn)^{m(\Gamma)}} \leq \frac{d^{2m}}{n} \cdot \#\{\text{possible } G_{\Gamma}\}.
\end{align*}
Since the number of vertices $\abs{V(\Gamma)}$ is at most $m$, the number of possible $G_{\Gamma}$ is at most $m^{m+1} \cdot \binom{m^{2}}{m} \leq m^{O(m)}$; the factor of $m \cdot \binom{m^{2}}{m}$ comes from choosing at most $m$ distinct pairs of vertices to form edges in $G_{\Gamma}$, a graph with at most $m$ vertices, and the extra factor of $m^{m}$ comes from the fact that each such edge can occur at most $m$ times. This may seem crude, but it is sharp up to the constant in the exponent, since the number of spanning trees of the complete graph on $m$ vertices is $m^{m-2}$ by Cayley's Formula. Putting these together, we have
\begin{align*}
    \Theta^{\mathsf{bad}} \leq \frac{1}{n} \exp\wrapp{O(m\log m) + O\wrapp{\frac{m^{2}}{n}}}
\end{align*}
Putting the upper bounds for $\Theta^{\mathsf{simple}}$ and $\Theta^{\mathsf{bad}}$ together, we obtain
\begin{align*}
    \sum_{\Gamma \in \mathscr{B}(\leq 2m)} \Pr[\Gamma \subseteq M] \lesssim \frac{1}{n} \cdot \exp(\exp(O_{d}(L)) + O(m\log m)).
\end{align*}
Combined with $\prod_{k=1}^{L} \frac{\lambda_{k}^{\bm{c}_{k}}}{\bm{c}_{k}!} \gtrsim m^{-O(m)} = \exp(-O(m\log m))$, this yields the relative upper bound
\begin{align*}
    \sum_{\Gamma \in \mathscr{B}(\bm{c})} \Pr[\Gamma \subseteq M] \leq \frac{\exp(\exp(O_{d}(L)) + O(m\log m))}{n} \cdot \prod_{k=1}^{L} \frac{\lambda_{k}^{\bm{c}_{k}}}{\bm{c}_{k}!}.
\end{align*}
Plugging in this upper bound completes the proof.
\end{proof}

{\sloppy
\printbibliography

@inproceedings{BHLLR26,
author = {Bencs, Ferenc and Huang, Brice and Lee, Daniel Z. and Liu, Kuikui and Regts, Guus},
title = {On Zeros and Algorithms for Disordered Systems: Mean-Field Spin Glasses},
year = {2026},
isbn = {9798400725364},
publisher = {Association for Computing Machinery},
address = {New York, NY, USA},
url = {https://doi.org/10.1145/3798129.3800769},
doi = {10.1145/3798129.3800769},
booktitle = {Proceedings of the 58th Annual ACM Symposium on Theory of Computing},
pages = {508–517},
numpages = {10},
location = {Salt Lake City, UT, USA},
series = {STOC '26}
}

@article{BGT13,
  author  = {Bayati, Mohsen and Gamarnik, David and Tetali, Prasad},
  title   = {Combinatorial approach to the interpolation method and scaling limits in sparse random graphs},
  journal = {The Annals of Probability},
  volume  = {41},
  number  = {6},
  pages   = {4080--4115},
  year    = {2013},
  doi     = {10.1214/12-AOP816},
  url     = {https://doi.org/10.1214/12-AOP816}
}

@article{CF16,
  author  = {Csikv\'{a}ri, P\'{e}ter and Frenkel, P\'{e}ter E.},
  title   = {{Benjamini--Schramm} continuity of root moments of graph polynomials},
  journal = {European Journal of Combinatorics},
  volume  = {52},
  pages   = {302--320},
  year    = {2016},
  doi     = {10.1016/j.ejc.2015.07.009},
  url     = {https://doi.org/10.1016/j.ejc.2015.07.009}
}

@article{KL26,
  author  = {Koehler, Frederic and Leung, Pui Kuen},
  title   = {Approximating the Permanent of a Random Matrix with Polynomially Small Mean: Zeros and Universality},
  journal = {arXiv preprint arXiv:2604.01367},
  year    = {2026},
  doi     = {10.48550/arXiv.2604.01367},
  url     = {https://arxiv.org/abs/2604.01367}
}

@article{CLMS22,
  title={The {I}sing antiferromagnet and max cut on random regular graphs},
  author={Coja-Oghlan, Amin and Loick, Philipp and Mezei, Bal{\'a}zs F and Sorkin, Gregory B},
  journal={SIAM J. Disc. Math.},
  volume={36},
  number={2},
  pages={1306--1342},
  year={2022}
}

@incollection{FL21,
  author    = {Christian Fabian and Philipp Loick},
  title     = {The Ising Antiferromagnet in the Replica Symmetric Phase},
  booktitle = {Extended Abstracts EuroComb 2021},
  series    = {Trends in Mathematics},
  publisher = {Birkh{\"a}user},
  address   = {Cham},
  pages      = {299--306},
  year       = {2021},
  doi        = {10.1007/978-3-030-83823-2_47},
  url        = {https://doi.org/10.1007/978-3-030-83823-2_47}
}

@article{BBD25,
  author  = {Roland Bauerschmidt and Thierry Bodineau and Benoit Dagallier},
  title   = {Kawasaki Dynamics Beyond the Uniqueness Threshold},
  journal = {Probability Theory and Related Fields},
  volume  = {192},
  number  = {1--2},
  pages   = {267--350},
  year    = {2025},
  doi     = {10.1007/s00440-024-01326-9},
  url     = {https://doi.org/10.1007/s00440-024-01326-9}
}

@article{BRZ95,
  title={On the purity of the limiting {G}ibbs state for the {I}sing model on the Bethe lattice},
  author={Bleher, Pavel M and Ruiz, Jean and Zagrebnov, Valentin A},
  journal={Journal of Statistical Physics},
  volume={79},
  number={1},
  pages={473--482},
  year={1995},
  publisher={Springer}
}

@article{HS26,
  title={Strong Topological Trivialization of Multi-Species Spherical Spin Glasses},
  author={Huang, Brice and Sellke, Mark},
  journal={Annals of Probability},
  volume={54},
  number={2},
  pages={1034--1107},
  year={2026}
}

@book{Roc70,
  title={Convex Analysis},
  author={Rockafellar, R. Tyrrell},
  year={1970},
  publisher={Princeton University Press}
}

@book{Bar16book,
  TITLE = {Combinatorics and Complexity of Partition Functions},
  AUTHOR = {Barvinok, Alexander},
  YEAR = {2016},
  PUBLISHER = {Springer Algorithms and Combinatorics},
  volume={30}
}

@inproceedings{CLV21a,
    author = {Chen, Zongchen and Liu, Kuikui and Vigoda, Eric},
    title = {Optimal Mixing of Glauber Dynamics: Entropy Factorization via High-Dimensional Expansion},
    year = {2021},
    publisher = {Association for Computing Machinery},
    address = {New York, NY, USA},
    booktitle = {Proceedings of the 53rd Annual ACM SIGACT Symposium on Theory of Computing},
    pages = {1537--1550},
    location = {Virtual, Italy},
    series = {STOC 2021}
}

@inproceedings{CLV20,
    author = {Chen, Zongchen and Liu, Kuikui and Vigoda, Eric},
    title={Rapid Mixing of Glauber Dynamics up to Uniqueness via Contraction},
    year={2020},
    booktitle={2020 IEEE 61st Annual Symposium on Foundations of Computer Science (FOCS)},
    pages={1307--1318},
    series={FOCS 2020}
}

@article{CCYZ25,
  title={Rapid Mixing at the Uniqueness Threshold},
  author={Chen, Xiaoyu and Chen, Zongchen and Yin, Yitong and Zhang, Xinyuan},
  journal={arXiv preprint arXiv:2411.03413},
  year={2025}
}

@article{Bar16,
 author = {Barvinok, Alexander},
 title = {Computing the Permanent of (Some) Complex Matrices},
 journal = {Found. Comput. Math.},
 issue_date = {April 2016},
 volume = {16},
 number = {2},
 month = apr,
 year = {2016},
 issn = {1615-3375},
 pages = {329--342},
 numpages = {14},
 acmid = {2898324},
 publisher = {Springer-Verlag New York, Inc.},
 address = {Secaucus, NJ, USA}
}

@article{Bar17perm,
title={Approximating permanents and hafnians},
author={Barvinok, Alexander},
year={2017},
journal={Discrete Analysis},
numpages={34},
volume={2},
}

@inproceedings{CFYZ21a,
  author={Chen, Xiaoyu and Feng, Weiming and Yin, Yitong and Zhang, Xinyuan},
  booktitle={2021 IEEE 62nd Annual Symposium on Foundations of Computer Science (FOCS)}, 
  title={Rapid mixing of Glauber dynamics via spectral independence for all degrees}, 
  year={2021},
  volume={},
  number={},
  pages={137-148},
}

@book {Stanleybookvol2,
    AUTHOR = {Stanley, Richard P.},
     TITLE = {Enumerative combinatorics. {V}ol. 2},
    SERIES = {Cambridge Studies in Advanced Mathematics},
    VOLUME = {62},
      NOTE = {With a foreword by Gian-Carlo Rota and appendix 1 by Sergey
              Fomin},
 PUBLISHER = {Cambridge University Press, Cambridge},
      YEAR = {1999},
     PAGES = {xii+581},
      ISBN = {0-521-56069-1; 0-521-78987-7},
   MRCLASS = {05A15 (05-02 05E05 05E10 68R05)},
  MRNUMBER = {1676282},
MRREVIEWER = {Ira\ Gessel},
       DOI = {10.1017/CBO9780511609589},
       URL = {https://doi-org.proxy.uba.uva.nl/10.1017/CBO9780511609589},
}

@article {BarBarIsing,
    AUTHOR = {Barvinok, Alexander and Barvinok, Nicholas},
     TITLE = {More on zeros and approximation of the {I}sing partition
              function},
   JOURNAL = {Forum Math. Sigma},
  FJOURNAL = {Forum of Mathematics. Sigma},
    VOLUME = {9},
      YEAR = {2021},
     PAGES = {Paper No. e46, 18},
      ISSN = {2050-5094},
   MRCLASS = {82B20 (30C15 68Q09)},
  MRNUMBER = {4271356},
       DOI = {10.1017/fms.2021.40},
       URL = {https://doi-org.proxy.uba.uva.nl/10.1017/fms.2021.40},
}

@article {GGHPcomplexityIsing,
    AUTHOR = {Galanis, Andreas and Goldberg, Leslie A. and Herrera-Poyatos,
              Andres},
     TITLE = {The complexity of approximating the complex-valued {I}sing
              model on bounded degree graphs},
   JOURNAL = {SIAM J. Discrete Math.},
  FJOURNAL = {SIAM Journal on Discrete Mathematics},
    VOLUME = {36},
      YEAR = {2022},
    NUMBER = {3},
     PAGES = {2159--2204},
      ISSN = {0895-4801,1095-7146},
   MRCLASS = {68R05 (05C31 82B20)},
  MRNUMBER = {4476941},
       DOI = {10.1137/21M1454043},
       URL = {https://doi-org.proxy.uba.uva.nl/10.1137/21M1454043},
}

@article {PRSdiskIsing,
    AUTHOR = {Patel, Viresh and Regts, Guus and Stam, Ayla},
     TITLE = {A near-optimal zero-free disk for the {I}sing model},
   JOURNAL = {Comb. Theory},
  FJOURNAL = {Combinatorial Theory},
    VOLUME = {4},
      YEAR = {2024},
    NUMBER = {2},
     PAGES = {Paper No. 9, 15},
      ISSN = {2766-1334},
   MRCLASS = {05C31 (68Q87 82B20)},
  MRNUMBER = {4807148},
}

@book {SteinShakarchi2003,
    AUTHOR = {Stein, Elias M. and Shakarchi, Rami},
     TITLE = {Fourier analysis},
    SERIES = {Princeton Lectures in Analysis},
    VOLUME = {1},
      NOTE = {An introduction},
 PUBLISHER = {Princeton University Press, Princeton, NJ},
      YEAR = {2003},
     PAGES = {xvi+311},
      ISBN = {0-691-11384-X},
   MRCLASS = {42-01},
  MRNUMBER = {1970295},
MRREVIEWER = {Steven\ George\ Krantz},
}

@inproceedings{AJKPV22,
    author = {Anari, Nima and Jain, Vishesh and Koehler, Frederic and Pham, Huy Tuan and Vuong, Thuy-Duong},
    title = {Entropic Independence: Optimal Mixing of down-up Random Walks},
    year = {2022},
    publisher = {Association for Computing Machinery},
    address = {New York, NY, USA},
    booktitle = {Proceedings of the 54th Annual ACM SIGACT Symposium on Theory of Computing},
    pages = {1418--1430},
    numpages = {13},
    location = {Rome, Italy},
    series = {STOC 2022}
}

@INPROCEEDINGS{CE22,
author = {Chen, Yuansi and Eldan, Ronen},
booktitle = {2022 IEEE 63rd Annual Symposium on Foundations of Computer Science (FOCS)},
title = {Localization Schemes: A Framework for Proving Mixing Bounds for Markov Chains (extended abstract)},
year = {2022},
volume = {},
issn = {},
pages = {110-122},
doi = {10.1109/FOCS54457.2022.00018},
url = {https://doi.ieeecomputersociety.org/10.1109/FOCS54457.2022.00018},
publisher = {IEEE Computer Society},
address = {Los Alamitos, CA, USA},
month = nov
}

@article{Eld13,
    title={Thin Shell Implies Spectral Gap Up to Polylog via a Stochastic Localization Scheme},
    author={Eldan, Ronen},
    year={2013},
    journal={Geometric and Functional Analysis},
    pages={532--569},
    volume={23},
}

@article{BB19,
title = {A very simple proof of the LSI for high temperature spin systems},
journal = {Journal of Functional Analysis},
volume = {276},
number = {8},
pages = {2582-2588},
year = {2019},
issn = {0022-1236},
doi = {https://doi.org/10.1016/j.jfa.2019.01.007},
url = {https://www.sciencedirect.com/science/article/pii/S0022123619300278},
author = {Roland Bauerschmidt and Thierry Bodineau},
}

@article{Len20,
  title={Beitrag zum Verst{\"a}ndnis der magnetischen Erscheinungen in festen K{\"o}rpern},
  author={W. Lenz},
  year={1920},
  journal={Physikalische Zeitschrift},
  volume={21},
  pages={613--615}
}

@article{Isi25,
    title={Beitrag zur Theorie des Ferromagnetismus},
    author={Ernst Ising},
    year={1925},
    journal={Zeitschrift f\"{u}r Physik},
    volume={31},
    pages={253--258},
}

@InProceedings{KLR22,
  title = 	 {Sampling Approximately Low-Rank Ising Models: MCMC meets Variational Methods},
  author =       {Koehler, Frederic and Lee, Holden and Risteski, Andrej},
  booktitle = 	 {Proceedings of Thirty Fifth Conference on Learning Theory},
  pages = 	 {4945--4988},
  year = 	 {2022},
  editor = 	 {Loh, Po-Ling and Raginsky, Maxim},
  volume = 	 {178},
  series = 	 {Proceedings of Machine Learning Research},
  month = 	 {7},
  publisher =    {PMLR},
  url = 	 {https://proceedings.mlr.press/v178/koehler22a.html},
}

@article{EKZ22,
    title={A spectral condition for spectral gap: fast mixing in high-temperature Ising models},
    author={Ronen Eldan and Frederic Koehler and Ofer Zeitouni},
    year={2022},
    journal={Probability Theory and Related Fields},
    volume={182},
    pages={1035--1051},
}

@inproceedings{Wei06,
    author = {Weitz, Dror},
    title = {Counting Independent Sets Up to the Tree Threshold},
    booktitle = {Proceedings of the 38th Annual ACM Symposium on Theory of Computing (STOC)},
    year = {2006},
    location = {Seattle, WA, USA},
    pages = {140--149},
    numpages = {10},
}

@inproceedings{Sly10,
  author={Sly, Allan},
  booktitle={2010 IEEE 51st Annual Symposium on Foundations of Computer Science}, 
  title={Computational Transition at the Uniqueness Threshold}, 
  year={2010},
  volume={},
  number={},
  pages={287--296},
  doi={10.1109/FOCS.2010.34}
}

@article{SS14,
author = {Sly, Allan and Sun, Nike},
title = {The Computational Hardness of Counting in Two-Spin Models on $d$-Regular Graphs},
journal={The Annals of Probability},
year={2014},
volume={42},
number={6},
pages={2383--2416}
}

@article{EKPS00,
title={Broadcasting on trees and the Ising model},
author={Evans, William and Kenyon, Claire and Peres, Yuval and Schulman, Leonard J.},
year={2000},
journal={The Annals of Applied Probability},
volume={10},
number={2},
pages={410--433},
}

@article{SST14,
author={Sinclair, Alistair
and Srivastava, Piyush
and Thurley, Marc},
title={Approximation Algorithms for Two-State Anti-Ferromagnetic Spin Systems on Bounded Degree Graphs},
journal={Journal of Statistical Physics},
year={2014},
volume={155},
number={4},
pages={666--686}
}

@article{PR17,
title = {Deterministic Polynomial-Time Approximation Algorithms for Partition Functions and Graph Polynomials},
journal = {SIAM Journal on Computing},
volume = {46},
pages = {1893--1919},
year = {2017},
number={6},
author = {Viresh Patel and Guus Regts}
}

@article{PR22,
    AUTHOR = {Patel, Viresh and Regts, Guus},
     TITLE = {Approximate counting using {T}aylor's theorem: a survey},
   JOURNAL = {Bull. Eur. Assoc. Theor. Comput. Sci. EATCS},
  FJOURNAL = {Bulletin of the European Association for Theoretical Computer
              Science. EATCS},
    NUMBER = {138},
      YEAR = {2022},
     PAGES = {40--70},
      ISSN = {0252-9742},
   MRCLASS = {68R05 (05C31)},
  MRNUMBER = {4608123},
MRREVIEWER = {Paul\ Andrew\ Dreyer, Jr.},
}

@article{LSS19b,
  title={Fisher zeros and correlation decay in the Ising model},
  author={Liu, Jingcheng and Sinclair, Alistair and Srivastava, Piyush},
  journal={Journal of Mathematical Physics},
  volume={60},
  number={10},
  pages = {103304},
  year={2019},
}

@article{JVV86,
  title={Random generation of combinatorial structures from a uniform distribution},
  author={Jerrum, Mark R. and Valiant, Leslie G. and Vazirani, Vijay V.},
  journal={Theoretical Computer Science},
  volume={43},
  pages={169--188},
  year={1986},
  publisher={Elsevier}
}

@article{Bar15,
 author = {Barvinok, Alexander},
 title = {Computing the Partition Function for Cliques in a Graph},
 year = {2015},
 pages = {339--355},
 doi = {10.4086/toc.2015.v011a013},
 publisher = {Theory of Computing},
 journal = {Theory of Computing},
 volume = {11},
 number = {13},
 URL = {http://www.theoryofcomputing.org/articles/v011a013},
}

@book{Bol01,
    series={Cambridge Studies in Advanced Mathematics},
    title={Random Graphs},
    publisher={Cambridge University Press},
    author={Bollob\'{a}s, B\'{e}la},
    year={2001},
}

@article{BCG22,
  author  = {Bordenave, Charles and Chafa{\"i}, Djalil and Garc{\'i}a-Zelada, David},
  title   = {Convergence of the spectral radius of a random matrix through its characteristic polynomial},
  journal = {Probability Theory and Related Fields},
  year    = {2022},
  volume  = {182},
  number  = {3--4},
  pages   = {1163--1181},
  month   = apr,
  doi     = {10.1007/s00440-021-01079-9},
  url     = {https://doi.org/10.1007/s00440-021-01079-9}
}

@INPROCEEDINGS{EM18,
  author={Eldar, Lior and Mehraban, Saeed},
  booktitle={2018 IEEE 59th Annual Symposium on Foundations of Computer Science (FOCS)}, 
  title={Approximating the Permanent of a Random Matrix with Vanishing Mean}, 
  year={2018},
  volume={},
  number={},
  pages={23-34},
  doi={10.1109/FOCS.2018.00012}}

@inproceedings{AKV24,
author = {Anari, Nima and Koehler, Frederic and Vuong, Thuy-Duong},
title = {Trickle-Down in Localization Schemes and Applications},
year = {2024},
isbn = {9798400703836},
publisher = {Association for Computing Machinery},
address = {New York, NY, USA},
url = {https://doi.org/10.1145/3618260.3649622},
doi = {10.1145/3618260.3649622},
booktitle = {Proceedings of the 56th Annual ACM Symposium on Theory of Computing},
pages = {1094–1105},
numpages = {12},
location = {Vancouver, BC, Canada},
series = {STOC 2024}
}

@book{Pan13b,
    title={The Sherrington-Kirkpatrick Model},
    author={Panchenko, Dmitry},
    year={2013},
    publisher={Springer New York, NY},
}
}

\appendix

\section{Basic moment estimates for real \texorpdfstring{$\beta$}{beta}}
\label{a:basic-mt}

In this appendix we prove~\cref{l:1mt-balanced} and \cref{l:basic-mt} using Laplace's method.

\subsection{First moment}

This subsection is devoted to the proof of~\cref{l:1mt-balanced}.
Let
\[
  {\mathcal{L}} = \left\{
    (\alpha, \gamma) \in \R^2: \alpha \in [0,1], \gamma \in [0,\min(\alpha,1-\alpha)]
  \right\}\,.
\]
For $x\in [0,1]$ let ${\mathcal{H}}(x) = -x\log x - (1-x)\log(1-x)$ denote the binary entropy function.
By slight abuse of notation, for a vector ${\underline{x}}$ with nonnegative entries and $\|{\underline{x}}\|_1 = 1$, let ${\mathcal{H}}({\underline{x}})$ denote the entropy of the discrete distribution with probabilities ${\underline{x}}$.
For $(\alpha, \gamma) \in {\mathcal{L}}$ let
\begin{align}\label{e:basic-mt-Phi}
  \Phi^\beta(\alpha,\gamma) &= {\mathcal{H}}(\alpha) + d \Psi^\beta(\alpha,\gamma)\,, \qquad \text{where} \\
  \nonumber
  \Psi^\beta(\alpha,\gamma) &= - {\mathcal{H}}(\alpha)
  + \frac12 {\mathcal{H}}\left(\gamma,\gamma,\alpha-\gamma,1-\alpha-\gamma\right)
  + \gamma \beta\,.
\end{align}
Let
\[
  {\mathcal{L}}_n = \left\{
    (\alpha, \gamma) \in {\mathcal{L}} :
    \text{$\alpha n$ is an integer and $(\alpha + \gamma) dn$ is an even integer}
  \right\}\,.
\]
For $\sigma \in \Sigma_n$ let $n_+(\sigma)$ denote the number of $+1$-valued entries of $\sigma$.
It is simple to check that there exists $\sigma \in \Sigma_n$ and a realization of $G \sim \configurationModel(n,d)$ with $n_+(\sigma) = \alpha n$ and $\isingHamiltonian_{G}(\sigma) = \gamma dn$ if and only if $(\alpha,\gamma) \in {\mathcal{L}}_n$.
\begin{lemma}\label{l:1mt-exp-rate}
  For any $(\alpha,\gamma) \in {\mathcal{L}}_n$ and $\beta \in \R$, the contribution to $\E[Z_{G}(\beta)]$ from $\sigma$ with $n_+(\sigma) = \alpha n$ and $\isingHamiltonian_{G}(\sigma) = \gamma dn$ is
  \begin{equation}\label{e:1mt-exp-rate}
    \sum_{\sigma \in \Sigma_n} \allone\{n_+(\sigma) = \alpha n\}
    \E\left[
      \allone\{\isingHamiltonian_{G}(\sigma) = \gamma dn\}
      e^{\beta \isingHamiltonian_{G}(\sigma)}
    \right]
    = \exp\left(
      n \Phi^\beta(\alpha,\gamma) + o(n)
    \right)\,.
  \end{equation}
  Consequently,
  \begin{align}
    \label{e:1mt-variational-inside}
    \E[Z_{G,< \varepsilon}(\beta)]
    &= \exp\left(
      n \sup\{
        \Phi^\beta(\alpha,\gamma) : (\alpha,\gamma) \in {\mathcal{L}}, |2\alpha - 1| < \varepsilon
      \} + o(n)
    \right)\,, \\
    \label{e:1mt-variational-outside}
    \E[Z_{G,\ge \varepsilon}(\beta)]
    &= \exp\left(
      n \sup\{
        \Phi^\beta(\alpha,\gamma) : (\alpha,\gamma) \in {\mathcal{L}}, |2\alpha - 1| \ge \varepsilon
      \} + o(n)
    \right)\,.
  \end{align}
\begin{proof}
  Let $m = \alpha n$ and $L = \gamma dn$.
  We claim that the left-hand side of \eqref{e:1mt-exp-rate} equals
  \begin{equation}\label{e:1mt-exact}
    \binom{n}{m} \cdot \frac{1}{(nd-1)!!} \cdot \binom{md}{L} \binom{(n-m)d}{L} L! (md-L-1)!! ((n-m)d-L-1)!! \cdot e^{\beta L}\,.
  \end{equation}
  The number of $\sigma$ with $n_+(\sigma) = \alpha n = m$ is clearly $\binom{n}{m}$.
  Consider any one of these $\sigma$.
  We have $\isingHamiltonian_{G}(\sigma) = \gamma dn$ if and only if $G$ has $L$ non-monochromatic edges (with respect to the vertex labeling $\sigma$).

  We next calculate the probability of this event for $G$ drawn from the configuration model.
  There are $(nd-1)!!$ possible perfect matchings of the half-edges of $G$.
  We claim that the number of perfect matchings resulting in $L$ non-monochromatic edges is
  \begin{equation}\label{e:1mt-count}
    \binom{md}{L} \binom{(n-m)d}{L} L! (md-L-1)!! ((n-m)d-L-1)!!\,.
  \end{equation}
  We label the $dn$ half-edges of $G$ by the spin value $\sigma_v$ of their attached vertex $v$.
  Let $S_+$ and $S_-$ be the sets of half-edges with labels $+1$ and $-1$, so that $|S_+| = md$, $|S_-| = (n-m) d$.
  A perfect matching with $L$ non-monochromatic edges can be specified by choosing $T_+\subseteq S_+$ and $T_- \subseteq S_-$ of size $|T_+| = |T_-| = L$ ($\binom{md}{L} \binom{(n-m)d}{L}$ choices), choosing a matching between $T_+$ and $T_-$ ($L!$ choices), and choosing a perfect matching among the remaining half-edges in $S_+\setminus T_+$ ($(md-L-1)!!$ choices) and in $S_- \setminus T_-$ ($((n-m)d-L-1)!!$ choices).
  This shows \eqref{e:1mt-count}.

  Finally, on this event, we have $e^{\beta \isingHamiltonian_{G}(\sigma)} = e^{\beta L}$.
  This proves the formula \eqref{e:1mt-exact}.
  A routine computation using Stirling's formula derives \eqref{e:1mt-exp-rate} from \eqref{e:1mt-exact}.
  The conclusions \eqref{e:1mt-variational-inside}, \eqref{e:1mt-variational-outside} are then standard from Laplace's method.
\end{proof}
\end{lemma}
\begin{lemma}\label{l:1mt-inner-stationarity}
  For any fixed $\alpha \in [0,1]$, $\beta \in \R$, the function $[0, {\color{blue} \min} 
  (\alpha,1-\alpha)] \ni \gamma \mapsto \Psi^\beta(\alpha,\gamma)$ is concave with unique maximum $\gamma = \gamma^\beta(\alpha)$ solving
  \begin{equation}\label{e:1mt-inner-stationarity}
    \log (\alpha-\gamma)
    + \log (1-\alpha-\gamma)
    - 2\log \gamma
    + 2\beta
    = 0\,.
  \end{equation}
\begin{proof}
  We directly compute that $2\partial_\gamma \Psi^\beta(\alpha,\gamma)$ equals the left-hand side of \eqref{e:1mt-inner-stationarity}.
  This is strictly decreasing in $\gamma$, with limits $+\infty$ and $-\infty$ as $\gamma \downarrow 0$ and $\gamma \uparrow {\color{blue} \min(\alpha,1-\alpha)}$
  So, $\gamma \mapsto \Psi^\beta(\alpha,\gamma)$ is concave and uniquely maximized at the $\gamma$ where $\partial_\gamma \Psi^\beta(\alpha,\gamma) = 0$, i.e. \eqref{e:1mt-inner-stationarity} holds.
\end{proof}
\end{lemma}
Then, for $\alpha \in [0,1]$, $\beta \in \R$ let
\begin{align*}
  \Psi^\beta(\alpha) = \sup\left\{
    \Psi^\beta(\alpha,\gamma) : \gamma \in \left[0, \min(\alpha,1-\alpha)\right]
  \right\} &= \Psi^\beta(\alpha,\gamma^\beta(\alpha))\,, \\
  \Phi^\beta(\alpha) = \sup\left\{
    \Phi^\beta(\alpha,\gamma) : \gamma \in \left[0, \min(\alpha,1-\alpha)\right]
  \right\} &= {\mathcal{H}}(\alpha) + d \Psi^\beta(\alpha)\,.
\end{align*}
\begin{fact}\label{f:1mt-beta-zero-trivial}
  At $\beta = 0$, we have $\Psi^0(\alpha) = 0$ for all $\alpha \in [0,1]$.
\begin{proof}
  From~\cref{l:1mt-inner-stationarity} it is straightforward to check that $\gamma^0(\alpha) = \alpha(1-\alpha)$.
  Plugging back into the definition of $\Psi$ yields $\Psi(\alpha,\gamma^0(\alpha),0) = 0$.
\end{proof}
\end{fact}
\begin{lemma}\label{l:1mt-alpha-beta-deriv}
  For any $\alpha \in [0,1]$, $\beta \in \R$, we have
  \[
    \partial_\alpha \partial_\beta \Psi^\beta(\alpha) = - \frac{2\alpha-1}{
      1 - 2\gamma (1-e^{-2\beta})
    }\,.
  \]
\begin{proof}
  Note that $\partial_\beta \Psi^\beta(\alpha) = \gamma^\beta(\alpha)$, and thus
  \[
    \partial_\alpha \partial_\beta \Psi^\beta(\alpha) = \partial_\alpha \gamma^\beta(\alpha)\,.
  \]
  We calculate this derivative using the implicit function theorem.
  From \eqref{e:1mt-inner-stationarity} we obtain (abbreviating $\gamma = \gamma^\beta(\alpha)$)
  \[
    \left(
      \frac{1}{\alpha-\gamma}
      - \frac{1}{1-\alpha-\gamma}
    \right)
    - \left(
      \frac{1}{\alpha-\gamma}
      + \frac{1}{1-\alpha-\gamma}
      + \frac{2}{\gamma}
    \right)
    \partial_\alpha \gamma^\beta(\alpha) = 0\,.
  \]
  This rearranges to
  \[
    \partial_\alpha \gamma^\beta(\alpha)
    = \frac{
      \frac{1}{\alpha-\gamma}
      - \frac{1}{1-\alpha-\gamma}
    }{
      \frac{1}{\alpha-\gamma}
      + \frac{1}{1-\alpha-\gamma}
      + \frac{2}{\gamma}
    }
    = - \frac{2\alpha-1}{
      1 - 2\gamma
      + \frac{2}{\gamma} (\alpha-\gamma)(1-\alpha-\gamma)
    }\,.
  \]
  From \eqref{e:1mt-inner-stationarity} we also have
  \[
    \frac{1}{\gamma}
    (\alpha-\gamma)(1-\alpha-\gamma)
    = 
    \gamma e^{-2\beta}\,.
  \]
  Thus,
  \[
    \partial_\alpha \gamma^\beta(\alpha)
    = - \frac{2\alpha-1}{
      1 - 2\gamma (1-e^{-2\beta})
    }. \qedhere
  \]
\end{proof}
\end{lemma}
\begin{lemma}\label{l:1mt-Phi-concave}
  For all $\alpha \in [\frac12,1]$ and $\beta \in [2\beta_\lowerBoundTag,+\infty)$, we have $\partial_\alpha \Phi^\beta(\alpha) \le -(2\alpha-1)$.
  (And by symmetry, for $\alpha \in [0,\frac12]$ we have $\partial_\alpha \Phi^\beta(\alpha) \ge 2\alpha-1$).
\begin{proof}
  Assume $\alpha \in [\frac12,1)$.
  Note that
  \[
    \frac{\differential}{\differential \alpha} {\mathcal{H}}(\alpha)
    = \log \frac{1-\alpha}{\alpha}
    \le -2(2\alpha-1)\,.
  \]
  So it suffices to show that
  \begin{equation}\label{e:1mt-Phi-concave-goal}
    \partial_\alpha \Psi^\beta(\alpha) \le \frac{2\alpha-1}{d}\,.
  \end{equation}
  If $\beta = 0$, this follows from~\cref{f:1mt-beta-zero-trivial}, which implies $\partial_\alpha \Psi(\alpha,0) = 0$.
  We now separately consider the cases $\beta > 0$ and $\beta \in [2\beta_\lowerBoundTag,0)$.

  \textit{Case 1: $\beta > 0$.}
  By~\cref{l:1mt-alpha-beta-deriv},
  \[
    \partial_\alpha \Psi^\beta(\alpha)
    = \partial_\alpha \Psi(\alpha,0) + \int_0^\beta \partial_\alpha \partial_\beta \Psi(\alpha,{\widetilde \beta}) \,\differential {\widetilde \beta} \\
    = - \int_0^\beta \frac{2\alpha-1}{
      1 - 2\gamma(\alpha,{\widetilde \beta}) (1-e^{-2{\widetilde \beta}})
    } \,\differential {\widetilde \beta}\,.
  \]
  Note the denominator $1 - 2\gamma(\alpha,{\widetilde \beta}) (1 - e^{-2{\widetilde \beta}})$ is always positive because $\gamma(\alpha,{\widetilde \beta}) \in [0,\min(\alpha, 1-\alpha)] \subseteq [0,\frac12]$.
  It follows that $\partial_\alpha \Psi^\beta(\alpha) < 0$ and \eqref{e:1mt-Phi-concave-goal} is proved.
  
  \textit{Case 2: $\beta \in [2\beta_\lowerBoundTag,0)$}.
  Then,
  \[
    \partial_\alpha \Psi^\beta(\alpha)
    = \partial_\alpha \Psi(\alpha,0) - \int_\beta^0 \partial_\alpha \partial_\beta \Psi(\alpha,{\widetilde \beta}) \,\differential {\widetilde \beta}
    = \int_\beta^0 \frac{2\alpha-1}{
      1 - 2\gamma(\alpha,{\widetilde \beta}) (1-e^{-2{\widetilde \beta}})
    } \,\differential {\widetilde \beta}\,.
  \]
  Recall that $|2\beta_\lowerBoundTag| = 2d^{-10} \le 2\cdot 3^{-10}$.
  For ${\widetilde \beta} \in [2\beta_\lowerBoundTag,0)$ we then have
  \[
    |2\gamma(\alpha,{\widetilde \beta}) (1 - e^{-2{\widetilde \beta}})| 
    \le 2 \cdot \frac12 \cdot 4|{\widetilde \beta}|
    \le 8 \cdot 3^{-10}\,.
  \]
  So, $1 - 2\gamma(\alpha,{\widetilde \beta}) (1 - e^{-2{\widetilde \beta}}) \ge \frac12$.
  It follows that
  \[
    \partial_\alpha \Psi^\beta(\alpha)
    \le \int_\beta^0 2(2\alpha-1) \,\differential {\widetilde \beta}
    = 2|\beta| (2\alpha-1)
    \le 4\cdot d^{-10} (2\alpha-1)
    \le 4\cdot 3^{-9} \cdot \frac{2\alpha-1}{d}
  \]
  and \eqref{e:1mt-Phi-concave-goal} is proved.
\end{proof}
\end{lemma}

\begin{proof}[Proof of~\cref{l:1mt-balanced}]
  The conclusions \eqref{e:1mt-variational-inside} and \eqref{e:1mt-variational-outside} of~\cref{l:1mt-exp-rate} can be rewritten as
  \begin{align*}
    \E[Z_{G,< \varepsilon}(\beta)]
    &= \exp\left(
      n \sup_{|2\alpha-1| < \varepsilon} \Phi^\beta(\alpha)
      + o(n)
    \right)\,, \\
    \E[Z_{G,\ge \varepsilon}(\beta)]
    &= \exp\left(
      n \sup_{|2\alpha-1| \ge \varepsilon} \Phi^\beta(\alpha)
      + o(n)
    \right)\,.
  \end{align*}
~\cref{l:1mt-Phi-concave} implies that for any $\beta \in [2\beta_\lowerBoundTag,+\infty)$, the function $[0,1] \ni \alpha \mapsto \Phi^\beta(\alpha)$ is maximized at $\alpha = \frac12$, with
  \[
    \sup_{|2\alpha-1| \ge \varepsilon} \Phi^\beta(\alpha)
    \le \Phi({\color{blue} 1/2},\beta) - \int_{1/2}^{(1+\varepsilon)/2} (2\alpha-1) \,\differential \alpha
    = \Phi({\color{blue} 1/2},\beta) - \frac{\varepsilon^2}{4}\,.
  \]
  This proves part \cref{i:1mt-balanced-deviation}.
  The same calculation with $s$ in place of $\varepsilon$ proves part \cref{i:1mt-balanced-subgaussian}.
\end{proof}

\subsection{Second and third moments}

In this subsection we prove ~\cref{l:basic-mt}.
We begin by deriving a formula for the (exponential rate of) the contribution to the $k$-th moment from any ``type'' (see~\cref{d:basic-mt-type}) which we will eventually apply with $k=2,3$.

For $k\ge 1$, recall $\Sigma_k = \{+1,-1\}^k$. 
For $\omega,\rho \in \Sigma_k$, let $\Delta(\omega,\rho)$ denote the Hamming distance between $\omega$ and $\rho$.
Let
\[
  {\mathcal{P}}_k = \left\{
    {\typeProportions} \in [0,1]^{\Sigma_k} : \|{\typeProportions}\|_1 = 1
  \right\}
\]
be the space of probability measures on the discrete set $\Sigma_k$.
For ${\typeProportions} \in {\mathcal{P}}_k$, let
\[
  {\mathcal{Q}}_k({\typeProportions}) = \left\{
    \Gamma = (\gamma_{\omega,\rho})_{\omega,\rho \in \Sigma_k} \in [0,1]^{\Sigma_k \times \Sigma_k} : 
    \text{$\Gamma$ is symmetric with row sums}
    \sum_{\rho\in \Sigma_k} \gamma_{\omega,\rho} = \alpha_\omega
  \right\}\,.
\]
Define
\[
  {\mathcal{L}}_k = \left\{
    ({\typeProportions}, \Gamma) : {\typeProportions} \in {\mathcal{P}}_k, \Gamma \in {\mathcal{Q}}_k({\typeProportions})
  \right\}\,.
\]
For $({\typeProportions}, \Gamma) \in {\mathcal{L}}_k$ let
\[
  \Phi_k^\beta({\typeProportions},\Gamma)
  = - (d-1) {\mathcal{H}}({\typeProportions}) + \frac{d}{2} {\mathcal{H}}(\Gamma) 
  + d\beta \sum_{(\omega,\rho) \in \binom{\Sigma_k}{2}} \gamma_{\omega,\rho} \Delta(\omega,\rho)\,.
\]
Here ${\mathcal{H}}(\Gamma)$ denotes the entropy of the discrete distribution on $\Sigma_k \times \Sigma_k$ with entries $\gamma_{\omega,\rho}$.
Note that in the case $k=1$, this reduces to $\Phi$ defined in \eqref{e:basic-mt-Phi} with the identification
\begin{align*}
  {\typeProportions} &= \begin{bmatrix}
    \alpha_+ \\ \alpha_-
  \end{bmatrix} = \begin{bmatrix}
    \alpha \\ 1-\alpha
  \end{bmatrix} \,, &
  \Gamma = \begin{bmatrix}
    \gamma_{+,+} & \gamma_{+,-} \\
    \gamma_{-,+} & \gamma_{-,-}
  \end{bmatrix} = \begin{bmatrix}
    \alpha - \gamma & \gamma \\
    \gamma & 1 - \alpha - \gamma
  \end{bmatrix}\,.
\end{align*}
Further define
\[
  {\mathcal{L}}_{k,n} = \left\{
    ({\typeProportions},\Gamma) \in {\mathcal{L}}_k : \text{${\typeProportions} n$ and $\Gamma dn$ have integer entries, and the diagonal entries of $\Gamma dn$ are even}
  \right\}\,.
\]
\begin{definition}
  \label{d:basic-mt-type}
  Let $\sigma^1,\ldots,\sigma^k \in \Sigma_n$ and $G$ be a $d$-regular graph on $[n]$.
  We say $(\sigma^1,\ldots,\sigma^k,G)$ has \textbf{type} $({\typeProportions},\Gamma)$, which we denote $(\sigma^1,\ldots,\sigma^k,G) \in {\mathcal{T}}_n({\typeProportions},\Gamma)$, if the following holds.
  \begin{enumerate}[label=(\roman*)]
    \item \label{i:basic-mt-type-spins} For each $\omega\in \Sigma_k$, there are exactly $\alpha_\omega n$ indices $i$ for which $(\sigma^1_i,\ldots,\sigma^k_i) = \omega$.
    Let $S(\omega)$ be the set of such $i$.
    \item \label{i:basic-mt-type-edges} For all distinct $\omega,\rho \in \Sigma_k$, $G$ has exactly $\gamma_{\omega,\rho} dn$ edges between $S(\omega)$ and $S(\rho)$.
    Consequently, for all $\omega \in \Sigma_k$, $G$ has $\gamma_{\omega,\omega} dn / 2$ edges with both endpoints in $S(\omega)$.
  \end{enumerate}
\end{definition}
It is simple to check that the set of $({\typeProportions},\Gamma)$ for which ${\mathcal{T}}_n({\typeProportions},\Gamma)$ is nonempty is precisely ${\mathcal{L}}_{k,n}$.
\begin{lemma}
  \label{l:kmt-exp-rate}
  For any $({\typeProportions},\Gamma) \in {\mathcal{L}}_{k,n}$ and $\beta \in \R$, the contribution to $\E[Z_{G}(\beta)^k]$ from $(\sigma^1,\ldots,\sigma^k,G)$ of type $({\typeProportions},\Gamma)$ is
  \begin{equation}\label{e:kmt-exp-rate}
    \sum_{\sigma^1,\ldots,\sigma^k \in \Sigma_n}
    \E\left[
      \allone\{(\sigma^1,\ldots,\sigma^k,G) \in {\mathcal{T}}_n({\typeProportions},\Gamma)\}
      e^{\beta \sum_{i=1}^k \isingHamiltonian_{G}(\sigma^i)}
    \right]
    = \exp\left(
      n \Phi_k^\beta({\typeProportions},\Gamma) + o(n)
    \right)\,.
  \end{equation}
\begin{proof}
  We claim the left-hand side of \eqref{e:kmt-exp-rate} is exactly
  \begin{align}\nonumber
    \binom{n}{{\typeProportions} n}
    &\times
    \frac{1}{(nd-1)!!}
    \times
    \prod_{\omega \in \Sigma_k}
    \binom{\alpha_\omega dn}{(\gamma_{\omega,\rho} dn : \rho \in \Sigma_k)}
    \times
    \prod_{(\omega,\rho) \in \binom{\Sigma_k}{2}} 
    (\gamma_{\omega,\rho} dn)! \\
    \label{e:kmt-exact}
    &\times
    \prod_{\omega \in \Sigma_k}
    (\gamma_{\omega,\omega} dn - 1)!!
    \times
    \prod_{(\omega,\rho) \in \binom{\Sigma_k}{2}} 
    e^{\beta \Delta(\omega,\rho) \cdot \gamma_{\omega,\rho} dn}
    \,.
  \end{align}
  This is proved by a similar counting argument as that of~\cref{l:1mt-exp-rate}.
  The number of $(\sigma^1,\ldots,\sigma^k)$ satisfying item \cref{i:basic-mt-type-spins} of~\cref{d:basic-mt-type} is precisely $\binom{n}{{\typeProportions} n}$.
  Then, out of the $(nd-1)!!$ possible realizations of $G$ from the configuration model, each given by a perfect matching of $nd$ half-edges, the number satisfying item \cref{i:basic-mt-type-edges} can be counted as follows.
  For each $\omega \in \Sigma_k$, out of the $\alpha_\omega dn$ half-edges emanating from $S(\omega)$, we choose $\gamma_{\omega,\rho} dn$ to match to half-edges emanating from $S(\rho)$ for each $\rho \in \Sigma_k$.
  This can be done in
  \[
    \binom{\alpha_\omega dn}{(\gamma_{\omega,\rho} dn : \rho \in \Sigma_k)}
  \]
  ways.
  Then, for distinct $\omega,\rho \in \Sigma_k$, the edges between $S(\omega)$ and $S(\rho)$ can be formed in $(\gamma_{\omega,\rho} dn)!$ ways, and the edges within $S(\omega)$ can be formed in $(\gamma_{\omega,\omega} dn - 1)!!$ ways.
  Finally the total energy of these configurations is
  \[
    \sum_{i=1}^k \isingHamiltonian_{G}(\sigma^i)
    = \sum_{(\omega,\rho) \in \binom{\Sigma_k}{2}} \Delta(\omega,\rho) \cdot \gamma_{\omega,\rho} dn\,.
  \]
  This proves \eqref{e:kmt-exact}.
  The conclusion \eqref{e:kmt-exp-rate} then follows by Stirling's formula.
\end{proof}
\end{lemma}
Let ${\underline{1}} \in \R^{\Sigma_k}$ be the all-$1$ vector.
We introduce a few notations for the following lemmas:
\begin{align*}
  {\mathcal{Q}}_k^+ &= \{\Gamma \in [0,+\infty)^{\Sigma_k\times \Sigma_k}\,\text{symmetric}\}\,, \\
  {\mathcal{P}}_k^\circ &= \{{\typeProportions} \in {\mathcal{P}}_k : \min_\omega \alpha_\omega > 0\}\,, \\
  {\mathcal{U}}_k &= \{{\underline{u}} \in \R^{\Sigma_k} : ({\underline{u}},{\underline{1}}) = 0\,\,\text{and}\,\,\|u\| = 1\}\,.
\end{align*}
\begin{lemma}\label{l:kmt-lagrange}
  For any $\Gamma \in {\mathcal{Q}}_k^+$, $\beta \in \R$ define
  \[
    \Lambda^\beta(\Gamma)
    = - \frac12 \sum_{\omega,\rho \in \Sigma_k} \gamma_{\omega,\rho} \log \gamma_{\omega,\rho}
    + \beta \sum_{\omega,\rho \in \binom{\Sigma_k}{2}} \gamma_{\omega,\rho} \Delta(\omega,\rho)\,.
  \]
  Then, for any ${\typeProportions} \in {\mathcal{P}}_k^\circ$, $\beta \in \R$,
  \begin{equation}\label{e:kmt-lagrange}
    \sup_{\Gamma \in {\mathcal{Q}}_k({\typeProportions})} \Lambda^\beta(\Gamma)
    = \inf_{{\underline{\lambda}} \in \R^{\Sigma_k}}
    \sup_{\Gamma \in {\mathcal{Q}}_k^+} \Lambda^\beta(\Gamma) 
    + \sum_{\omega \in \Sigma_k} \lambda_\omega \left(\sum_{\rho \in \Sigma_k} \gamma_{\omega,\rho} - \alpha_\omega \right)\,.
  \end{equation}
  The inf-sup on the right-hand side is attained at a unique stationary point $({\underline{\lambda}},\Gamma) = ({\underline{\lambda}}^\beta({\typeProportions}), \Gamma^\beta({\typeProportions}))$, which satisfies
  \begin{equation}\label{e:kmt-stationarity-lambda}
    \gamma_{\omega,\rho} = e^{-1 + \lambda_\omega + \lambda_\rho + \beta \Delta(\omega,\rho)} \qquad\qquad
    \text{for all $\omega,\rho \in \Sigma_k$}\,,
  \end{equation}
  and 
  \begin{equation}
    \label{e:kmt-stationarity-rows}
    \sum_{\rho \in \Sigma_k} \gamma_{\omega,\rho} = \alpha_\omega \qquad\qquad
    \text{for all $\omega \in \Sigma_k$}\,.
  \end{equation}
  Finally, ${\underline{\lambda}}^\beta$, $\Gamma^\beta$ are continuously differentiable functions of ${\typeProportions} \in {\mathcal{P}}_k^\circ$.
\begin{proof}
  The function $\Lambda^\beta$ is strictly concave on ${\mathcal{Q}}_k^+$, with the convention $0\log 0 = 0$.
  Since ${\typeProportions} \in {\mathcal{P}}_k^\circ$, the affine constraint $\Gamma \in {\mathcal{Q}}_k({\typeProportions})$ has a strictly feasible point, for instance $\gamma_{\omega,\rho} = \alpha_\omega \alpha_\rho$.
  Hence the Karush--Kuhn--Tucker theorem (equivalently, strong Lagrange duality for a concave maximization problem with affine constraints and Slater's condition; see \cite[Theorem~28.2]{Roc70}) gives \eqref{e:kmt-lagrange}, and the primal optimizer is unique by strict concavity.

  We next identify the saddle point, prove its uniqueness, and prove its regularity in ${\typeProportions}$.
  For fixed ${\underline{\lambda}}$, the supremum over $\Gamma \in {\mathcal{Q}}_k^+$ in \eqref{e:kmt-lagrange} is attained at the unique point
  \[
    \gamma_{\omega,\rho}({\underline{\lambda}})
    = e^{-1 + \lambda_\omega + \lambda_\rho + \beta \Delta(\omega,\rho)}\,,
  \]
  which is strictly positive.
  Thus the dual problem is the minimization of the smooth function
  \[
    \frac12 \sum_{\omega,\rho \in \Sigma_k}
    e^{-1 + \lambda_\omega + \lambda_\rho + \beta \Delta(\omega,\rho)}
    - \sum_{\omega \in \Sigma_k} \lambda_\omega \alpha_\omega
  \]
  over ${\underline{\lambda}} \in \R^{\Sigma_k}$.
  Its Hessian at ${\underline{\lambda}}$ is
  \[
    \left(
      \allone\{\omega=\rho\}\sum_{\tau \in \Sigma_k}\gamma_{\omega,\tau}({\underline{\lambda}})
      + \gamma_{\omega,\rho}({\underline{\lambda}})
    \right)_{\omega,\rho \in \Sigma_k}\,,
  \]
  and for any nonzero ${\underline{v}} \in \R^{\Sigma_k}$ the associated quadratic form is
  \[
    \frac12 \sum_{\omega,\rho \in \Sigma_k}
    \gamma_{\omega,\rho}({\underline{\lambda}}) (v_\omega + v_\rho)^2 > 0\,.
  \]
  Therefore the dual objective is strictly convex, so the minimizing ${\underline{\lambda}}^\beta({\typeProportions})$ is unique.
  Its first-order conditions are precisely the row-sum equations \eqref{e:kmt-stationarity-rows}, and the corresponding unique $\Gamma^\beta({\typeProportions})$ is given by \eqref{e:kmt-stationarity-lambda}.
  Since the Hessian above is nonsingular at the minimizer, the implicit function theorem applied to these first-order conditions implies that ${\underline{\lambda}}^\beta({\typeProportions})$, and hence $\Gamma^\beta({\typeProportions})$, depend continuously differentiably on ${\typeProportions} \in {\mathcal{P}}_k^\circ$.
\end{proof}
\end{lemma}
\begin{lemma}\label{l:kmt-lambda-deriv}
  Suppose ${\typeProportions} \in {\mathcal{P}}_k^\circ$ and ${\underline{u}} \in {\mathcal{U}}_k$.
  Let $D({\typeProportions}) \in \R^{\Sigma_k\times \Sigma_k}$ be the diagonal matrix with entries $D({\typeProportions})_{\omega,\omega} = \alpha_\omega$ and
  \[
    Q^\beta({\typeProportions}) = D({\typeProportions}) + \Gamma^\beta({\typeProportions})\,.
  \]
  Then for all $\beta \in \R$, $Q^\beta({\typeProportions}) \succ 0$.
  In particular $Q^\beta({\typeProportions})$ is invertible and
  \[
    \frac{\differential}{\differential s} {\underline{\lambda}}^\beta({\typeProportions} + s{\underline{u}}) \Big|_{s=0} = Q^\beta({\typeProportions})^{-1} {\underline{u}}\,.
  \]
\begin{proof}
  Note that $D({\typeProportions}) - \Gamma^\beta({\typeProportions}) \succeq 0$ because for any ${\underline{v}} \in \R^{\Sigma_k}$,
  \[
    \langle {\underline{v}}, (D({\typeProportions}) - \Gamma^\beta({\typeProportions})) {\underline{v}}\rangle
    = \sum_{(\omega,\rho) \in \binom{\Sigma_k}{2}} \gamma_{\omega,\rho}^\beta({\typeProportions}) (v_\omega - v_\rho)^2 \ge 0\,.
  \]
  Let ${\widetilde Q}^\beta({\typeProportions})$ be the matrix whose entries are the absolute values of those of $D({\typeProportions}) - \Gamma^\beta({\typeProportions})$. 
  That is for distinct $\omega,\rho \in \Sigma_k$,
  \begin{align*}
    {\widetilde Q}^\beta_{\omega,\rho}({\typeProportions}) &= \gamma_{\omega,\rho}^\beta({\typeProportions}), &
    {\widetilde Q}^\beta_{\omega,\omega}({\typeProportions}) &= \alpha_\omega - \gamma_{\omega,\omega}^\beta({\typeProportions})\,.
  \end{align*}
  Then, \cite[Lemma 2.6]{HS26} implies ${\widetilde Q}^\beta({\typeProportions}) \succeq 0$.
  ~\cref{l:kmt-lagrange} implies each $\gamma_{\omega,\omega}^\beta({\typeProportions})$ is strictly positive, so
  \[
    Q^\beta_{\omega,\omega}({\typeProportions}) = \alpha_\omega + \gamma_{\omega,\omega}^\beta({\typeProportions}) > {\widetilde Q}^\beta_{\omega,\omega}({\typeProportions})\,.
  \]
  Thus $Q^\beta({\typeProportions}) \succ {\widetilde Q}^\beta({\typeProportions}) \succeq 0$, proving the first conclusion.
  We will prove the second conclusion using the implicit function theorem.
  Let ${\typeProportions}(s) = {\typeProportions} + s{\underline{u}}$ and
  \[
    {\underline{w}} = ({\underline{\lambda}}^\beta \circ {\typeProportions})'(0)
    = \frac{\differential}{\differential s} {\underline{\lambda}}^\beta({\typeProportions}(s)) \big|_{s=0}\,.
  \]
  Equations \eqref{e:kmt-stationarity-lambda}--\eqref{e:kmt-stationarity-rows}, applied to ${\typeProportions}(s)$, imply
  \[
    \alpha_\omega(s)
    = \sum_{\rho \in \Sigma_k} e^{-1 + \lambda^\beta_\omega({\typeProportions}(s)) + \lambda^\beta_\rho({\typeProportions}(s)) + \beta \Delta(\omega,\rho)}
  \]
  for all $\omega \in \Sigma_k$.
  Differentiating in $s$ and evaluating at $s=0$ yields
  \begin{align*}
    u_\omega
    &= \sum_{\rho \in \Sigma_k} e^{-1 + \lambda^\beta_\omega({\typeProportions}(0)) + \lambda^\beta_\rho({\typeProportions}(0)) + \beta \Delta(\omega,\rho)} ((\lambda^\beta_\omega \circ {\typeProportions})'(0) + (\lambda^\beta_\rho \circ {\typeProportions})'(0)) \\
    &= \sum_{\rho \in \Sigma_k} \gamma_{\omega,\rho}^\beta({\typeProportions}) (w_\omega + w_\rho)
    = \alpha_\omega w_\omega + \sum_{\rho \in \Sigma_k} \gamma_{\omega,\rho}^\beta({\typeProportions}) w_\rho
    = \left[Q^\beta({\typeProportions}) {\underline{w}}\right]_\omega\,.
  \end{align*}
  Thus ${\underline{u}} = Q^\beta({\typeProportions}) {\underline{w}}$, and by invertibility of $Q^\beta({\typeProportions})$ the result follows.
\end{proof}
\end{lemma}
\begin{lemma}\label{l:kmt-gradient-hessian}
  Suppose ${\typeProportions} \in {\mathcal{P}}_k^\circ$ and ${\underline{u}} \in {\mathcal{U}}_k$.
  Let
  \[
    \Phi_k^\beta({\typeProportions})
    = \sup_{\Gamma \in {\mathcal{Q}}_k({\typeProportions})} \Phi_k^\beta({\typeProportions},\Gamma)
    = - (d-1) {\mathcal{H}}({\typeProportions}) + d \sup_{\Gamma \in {\mathcal{Q}}_k({\typeProportions})} \Lambda^\beta(\Gamma)\,.
  \]
  Then,
  \begin{equation}\label{e:kmt-gradient}
    \frac{\differential}{\differential s} \Phi_k^\beta({\typeProportions} + s{\underline{u}}) \Big|_{s=0}
    = (d-1) \sum_{\omega \in \Sigma_k} u_\omega \log \alpha_\omega
    - d \sum_{\omega \in \Sigma_k} u_\omega \lambda_\omega({\typeProportions})
  \end{equation}
  and
  \begin{equation}\label{e:kmt-hessian}
    \frac{\differential^2}{\differential^2 s} \Phi_k^\beta({\typeProportions} + s{\underline{u}}) \Big|_{s=0}
    = \langle {\underline{u}}, M^\beta({\typeProportions}) {\underline{u}}\rangle \quad \text{where} \quad
    M^\beta({\typeProportions}) = (d-1) D({\typeProportions})^{-1} - d Q^\beta({\typeProportions})^{-1}\,.
  \end{equation}
\begin{proof}
  Using the representation \eqref{e:kmt-lagrange},
  \[
    \Phi_k^\beta({\typeProportions})
    = - (d-1) {\mathcal{H}}({\typeProportions}) 
    + d \Lambda^\beta(\Gamma^\beta({\typeProportions})) 
    + d \sum_{\omega \in \Sigma_k} \lambda^\beta_\omega({\typeProportions}) \left(\sum_{\rho \in \Sigma_k} \gamma^\beta_{\omega,\rho}({\typeProportions}) - \alpha_\omega \right)\,.
  \]
  Let ${\typeProportions}(s) = {\typeProportions} + s{\underline{u}}$ and $f(s) = \Phi_k^\beta({\typeProportions}(s))$.
  Differentiating in $s$ yields
  \begin{align*}
    f'(s)
    &= (d-1) \sum_{\omega \in \Sigma_k} u_\omega (1 + \log \alpha_\omega(s))
    - d \sum_{\omega \in \Sigma_k} u_\omega \lambda_\omega({\typeProportions}(s)) \\
    &= (d-1) \sum_{\omega \in \Sigma_k} u_\omega \log \alpha_\omega(s)
    - d \sum_{\omega \in \Sigma_k} u_\omega \lambda_\omega({\typeProportions}(s))\,.
  \end{align*}
  Note that the derivatives through $\lambda^\beta_\omega({\typeProportions}(s))$ and $\gamma^\beta_{\omega,\rho}({\typeProportions}(s))$ vanish by stationarity of $({\underline{\lambda}}^\beta({\typeProportions}), \Gamma^\beta({\typeProportions}))$.
  Setting $s=0$ in the last display proves \eqref{e:kmt-gradient}.
  Differentiating again using ~\cref{l:kmt-lambda-deriv} and evaluating at $s=0$ shows
  \[
    f''(0)
    = (d-1) \sum_{\omega \in \Sigma_k} \frac{(u_\omega)^2}{\alpha_\omega}
    - d \sum_{\omega \in \Sigma_k} u_\omega \left[Q^\beta({\typeProportions})^{-1} {\underline{u}}\right]_\omega
    = \langle {\underline{u}}, M^\beta({\typeProportions}) {\underline{u}}\rangle\,,
  \]
  proving \eqref{e:kmt-hessian}.
\end{proof}
\end{lemma}
Henceforth, let ${\typeProportions}_* = 2^{-k} {\underline{1}} \in {\mathcal{P}}_k^\circ$.
\begin{lemma}\label{l:kmt-stationarity}
  For all $k\ge 1$, $\beta \in \R$, and ${\underline{u}} \in {\mathcal{U}}_k$, 
  \[
    \frac{\differential}{\differential s} \Phi_k^\beta({\typeProportions}_* + s{\underline{u}}) \Big|_{s=0} = 0\,.
  \]
\begin{proof}
  It is easy to check that both ${\typeProportions}_*$ and ${\underline{\lambda}}^\beta({\typeProportions}_*)$ are multiples of ${\underline{1}}$.
  The result is immediate from \eqref{e:kmt-gradient}.
\end{proof}
\end{lemma}
\begin{lemma}\label{l:kmt-concavity}
  For each $\delta > 0$, $k\ge 1$ there exists $\iota = \iota(d,\delta,k) > 0$ such that for all $\beta \in [2\beta_\lowerBoundTag,\betaKS - \delta]$,
  \[
    M^\beta({\typeProportions}_*) - d {\underline{1}} {\underline{1}}^\matrixtranspose \preceq - \iota I\,.
  \]
\begin{proof}
  By \eqref{e:kmt-stationarity-lambda}--\eqref{e:kmt-stationarity-rows}, $\Gamma^\beta({\typeProportions}_*)$ has entries
  \[
    \gamma_{\omega,\rho}^\beta({\typeProportions}_*)
    = \frac{e^{\beta \Delta(\omega,\rho)}}{2^k (1+e^\beta)^k}\,.
  \]
  If we regard $\R^{\Sigma_k}$ as the tensor product $(\R^2)^{\otimes k}$, then $\Gamma^\beta({\typeProportions}_*)$ also tensorizes as
  \[
    \Gamma^\beta({\typeProportions}_*)
    = \frac{1}{2^k} \Xi^{\otimes k}\,, \qquad \text{where} \qquad
    \Xi = \frac{1}{1+e^\beta}
    \begin{bmatrix}
      1 & e^\beta \\
      e^\beta & 1
    \end{bmatrix}\,.
  \]
  Let $\hat{{\underline{1}}} = {\underline{1}} / \|{\underline{1}}\| = {\underline{1}} / 2^{k/2}$.
  Since $D({\typeProportions}_*) = \frac{1}{2^k} I$, a simple calculation shows the matrix we are interested in is
  \begin{equation}\label{e:kmt-concavity-tilde-M}
    M^\beta({\typeProportions}_*) - d {\underline{1}} {\underline{1}}^\matrixtranspose
    = 2^k \left[
      (d-1) I - d \left(
        I + \Xi^{\otimes k}
      \right)^{-1} - d\hat{{\underline{1}}} \hat{{\underline{1}}}^\matrixtranspose
    \right]
    \equiv 2^k {\widetilde M}\,.
  \end{equation}
  Note that $\Xi$ has eigenvectors $(1,1)$ and $(1,-1)$, with eigenvalues $1$ and 
  \[
    \frac{1-e^\beta}{1+e^\beta} = - \tanh(\beta/2)\,.
  \]
  Thus the distinct eigenvalues of $\Xi^{\otimes k}$ are of the form 
  \[
    a_j = (-\tanh(\beta/2))^j
  \]
  for $0\le j\le k$.
  The eigenvalue $a_0 = 1$ has multiplicity one, with corresponding eigenvector $\hat{{\underline{1}}}$.
  So, if ${\underline{u}} = \hat{{\underline{1}}}$, then
  \[
    \langle {\widetilde M} {\underline{u}}, {\underline{u}}\rangle 
    = d-1 - d (1+1)^{-1} - d = -\frac{d}{2} - 1\,.
  \]
  We will next argue that there exists $\iota_1 = \iota_1(d,\delta) > 0$ such that $\beta \in [2\beta_\lowerBoundTag,\betaKS - \delta]$, the remaining eigenvalues satisfy
  \begin{equation}\label{e:kmt-concavity-eigenvalues}
    \max_{1\le j\le k} a_j \le \frac{1}{d-1} - \iota_1\,.
  \end{equation}
  Indeed, if $\beta \in [0,\betaKS - \delta]$ and $k\ge 2$, then
  \[
    \max_{1\le j\le k} a_j = a_2 = \tanh^2(\beta/2)
    \le \tanh^2(\betaKS/2) - \iota_1
    = \frac{1}{d-1} - \iota_1\,,
  \]
  while if $\beta \in [2\beta_\lowerBoundTag,0]$ or $k=1$, then
  \[
    \max_{1\le j\le k} a_j = a_1 = -\tanh(\beta/2)
    \le \tanh(|\beta_\lowerBoundTag|) \le |\beta_\lowerBoundTag| = d^{-10} \le \frac{1}{d-1} - \iota_1\,.
  \]
  This proves \eqref{e:kmt-concavity-eigenvalues}. Finally, if ${\underline{u}}$ is a (unit) eigenvector corresponding to $a_j$, then for suitable $\iota_2 = \iota_2(d,\delta) > 0$,
  \[
    \langle {\widetilde M} {\underline{u}}, {\underline{u}}\rangle
    = d-1 - d(1 + a_j)^{-1}
    \le d-1 - d\left(1 + \frac{1}{d-1} - \iota_1\right)^{-1}
    \le -\iota_2\,.
  \]
  Together with \eqref{e:kmt-concavity-tilde-M} this implies the conclusion.
\end{proof}
\end{lemma}

\begin{lemma}\label{l:kmt-lambda-min-Q}
  For all $k\ge 1$ and all $\beta \in [2\beta_\lowerBoundTag,\betaKS]$, we have
  \[
    \lambda_{\min}(Q^\beta({\typeProportions}_*)) \ge 2^{-(k+2)} \equiv \lambda_{\lowerBoundTag,k}\,.
  \]
\begin{proof}
  As calculated in the proof of ~\cref{l:kmt-concavity}, for $\Xi$ defined therein,
  \[
    Q^\beta({\typeProportions}_*) = 2^{-k} \left(I + \Xi^{\otimes k}\right)
  \]
  Since $\Xi^{\otimes k}$ has eigenvalues $(-\tanh(\beta/2))^j$ for $0\le j\le k$, we have
  \[
    \lambda_{\min}(Q^\beta({\typeProportions}_*)) 
    \ge 2^{-k} (1 - \tanh(\betaKS/2)) 
    = 2^{-k} \left(1 - \frac{1}{\sqrt{d-1}} \right)
    \ge 2^{-(k+2)}\,,
  \]
  where in the last inequality we used that $d\ge 3$.
\end{proof}
\end{lemma}

\begin{lemma}\label{l:kmt-Q-lipschitz}
  For each $k\ge 1$, there exists $s_\upperBoundTag = s_\upperBoundTag(d,k) > 0$ such that for all $\beta \in [2\beta_\lowerBoundTag,\betaKS]$, the following holds.
  For all ${\typeProportions} \in {\mathcal{P}}_k^\circ$ such that $\|{\typeProportions} - {\typeProportions}_*\| \le s_\upperBoundTag$,
  \[
    \|Q^\beta({\typeProportions}) - Q^\beta({\typeProportions}_*)\|_{\mathsf{op}} \le \frac{\|{\typeProportions} - {\typeProportions}_*\|}{s_\upperBoundTag}\,.
  \]
\begin{proof}
  Any such ${\typeProportions}$ can be expressed in the form ${\typeProportions}(s) = {\typeProportions}_* + s{\underline{u}}$ for some ${\underline{u}} \in {\mathcal{U}}_k$.
  So it suffices to prove that for all ${\underline{u}} \in {\mathcal{U}}_k$ and $s \in [0,s_\upperBoundTag]$,
  \[
    \|Q^\beta({\typeProportions}(s)) - Q^\beta({\typeProportions}(0))\|_{\mathsf{op}} \le s / s_\upperBoundTag\,.
  \]
  Write
  \[
    L(s) = \|Q^\beta({\typeProportions}(s)) - Q^\beta({\typeProportions}(0))\|_{\mathsf{op}}\,.
  \]
  We estimate
  \begin{align}\nonumber
    L(s,{\underline{u}})
    &\le \|D({\typeProportions}(s)) - D({\typeProportions}(0))\|_{\mathsf{op}}
    + \|\Gamma^\beta({\typeProportions}(s)) - \Gamma^\beta({\typeProportions}(0))\|_{\mathsf{op}} \\
    \label{e:kmt-L-bd1}
    &\le s + \int_0^s \|(\Gamma^\beta \circ {\typeProportions})'(r)\|_{\mathsf{op}}\,\differential r\,.
  \end{align}
  Since the entries of $\Gamma^\beta$ have the form \eqref{e:kmt-stationarity-lambda}, we have
  \[
    (\Gamma^\beta \circ {\typeProportions})'(r)
    = D(({\underline{\lambda}}^\beta \circ {\typeProportions})'(r)) \Gamma^\beta({\typeProportions}(r))
    + \Gamma^\beta({\typeProportions}(r)) D(({\underline{\lambda}}^\beta \circ {\typeProportions})'(r))\,,
  \]
  where $D(({\underline{\lambda}}^\beta \circ {\typeProportions})'(r))$ is the diagonal matrix with $(\omega,\omega)$ entry $({\underline{\lambda}}^\beta_\omega \circ {\typeProportions})'(r)$.
  Since $\Gamma^\beta({\typeProportions}(r))$ is symmetric, nonnegative, and has row sums bounded by $1$, we have $\|\Gamma^\beta({\typeProportions}(r))\|_{\mathsf{op}} = 1$.
  Thus
  \begin{align*}
    \|(\Gamma^\beta \circ {\typeProportions})'(r)\|_{\mathsf{op}}
    &\le 2\|D(({\underline{\lambda}}^\beta \circ {\typeProportions})'(r))\|_{\mathsf{op}}
    = 2\|D(Q^\beta({\typeProportions}(r))^{-1} {\underline{u}})\|_{\mathsf{op}} \\
    &\le 2 \|Q^\beta({\typeProportions}(r))^{-1} {\underline{u}}\|
    \le 2\|Q^\beta({\typeProportions}(r))^{-1}\|_{{\mathsf{op}}}\,.
  \end{align*}
  Plugging back into \eqref{e:kmt-L-bd1} yields
  \begin{equation}\label{e:kmt-L-bd2}
    L(s,{\underline{u}}) \le s + 2\int_0^s \|Q^\beta({\typeProportions}(r))^{-1}\|_{{\mathsf{op}}} \,\differential r\,.
  \end{equation}
  Let $\lambda_\lowerBoundTag = \lambda_\lowerBoundTag(k)$ be defined in ~\cref{l:kmt-lambda-min-Q} and
  \[
    s_+({\underline{u}}) = \inf\left\{
      r \ge 0 : L(r,{\underline{u}}) \ge \frac12 \lambda_\lowerBoundTag
    \right\}\,,
  \]
  so that $L(r,{\underline{u}}) < \frac12 \lambda_\lowerBoundTag$ for all $r < s_+({\underline{u}})$.
  For such $r$, ~\cref{l:kmt-lambda-min-Q} implies
  \[
    \lambda_{\min}(Q^\beta({\typeProportions}(r)))
    \ge \lambda_{\min}(Q^\beta({\typeProportions}(0))) - L(r,{\underline{u}})
    \ge \lambda_\lowerBoundTag / 2\,.
  \]
  Plugging into \eqref{e:kmt-L-bd2} shows that for all $s\le s_+({\underline{u}})$,
  \begin{equation}\label{e:kmt-L-bd3}
    L(s,{\underline{u}})
    \le s + 2\int_0^s (\lambda_\lowerBoundTag / 2)^{-1} \,\differential r
    = \frac{\lambda_\lowerBoundTag + 4}{\lambda_\lowerBoundTag} \cdot s\,.
  \end{equation}
  On the other hand, $s \mapsto L(s,{\underline{u}})$ is continuous.
  We conclude that
  \[
    \frac12 \lambda_\lowerBoundTag \le L(s_+({\underline{u}}),{\underline{u}}) \le \frac{\lambda_\lowerBoundTag + 4}{\lambda_\lowerBoundTag} \cdot s_+({\underline{u}})\,.
  \]
  Therefore
  \[
    s_+({\underline{u}}) \ge \frac12 \cdot \frac{(\lambda_\lowerBoundTag)^2}{\lambda_\lowerBoundTag + 4}\,.
  \]
  Let us now take
  \[
    s_\upperBoundTag = \min\left(
      \frac12 \cdot \frac{(\lambda_\lowerBoundTag)^2}{\lambda_\lowerBoundTag + 4}, \frac{\lambda_\lowerBoundTag}{\lambda_\lowerBoundTag + 4}
    \right)\,.
  \]
  Then $s_+({\underline{u}}) \ge s_\upperBoundTag$.
  So, for all $s\le s_\upperBoundTag$, \eqref{e:kmt-L-bd3} (which applies because $s\le s_+({\underline{u}})$) implies $L(s,{\underline{u}}) \le s / s_\upperBoundTag$.
  As this choice of $s_\upperBoundTag$ works for all ${\underline{u}} \in {\mathcal{U}}_k$, the result follows.
\end{proof}
\end{lemma}
\begin{proposition}\label{p:kmt-concavity-in-region}
  For each $\delta > 0$, $k\ge 1$ there exists $\iota = \iota(d,\delta,k) > 0$ and $s_0 = s_0(d,\delta,k) > 0$ such that for all $\beta \in [2\beta_\lowerBoundTag,\betaKS - \delta]$, the following holds.
  For all ${\typeProportions} \in {\mathcal{P}}_k^\circ$ with $\|{\typeProportions} - {\typeProportions}_*\| \le s_0$,
  \[
    M^\beta({\typeProportions}) - d {\underline{1}} {\underline{1}}^\matrixtranspose \preceq - \iota I\,.
  \]
\begin{proof}
  Let $\iota^{\cref{l:kmt-concavity}}(d,\delta,k)$ be given by ~\cref{l:kmt-concavity}, and set $\iota = \iota^{\cref{l:kmt-concavity}}(d,\delta,k) / 2$.
  Let $\lambda_\lowerBoundTag = \lambda_{\lowerBoundTag,k}$ be given by ~\cref{l:kmt-lambda-min-Q} and $s_\upperBoundTag = s_\upperBoundTag(d,k)$ be given by ~\cref{l:kmt-Q-lipschitz}.
  Then set
  \[
    s_0 = \min\left(
      2^{-(k+1)}, 
      \frac{\iota}{2^{2k+2}(d-1)},
      \frac12 \lambda_\lowerBoundTag s_\upperBoundTag, 
      \frac{\iota (\lambda_\lowerBoundTag)^2 s_\upperBoundTag}{4d}
    \right)\,,
  \]
  ~\cref{l:kmt-concavity} ensures
  \[
    M^\beta({\typeProportions}_*) - d {\underline{1}} {\underline{1}}^\matrixtranspose \preceq - 2\iota I\,.
  \]
  We will show that for all ${\typeProportions} \in {\mathcal{P}}_k^\circ$ with $\|{\typeProportions} - {\typeProportions}_*\| \le s_0$,
  \[
    \|M^\beta({\typeProportions}) - M^\beta({\typeProportions}_*)\|_{\mathsf{op}} \le \iota\,,
  \]
  which clearly implies the result.
  We estimate
  \[
    \|M^\beta({\typeProportions}) - M^\beta({\typeProportions}_*)\|_{\mathsf{op}}
    \le (d-1) \|D({\typeProportions})^{-1} - D({\typeProportions}_*)^{-1}\|_{\mathsf{op}} 
    + d \|Q^\beta({\typeProportions})^{-1} - Q^\beta({\typeProportions}_*)^{-1}\|_{\mathsf{op}}\,.
  \]
  Let $s = \|{\typeProportions} - {\typeProportions}_*\|$. Since $s\le s_0$,
  \begin{align*}
    (d-1) \|D({\typeProportions})^{-1} - D({\typeProportions}_*)^{-1}\|_{\mathsf{op}} 
    &\le (d-1) \left[\frac{1}{2^{-k} - s} - \frac{1}{2^{-k}}\right] \\
    &\le \frac{(d-1) s}{(2^{-k} - s) 2^{-k}}
    \le \frac{(d-1) s}{2^{-(k+1)} \cdot 2^{-k}}
    \le \iota/2\,.
  \end{align*}
  Also, using Lemmas~\cref{l:kmt-lambda-min-Q} and \cref{l:kmt-Q-lipschitz},
  \begin{align*}
    d \|Q^\beta({\typeProportions})^{-1} - Q^\beta({\typeProportions}_*)^{-1}\|_{\mathsf{op}}
    &\le d \|Q^\beta({\typeProportions})^{-1}\|_{\mathsf{op}} \|Q^\beta({\typeProportions}) - Q^\beta({\typeProportions}_*)\|_{\mathsf{op}} \|Q^\beta({\typeProportions}_*)^{-1}\|_{\mathsf{op}} \\
    &\le d \left(\lambda_\lowerBoundTag - \frac{s}{s_\upperBoundTag}\right)^{-1} (s/s_\upperBoundTag) (\lambda_\lowerBoundTag)^{-1} \\
    &\le d (\lambda_\lowerBoundTag/2)^{-1} (s/s_\upperBoundTag) (\lambda_\lowerBoundTag)^{-1}
    \le \iota/2\,. \qedhere
  \end{align*}
\end{proof}
\end{proposition}
For the proof of ~\cref{l:kmt-k-3} below we need the following fact, which is proved by direct algebra.
\begin{fact}\label{f:quadratic-system}
  For $A,B,u,v > 0$, consider the system
  \begin{align*}
    Ax^2 + Bxy &= u\,, & Ay^2 + Bxy &= v\,.
  \end{align*}
  This has a unique nonnegative solution $(x,y)$, which is furthermore positive and satisfies
  \[
    \frac{x}{y} = \frac{\vartheta t + \sqrt{1 + (\vartheta^2 - 1)t^2}}{1-t}\,,
  \]
  where $\vartheta = B/A$ and $t = \frac{u-v}{u+v}$.
\begin{proof}
  If $A=B$, dividing the given equations gives
  \[
    \frac{x}{y} = \frac{u}{v} = \frac{1+t}{1-t} = \frac{\vartheta t + \sqrt{1 + (\vartheta^2 - 1)t^2}}{1-t}\,,
  \]
  as desired. Otherwise assume $A\neq B$.
  Write $s = x+y$, $r = x-y$. Adding and subtracting the given equations yields
  \begin{align*}
    Asr &= u-v\,, & (A+B) s^2 + (A-B) r^2 = 2(u+v)\,.
  \end{align*}
  The first equation gives $r = \frac{u-v}{As}$, and plugging into the second gives the quadratic equation in $s^2$
  \[
    A(A+B) s^4 - 2A(u+v) s^2 + \frac{A-B}{A} (u-v)^2 = 0\,.
  \]
  This has solution
  \begin{align*}
    s^2 &= \frac{A(u+v) \pm \sqrt{A^2 (u+v)^2 + (B^2-A^2) (u-v)^2}}{A(A+B)}\,, \\
    r^2 = \frac{(u-v)^2}{A^2s^2} &= 
    \frac{A(u+v) \mp \sqrt{A^2 (u+v)^2 + (B^2-A^2) (u-v)^2}}{A(A-B)}\,.
  \end{align*}
  Note that the assumption $A,B,u,v > 0$ ensures this solution is real.
  At any nonnegative solution $(x,y)$ we must have $|s| \ge |r|$.
  Thus the choice of signs $(\pm,\mp)$ above must be $(+,-)$.
  Since $A,B,u,v > 0$, any nonnegative solution $(x,y)$ must in fact be positive.
  Then,
  \begin{align*}
    \frac{x}{y} &= \frac{s+r}{s-r} 
    = \frac{s + \frac{u-v}{As}}{s - \frac{u-v}{As}}
    = \frac{As^2 + (u-v)}{As^2 - (u-v)} \\
    &= \frac{
      A(u+v) + \sqrt{A^2 (u+v)^2 + (B^2-A^2) (u-v)^2} + (A+B)(u-v)
    }{
      A(u+v) + \sqrt{A^2 (u+v)^2 + (B^2-A^2) (u-v)^2} - (A+B)(u-v)
    } \\
    &= \frac{
      1 + \sqrt{1 + (\vartheta^2 - 1) t^2} + (\vartheta + 1) t
    }{
      1 + \sqrt{1 + (\vartheta^2 - 1) t^2} - (\vartheta + 1) t
    }= \frac{\vartheta t + \sqrt{1 + (\vartheta^2 - 1)t^2}}{1-t}\,,
  \end{align*}
  as desired.
  Call this value $\kappa$.
  The positive ray $x/y = \kappa$ clearly intersects $Ax^2 + Bxy = u$ in exactly one point.
  This proves that any nonnegative solution must be unique.
  Finally, it is straightforward to check that this intersection also solves $Ay^2 + Bxy = v$, so the solution exists.
\end{proof}
\end{fact}

\begin{lemma}\label{l:kmt-k-3}
  Consider $k=3$ and any $\beta \ge 2\beta_\lowerBoundTag$.
  On the dimension-$1$ affine space
  \[
    {\mathcal{P}}_3^\dagger = \left\{
      {\typeProportions} \in {\mathcal{P}}_3 : 
      \alpha_{+++} = 
      \alpha_{+--} = 
      \alpha_{-+-} = 
      \alpha_{--+}
      \,\,\text{and}\,\,
      \alpha_{++-} = 
      \alpha_{+-+} = 
      \alpha_{-++} = 
      \alpha_{---}
    \right\}\,,
  \]
  the function $\Phi_3^\beta$ is maximized uniquely at ${\typeProportions}_*$.
\begin{proof}
  We may parametrize ${\mathcal{P}}_3^\dagger = \{{\typeProportions}^\dagger(t) : t\in [-1,1]\}$, where ${\typeProportions}^\dagger(t)$ is defined by
  \begin{align*}
    \alpha^\dagger_{+++}(t) =
    \alpha^\dagger_{+--}(t) = 
    \alpha^\dagger_{-+-}(t) = 
    \alpha^\dagger_{--+}(t) &=
    \frac{1+t}{8}\,, \\
    \alpha^\dagger_{++-}(t) =
    \alpha^\dagger_{+-+}(t) = 
    \alpha^\dagger_{-++}(t) = 
    \alpha^\dagger_{---}(t) &=
    \frac{1-t}{8}\,.
  \end{align*}
  By symmetry, we may assume without loss $t\in [0,1]$.
  For $t\in [0,1)$, let ${\underline{\lambda}}^\dagger(t)$ be the vector ${\underline{\lambda}}$ associated to ${\typeProportions}^\dagger(t)$ by ~\cref{l:kmt-lagrange}, see \eqref{e:kmt-stationarity-lambda}--\eqref{e:kmt-stationarity-rows} therein.
  By symmetry, ${\underline{\lambda}}^\dagger(t)$ satisfies
  \[
    \lambda^\dagger_{+++}(t) =
    \lambda^\dagger_{+--}(t) =
    \lambda^\dagger_{-+-}(t) =
    \lambda^\dagger_{--+}(t)
    \quad\text{and}\quad
    \lambda^\dagger_{++-}(t) =
    \lambda^\dagger_{+-+}(t) =
    \lambda^\dagger_{-++}(t) =
    \lambda^\dagger_{---}(t)\,.
  \]
  To reduce notation, we henceforth write
  \begin{align*}
    \alpha_+ &= \alpha^\dagger_{+++}(t)\,, &
    \alpha_- &= \alpha^\dagger_{---}(t)\,, &
    \lambda_+ &= \lambda^\dagger_{+++}(t)\,, &
    \lambda_- &= \lambda^\dagger_{---}(t)\,,
  \end{align*}
  and suppress these variables' dependence on $t$.
  We will show that for all $t\in (0,1)$,
  \begin{equation}\label{e:kmt-k-3-goal-1}
    (d-1) \log \frac{\alpha_+}{\alpha_-}
    - d (\lambda_+ - \lambda_-) < 0\,.
  \end{equation}
  By the formula \eqref{e:kmt-gradient} from ~\cref{l:kmt-gradient-hessian}, this implies that $(\Phi_3^\beta \circ {\typeProportions}^\dagger)'(t) < 0$ for all $t\in (0,1)$.
  Then $\Phi_3^\beta({\typeProportions}^\dagger(t))$ is maximized over $t\in [0,1]$ at $t=0$, corresponding to ${\typeProportions}^\dagger(0) = {\typeProportions}_*$, which implies the conclusion.
  Define
  \begin{align*}
    x_+ &= e^{\lambda_+ - \frac12}\,, & 
    x_- &= e^{\lambda_- - \frac12}\,.
  \end{align*}
  Then \eqref{e:kmt-k-3-goal-1} is equivalent to
  \begin{equation}\label{e:kmt-k-3-goal-2}
    (d-1) \log \frac{\alpha_+}{\alpha_-}
    < d \log \frac{x_+}{x_-}
  \end{equation}
  for all $t\in (0,1)$.
  The identities \eqref{e:kmt-stationarity-lambda}--\eqref{e:kmt-stationarity-rows} imply that for $A = 1 + 3e^{2\beta}$, $B = e^\beta (3 + e^{2\beta})$, we have
  \begin{align*}
    A (x_+)^2 + B x_+ x_- &= \alpha_+\,, &
    A (x_-)^2 + B x_+ x_- &= \alpha_-\,.
  \end{align*}
  Write $\vartheta = B/A$, and note that $t = \frac{\alpha_+ - \alpha_-}{\alpha_+ + \alpha_-}$, like in the setup of Fact~\cref{f:quadratic-system}.
  Fact~\cref{f:quadratic-system} implies that this system has a unique nonnegative solution $(x_+,x_-)$, which is positive and satisfies
  \[
    \frac{x_+}{x_-} = \frac{\vartheta t + \sqrt{1 + (\vartheta^2 - 1)t^2}}{1-t}\,.
  \]
  Therefore, \eqref{e:kmt-k-3-goal-2} reduces to
  \begin{equation}\label{e:kmt-k-3-goal-3}
    (d-1) \log \frac{1+t}{1-t}
    < d \log \frac{\vartheta t + \sqrt{1 + (\vartheta^2 - 1)t^2}}{1-t}
  \end{equation}
  for all $t\in (0,1)$.
  We now recall
  \[
    \vartheta = \frac{B}{A} = \frac{e^\beta (3 + e^{2\beta})}{1 + 3e^{2\beta}}\,.
  \]
  Let $f(x) = \frac{x(3+x^2)}{1+3x^2}$.
  A simple calculation shows $f'(x) \ge 0$ on $x\in [0,+\infty)$, with strict equality except at $x=1$.
  Thus $f$ is strictly increasing on $[0,+\infty)$, which implies $\vartheta$ is strictly increasing in $\beta$.
  Since the right-hand side of \eqref{e:kmt-k-3-goal-3} is nondecreasing in $\vartheta$, it suffices to prove \eqref{e:kmt-k-3-goal-3} at $\beta = 2\beta_\lowerBoundTag$.
  Recall $\beta_\lowerBoundTag = -d^{-10}$, so $e^\beta \in [1-2d^{-10}, 1]$.
  We estimate
  \[
    f(1 - 2d^{-10}) \ge \frac{(1-2d^{-10})(3 + (1 - 2 \cdot 2d^{-10}))}{1 + 3\cdot 1}
    \ge 1 - 3d^{-10}\,.
  \]
  Thus
  \[
    \vartheta \in [f(1-2d^{-10}), f(1)] \subseteq [1 - 3d^{-10}, 1]\,.
  \]
  For such $\vartheta$,
  \[
    \vartheta t + \sqrt{1 + (\vartheta^2 - 1)t^2}
    \ge \vartheta t + 1 + (\vartheta^2 - 1)t^2
    \ge 1 + [1 - 3(1 - \vartheta)] t
    \ge 1 + (1 - 9d^{-10}) t\,.
  \]
  Let $\iota = 9d^{-10}$.
  Then \eqref{e:kmt-k-3-goal-3} reduces to showing
  \begin{equation}\label{e:kmt-k-3-goal-4}
    (d-1) \log (1+t) + \log(1-t) < d \log (1 + (1 - \iota)t)
  \end{equation}
  for all $t\in (0,1)$.
  Note that equality holds in \eqref{e:kmt-k-3-goal-4} at $t=0$, while
  \begin{align*}
    &\frac{\differential}{\differential t}\Big\{
      d \log (1 + (1 - \iota)t)
      - (d-1) \log (1+t) - \log(1-t)
    \Big\}
    = d\left[
      \frac{1-\iota}{1 + (1 - \iota)t} - \frac{1}{1+t}
    \right] + \frac{1}{1+t} + \frac{1}{1-t} \\
    &= - \frac{d\iota}{(1 + (1 - \iota)t)(1+t)} + \frac{1}{1+t} + \frac{1}{1-t}
    \ge - \frac{d\iota}{1+t} + \frac{1}{1+t} > 0
  \end{align*}
  since $d\iota = 9d^{-9} < 1$.
  This proves \eqref{e:kmt-k-3-goal-4} and the conclusion follows.
\end{proof}
\end{lemma}

\begin{proof}[Proof of ~\cref{l:basic-mt}]
  Let $s_0 = \min(s_0(d,\delta,2), s_0(d,\delta,3))$ be given by ~\cref{p:kmt-concavity-in-region}.
  The partition function $Z_{G,< \varepsilon}^{(2)}(\beta)$ is a sum over $(\sigma^1,\sigma^2) \in (\Sigma_n)^2$ satisfying
  \[
    \max\left(|R(\sigma^1,\allone)|, |R(\sigma^2,\allone)|, |R(\sigma^1,\sigma^2)| \right) < \varepsilon\,.
  \]
  By setting $\varepsilon_0$ small enough, we can ensure that $\varepsilon_0 < 1$ and if $\varepsilon < \varepsilon_0$, the ${\typeProportions}$ associated to such $(\sigma^1,\sigma^2)$ (in the sense of ~\cref{d:basic-mt-type}~\cref{i:basic-mt-type-spins}) satisfies $\|{\typeProportions} - {\typeProportions}_*\| \le s_0$.
  Then ~\cref{l:kmt-exp-rate} yields
  \begin{align*}
    \E[Z_{G,< \varepsilon}^{(2)}(\beta)]
    &\le \exp\left(
      n \sup\{\Phi_2^\beta({\typeProportions},\Gamma) : \|{\typeProportions} - {\typeProportions}_*\| \le s_0\} + o(n)
    \right) \\
    &= \exp\left(
      n \sup\{\Phi_2^\beta({\typeProportions}) : \|{\typeProportions} - {\typeProportions}_*\| \le s_0\} + o(n)
    \right)\,,
  \end{align*}
  where we recall $\Phi_2^\beta({\typeProportions}) = \sup_{\Gamma \in {\mathcal{Q}}_2({\typeProportions})} \Phi_2^\beta({\typeProportions},\Gamma)$ is defined in ~\cref{l:kmt-gradient-hessian}.
  ~\cref{l:kmt-gradient-hessian}, together with ~\cref{l:kmt-stationarity} and ~\cref{p:kmt-concavity-in-region}, ensure that in the set
  \[
    \{{\typeProportions} \in {\mathcal{P}}^\circ_2 : \|{\typeProportions} - {\typeProportions}_*\| \le s_0\}\,,
  \]
  the function ${\typeProportions} \mapsto \Phi_2^\beta({\typeProportions})$ is $\iota$-strongly concave and stationary at ${\typeProportions}_*$.
  It follows that
  \[
    \E[Z_{G,< \varepsilon}^{(2)}(\beta)]
    = \exp(n\Phi_2^\beta({\typeProportions}_*) + o(n))
    = e^{o(n)} \E[Z_{G,< \varepsilon}(\beta)]^2\,,
  \]
  which proves part \cref{i:basic-mt-annealed-2}.
  The same argument with $s$ in place of $\varepsilon$ proves part \cref{i:basic-mt-annealed-2-subgaussian}.

  We turn to part \cref{i:basic-mt-annealed-3}.
  ~\cref{l:kmt-k-3} implies that for ${\mathcal{P}}_3^\dagger$ defined therein,
  \[
    \sup\left\{
      \Phi_3^\beta({\typeProportions}) : {\typeProportions} \in {\mathcal{P}}_3^\dagger, \|{\typeProportions} - {\typeProportions}_*\| > s_0
    \right\} < \Phi_3^\beta({\typeProportions}_*)\,.
  \]
  Note that $\Phi_3^\beta$ is continuous on compact domain ${\mathcal{P}}_3$, and thus uniformly continuous.
  So there exists $\iota' > 0$ such that
  \begin{equation}\label{e:cP-3-dagger-away-from-center}
    \sup\left\{
      \Phi_3^\beta({\typeProportions}) : {\typeProportions} \in {\mathcal{P}}_3, d({\typeProportions}, {\mathcal{P}}_3^\dagger) \le \iota', \|{\typeProportions} - {\typeProportions}_*\| > s_0
    \right\} < \Phi_3^\beta({\typeProportions}_*)\,,
  \end{equation}
  where $d({\typeProportions}, {\mathcal{P}}_3^\dagger)$ denotes point-to-set distance.
  Note that $Z_{G,< \varepsilon}^{(3)}(\beta)$ is a sum over $(\sigma^1,\sigma^2,\sigma^3) \in (\Sigma_n)^3$ satisfying
  \[
    \max\left(
      |R(\sigma^1,\allone)|, 
      |R(\sigma^2,\allone)|, 
      |R(\sigma^3,\allone)|, 
      |R(\sigma^1,\sigma^2)|, 
      |R(\sigma^1,\sigma^3)|, 
      |R(\sigma^2,\sigma^3)|
    \right) < \varepsilon\,.
  \]
  By setting $\varepsilon_0$ small enough, we can ensure that $\varepsilon_0 < 1$ and if $\varepsilon < \varepsilon_0$, the ${\typeProportions}$ associated to such $(\sigma^1,\sigma^2,\sigma^3)$ satisfy $d({\typeProportions},{\mathcal{P}}_3^\dagger) \le \iota'$.
  ~\cref{l:kmt-exp-rate} implies
  \[
    \E[Z_{G,< \varepsilon}^{(3)}(\beta)]
    \le \exp\left(
      n \sup\{\Phi_3^\beta({\typeProportions}) : d({\typeProportions},{\mathcal{P}}_3^\dagger) \le \iota'\} + o(n)
    \right)\,.
  \]
  The estimate \eqref{e:cP-3-dagger-away-from-center} bounds this supremum by $\Phi_3^\beta({\typeProportions}_*)$ for the subset of ${\typeProportions}$ where $\|{\typeProportions} - {\typeProportions}_*\| > s_0$.
  The same proof as for part \cref{i:basic-mt-annealed-2} shows
  \[
    \sup\{\Phi_3^\beta({\typeProportions}) : \|{\typeProportions} - {\typeProportions}_*\| \le s_0\} = \Phi_3^\beta({\typeProportions}_*)\,.
  \]
  Thus
  \[
    \E[Z_{G,< \varepsilon}^{(3)}(\beta)]
    = \exp(n\Phi_3^\beta({\typeProportions}_*) + o(n))
    = e^{o(n)} \E[Z_{G,< \varepsilon}(\beta)]^3\,,
  \]
  which proves part \cref{i:basic-mt-annealed-3}.
\end{proof}

\section{Complex truncated second moment}
\label{s:deferred-complex-moment}

Here we provide a proof of \cref{prop:2nd-moment}.
Conceptually, it will go along the same lines as our proof of \cref{l:1mt}. 
We first follow the same steps as in the proof of~\cref{l:1mt} to proof the unconditioned statement, after which we give the full proof.

Throughout this section, put
\[
    N:=nd,\qquad p:=e^\beta,\qquad q:=e^{\overline{\beta}},\qquad h(\beta):=(1+p)(1+q),
\]
and
\[
    \delta_2(\beta):=\frac{p-1}{p+1}\frac{q-1}{q+1}
    =\left|\tanh\left(\frac{\beta}{2}\right)\right|^2.
\]
Thus, for real $\beta\geq0$,
\[
    \betaKS
    =
    2\arctanh\left(\frac{1}{\sqrt{d-1}}\right),
    \qquad
    (d-1)\delta_2(\betaKS)=1.
\]

Fix $0<\eta<1$. For $\rho>0$, define
\[
    K_\eta(\rho)
    :=
    \left\{
        \beta\in\C:
        \dist\left(
            \beta,[0,(1-\eta)\betaKS]
        \right)
        \leq\rho
    \right\}.
\]
After choosing $\rho=\rho(d,\eta)>0$ sufficiently small, there exists $c_\eta>0$ such that
\begin{equation}
\label{eq:second-moment-KS-gap}
    1-(d-1)\delta_2(\beta)\geq c_\eta
\end{equation}
for every $\beta\in K_\eta(\rho)$.

We start with obtaining a formula for the (unconditioned) truncated second moment, after which we shall handle the general case.
As for the first moment, we need a representation of the second moment in terms of a contour integral.

\subsection{A contour integral representation of the second moment}
Let us recall that we work with the configuration model on $N=nd$ sockets and $n$ clouds. 

For a pair of configurations $(\sigma,\sigma')\in\{-1,1\}^n\times\{-1,1\}^n$, let
\[
    \momentCounts=(m_{++},m_{+-},m_{-+},m_{--}),
\]
where $m_{ab}$ is the number of coordinates $i$ for which
\[
    (\sigma_i,\sigma'_i)=(a,b).
\]
Thus
\[
    m_{++}+m_{+-}+m_{-+}+m_{--}=n.
\]
The number of pairs $(\sigma,\sigma')$ having type-count vector $\momentCounts$ is
\[
    \binom{n}{\momentCounts}
    :=
    \frac{n!}{
        m_{++}!m_{+-}!m_{-+}!m_{--}!
    }.
\]

For a fixed pair $(\sigma,\sigma')$ with type-count vector $\momentCounts$, put
\[
    f_N(\beta,\momentCounts)
    :=
    \E\left[
        \exp\left(
            \beta \isingHamiltonian(\sigma)
            +
            \overline{\beta}\isingHamiltonian(\sigma')
        \right)
    \right].
\]
A socket incident to a cloud of type $a\in\{++,+-,-+,--\}$ is said to have type $a$. There are $dm_a$ sockets of type $a$.

Define
\[
    Q(p,q)
    :=
    \begin{pmatrix}
        1&q&p&pq\\
        q&1&pq&p\\
        p&pq&1&q\\
        pq&p&q&1
    \end{pmatrix}.
\]
The entry $Q(p,q)_{ab}$ is the weight contributed by an edge joining sockets of types $a$ and $b$.
Define 
\begin{equation}
\label{eq:second-moment-phase}
    \Psi_{\typeProportions,\beta}^{(N)}(\contourVariables)
    :=
    \frac{1}{2}\contourVariables^{\matrixtranspose}Q(p,q)\contourVariables
    -
    \left(
        \typeProportions+\frac{1}{N}\vec{1}
    \right)^{\matrixtranspose}\log\contourVariables,
\end{equation}
and let
\[
    \mathcal A_{n,\varepsilon}^{(2)}
    :=
    \left\{
        \typeProportions\in
        \left(\frac{1}{n}\Z_{\geq0}\right)^4:
        \sum_a\alpha_a=1,
        \quad
        \left\|
            \typeProportions-\frac{1}{4}\vec{1}
        \right\|_\infty<\varepsilon
    \right\}.
\]

\begin{lemma}\label{lem:2nd moment contour representation}
We have the following representation for the truncated second moment:
\begin{equation}
\label{eq:second-moment-exact-sum}
\begin{aligned}
    \E_M\left[
        Z_{M,<\varepsilon}^{(2)}(\beta)
    \right]
    &=
    \sum_{\typeProportions\in\mathcal A_{n,\varepsilon}^{(2)}}
    \binom{n}{n\typeProportions}
    \frac{
        \prod_a(\alpha_aN)!
    }{
        (N-1)!!\,N^{\frac{N}{2}}
    }\\
    &\qquad\times
    \frac{1}{(2\pi i)^4}
    \oint
    \exp\left(
        N\Psi_{\typeProportions,\beta}^{(N)}(\contourVariables)
    \right)
    \,\differential\contourVariables.
\end{aligned}
\end{equation}    
\end{lemma}
\begin{proof}
The exponential formula for matchings (see e.g.~\cite[Chapter 5]{Stanleybookvol2}) gives
\begin{equation}
\label{eq:second-moment-coefficient}
    f_N(\beta,\momentCounts)
    =
    \frac{
        \prod_a(dm_a)!
    }{
        (N-1)!!
    }
    [\contourCoefficients^{d\momentCounts}]
    \exp\left(
        \frac{1}{2}\contourCoefficients^{\matrixtranspose}Q(p,q)\contourCoefficients
    \right),
\end{equation}
where
\[
    \contourCoefficients=(u_{++},u_{+-},u_{-+},u_{--})
\]
and
\[
    [\contourCoefficients^{d\momentCounts}]
    :=
    [u_{++}^{dm_{++}}u_{+-}^{dm_{+-}}u_{-+}^{dm_{-+}}u_{--}^{dm_{--}}].
\]
Equivalently,
\begin{align}
    f_N(\beta,\momentCounts)
    &=
    \frac{
        \prod_a(dm_a)!
    }{
        (N-1)!!
    }
    \frac{1}{(2\pi i)^4}
    \oint
    \exp\left(
        \frac{1}{2}\contourCoefficients^{\matrixtranspose}Q(p,q)\contourCoefficients
    \right)
    \prod_a\frac{\differential u_a}{u_a^{dm_a+1}}.
\label{eq:second-moment-contour-unscaled}
\end{align}

Write
\[
    \momentCounts=n\typeProportions,
    \qquad
    \contourCoefficients=\sqrt{N}\,\contourVariables.
\]
Then \eqref{eq:second-moment-contour-unscaled} becomes
\begin{equation}
\label{eq:second-moment-contour-scaled}
    f_N(\beta,n\typeProportions)
    =
    \frac{
        \prod_a(\alpha_aN)!
    }{
        (N-1)!!\,N^{\frac{N}{2}}
    }
    \frac{1}{(2\pi i)^4}
    \oint
    \exp\left(
        N\Psi_{\typeProportions,\beta}^{(N)}(\contourVariables)
    \right)
    \,\differential\contourVariables.
\end{equation}
Summing over all $\alpha\in  \mathcal A_{n,\varepsilon}^{(2)}$ gives the required representation.
\end{proof}

\subsection{The contour integral}\label{sec:2nd m contour integral}
Here we focus on the contour integral in \cref{lem:2nd moment contour representation}

Write
\[
    M(p,q):=Q(p,q)-I.
\]
The eigenvectors of $M(p,q)$ are
\[
    v_1=(1,1,1,1)^{\matrixtranspose},\qquad
    v_2=(-1,-1,1,1)^{\matrixtranspose},
\]
\[
    v_3=(-1,1,-1,1)^{\matrixtranspose},\qquad
    v_4=(1,-1,-1,1)^{\matrixtranspose},
\]
with respective eigenvalues
\[
    p+q+pq,\qquad
    -p+q-pq,\qquad
    p-q-pq,\qquad
    -p-q+pq.
\]
Let \[\lambda_\beta^{(N)}
    =
    \sqrt{\frac{
        \frac{1}{4}+\frac{1}{N}
    }{
        h(\beta)
    }}.
\]

\begin{theorem}[Uniform four-dimensional contour estimate]
\label{thm:second-moment-contour}
There exist $\rho=\rho(d,\eta)>0$, $\varepsilon_0=\varepsilon_0(d,\eta)>0$, and $N_0>0$ such that the following holds for every $N\geq N_0$, every $\beta\in K_\eta(\rho)$, and every $\typeProportions$ satisfying
\[
    \left\|
        \typeProportions-\frac{1}{4}\vec{1}
    \right\|_\infty
    \leq\varepsilon_0.
\]
There exists a unique stationary point $\contourVariables^{(N)}(\typeProportions,\beta)=\lambda_\beta^{(N)}\vec{1}$ of $\Psi_{\typeProportions,\beta}^{(N)}$.

Moreover,
\begin{align}
    &\oint
    \exp\left(
        N\Psi_{\typeProportions,\beta}^{(N)}(\contourVariables)
    \right)
    \,\differential\contourVariables
    \nonumber\\
    &\qquad=
    \left(
        1+O\left(\frac{1}{n}\right)
    \right)
    \frac{8\pi^2}{N^2}
    \frac{
        \exp\left(
            N\Psi_{\typeProportions,\beta}^{(N)}
            \left(
                \contourVariables^{(N)}(\typeProportions,\beta)
            \right)
        \right)
    }{
        \sqrt{
            \det\left(
                \nabla_{\contourVariables}^2
                \Psi_{\typeProportions,\beta}^{(N)}
                \left(
                    \contourVariables^{(N)}(\typeProportions,\beta)
                \right)
            \right)
        }
    }.
\label{eq:second-moment-contour-asymptotic}
\end{align}
The square-root branch is chosen by continuation from real $\beta$, where the determinant at the symmetric saddle is positive.
\end{theorem}

\begin{proof}
The stationary-point equation for $\Psi_{\typeProportions,\beta}^{(N)}$ is
\begin{equation}
\label{eq:second-moment-stationary}
    x_a\left(Q(p,q)\contourVariables\right)_a
    =
    \alpha_a+\frac{1}{N},
    \qquad
    a\in\{++,+-,-+,--\}.
\end{equation}
At $\typeProportions_\star:=\tfrac{1}{4}\vec{1}$, and $\beta\geq 0$, the two relevant solutions are
\[
    \contourVariables=\pm\lambda_\beta^{(N)}\vec{1}.
\]
The Jacobian of \eqref{eq:second-moment-stationary} with respect to $\contourVariables$, evaluated at the positive solution, is
\[
    \lambda_\beta^{(N)}
    \left(
        h(\beta)I+Q(p,q)
    \right)
    =
    \lambda_\beta^{(N)}
    \left(
        (1+h(\beta))I+M(p,q)
    \right).
\]
Its eigenvalues are
\[
    2\lambda_\beta^{(N)}h(\beta),\qquad
    2\lambda_\beta^{(N)}(1+q),
\]
\[
    2\lambda_\beta^{(N)}(1+p),\qquad
    2\lambda_\beta^{(N)}(1+pq).
\]
These eigenvalues are uniformly bounded away from zero for $\beta\in K_\eta(\rho)$ after decreasing $\rho$ if necessary. Applying the holomorphic implicit-function theorem in the independent variables $(p,q)$, and then specializing to
\[
    p=e^\beta,
    \qquad
    q=e^{\overline{\beta}},
\]
gives the uniqueness of the asserted stationary points uniformly over the parameter set.

We next use Cauchy's theorem to evaluate the contour integral on the torus $\T=\{\lambda^{(N)}_{\beta}e^{i\vec{\theta}}\mid \vec{\theta}\in [0,2\pi)^4\}$.
It will be convenient to change coordinates accordingly. 
Write
\[
    \contourVariables^{(N)}(\typeProportions,\beta)
    =
     e^{i\vec{\theta}^{\,*}}
\]
coordinatewise, and define
\[
    \tilde{\Psi}_{\typeProportions,\beta}^{(N)}(\vec{\theta})
    :=
    \Psi_{\typeProportions,\beta}^{(N)}
    \left(
         e^{i\vec{\theta}}
    \right).
\]
Let $\vec{\vartheta^*}$ denote the stationary point corresponding to the continuation of $\contourVariables(\tfrac14\vec{1},\beta))$ for $\beta\geq 0$.

We next identify the dominant points of the torus. 
First assume that $\beta$ is real and that $\typeProportions=\typeProportions_\star=\tfrac14\vec{1}$. Write
\[
    x_a=\lambda_\beta^{(N)}e^{i\theta_a}.
\]
Since all entries of $Q(e^\beta,e^\beta)$ are positive,
\begin{align}
    &\Psi_{\typeProportions_\star,\beta}^{(N)}
    \left(
        \lambda_\beta^{(N)}\vec{1}
    \right)
    -
    \re\Psi_{\typeProportions_\star,\beta}^{(N)}
    \left(
        \lambda_\beta^{(N)}e^{i\vec{\theta}}
    \right)
    \nonumber\\
    &\qquad=
    \frac{
        \left(\lambda_\beta^{(N)}\right)^2
    }{2}
    \sum_{a,b}
    Q(e^\beta,e^\beta)_{ab}
    \left(
        1-\cos(\theta_a+\theta_b)
    \right).
\label{eq:second-moment-real-gap}
\end{align}
The right-hand side of \eqref{eq:second-moment-real-gap} is non negative. It vanishes precisely when
\[
    \vec{\theta}=\vec{0}
    \qquad\text{or}\qquad
    \vec{\theta}=\pi\vec{1}.
\]
Thus these are the only global maxima of the real part on the torus.

At $\contourVariables=\lambda^{(N)}_{\beta}\vec{1}$, the Hessian, $H$ of $\Psi^{(N)}_{\typeProportions^\star,\beta}$ is given by
\begin{equation}\label{eq:2nd m formula Hessian Psi}
    \left(
        h(\beta)I+Q(e^\beta,e^\beta)
    \right),
\end{equation}
which is uniformly positive definite for real $\beta$ in any bounded set and in particular for
\[
    \beta\in[0,(1-\eta)\betaKS].
\]
The angular Hessian at $\vec{0}$ satisfies
\[
    A
    :=
    \nabla_{\vec{\theta}}^2
    \tilde{\Psi}_{\typeProportions^\star,\beta}^{(N)}
    \left(
        \vec{\theta}^{\,*}
    \right)=-\left(\lambda_\beta^{(N)}\right)^2H.
\]
Thus $A$ is uniformly negative definite.


It thus follows from the above and compactness, that there exists constants $\gamma,\delta,c>0$ such that, uniformly in $\beta\in[0,(1-\eta)\betaKS]$ and all sufficiently large $N$ we have
\begin{itemize}
    \item the eigenvalues of the real part of the Hessian of $\Tilde \Psi^{(N)}_{\typeProportions_\star,\beta}$ are bounded above by $-2c$ in the $\gamma$-neighborhoods of $\vec 0$ (and $\vec\pi$ by symmetry);
    \item outside these two neighborhoods, 
    \[
    \re \tilde\Psi_{\typeProportions_\star,\beta}^{(N)}(\vec\theta)\le \re\tilde\Psi_{\typeProportions_\star,\beta}^{(N)}(\vec 0)-2\delta.
    \]
\end{itemize}
We now perturb $(\typeProportions_\star,\beta)$ and extend these properties to the general setting.
To this end note that the higher order derivatives of $\tilde\Psi$ and the stationary point $\contourVariables(\typeProportions,\beta))$ depend continuously on $(\typeProportions,\beta,1/N)$. 
Therefore, provided $\varepsilon$ and $\rho$ are small enough, the proceeding uniform real-case estimates, continuity, symmetry and compactness imply that there exists a neighbourhood $U_1$ of $\vec 0$ containing $\vec{\vartheta^*}$ such that the following holds for every $\beta\in K_\eta(\rho)$ and every $\typeProportions$ satisfying
\[
    \left\|
        \typeProportions-\typeProportions_\star
    \right\|_\infty
    \leq\varepsilon,
\]
\begin{enumerate}
    \item[P$^{(2)}$(1)] the eigenvalues of the real part of the $\tilde\Psi^{(N)}_{\typeProportions,\beta}$ are bounded above by $-c$
    on $U_1$ and on $U_2:=U_1+\pi\vec{1}$;
    \item[P$^{(2)}$(2)] for all $\vec \theta\in \T\setminus (U_1\cup U_2)$,
    \begin{equation}
    \label{eq:second-moment-perturbed-gap}
        \re\tilde{\Psi}^{(N)}_{\typeProportions,\beta}(\vec\theta)
        \leq
        \re\tilde{\Psi}^{(N)}_{\typeProportions,\beta}(\vec{\vartheta^*})   - \delta;
    \end{equation}

    \item[P$^{(2)}$(3)] all derivatives $\tilde\Psi$ of order at most four are uniformly bounded on $U_1\cup U_2$.
\end{enumerate}

Taylor's theorem through fourth order gives 
\[
\begin{aligned}\label{eq:taylor second mom}
    \tilde{\Psi}_{\typeProportions,\beta}^{(N)}
    \left(
        \vec{\vartheta}^{\,*}+\vec{s}
    \right)
    &=
    \tilde{\Psi}_{\typeProportions,\beta}^{(N)}
    \left(
        \vec{\vartheta}^{\,*}
    \right)
    +
    \frac{1}{2}\vec{s}^{\,\matrixtranspose}A\vec{s}
    +
    P_3(\vec{s})
    +
    \mathcal{E}(\vec{s})
\end{aligned}
\]
where $P_3$ is an odd homogeneous polynomials of degree three, and
\[
    | \mathcal{E}(\vec{s})|
    \leq
    C\|\vec{s}\|^4
\]
for some constant $C>0$.
After the change of variables $\vec{s}=\frac{\vec{t}}{\sqrt{N}}$,
we can write
\begin{align}
    &N\left[
        \tilde{\Psi}_{\typeProportions,\beta}^{(N)}
        \left(
            \vec{\theta}^{\,*}
            +
            \frac{\vec{t}}{\sqrt{N}}
        \right)
        -
        \tilde{\Psi}_{\typeProportions,\beta}^{(N)}
        \left(
            \vec{\theta}^{\,*}
        \right)
    \right]
    +
    \frac{i}{\sqrt{N}}
    \sum_a t_a
    \nonumber\\
    &\qquad=
    \frac{1}{2}\vec{t}^{\,\matrixtranspose}A\vec{t}
    +
    \frac{1}{\sqrt{N}}Q_3(\vec{t})
    +
    \frac{1}{N}\tilde{\mathcal{E}}(\vec{t}),
\label{eq:second-moment-local-expansion}
\end{align}
where $Q_3$ is odd, and
\[
    |\tilde{\mathcal{E}}(\vec{t})|
    \leq
    C\|\vec{t}\|^4.
\]

Next let
\[
    B_N:=
    \left[
        -N^{-\frac{2}{5}},
        N^{-\frac{2}{5}}
    \right]^4.
\]
Since $B_N$ is symmetric under $\vec{t}\mapsto-\vec{t}$, we use the same argument as in the proof of the complex first moment to see that,
\begin{align}
   & \int_{B_N}\exp \left(N\left(
        \tilde{\Psi}_{\typeProportions,\beta}^{(N)}
        \left(
            \vec{\theta}^{\,*}
            +
            \vec{s}
        \right)
        -
        \tilde{\Psi}_{\typeProportions,\beta}^{(N)}
        \left(
            \vec{\theta}^{\,*}
        \right)
    \right)\right)\differential(\vec s)=
    \\
  &  \frac{1}{N^2} \int_{
        \left[
            -N^{\frac{1}{10}},
            N^{\frac{1}{10}}
        \right]^4
    }
    \exp\left(
        \frac{1}{2}\vec{t}^{\,\matrixtranspose}A\vec{t}
        +
        \frac{1}{\sqrt{N}}Q_3(\vec{t})
        +
        \frac{1}{N}\tilde{\mathcal{E}}(\vec{t})
    \right)
    \,\differential\vec{t}
    \nonumber\\
    &\qquad=
   \frac{1}{N^2} \int_{\R^4}
    \exp\left(
        \frac{1}{2}\vec{t}^{\,\matrixtranspose}A\vec{t}
    \right)
    \,\differential\vec{t}
    +
    O\left(\frac{1}{n}\right).
\label{eq:second-moment-local-gaussian}
\end{align}
The complex Gaussian identity gives
\[
    \int_{\R^4}
    \exp\left(
        \frac{1}{2}\vec{t}^{\,\matrixtranspose}A\vec{t}
    \right)
    \,\differential\vec{t}
    =
    \frac{(2\pi)^2}{\sqrt{\det(-A)}}.
\]
The contribution around each stationary point is therefore given by
\[
    \left(
        1+O\left(\frac{1}{n}\right)
    \right)
    \frac{(2\pi)^2}{N^2}
    \frac{
        \exp\left(
            N\Psi_{\typeProportions,\beta}^{(N)}
            \left(
                \contourVariables^{(N)}(\typeProportions,\beta)
            \right)
        \right)
    }{
        \sqrt{
            \det\left(
                \nabla_{\contourVariables}^2
                \Psi_{\typeProportions,\beta}^{(N)}
                \left(
                    \contourVariables^{(N)}(\typeProportions,\beta)
                \right)
            \right)
        }
    }.
\]
By P$^{(2)}$(1) and P$^{(2)}$(2), the integral outside the two $N^{-\frac{2}{5}}$ neighborhoods is bounded in absolute value by
\[
   O\left(
    \exp\left(
        N\re\Psi_{\typeProportions,\beta}^{(N)}
        \left(
            \contourVariables^{(N)}(\typeProportions,\beta)
        \right)
    \right)
    e^{-cN^{\frac{1}{5}}}\right),
\]
which is $O(n^{-1})$ relative to the two local contributions. Adding the two contributions proves \eqref{eq:second-moment-contour-asymptotic}.
\end{proof}

\subsection{Derivative calculations}
Here we calculate the first and second derivative of $\Psi^{(N)}_{\typeProportions,\beta}$ at the stationary point, when regarded as a function of $\typeProportions$. 
We will use these to prove a formula for the truncated second moment in the next subsection.

To avoid ambiguity we denote for fixed $\beta$ the map $\typeProportions\mapsto \Psi^{(N)}_{\typeProportions,\beta}(\contourVariables(\typeProportions,\beta))$ by $\Psi^{(N)}_\beta(\typeProportions)$. 
We want to compute the gradient and Hessian of this map near $\typeProportions=\tfrac{1}{4}\vec{1}.$
Write
\[
    \alpha_{--}
    =
    1-\alpha_{++}-\alpha_{+-}-\alpha_{-+},
\]
Let
\[
    B
    :=
    \begin{pmatrix}
        -1&0&0\\
        0&-1&0\\
        0&0&-1\\
        1&1&1
    \end{pmatrix}.
\]

\begin{lemma}
\label{lem:second-moment-saddle-derivatives}
At $\typeProportions_\star=\frac{1}{4}\vec{1}$,
we have
\[
    \nabla_{\typeProportions}\Psi^{(N)}_{\beta}(\typeProportions_\star)=0
\]
and
\begin{equation}
\label{eq:second-moment-saddle-hessian}
    \nabla_{\typeProportions}^2\Psi^{(N)}_{\beta}(\typeProportions_\star)
    =
    -\frac{h(\beta)}{\frac{1}{4}+\frac{1}{N}}
    B^{\matrixtranspose}
    \left(
        (1+h(\beta))I+M(p,q)
    \right)^{-1}
    B.
\end{equation}

\end{lemma}

\begin{proof}
Since $\contourVariables^{(N)}(\typeProportions,\beta)$ is a stationary point,
\[
    \nabla_{\contourVariables}
    \Psi_{\typeProportions,\beta}^{(N)}
    \left(
        \contourVariables^{(N)}(\typeProportions,\beta)
    \right)
    =
    0.
\]
Therefore the chain rule gives
\[
    \nabla_{\typeProportions}\Psi^{(N)}_{\beta}(\typeProportions)
    =
    \nabla_{\typeProportions}
    \Psi_{\typeProportions,\beta}^{(N)}
    \left(
        \contourVariables^{(N)}(\typeProportions,\beta)
    \right).
\]
At $\typeProportions=\typeProportions_\star$, all four coordinates of $  \contourVariables^{(N)}(\typeProportions,\beta)$ are equal, and hence
\[
    \nabla_{\typeProportions}\Psi^{(N)}_{\beta}(\typeProportions_\star)=0.
\]

Differentiating the stationary-point equation with respect to the three independent coordinates of $\typeProportions$ gives
\[
    \nabla_{\typeProportions}\contourVariables^{(N)}
    =
    -
    \left(
        \nabla_{\contourVariables}^2
        \Psi_{\typeProportions,\beta}^{(N)}
    \right)^{-1}
    \nabla_{\typeProportions}\nabla_{\contourVariables}
    \Psi_{\typeProportions,\beta}^{(N)}.
\]
Consequently,
\[
    \nabla_{\typeProportions}^2\Psi^{(N)}_{\beta}
    =
    -
    \left(
        \nabla_{\typeProportions}\nabla_{\contourVariables}
        \Psi_{\typeProportions,\beta}^{(N)}
    \right)^{\matrixtranspose}
    \left(
        \nabla_{\contourVariables}^2
        \Psi_{\typeProportions,\beta}^{(N)}
    \right)^{-1}
    \nabla_{\typeProportions}\nabla_{\contourVariables}
    \Psi_{\typeProportions,\beta}^{(N)}.
\]
At $\typeProportions=\typeProportions_\star$,
\[
    \contourVariables^{(N)}(\typeProportions_\star,\beta)
    =
    \lambda_\beta^{(N)}\vec{1},
\]
and
\[
    \nabla_{\contourVariables}^2
    \Psi_{\typeProportions_\star,\beta}^{(N)}
    \left(
        \lambda_\beta^{(N)}\vec{1}
    \right)
    =
    (1+h(\beta))I+M(p,q).
\]
Moreover,
\[
    \nabla_{\typeProportions}\nabla_{\contourVariables}
    \Psi_{\typeProportions_\star,\beta}^{(N)}
    \left(
        \lambda_\beta^{(N)}\vec{1}
    \right)
    =
    \frac{1}{\lambda_\beta^{(N)}}B.
\]
Since
\[
    \left(\lambda_\beta^{(N)}\right)^{-2}
    =
    \frac{h(\beta)}{\frac{1}{4}+\frac{1}{N}},
\]
substitution gives \eqref{eq:second-moment-saddle-hessian}.
\end{proof}

\subsection{A formula for the second moment}
Define
\[
    H(\typeProportions):=-\sum_a\alpha_a\log\alpha_a
\]
and
\begin{equation}\label{eq:2nd m combined entropy and energy}
    \mathcal F_{N,\beta}(\typeProportions)
    :=
    d\left(
        \Psi^{(N)}_{\beta}(\typeProportions)-\frac{1}{2}
    \right)
    +
    (1-d)H(\typeProportions).
\end{equation}
We will need the following lemma.
\begin{lemma}\label{lem:Hessian determinant}
For $\beta\in K_\eta(\rho)$ and all sufficiently large $n$ we have that
\[
    \nabla_{\typeProportions}^2\mathcal F_{N,\beta}(\typeProportions_\star)
\]
is uniformly negative definite,
Moreover,
\begin{align}
    &\det\left(
        -\nabla_{\typeProportions}^2
        \mathcal F_{N,\beta}(\typeProportions_\star)
    \right)
    \nonumber\\
    &\qquad=
    32
    \left(d(p-1)+2\right)
    \left(d(q-1)+2\right)
    \frac{h(\beta)}{1+pq}
    \left(
        1-(d-1)\delta_2(\beta)
    \right)
    \left(
        1+O\left(\frac{1}{n}\right)
    \right).
\label{eq:second-moment-rate-determinant}
\end{align}
\end{lemma}
\begin{proof}
At $\typeProportions_\star$ we have
\[
    \nabla_{\typeProportions}H(\typeProportions_\star)=0
\]
and
\[
    \nabla_{\typeProportions}^2H(\typeProportions_\star)
    =
    -4B^{\matrixtranspose}B.
\]
Combining this identity with~\cref{lem:second-moment-saddle-derivatives} gives
\begin{equation}
\label{eq:second-moment-rate-hessian}
    \nabla_{\typeProportions}^2\mathcal F_{N,\beta}(\typeProportions_\star)
    =
    B^{\matrixtranspose}
    \left[
        4(d-1)I
        -
        \frac{dh(\beta)}{\frac{1}{4}+\frac{1}{N}}
        \left(
            (1+h(\beta))I+M(p,q)
        \right)^{-1}
    \right]
    B.
\end{equation}

The image of $B$ is the orthogonal complement of $v_1$, and is spanned by $v_2,v_3,v_4$ (which we recall were defined in \cref{sec:2nd m contour integral}). The eigenvalues of the matrix in square brackets in \eqref{eq:second-moment-rate-hessian} in these three directions are
\[
    \mu_{2,N}
    =
    4(d-1)
    -
    \frac{d(1+p)}{
        2\left(
            \frac{1}{4}+\frac{1}{N}
        \right)
    },
\]
\[
    \mu_{3,N}
    =
    4(d-1)
    -
    \frac{d(1+q)}{
        2\left(
            \frac{1}{4}+\frac{1}{N}
        \right)
    },
\]
and
\[
    \mu_{4,N}
    =
    4(d-1)
    -
    \frac{dh(\beta)}{
        2\left(
            \frac{1}{4}+\frac{1}{N}
        \right)(1+pq)
    }.
\]
Thus
\[
    \mu_{2,N}
    =
    -2\left(d(p-1)+2\right)
    +
    O\left(\frac{1}{n}\right),
\]
\[
    \mu_{3,N}
    =
    -2\left(d(q-1)+2\right)
    +
    O\left(\frac{1}{n}\right),
\]
and
\[
    \mu_{4,N}
    =
    -\frac{2h(\beta)}{1+pq}
    \left(
        1-(d-1)\delta_2(\beta)
    \right)
    +
    O\left(\frac{1}{n}\right).
\]
Therefore by \eqref{eq:second-moment-KS-gap}, the real part of
\[
    \nabla_{\typeProportions}^2\mathcal F_{N,\beta}(\typeProportions_\star)
\]
is uniformly negative definite for $\beta\in K_\eta(\rho)$ and all sufficiently large $n$.
Finally, for the determinant we use 
\[
    \det(B^{\matrixtranspose}B)=4,
\]
to arrive at~\cref{eq:second-moment-rate-determinant}.
\end{proof}

We now state and prove the promised formula for the truncated second moment.
\begin{theorem}[Unconditioned second moment]
\label{thm:unconditioned-second-moment}
Fix $0<\eta<1$. There exist $\rho=\rho(d,\eta)>0$ and $\varepsilon_0=\varepsilon_0(d,\eta)>0$ such that, for every $\beta\in K_\eta(\rho)$ and every $0<\varepsilon\leq\varepsilon_0$,
\begin{align}
    \E_M\left[
        Z_{M,<\varepsilon}^{(2)}(\beta)
    \right]
    &=
    \left(
        1+O\left(\frac{1}{n}\right)
    \right)
    4^n
    \left(
        \frac{h(\beta)}{4}
    \right)^{\frac{N}{2}}
    \frac{
        \sqrt{h(\beta)}
    }{
        \sqrt{
            \left(d(p-1)+2\right)
            \left(d(q-1)+2\right)
        }
    }
    \nonumber\\
    &\qquad\times
    \frac{1}{
        \sqrt{
            1-(d-1)\delta_2(\beta)
        }
    }.
\label{eq:unconditioned-second-moment}
\end{align}
\end{theorem}
\begin{proof}
    We follow the same steps as in the proof of \cref{l:1mt}. 
    Let $z=e^\beta$. Then
    Stirling-approximations give
    \begin{align}
        &\binom{n}{n\typeProportions}
        \frac{
            \prod_a(\alpha_aN)!
        }{
            (N-1)!!\,N^{\frac{N}{2}}
        }
        \nonumber\\
        &\qquad=
        \left(
            1+O\left(\frac{1}{n}\right)
        \right)
        \sqrt{\pi}\,d^2n^{\frac{1}{2}}
        \exp\left(
            n(1-d)H(\typeProportions)-\frac{N}{2}
        \right),
    \label{eq:second-moment-stirling}
    \end{align}
    uniformly in a fixed neighborhood of $\typeProportions_\star$. Combining
    \eqref{eq:second-moment-exact-sum},
    \eqref{eq:second-moment-contour-asymptotic}, and
    \eqref{eq:second-moment-stirling} gives
    \begin{equation}
    \label{eq:second-moment-pre-laplace}
    \begin{aligned}
        \E_M\left[
            Z_{M,<\varepsilon}^{(2)}(\beta)
        \right]
        &=
        \frac{1}{2(\pi n)^{\frac{3}{2}}}
        \sum_{\typeProportions\in\mathcal A_{n,\varepsilon}^{(2)}}
        \frac{
            \exp\left(
                -n \mathcal{F}_{N,z}(\typeProportions)
            \right)
        }{
            \sqrt{
                \det\left(
                    \nabla_{\contourVariables}^2
                    \Psi_{\typeProportions,z}^{(N)}
                    \left(
                        \contourVariables^{(N)}(\typeProportions,z)
                    \right)
                \right)
            }
        }\\
        &\qquad\times
        \left(
            1+O\left(\frac{1}{n}\right)
        \right).
    \end{aligned}
    \end{equation}
    Let us denote by
    \[
        a_{N,z}(\vec\alpha)=
        \left(
                \det\left(
                    \nabla_{\contourVariables}^2
                    \Psi_{\typeProportions,z}^{(N)}
                    \left(
                        \contourVariables^{(N)}(\typeProportions,z)
                    \right)
                \right)
        \right)^{-1/2},
    \]
    and note that by~\cref{eq:2nd m formula Hessian Psi} this is uniformly bounded away from $0$ and $\infty$ for $\varepsilon$ small enough, since the stationary point depends continuously on $\typeProportions.$

By the previous lemma the Hessian of $\mathcal{F}_{N,\beta}(\typeProportions)$ is uniformly negative definite, therefore the contribution to
\eqref{eq:second-moment-pre-laplace} from $\typeProportions$ satisfying
\[
    \left\|
        \typeProportions-\typeProportions_\star
    \right\|_\infty
    >
    n^{-\frac{2}{5}}
\]
are exponentially small relative to the main contribution.

On the remaining region, write
\[
    \typeProportions
    =
    \typeProportions_\star+\frac{\vec{t}}{\sqrt{n}}.
\]
Taylor's theorem gives
\begin{align}
    n\left(
        \mathcal F_{N,\beta}(\typeProportions)
        -
        \mathcal F_{N,\beta}(\typeProportions_\star)
    \right)
    &=
    \frac{1}{2}
    \vec{t}^{\,\matrixtranspose}
    \nabla_{\typeProportions}^2
    \mathcal F_{N,\beta}(\typeProportions_\star)
    \vec{t}
    +
    \frac{1}{\sqrt{n}}P_3(\vec{t})
    \nonumber\\
    &\qquad+
    O\left(
        \frac{1+\|\vec{t}\|^4}{n}
    \right),
\label{eq:second-moment-outer-rate-expansion}
\end{align}
where $P_3$ is homogeneous of degree three. Similarly,
\begin{equation}
\label{eq:second-moment-outer-amplitude-expansion}
     a_{N,z}(\typeProportions)
    =
     a_{N,z}(\typeProportions_\star)
    +
    \frac{1}{\sqrt{n}}L(\vec{t})
    +
    O\left(
        \frac{1+\|\vec{t}\|^2}{n}
    \right),
\end{equation}
where $L$ is linear. The complete term of order $n^{-\frac{1}{2}}$ obtained from
\cref{eq:second-moment-outer-rate-expansion} and
\cref{eq:second-moment-outer-amplitude-expansion} is odd.

Write
\[
    A_{N,\beta}:=- \nabla_{\typeProportions}^2
    {\mathcal F}_{N,\beta}(\typeProportions^\star).
\]
In the three independent coordinates, the possible values of
$\vec t=\sqrt n(\typeProportions-
\typeProportions^\star)$ form the affine lattice
\[
    \Lambda_n
    =
    \frac1{\sqrt n}
    \left(
        \mathbb Z^3-\frac n4\vec 1
    \right).
\]
Combining the expansions above, we obtain
\begin{align*}
    &a_{N,\beta}(\typeProportions)
    \exp\left(
        n\bigl(
            {\mathcal F}_{N,\beta}(\typeProportions)
            -
            {\mathcal F}_{N,\beta}(\typeProportions^\star)
        \bigr)
    \right)
    \\
    &\quad =
    e^{-\frac12\vec t^{\mathsf T}
        A_{N,\beta}\vec t}
    \left[
        a_{N,\beta}(\typeProportions^\star)
        +
        \frac1{\sqrt n}R_{N,\beta}(\vec t)
        +
        O\left(
            \frac{1+\|\vec t\|^6}{n}
        \right)
    \right],
\end{align*}
where $R_{N,\beta}$ is an odd polynomial of degree at most
three, with uniformly bounded coefficients.

Applying \cref{cor:poisson-affine-lattice-truncated} with
$\gamma=1/10$ and $P\equiv1$ gives
\[
    \frac1{n^{3/2}}
    \sum_{\substack{\vec t\in\Lambda_n\\
                    \|\vec t\|_\infty\leq n^{1/10}}}
    e^{-\frac12\vec t^{\mathsf T}
        A_{N,\beta}\vec t}
    =
    \frac{(2\pi)^{3/2}}
         {\sqrt{\det A_{N,\beta}}}
    +
    O(e^{-c n^{1/5}}).
\]
Applying the same corollary with $P=R_{N,\beta}$, and using
that $R_{N,\beta}$ is odd, gives
\[
    \frac1{n^{3/2}}
    \sum_{\substack{\vec t\in\Lambda_n\\
                    \|\vec t\|_\infty\leq n^{1/10}}}
    R_{N,\beta}(\vec t)
    e^{-\frac12\vec t^{\mathsf T}
        A_{N,\beta}\vec t}
    =
    O(e^{-c n^{1/5}}).
\]
Finally, the polynomially weighted Gaussian estimate in the
same corollary shows that the remaining terms contribute a
relative error $O(n^{-1})$. 
Consequently,
\begin{align}
    &\sum_{\typeProportions \in A_{n,\varepsilon}^{(2)}}
    a_{N,\beta}(\typeProportions)
    e^{nF_{N,\beta}(\typeProportions)} \label{eq:unconditioned type summation}
    \\
    &\quad =
    n^{3/2}
    \frac{(2\pi)^{3/2}
          a_{N,\beta}(\typeProportions^\star)
          e^{nF_{N,\beta}(\typeProportions^\star)}}
         {\sqrt{
             \det\bigl(
                 -\nabla_{\typeProportions}^2
                 F_{N,\beta}(\typeProportions^\star)
             \bigr)
         }}
    \left(1+O(n^{-1})\right).\nonumber
\end{align}
By \eqref{eq:second-moment-pre-laplace} we have
\begin{equation}
\label{eq:second-moment-inner-determinant-center}
    \det\left(
        \nabla_{\contourVariables}^2
        \Psi_{\typeProportions_\star,\beta}^{(N)}
        \left(
            \lambda_\beta^{(N)}\vec{1}
        \right)
    \right)
    =
    16h(\beta)^2(1+pq).
\end{equation}
Moreover,
\begin{equation}
\label{eq:second-moment-rate-center}
    \exp\left(
        n\mathcal F_{N,\beta}(\typeProportions_\star)
    \right)
    =
    (4h(\beta))^2
    4^n
    \left(
        \frac{h(\beta)}{4}
    \right)^{\frac{N}{2}}
    \left(
        1+O\left(\frac{1}{n}\right)
    \right).
\end{equation}
Combining these two expressions finishes the proof.    
\end{proof}

\subsection{Contour integral for conditioning on a fixed matching \texorpdfstring{$F$}{F}}
We now extend the preceding second-moment calculation to the setting where we condition on a fixed partial matching $F\subseteq {\bf M}$. The main point is that, after conditioning on $F$, the remaining unmatched sockets are still paired uniformly at random. Thus the same contour-integral calculation as in the unconditioned case continues to apply, but with slightly shifted type counts.

Let $F$ be a partial matching. Put
\[
    g:=|F|,
    \qquad
    v:=|V(F)|,
    \qquad
    V:=[n]\setminus V(F).
\]
Fix boundary assignments
\[
    \sigma^F,{\sigma'}^F\in\{-1,1\}^{V(F)}.
\]
For $a\in\{++,+-,-+,--\}$, let $m_a^F$ be the number of clouds in $V(F)$ of type $a$, and let $s_a^F$ be the number of sockets of type $a$ already matched by $F$. Thus
\[
    \sum_a m_a^F=v,
    \qquad
    \sum_a s_a^F=2g.
\]
Define
\[
    b_a^F:=m_a^F-\frac{s_a^F}{d},
    \qquad
    B_F:=\sum_a b_a^F=v-\frac{2g}{d}.
\]

For an assignment $(\sigma,\sigma')$ on $V$, let $m_a$ be the number of clouds of type $a$. The number of unmatched sockets of type $a$ is
\[
    d(m_a+m_a^F)-s_a^F
    =
    d(m_a+b_a^F).
\]
Define
\[
    m_a':=m_a+b_a^F,
    \qquad
    N':=N-2g,
    \qquad
    n':=\frac{N'}{d}=n-\frac{2g}{d}.
\]
Then
\[
    \sum_a m_a'=n'.
\]
Although $n'$ need not be an integer, each $dm_a'$ and $N'$ is an integer.

Next, write
\[
    \typeProportions:=\frac{\momentCounts'}{n'},
\]
and define
\begin{equation}
\label{eq:conditioned-coefficient}
    f_{N'}(\beta,\momentCounts')
    =
    \frac{
        \prod_a(\alpha_aN')!
    }{
        (N'-1)!!\,(N')^{\frac{N'}{2}}
    }
    \frac{1}{(2\pi i)^4}
    \oint
    \exp\left(
        N'\Psi_{\typeProportions,\beta}^{(N')}(\contourVariables)
    \right)
    \,\differential\contourVariables.
\end{equation}
Since after conditioning on $F$, the unmatched sockets are paired by a uniform perfect matching, the proof of~\cref{lem:2nd moment contour representation} immediately gives the following.

\begin{lemma}\label{lem:conditioned 2nd moment contour representation}
With notation as above we have 
\begin{align}
    &\E_M\left[
        Z_{M,<\varepsilon}^{(2)}(\beta)
        \,\middle|\,
        F\subseteq M
    \right]
    \nonumber\\
    &\qquad=
    \sum_{\sigma^F,{\sigma'}^F}
    \exp\left(
        \beta \isingHamiltonian_{F}(\sigma^F)
        +
        \overline{\beta}\isingHamiltonian_{F}({\sigma'}^F)
    \right)
    \mathcal S_F(\sigma^F,{\sigma'}^F),
\label{eq:conditioned-boundary-sum}
\end{align}
where
\begin{equation}
\label{eq:conditioned-fixed-boundary}
    \mathcal S_F(\sigma^F,{\sigma'}^F)
    :=
    \sum_{\momentCounts'}
    \binom{n-v}{\momentCounts'-\vec{b}^{\,F}}
    f_{N'}(\beta,\momentCounts'),
\end{equation}
restricted to those $\vec{m'}$ satisfying the truncation condition.
\end{lemma}

For our proof of~\cref{prop:2nd-moment} at the level of the contour integral, no new analysis is required: we may reuse the uniform contour estimate proved above with $N'$ in place of $N$.
The only difference from the unconditioned second moment is now in the multinomial coefficient. 
Indeed, the contour term depends naturally on the effective proportions $\typeProportions=\momentCounts'/n'$, whereas the number of spin assignments on $V$ is determined by the proportions on the actual unexposed vertex set $V$. 
The next subsection deals with this.

\subsection{Proof of \cref{prop:2nd-moment}}
We start by estimating  $\mathcal{S}_F$ for a fixed matching $F$.
To this end we introduce some notation.

Let $L:=n-v$ and define
\[
    \typeProportions^V
    :=
    \frac{
        \momentCounts'-\vec{b}^{\,F}
    }{
        L
    }.
\]
Since $n'=L+B_F$, we have
\begin{equation}
\label{eq:boundary-alpha-relation}
    \typeProportions^V-\typeProportions
    =
    \frac{
        B_F\typeProportions-\vec{b}^{\,F}
    }{
        L
    }.
\end{equation}
In particular,
\[
    \left\|
        \typeProportions^V-\typeProportions
    \right\|_\infty
    =
    O\left(\frac{g}{n}\right)
\]
uniformly over the boundary assignments.

Define
\begin{equation}
\label{eq:boundary-entropy-correction}
    \Delta_F(\typeProportions)
    :=
    LH(\typeProportions^V)-n'H(\typeProportions)+B_F\log4
\end{equation}
and
\begin{equation}
\label{eq:boundary-amplitude-correction}
    R_F(\typeProportions)
    :=
    \left(
        \frac{n'}{L}
    \right)^{\frac{3}{2}}
    \prod_a
    \left(
        \frac{\alpha_a}{\alpha_a^V}
    \right)^{\frac{1}{2}}.
\end{equation}

\begin{lemma}
\label{lem:boundary-shifted-laplace}
Fix $0<\eta<1$. There exist $\rho=\rho(d,\eta)>0$, $\varepsilon_0=\varepsilon_0(d,\eta)>0$, and $C=C(d,\eta)>0$ such that the following holds uniformly for $\beta\in K_\eta(\rho)$, $0<\varepsilon\leq\varepsilon_0$, every partial matching $F$ of size $g=O(\log n)$, and every pair of boundary assignments $\sigma^F,{\sigma'}^F$, we have
\begin{align}
    \mathcal S_F(\sigma^F,{\sigma'}^F)
    &=
    \left(
        1+
        C\left(
            \frac{(1+g)^2}{n}
        \right)
    \right)
    4^{n-v}
    \left(
        \frac{h(\beta)}{4}
    \right)^{\frac{N-2g}{2}}
    \frac{
        \sqrt{h(\beta)}
    }{
        \sqrt{
            \left(d(p-1)+2\right)
            \left(d(q-1)+2\right)
        }
    }
    \nonumber\\
    &\qquad\times
    \frac{1}{
        \sqrt{
            1-(d-1)\delta_2(\beta)
        }
    }.
\label{eq:conditioned-bulk-estimate}
\end{align}
\end{lemma}

\begin{proof}
Uniform Stirling estimates give
\begin{align}
    \binom{L}{\momentCounts'-\vec{b}^{\,F}}
    &=
    4^{-B_F}
    R_F(\typeProportions)
    \exp\left(
        \Delta_F(\typeProportions)
    \right)
    \frac{
        \exp\left(
            n'H(\typeProportions)
        \right)
    }{
        (2\pi n')^{\frac{3}{2}}
        \prod_a\alpha_a^{\frac{1}{2}}
    }
    \left(
        1+O\left(\frac{1+g}{n}\right)
    \right).
\label{eq:boundary-stirling}
\end{align}
A Taylor expansion, using
\eqref{eq:boundary-alpha-relation}, gives
\begin{equation}
\label{eq:boundary-entropy-center}
    \Delta_F(\typeProportions_\star)
    =
    O\left(\frac{g^2}{n}\right).
\end{equation}
Moreover, uniformly on a fixed neighborhood of $\typeProportions_\star$,
\begin{equation}
\label{eq:boundary-entropy-derivatives}
    \left|
        \frac{\partial}{\partial\typeProportions_a} \Delta_F(\typeProportions)
    \right|
    \leq
    C_a g,
    \qquad
    a\in \{--,-+,+-,++\}.
\end{equation}
Similarly,
\begin{equation}
\label{eq:boundary-amplitude-bound}
    R_F(\typeProportions)
    =
    1+O\left(\frac{g}{n}\right),
\end{equation}
and all derivatives of $R_F$ of order at most four are $O(g/n)$.

Substituting \eqref{eq:conditioned-coefficient} and
\eqref{eq:boundary-stirling} into
\eqref{eq:conditioned-fixed-boundary}, and using the contour estimate with $N'$ in place of $N$, gives the same three-dimensional Laplace sum as in the unconditioned calculation~\cref{eq:unconditioned type summation}, multiplied by the factor
\[
     4^{-B_F}
    R_F(\typeProportions)
    \exp\left(
        \Delta_F(\typeProportions)
    \right).   
\]
The replacement of $N$ and $n$ by $N'$ and $n'$ changes the uniform constants only by
\[
    1+O\left(\frac{g}{n}\right).
\]
The contribution outside the $n^{-\frac{2}{5}}$ neighborhood of $\typeProportions_\star$ is exponentially small, uniformly in the boundary assignments provided $g$ is small enough; $g=O(\log n)$ clearly suffices. 
For $\typeProportions$ contained in a $n^{-\frac{2}{5}}$ neighborhood of $\typeProportions_\star$ we rescale 
\[
    \typeProportions
    =
    \typeProportions_\star+\frac{\vec{t}}{\sqrt{n'}},
\]
with $\|\vec t\|_\infty\leq n^{1/10}$.
Equations
\eqref{eq:boundary-entropy-center},
\eqref{eq:boundary-entropy-derivatives}, and
\eqref{eq:boundary-amplitude-bound} imply
\begin{align}
    R_F(\typeProportions)
    \exp\left(
        \Delta_F(\typeProportions)
    \right)
    &=
    1
    +
    \frac{1}{\sqrt{n'}}L_F(\vec{t})
    \nonumber\\
    &\qquad+
    O\left(
        \frac{
            (1+g)^2
            \left(
                1+\|\vec{t}\|^{2}
            \right)
        }{n}
    \right),
\label{eq:boundary-local-expansion}
\end{align}
where $L_F$ is linear and its coefficients are $O(g)$. The order $(n')^{-\frac{1}{2}}$ term in the unconditioned Laplace expansion is also odd. Hence the complete term of order $(n')^{-\frac{1}{2}}$ is odd. 

In the three independent coordinates, the possible values of
\[
    \vec t
    =
    \sqrt{n'}
    \bigl(
        \typeProportions- \typeProportions^\star
    \bigr)
\]
form the affine lattice
\[
    \Lambda_{n',F}
    :=
    \frac1{\sqrt{n'}}
    \left(
        \mathbb Z^3+\vec b^F_{[3]}
        -\frac{n'}4\vec 1
    \right).
\]
Write
\[
    A_{N',\beta}
    :=
    -\nabla_{\typeProportions}^2
    F_{N',\beta}(\typeProportions^\star).
\]

Combining the unconditioned Laplace expansion with
\cref{eq:boundary-local-expansion}, we obtain, uniformly
for $\|\vec t\|_\infty\leq n^{1/10}$,
\begin{align*}
    &a_{N',\beta}(\typeProportions)
    R_F(\typeProportions)e^{\Delta_F(\typeProportions)}
    \exp\left(
        n'\bigl(
            {\mathcal F}_{N',\beta}(\typeProportions)
            -
            {\mathcal F}_{N',\beta}(\typeProportions^\star)
        \bigr)
    \right)
    \\
    &\quad =
    e^{-\frac12\vec t^{\mathsf T}
        A_{N',\beta}\vec t}
    \left[
        a_{N',\beta}(\typeProportions^\star)
        +
        \frac1{\sqrt{n'}}
        R_{N',\beta,F}(\vec t)
        +
        O\left(
            \frac{(1+g)^2(1+\|\vec t\|^6)}{n}
        \right)
    \right],
\end{align*}
where $R_{N',\beta,F}$ is an odd polynomial of degree at most
three whose coefficients are $O(1+g)$.

Applying \cref{cor:poisson-affine-lattice-truncated} with
scaling parameter $n'$, $\gamma=1/10$, and $P\equiv1$, gives
\[
    \frac1{(n')^{3/2}}
    \sum_{\substack{\vec t\in\Lambda_{n',F}\\
                    \|\vec t\|_\infty\leq n^{1/10}}}
    e^{-\frac12\vec t^{\mathsf T}
        A_{N',\beta}\vec t}
    =
    \frac{(2\pi)^{3/2}}
         {\sqrt{\det A_{N',\beta}}}
    +
    O(e^{-cn^{1/5}}).
\]
Since $R_{N',\beta,F}$ is odd, another application of the
corollary gives
\[
    \frac1{(n')^{3/2}}
    \sum_{\substack{\vec t\in\Lambda_{n',F}\\
                    \|\vec t\|_\infty\leq n^{1/10}}}
    R_{N',\beta,F}(\vec t)
    e^{-\frac12\vec t^{\mathsf T}
        A_{N',\beta}\vec  t}
    =
    O\left((1+g)e^{-cn^{1/5}}\right).
\]
Finally, the polynomially weighted Gaussian estimate in the
same corollary shows that the remaining terms contribute a
relative error
\[
    O\left(\frac{(1+g)^2}{n}\right).
\]
Consequently, the conditioned central sum is equal to
$4^{-B_F}$ times the unconditioned central expression with
$n'$ and $N'$ in place of $n$ and $N$, up to a relative error
$O((1+g)^2/n)$.

Thus the conditioned sum equals $4^{-B_F}$ times the unconditioned central expression with $n'$ and $N'$ in place of $n$ and $N$, up to the asserted relative error. Since
\[
    4^{-B_F}4^{n'}=4^{n-v},
\]
this gives \eqref{eq:conditioned-bulk-estimate}.
\end{proof}

With the first moment calculation we can now prove \cref{prop:2nd-moment}
\begin{proof}[Proof of \cref{prop:2nd-moment}]
We assume that $\varepsilon>0$ and $\eta>0$ are small enough as demanded by the previous lemma as well by~\cref{l:1mt}.

The full conditioned second moment is given by
\begin{align}
    &\E_M\left[
        Z_{M,<\varepsilon}^{(2)}(\beta)
        \,\middle|\,
        F\subseteq M
    \right]
    \nonumber\\
    &\qquad=
    \sum_{\sigma^F,{\sigma'}^F}
    \exp\left(
        \beta \isingHamiltonian_{F}(\sigma^F)
        +
        \overline{\beta}\isingHamiltonian_{F}({\sigma'}^F)
    \right)
    \mathcal S_F(\sigma^F,{\sigma'}^F).
\label{eq:conditioned-boundary-sum}
\end{align}
The main term in ~\cref{lem:boundary-shifted-laplace} is independent of the boundary assignments. Moreover,
\begin{align}
    &\sum_{\sigma^F,{\sigma'}^F}
    \exp\left(
        \beta \isingHamiltonian_{F}(\sigma^F)
        +
        \overline{\beta}\isingHamiltonian_{F}({\sigma'}^F)
    \right)
    \nonumber\\
    &\qquad=
    Z_F(\beta)\overline{Z_F(\beta)}
    =
    |Z_F(\beta)|^2.
\label{eq:boundary-partition-function}
\end{align}

To control the error without assuming a lower bound on $|Z_F(\beta)|$, define
\[
    Z_F^{\mathrm{abs}}(\beta)
    :=
    \sum_{\sigma^F}
    \exp\left(
        \re(\beta)\isingHamiltonian_{F}(\sigma^F)
    \right).
\]
Then
\[
    \sum_{\sigma^F,{\sigma'}^F}
    \left|
        \exp\left(
            \beta \isingHamiltonian_{F}(\sigma^F)
            +
            \overline{\beta}\isingHamiltonian_{F}({\sigma'}^F)
        \right)
    \right|
    =
    \left(
        Z_F^{\mathrm{abs}}(\beta)
    \right)^2.
\]
Since $v\leq2g$ and $\beta$ ranges over the compact set $K_\eta(\rho)$, there exists $C_\eta>1$ such that
\begin{equation}
\label{eq:boundary-absolute-bound}
    4^{-v}
    \left|
        \frac{h(\beta)}{4}
    \right|^{-g}
    \left(
        Z_F^{\mathrm{abs}}(\beta)
    \right)^2
    \leq
    C_\eta^g.
\end{equation}

Combining
\eqref{eq:conditioned-bulk-estimate},
\eqref{eq:conditioned-boundary-sum},
\eqref{eq:boundary-partition-function}, and
\eqref{eq:boundary-absolute-bound}, we obtain
\begin{align}
    &\E_M\left[
        Z_{M,<\varepsilon}^{(2)}(\beta)
        \,\middle|\,
        F\subseteq M
    \right]
    \nonumber\\
    &\qquad=
    4^{n-v}
    \left(
        \frac{h(\beta)}{4}
    \right)^{\frac{N-2g}{2}}
    \frac{
        \sqrt{h(\beta)}
    }{
        \sqrt{
            \left(d(p-1)+2\right)
            \left(d(q-1)+2\right)
        }
    }
    \frac{
        |Z_F(\beta)|^2
    }{
        \sqrt{
            1-(d-1)\delta_2(\beta)
        }
    }
    +
    \mathcal R_F,
\label{eq:conditioned-before-normalization}
\end{align}
where
\begin{equation}
\label{eq:conditioned-error-bound}
    |\mathcal R_F|
    \leq
    C_\eta^g
    \frac{(1+g)^2}{n}
    4^n
    \left|
        \frac{h(\beta)}{4}
    \right|^{\frac{N}{2}}
    \frac{
        \sqrt{h(\beta)}
    }{
        \sqrt{
            \left(d(p-1)+2\right)
            \left(d(q-1)+2\right)
        }
    }.
\end{equation}

With the first moment calculation we can now obtain the following

\begin{align}
    \frac{
        \E_M\left[
            Z_{M,<\varepsilon}^{(2)}(\beta)
            \,\middle|\,
            F\subseteq M
        \right]
    }{
        \left|
            \E_M\left[
                Z_{M,<\varepsilon}(\beta)
            \right]
        \right|^2
    }
    &=
    \frac{
        |\widetilde Z_F(\beta)|^2
    }{
        \sqrt{
            1-(d-1)\delta_2(\beta)
        }
    }
    +
    O\left(
        \frac{
            \mathcal E(\eta,m)
        }{n}
    \right),
\end{align}
where one may take
\[
    \mathcal E(\eta,m)
    :=
    C_\eta^m(1+m)^2.
\]
This finishes the proof.
\end{proof}

\section{Broadcast analysis of augmented planted model}
\label{s:planted-augmented}

In this appendix, we will prove ~\cref{l:planted-augmented}.
Our argument will reduce to analyzing the following broadcast problem on a tree.
This is done by a slight adaptation of an argument of \cite{BRZ95}.
\begin{definition}[Antiferromagnetic Ising broadcast model with side information]\label{d:antiferro-broadcast}
  Let $d\ge 3$, $\beta \in \R$, $t \ge 0$, and $R \in \N$.
  Let $\T$ be a tree with root $o$ and depth $R$, where every vertex of depth $<R$ has degree $d$.
  For $0\le r\le R$, let $V_r$ be the set of vertices of depth $r$.
  For each $u \in V_r$ where $r<R$, write $S(u)$ for the set of children of $u$.
  Let $\nu^{\broadcastTag,R}_{\beta,t}$ denote the law of $(\sigma_\T, {\boldsymbol{v}}_\T) = (\sigma_u, v_u : u \in \T)$ generated as follows:
  \begin{itemize}
    \item Sample $\sigma_o \sim \unif(\{\pm 1\})$.
    \item For each $u \in V_r$ where $r<R$, conditional on $\sigma_u$, sample the spins of its children $(\sigma_w : w \in S(u))$ IID from the distribution
    \[
      \Pr(\sigma_w = \sigma_u) = \frac{1}{1+e^\beta}\,, \qquad 
      \Pr(\sigma_w = -\sigma_u) = \frac{e^\beta}{1+e^\beta}\,.
    \]
    \item For each $u\in \T$, sample $v_u = \sqrt{t} \sigma_u + Z_u$, for $Z_u \sim {\mathcal{N}}(0,1)$ independent of everything else.
  \end{itemize}
  Let $\E^{\broadcastTag,R}_{\beta,t}$ denote expectation with respect to $\nu^{\broadcastTag,R}_{\beta,t}$.
  Further define ${\mathbb{L}} = V_R$ and $\sigma_{\mathbb{L}} = (\sigma_u : u\in {\mathbb{L}})$.
  Note that $\T,{\mathbb{L}}$ both depend on $d,R$ but we suppress this from the notation.
\end{definition}
We will interpret this model as the problem of inferring the root spin $\sigma_o$ from $(\sigma_{\mathbb{L}}, {\boldsymbol{v}}_\T)$.
\begin{lemma}\label{l:broadcast}
  For all $\varepsilon_2 > 0$, there exists $t_0 = t_0(d,\delta,\varepsilon_2) > 0$ and $R_0 = R_0(d,\delta,\varepsilon_2) > 0$ such that for all $\beta \in [2\beta_\lowerBoundTag, \betaKS - \delta]$, $t \in [0,t_0]$, and $R \ge R_0$ the following holds.
  Let $(\sigma_\T, {\boldsymbol{v}}_\T) \sim \nu^{\broadcastTag,R}_{\beta,t}$.
  Then,
  \[
    \E^{\broadcastTag,R}_{\beta,t} \Big[
      \,\big|\E^{\broadcastTag,R}_{\beta,t}[\sigma_o | \sigma_{\mathbb{L}}, {\boldsymbol{v}}_\T]\big|\,
    \Big] \le \varepsilon_2\,.
  \]
\begin{proof}
  In this proof we write $\E$ for $\E^{\broadcastTag,R}_{\beta,t}$.
  For vertex $u\in \T$, let 
  \[
    \observedData(u) = \Big(
      (\sigma_w : w \in {\mathbb{L}}, w\,\text{descendant of}\,u), 
      (v_w : w \in \T, w\,\text{descendant of}\,u)
    \Big)
  \]
  denote the part of $(\sigma_{\mathbb{L}}, {\boldsymbol{v}}_\T)$ in the descendant subtree of $u$ (including $v_u$).
  Note that $\observedData(o) = (\sigma_{\mathbb{L}}, {\boldsymbol{v}}_\T)$.
  Define the belief propagation message
  \[
    h_u = \frac12 \log \frac{
      \Pr(\sigma_u = 1 | \observedData(u))
    }{
      \Pr(\sigma_u = -1 | \observedData(u))
    }\,,
  \]
  so that $\E[\sigma_u | \observedData(u)] = \tanh(h_u)$.
  Note that for $u\in V_R$, we have $h_u = \sigma_u \cdot \infty$ because $\sigma_u$ is revealed.
  For any non-leaf $u$,
  \[
    \frac{
      \Pr(\sigma_u = 1 | \observedData(u))
    }{
      \Pr(\sigma_u = -1 | \observedData(u))
    }
    = \frac{p(v_u | \sigma_u = 1)}{p(v_u | \sigma_u = -1)} 
    \prod_{w\in S(u)} \frac{
      \Pr(\sigma_w = 1 | \observedData(w)) + e^\beta \Pr(\sigma_w = -1 | \observedData(w))
    }{
      e^\beta \Pr(\sigma_w = 1 | \observedData(w)) + \Pr(\sigma_w = -1 | \observedData(w))
    }
  \]
  This can be simplified by
  \[
    \frac{p(v_u | \sigma_u = 1)}{p(v_u | \sigma_u = -1)} 
    = \frac{
      p(Z_u = v_u - \sqrt{t})
    }{
      p(Z_u = v_u + \sqrt{t})
    }
    = \frac{
      \exp(-\frac{1}{2} (v_u - \sqrt{t})^2)
    }{
      \exp(-\frac{1}{2} (v_u + \sqrt{t})^2)
    }
    = e^{2\sqrt{t} v_u}
  \]
  and
  \[
    \frac{
      \Pr(\sigma_w = 1 | \observedData(w)) + e^\beta \Pr(\sigma_w = -1 | \observedData(w))
    }{
      e^\beta \Pr(\sigma_w = 1 | \observedData(w)) + \Pr(\sigma_w = -1 | \observedData(w))
    }
    = \frac{
      e^{-\beta/2} e^{h_w} + e^{\beta/2} e^{-h_w} 
    }{
      e^{\beta/2} e^{h_w} + e^{-\beta/2} e^{-h_w} 
    }\,.
  \]
  Thus
  \[
    h_u
    = \sqrt{t} v_u + \sum_{w\in S(u)} \frac12 \log \frac{
      e^{-\beta/2} e^{h_w} + e^{\beta/2} e^{-h_w} 
    }{
      e^{\beta/2} e^{h_w} + e^{-\beta/2} e^{-h_w} 
    }
    = \sqrt{t} v_u - \sum_{w\in S(u)} \arctanh(\theta \tanh(h_w))\,,
  \]
  for $\theta = \tanh(\beta/2)$.
  Let
  \[
    \tilde h_u = - \sum_{w\in S(u)} \arctanh(\theta \tanh(h_w))\,.
  \]
  Then, \cite[Lemma 4.1]{BRZ95} shows that for each non-root, non-leaf $u$ and any $w\in S(u)$,
  \[
    \E[\sinh^2(\tilde h_u)] \le (d-1) \theta^2 \E[\sinh^2(h_w)].
  \]
  where the expectation is over the randomness of the generative process described in ~\cref{d:antiferro-broadcast}.
  Thus
  \begin{align*}
    \E[\sinh^2(h_u)]
    &= \frac14 \E\left[
      \exp(2h_u) + \exp(-2h_u) - 2
    \right] \\
    &= \frac14 \E\left[
      \exp(2\tilde h_u + 2t \sigma_u + 2\sqrt{t} Z_u) + \exp(- 2\tilde h_u - 2t \sigma_u - 2\sqrt{t} Z_u) - 2
    \right] \\
    &= \frac14 \E\left[
      \exp(2\tilde h_u + 2t \sigma_u + 2t) + \exp(- 2\tilde h_u - 2t \sigma_u + 2t) - 2
    \right] \\
    &\le \frac14 \E\left[
      e^{4t} (\exp(2\tilde h_u) + \exp(- 2\tilde h_u)) - 2
    \right]
    = e^{4t} \E\sinh^2(\tilde h_u) + \frac12 (e^{4t}-1) \\
    &\le e^{4t} (d-1) \theta^2 \E[\sinh^2(h_w)] + (e^{4t}-1)\,.
  \end{align*}
  \cite[Lemma 4.1]{BRZ95} also shows that at the root $o$, and $w\in S(o)$,\footnote{The root requires special handling because it has $d$ children, whereas all other nodes have $d-1$ children. The argument in \cite{BRZ95} avoids this issue by working with a $k$-ary tree (where $k=d-1$) rather than a $d$-regular tree.} 
  \[
    \E[\sinh^2(\tilde h_o)] \le d \theta^2 \E[\sinh^2(h_w)],
  \]
  and arguing similarly as above
  \[
    \E[\sinh^2(h_o)] \le e^{4t} d \theta^2 \E[\sinh^2(h_w)] + (e^{4t}-1)\,.
  \]
  Let $x_r$ be the value of $\E[\sinh^2(h_u)]$ for $u\in V_r$, which clearly depends on $u$ only through its depth $r$.
  We have thus shown
  \begin{align}
    \nonumber
    x_0 &\le e^{4t} d\theta^2 x_1 + (e^{4t}-1)\,, \\
    \label{e:sinh-recursion}
    x_r &\le e^{4t} (d-1)\theta^2 x_{r+1} + (e^{4t}-1)\,, \qquad \text{for}\, 1\le r\le R-2\,.
  \end{align}
  We next bound $x_{R-1}$ to obtain an initial condition for this recursion.
  (Note that we cannot start from $x_R = \infty$.)
  For any $u\in V_{R-1}$, almost surely
  \[
    |\tilde h_u| \le \sum_{w\in S(u)} |\arctanh(\theta)| = \frac{(d-1) |\beta|}{2} \equiv h_{\upperBoundTag}\,.
  \]
  Thus,
  \begin{align*}
    x_{R-1} &= \E[\sinh^2(h_u)]
    = \frac14 \E[e^{2\tilde h_u + 2t\sigma_u + 2\sqrt{t}Z_u} + e^{-2\tilde h_u - 2t\sigma_u - 2\sqrt{t}Z_u} - 2] \\
    &\le \frac14 \E[(e^{2\tilde h_u} + e^{-2\tilde h_u}) e^{4t}]
    \le e^{2h_{\upperBoundTag} + 4t}
    \equiv x_{R-1,\upperBoundTag}\,.
  \end{align*}
  Since $\beta \in [2\beta_\lowerBoundTag, \betaKS - \delta]$, and $|2\beta_\lowerBoundTag| < \betaKS$, there exists $\iota = \iota(d,\delta) > 0$ such that
  \[
    \theta^2 = \tanh^2(\beta/2) 
    \le \tanh^2(\betaKS/2)(1-2\iota)
    = \frac{1-2\iota}{d-1}\,.
  \]
  The last equality follows from the definition $\betaKS = 2 \arctanh(\frac{1}{\sqrt{d-1}})$.
  Thus, there exists $t_{0,+} = t_{0,+}(d,\delta)$ such that for all $t\in [0,t_{0,+}]$,
  \[
    e^{4t} (d-1)\theta^2 \le 1-\iota\,.
  \]
  Iteratively applying \eqref{e:sinh-recursion} yields
  \[
    x_1 \le (e^{4t}-1) \sum_{\ell = 0}^{R-3} (1-\iota)^\ell + (1-\iota)^{R-2} x_{R-1,\upperBoundTag}
    \le \frac{e^{4t}-1}{\iota} + (1-\iota)^{R-2} x_{R-1,\upperBoundTag}\,.
  \]
  Note that
  \[
    e^{4t} d\theta^2 \le \frac{d}{d-1} (1-\iota) \le 2\,,
  \]
  so
  \[
    x_0 \le 2x_1 + (e^{4t}-1)
    \le \frac{3(e^{4t}-1)}{\iota} + 2(1-\iota)^{R-2} x_{R-1,\upperBoundTag}\,.
  \]
  We then set $t_0$ and $R_0$, depending on only $d,\delta,\varepsilon_2$, such that $t_0 \le t_{0,+}$ and for all $t\in [0, t_0]$, $R \ge R_0$,
  \[
    \max\left(
      \frac{3(e^{4t}-1)}{\iota},
      2(1-\iota)^{R-2} x_{R-1,\upperBoundTag}
    \right) \le (\varepsilon_2)^2/2\,.
  \]
  This ensures $x_0 \le (\varepsilon_2)^2$. Then,
  \[
    \E [|\E[\sigma_o | \observedData(o)]|]^2
    \le \E [\E[\sigma_o | \observedData(o)]^2]
    = \E[\tanh^2(h_o)]
    \le \E[\sinh^2(h_o)]
    = x_0 \le (\varepsilon_2)^2\,,
  \]
  as desired.
\end{proof}
\end{lemma}
We will also need the following two-replica version of ~\cref{d:antiferro-broadcast}.
\begin{definition}[Two-replica broadcast model]\label{d:antiferro-broadcast-two-replica}
  Let $\nu^{\broadcastTag,2,R}_{\beta,t}$ denote the law of $(\sigma_\T, {\sigma'}_\T, {\boldsymbol{v}}_\T)$ where $(\sigma_\T, {\boldsymbol{v}}_\T) \sim \nu^{\broadcastTag,R}_{\beta,t}$ and ${\sigma'}_\T$ is an independent copy of $\sigma_\T$.
  Let $\E^{\broadcastTag,2,R}_{\beta,t}$ denote expectation with respect to $\nu^{\broadcastTag,2,R}_{\beta,t}$.
\end{definition}
We interpret this model as the problem of inferring $\sigma_o$ from $(\sigma_{\mathbb{L}}, {\sigma'}_\T, {\boldsymbol{v}}_\T)$.
Of course, ${\sigma'}_\T$ gives no information about $\sigma_o$ so the following corollary is immediate.
\begin{corollary}\label{c:broadcast}
  For all $\varepsilon_2 > 0$, there exists $t_0 = t_0(d,\delta,\varepsilon_2) > 0$ and $R_0 = R_0(d,\delta,\varepsilon_2) > 0$ such that for all $\beta \in [2\beta_\lowerBoundTag, \betaKS - \delta]$, $t \in [0,t_0]$, and $R \ge R_0$ the following holds.
  Let $(\sigma_\T, {\sigma'}_\T, {\boldsymbol{v}}_\T)$ be generated as in ~\cref{d:antiferro-broadcast-two-replica} with parameters $d,\beta,t,R$.
  Then,
  \[
    \E [|\E[\sigma_o | \sigma_{\mathbb{L}}, {\sigma'}_\T, {\boldsymbol{v}}_\T]|] \le \varepsilon_2\,.
  \]
\end{corollary}
The proof of ~\cref{l:planted-augmented} given ~\cref{l:broadcast} and Corollary~\cref{c:broadcast} is similar to \cite[Lemmas 4.7--4.9 and Corollary 4.10]{CLMS22}.
For $u$ a vertex in $G$, let $B_R(u)$ be the induced subgraph on $\{w\in G : \dist(w,u) \le R\}$ and ${\overline B}_R(u)$ the induced subgraph on $\{w\in G : \dist(w,u) \ge R\}$.
For $G'$ a subgraph of $G$, let $\sigma_{G'} = (\sigma_u: u\in G')$, which we will interpret as the graph $G'$ labeled by the $\sigma_u$.
We will also interpret a sample $(\sigma^1_\T,{\boldsymbol{v}}^1_\T,\sigma^2_\T,{\boldsymbol{v}}^2_\T) \sim (\nu^{\broadcastTag,R}_{\beta,t})^{\otimes 2}$ as the pair of labeled graphs 
\[
  (\sigma_{\T \cup \T'}, {\boldsymbol{v}}_{\T \cup \T'})
  = \Big((\sigma^1_\T,\sigma^2_\T),({\boldsymbol{v}}^1_\T,{\boldsymbol{v}}^2_\T)\Big)
\] 
where $\T'$ is a disjoint copy of $\T$.
Similarly we will interpret a sample from $(\nu^{\broadcastTag,2,R}_{\beta,t})^{\otimes 2}$ as a triple of labeled graphs $(\sigma_{\T \cup \T'}, {\sigma'}_{\T \cup \T'}, {\boldsymbol{v}}_{\T \cup \T'})$.

Recall the planted models $\nu^\plantedTag_{\beta,t}$, $\nu^{\plantedTag,2}_{\beta,t}$ from Definitions~\cref{d:augmented-planted-model} and \cref{d:two-replica-augmented-planted-model}, and their associated posterior measures $\mu_{G,\beta,t,<\varepsilon}^{{\boldsymbol{v}}}$ and $\mu_{G,\beta,t,<\varepsilon}^{2,{\sigma'}^*,{\boldsymbol{v}}}$.
These definitions all depend on $\varepsilon$: for example, any $(\sigma^*,G,{\boldsymbol{v}})$ in the support of $\nu^\plantedTag_{\beta,t}$ must satisfy $\sigma^* \in \Sigma_n(\varepsilon)$, and any $(\sigma^*,{\sigma'}^*,G,{\boldsymbol{v}})$ in the support of $\nu^{\plantedTag,2}_{\beta,t}$ must satisfy $\sigma^*,{\sigma'}^* \in \Sigma_n(\varepsilon)$, $|R(\sigma^*,{\sigma'}^*)| \le \varepsilon$.
In this subsection, we write these measures as $\nu^\plantedTag_{\beta,t,\varepsilon}$, $\nu^{\plantedTag,2}_{\beta,t,\varepsilon}$, $\mu_{G,\beta,t,<\varepsilon}^{{\boldsymbol{v}}}$, $\mu_{G,\beta,t,<\varepsilon}^{2,{\sigma'}^*,{\boldsymbol{v}}}$ to make this dependence explicit.
Similarly we write $\E^\plantedTag_{\beta,t,\varepsilon}$, $\E^{\plantedTag,2}_{\beta,t,\varepsilon}$ instead of $\E^\plantedTag_{\beta,t}$, $\E^{\plantedTag,2}_{\beta,t}$ for expectation with respect to $\nu^\plantedTag_{\beta,t,\varepsilon}$, $\nu^{\plantedTag,2}_{\beta,t,\varepsilon}$.
The following lemma connects $\nu^\plantedTag_{\beta,t,\varepsilon}$, $\nu^{\plantedTag,2}_{\beta,t,\varepsilon}$ to $\nu^{\broadcastTag,R}_{\beta,t}$, $\nu^{\broadcastTag,2,R}_{\beta,t}$.
\begin{lemma}\label{l:broadcast-tv-error}
  Consider any $R \in \N$ independent of $n$.
  Sample $(\sigma^*,G,{\boldsymbol{v}}) \sim \nu^\plantedTag_{\beta,t,\varepsilon}$ and independent uniformly random vertices $u,w$ of $G$.
  Write $B_R(uw) = B_R(u) \cup B_R(w)$. Then,
  \begin{equation}\label{e:broadcast-tv-error-one-replica}
    \TV\Big(
      \Law(\sigma^*_{B_R(uw)},{\boldsymbol{v}}_{B_R(uw)}),
      (\nu^{\broadcastTag,R}_{\beta,t})^{\otimes 2}
    \Big) = o_n(1)\,.
  \end{equation}
  Likewise, for $(\sigma^*,{\sigma'}^*,G,{\boldsymbol{v}}) \sim \nu^{\plantedTag,2}_{\beta,t,\varepsilon}$ and independent uniformly random vertices $u,w$ of $G$,
  \begin{equation}\label{e:broadcast-tv-error-two-replica}
    \TV\Big(
      \Law(\sigma^*_{B_R(uw)},{\sigma'}^*_{B_R(uw)},{\boldsymbol{v}}_{B_R(uw)}),
      (\nu^{\broadcastTag,2,R}_{\beta,t})^{\otimes 2}
    \Big) = o_n(1)\,.
  \end{equation}
  Finally, $\TV(\nu^\plantedTag_{\beta,t,\varepsilon},\nu^\plantedTag_{\beta,t,\varepsilon/2}) = o_n(1)$ and $\TV(\nu^{\plantedTag,2}_{\beta,t,\varepsilon},\nu^{\plantedTag,2}_{\beta,t,\varepsilon/2}) = o_n(1)$.
\end{lemma}
\begin{remark}
  Note that (e.g.) $\sigma^*_{B_R(uw)}$ is a labeling of the graph $B_R(uw)$, which is itself random and not necessarily isomorphic to $\T \cup \T'$.
  ~\cref{l:broadcast-tv-error} includes the claim that $B_R(uw) \cong \T \cup \T'$ with probability $1-o_n(1)$.
\end{remark}
\begin{proof}
  Let $(\nu^{\broadcastTag,R}_\beta)^{\otimes 2}$ denote the marginal law of $\sigma_{\T \cup \T'}$ under $(\nu^{\broadcastTag,R}_{\beta,t})^{\otimes 2}$, and $(\nu^{\broadcastTag,2,R}_\beta)^{\otimes 2}$ denote the marginal law of $(\sigma_{\T \cup \T'},{\sigma'}_{\T \cup \T'})$ under $(\nu^{\broadcastTag,2,R}_{\beta,t})^{\otimes 2}$.
  We will first prove that
  \begin{align}
    \label{e:broadcast-tv-error-one-replica-no-v}
    \TV\Big(
      \Law(\sigma^*_{B_R(uw)}),
      (\nu^{\broadcastTag,R}_\beta)^{\otimes 2}
    \Big) &= o_n(1)\,, \\
    \label{e:broadcast-tv-error-two-replica-no-v}
    \TV\Big(
      \Law(\sigma^*_{B_R(uw)},{\sigma'}^*_{B_R(uw)}),
      (\nu^{\broadcastTag,2,R}_\beta)^{\otimes 2}
    \Big) &= o_n(1)\,,
  \end{align}
  i.e. that \eqref{e:broadcast-tv-error-one-replica}, \eqref{e:broadcast-tv-error-two-replica} hold if we omit the gaussians ${\boldsymbol{v}}_{B_R(uw)}$ and ${\boldsymbol{v}}_{\T \cup \T'}$.
  We treat the one- and two-replica cases simultaneously.
  For $k\in\{1,2\}$, set $\Sigma_k=\{+1,-1\}^k$ and let $\Delta$ denote Hamming distance on $\Sigma_k$.
  Let $\xi_v=\sigma_v^*$ for $k=1$, and let $\xi_v=(\sigma_v^*,{\sigma'}_v^*)$ for $k=2$.
  For $\omega\in\Sigma_k$, let
  \begin{align*}
    S(\omega) &= \{v\in[n]:\xi_v=\omega\}\,, &
    \alpha_\omega &= |S(\omega)| / n\,.
  \end{align*}
  Let $\gamma_{\omega,\rho}dn$ be the number of half-edges incident to $S(\omega)$ whose matched half-edges are incident to $S(\rho)$.
  Write ${\typeProportions}=(\alpha_\omega)_\omega \in \R^{\Sigma_k}$ and $\Gamma=(\gamma_{\omega,\rho})_{\omega,\rho} \in \R^{\Sigma_k \times \Sigma_k}$.
  Thus $\sum_{\rho \in \Sigma_k} \gamma_{\omega,\rho} = \alpha_\omega$ for every $\omega \in \Sigma_k$.
  For $\omega,\rho\in\Sigma_k$, define
  \[
    P^{(k)}_\beta(\omega,\rho)
    = \frac{e^{\beta\Delta(\omega,\rho)}}{(1+e^\beta)^k}\,.
  \]
  Also let
  \begin{align*}
    \alpha^*_{\omega} &= 2^{-k}\,,&
    \gamma^*_{\omega,\rho} &= 2^{-k}P^{(k)}_\beta(\omega,\rho)\,,
  \end{align*}
  for all $\omega,\rho \in \Sigma_k$, and write ${\typeProportions}^*=(\alpha^*_{\omega})_\omega$ and $\Gamma^*=(\gamma^*_{\omega,\rho})_{\omega,\rho}$.

  Recall from Remark~\cref{r:augmented-planted-interpretation} that (if $k=1$) $(\sigma^*,G) \sim \nu^\plantedTag_{\beta,\varepsilon}$, and (if $k=2$) $(\sigma^*,{\sigma'}^*,G) \sim \nu^{\plantedTag,2}_{\beta,\varepsilon}$.
  ~\cref{l:1mt-balanced}~\cref{i:1mt-balanced-subgaussian} for $k=1$ and ~\cref{l:basic-mt}~\cref{i:basic-mt-annealed-2-subgaussian} for $k=2$ imply that, with $\Pr$ denoting probability under $\nu^\plantedTag_{\beta,\varepsilon}$ or $\nu^{\plantedTag,2}_{\beta,\varepsilon}$, for every fixed $\eta > 0$,
  \[
    \Pr\left(
      \max_{\omega\in\Sigma_k}
      |\alpha_\omega-\alpha^*_\omega| \le \eta
    \right) = 1 - o_n(1)\,.
  \]
  Conditional on ${\typeProportions}$, a standard configuration-model count (see the proof of ~\cref{l:kmt-exp-rate}) and Stirling's formula show that the probability of realizing $\Gamma$ is, up to factors polynomial in $n$, proportional to $\exp(dn\Lambda_\beta(\Gamma))$, where
  \[
    \Lambda_\beta(\Gamma)
    =-\frac12\sum_{\omega,\rho\in\Sigma_k}
      \gamma_{\omega,\rho}\log\gamma_{\omega,\rho}
      +\beta\sum_{(\omega,\rho)\in\binom{\Sigma_k}{2}}
      \gamma_{\omega,\rho}\Delta(\omega,\rho)\,.
  \]
  The functional $\Lambda_\beta$ is strictly concave on the convex set of symmetric nonnegative matrices with row sums ${\typeProportions}$, and its unique maximizer $\Gamma^\beta({\typeProportions})$ is continuous in ${\typeProportions}$ near ${\typeProportions}^*$.
  A standard first-order check shows $\Gamma^\beta({\typeProportions}^*) = \Gamma^*$.
  It follows that
  \[
    \Pr\left(
      \max_{\omega\in\Sigma_k}
      |\alpha_\omega-\alpha^*_\omega|
      + \max_{\omega,\rho\in\Sigma_k}
      |\gamma_{\omega,\rho}-\gamma^*_{\omega,\rho}|
      \le \eta
    \right) = 1 - o_n(1)\,.
  \]
  Fix $\eta > 0$ and condition on the event that the ${\typeProportions}, \Gamma$ associated to $(\sigma^*,G)$ or $(\sigma^*,{\sigma'}^*,G)$ satisfies the event in the last display.
  Condition further on the sets $(S(\omega) : \omega \in \Sigma_k)$.
  Then, the conditional distribution of the edges of $G$ is uniform among all perfect matchings of half-edges consistent with $\Gamma$, i.e. where exactly $\gamma_{\omega,\rho} dn$ half-edges incident to $S(\omega)$ are matched to half-edges incident to $S(\rho)$, for all $\omega,\rho \in \Sigma_k$.

  Let $O_R(1)$ denote a constant that may depend on $R$.
  We will explore $B_R(uv)$ by breadth-first search and record the spins $(\xi_w \in B_R(uv))$ encountered; note that this exploration involves only $O_R(1)$ steps.
  First, $\xi_u$ and $\xi_v$ are independent with distribution ${\typeProportions}$.
  At each step in the BFS exploration, we randomly match a half-edge $h$ incident to a discovered vertex $w$, to a half-edge $h'$ incident to a vertex $w'$.
  Suppose $w$ has spins $\xi_w = \omega$.
  Then, with ${\mathcal{F}}$ denoting any realization of the BFS history up to this point, we have
  \[
    \Pr(\text{$w'$ is undiscovered and $\xi_{w'} = \rho$} | {\mathcal{F}})
    = \frac{\gamma_{\omega,\rho} dn + O_R(1)}{\alpha_\omega dn + O_R(1)}
    = P^{(k)}_\beta(\omega,\rho) + O_R(1) (\eta + n^{-1})\,.
  \]
  We will couple this BFS to a BFS of $\T \cup \T'$.
  For $v\in \T \cup \T'$, similarly let $\xi_v = \sigma_v$ if $k=1$, and $\xi_v = (\sigma_v,{\sigma'}_v)$ if $k=2$.
  Write $o,o'$ for the roots of $\T,\T'$.
  Then, under the laws $\nu^{\broadcastTag,R}_\beta$ or $\nu^{\broadcastTag,2,R}_{\beta}$, $\xi_o$ and $\xi_{o'}$ are independent with distribution ${\typeProportions}^*$.
  By Definitions~\cref{d:antiferro-broadcast} and \cref{d:antiferro-broadcast-two-replica}, at each step in the BFS exploration of $\T \cup \T'$, if a vertex $w$ has spin $\xi_w = \omega$, its children have spins distributed as
  \[
    \Pr(\xi_{w'} = \rho) = P^{(k)}_\beta(\omega,\rho)\,.
  \]
  It follows that we can couple these BFS explorations to coincide with failure probabilty $O_R(1) (\eta + n^{-1})$.
  Sending $n\to\infty$ and then $\eta \downarrow 0$ proves \eqref{e:broadcast-tv-error-one-replica-no-v}--\eqref{e:broadcast-tv-error-two-replica-no-v}.

  Conditional on $G$ and either $\sigma^*$ or $(\sigma^*,{\sigma'}^*)$, the coordinates of ${\boldsymbol{v}}$ are independent with $v_i\sim{\mathcal{N}}(\sqrt{t}\sigma_i^*,1)$.
  The same Markov kernel generates ${\boldsymbol{v}}_{\T \cup \T'}$ from either $\sigma_{\T \cup \T'} \sim \nu^{\broadcastTag,R}_\beta$ or $(\sigma_{\T \cup \T'}, {\sigma'}_{\T \cup \T'}) \sim \nu^{\broadcastTag,2,R}_\beta$.
  Since total variation distance does not increase under a Markov kernel, the estimates \eqref{e:broadcast-tv-error-one-replica-no-v}, \eqref{e:broadcast-tv-error-two-replica-no-v} imply the conlusions \eqref{e:broadcast-tv-error-one-replica}, \eqref{e:broadcast-tv-error-two-replica}.
  For the final assertion of the lemma, let
  \begin{align*}
    A_1 &= \{|R(\sigma^*,\allone)|<\varepsilon/2\}\,, & 
    A_2 &= \left\{
      \max\left(|R(\sigma^*,\allone)|,|R({\sigma'}^*,\allone)|,|R(\sigma^*,{\sigma'}^*)|\right)<\varepsilon/2
    \right\}\,.
  \end{align*}
  By definition, $\nu^\plantedTag_{\beta,t,\varepsilon/2}$ is $\nu^\plantedTag_{\beta,t,\varepsilon}$ conditioned on $A_1$, so
  \[
    \TV(\nu^\plantedTag_{\beta,t,\varepsilon},\nu^\plantedTag_{\beta,t,\varepsilon/2})
    = \nu^\plantedTag_{\beta,t,\varepsilon}(A_1^c)=o_n(1)
  \]
  by ~\cref{l:1mt-balanced}~\cref{i:1mt-balanced-subgaussian}.
  Similarly, $\nu^{\plantedTag,2}_{\beta,t,\varepsilon/2}$ is $\nu^{\plantedTag,2}_{\beta,t,\varepsilon}$ conditioned on $A_2$, so ~\cref{l:basic-mt}~\cref{i:basic-mt-annealed-2-subgaussian} gives
  \[
    \TV(\nu^{\plantedTag,2}_{\beta,t,\varepsilon},\nu^{\plantedTag,2}_{\beta,t,\varepsilon/2})
    = \nu^{\plantedTag,2}_{\beta,t,\varepsilon}(A_2^c)=o_n(1)\,. \qedhere
  \]
\end{proof}
Let $\E_{u,w}$ denote expectation with respect to $u,w$ independent uniformly random vertices of $G$.
\begin{proposition}\label{p:averaged-rerandomized-magnetization}
  For any $\varepsilon_3 > 0$, there exists $t_0 = t_0(d,\delta,\varepsilon_3) > 0$ 
  such that for all $\beta \in [2\beta_\lowerBoundTag, \betaKS - \delta]$, $t \in [0,t_0]$, 
  \begin{equation}\label{e:averaged-rerandomized-magnetization-one-replica}
    \E^\plantedTag_{\beta,t,\varepsilon} \E_{u,w}
    |\mu_{G,\beta,t,<\varepsilon}^{{\boldsymbol{v}}}(\sigma_u \sigma_w)|
    \le \varepsilon_3\,.
  \end{equation}
  In this inequality, recall that $\E^\plantedTag_{\beta,t,\varepsilon}$ denotes expectation with respect to $(\sigma^*,G,{\boldsymbol{v}}) \sim \nu^\plantedTag_{\beta,t,\varepsilon}$, and (recalling \eqref{e:gibbs-avg-of-function}) $\mu_{G,\beta,t,<\varepsilon}^{{\boldsymbol{v}}}(\cdot)$ denotes expectation over $\sigma \sim \mu_{G,\beta,t,<\varepsilon}^{{\boldsymbol{v}}}$ conditional on $G,{\boldsymbol{v}}$.
  Similarly,
  \begin{equation}\label{e:averaged-rerandomized-magnetization-two-replica}
    \E^{\plantedTag,2}_{\beta,t,\varepsilon} \E_{u,w}
    |\mu_{G,\beta,t,<\varepsilon}^{2,{\sigma'}^*,{\boldsymbol{v}}}(\sigma_u \sigma_w)|
    \le \varepsilon_3\,.
  \end{equation}
\begin{proof}
  For an induced subgraph $J \subseteq G$ and spins ${\boldsymbol{z}} \in \R^n$, write $\mu_{G,\beta,t,<\varepsilon}^{{\boldsymbol{v}}}(\cdot | \sigma_J = {\boldsymbol{z}}_J)$ for the conditional Gibbs measure
  \begin{equation}\label{e:conditional-gibbs-msr}
    \mu_{G,\beta,t,<\varepsilon}^{{\boldsymbol{v}}}(\sigma | \sigma_J = {\boldsymbol{z}}_J)
    \propto \allone\{\sigma \in \Sigma_n(\varepsilon), \sigma_J = {\boldsymbol{z}}_J\} \exp(\isingHamiltonian_{G,\beta,t}(\sigma,{\boldsymbol{v}}))\,,
  \end{equation}
  where $\isingHamiltonian_{G,\beta,t}$ is defined in \eqref{e:augmented-planted-model}.
  Here $G,{\boldsymbol{v}}$ are also given, so this measure is a function of $G,{\boldsymbol{v}},{\boldsymbol{z}}_J$.
  We extend the notation \eqref{e:gibbs-avg-of-function} by writing $\mu_{G,\beta,t,<\varepsilon}^{{\boldsymbol{v}}}(f | \sigma_J = {\boldsymbol{z}}_J)$ for the expectation of $f$ with respect to this measure.
  Let $R$ be a constant (independent of $n$) we will set below.
  Define
  \[
    g_\varepsilon(G,{\boldsymbol{v}},{\boldsymbol{z}},u,w) = \mu_{G,\beta,t,<\varepsilon}^{{\boldsymbol{v}}}\Big(
        \sigma_u \sigma_w
        \Big| 
        \sigma_{{\overline B}_R(u) \cap {\overline B}_R(w)} = {\boldsymbol{z}}_{{\overline B}_R(u) \cap {\overline B}_R(w)}
    \Big)\,.
  \]
  By the tower property of expectations and Nishimori's property,
  \begin{align}\label{e:averaged-rerandomized-magnetization-step1}
    \E^\plantedTag_{\beta,t,\varepsilon} \E_{u,w}
    |\mu_{G,\beta,t,<\varepsilon}^{{\boldsymbol{v}}}(\sigma_u \sigma_w)|
    &= \E^\plantedTag_{\beta,t,\varepsilon} \E_{u,w}
    |\mu_{G,\beta,t,<\varepsilon}^{{\boldsymbol{v}}}(g_\varepsilon(G,{\boldsymbol{v}},\sigma,u,w))| \\
    \nonumber
    &\le \E^\plantedTag_{\beta,t,\varepsilon} \E_{u,w}
    \mu_{G,\beta,t,<\varepsilon}^{{\boldsymbol{v}}}(|g_\varepsilon(G,{\boldsymbol{v}},\sigma,u,w)|) \\
    \nonumber
    &= \E^\plantedTag_{\beta,t,\varepsilon} \E_{u,w} |g_\varepsilon(G,{\boldsymbol{v}},\sigma^*,u,w)|\,.
  \end{align}
  Since $|g_\varepsilon|$ is bounded by $1$, by the final assertion of ~\cref{l:broadcast-tv-error},
  \begin{equation}\label{e:averaged-rerandomized-magnetization-step2}
    \E^\plantedTag_{\beta,t,\varepsilon} \E_{u,w} |g_\varepsilon(G,{\boldsymbol{v}},\sigma^*,u,w)|
    = \E^\plantedTag_{\beta,t,\varepsilon/2} \E_{u,w} |g_\varepsilon(G,{\boldsymbol{v}},\sigma^*,u,w)| + o_n(1)\,.
  \end{equation}
  Any $(\sigma^*,G,{\boldsymbol{v}})$ in the support of $\nu^\plantedTag_{\beta,t,\varepsilon/2}$ satisfies $\sigma^* \in \Sigma_n(\varepsilon/2)$.
  For any $u,w \in G$, ${\overline B}_R(u) \cap {\overline B}_R(w)$ contains all but at most $2d^R \ll n\varepsilon/2$ vertices.
  Thus, any $\sigma$ such that $\sigma_{{\overline B}_R(u) \cap {\overline B}_R(w)} = \sigma^*_{{\overline B}_R(u) \cap {\overline B}_R(w)}$ satisfies $\sigma \in \Sigma_n(\varepsilon)$.
  So
  \[
    \mu_{G,\beta,t,<\varepsilon}^{{\boldsymbol{v}}}\Big(\sigma \Big| \sigma_{{\overline B}_R(u) \cap {\overline B}_R(w)} = \sigma^*_{{\overline B}_R(u) \cap {\overline B}_R(w)}\Big)
    \propto \allone\{\sigma_{{\overline B}_R(u) \cap {\overline B}_R(w)} = \sigma^*_{{\overline B}_R(u) \cap {\overline B}_R(w)}\} \exp(\isingHamiltonian_{G,\beta,t}(\sigma,{\boldsymbol{v}}))\,,
  \]
  as in the definition \eqref{e:conditional-gibbs-msr} we can remove the redundant truncation to $\sigma \in \Sigma_n(\varepsilon)$.
  Hence this measure is a Markov random field, and can be written explicitly as
  \begin{align}\label{e:conditional-gibbs-msr-as-markov-field}
    \mu_{G,\beta,t,<\varepsilon}^{{\boldsymbol{v}}}\Big(\sigma \Big| \sigma_{{\overline B}_R(u) \cap {\overline B}_R(w)} = &\sigma^*_{{\overline B}_R(u) \cap {\overline B}_R(w)}\Big)
    \propto \allone\{\sigma_{{\overline B}_R(u) \cap {\overline B}_R(w)} = \sigma^*_{{\overline B}_R(u) \cap {\overline B}_R(w)}\} \\
    \nonumber
    &\times \exp\left(
        \beta \sum_{ij \in E(B_R(uw))} \frac{1 - \sigma_i\sigma_j}{2} - \frac12 \sum_{i\in B_R(uw)} (\sqrt{t} \sigma_i - v_i)^2
    \right)\,.
  \end{align}
  In particular, if we define the boundary set
  \[
    \partial B_R(uw) = \{v\in [n] : \min(\dist(v,u), \dist(v,w)) = R\}\,,
  \]
  then the measure \eqref{e:conditional-gibbs-msr-as-markov-field} depends on only $(\sigma_{\partial B_R(uw)}, {\boldsymbol{v}}_{B_R(uw)})$.
  We can therefore write
  \[
    g_\varepsilon(G,{\boldsymbol{v}},\sigma^*,u,w) = {\tilde g}_\varepsilon(\sigma_{\partial B_R(uw)}, {\boldsymbol{v}}_{B_R(uw)})\,,
  \]
  where ${\widetilde g}_\varepsilon$ denotes the average of $\sigma_u\sigma_w$ under \eqref{e:conditional-gibbs-msr-as-markov-field}.
  Since $|{\widetilde g}_\varepsilon|$ is bounded by $1$, we may use ~\cref{l:broadcast-tv-error} to compare $\E^\plantedTag_{\beta,t,\varepsilon} \E_{u,w}$ to expectation over $(\nu^{\broadcastTag,R}_{\beta,t})^{\otimes 2}$.
  So, for $\T'$ a disjoint copy of $\T$ with leaf set ${\mathbb{L}}'$,
  \begin{align}\nonumber
    \E^\plantedTag_{\beta,t,\varepsilon/2} \E_{u,w} |g_\varepsilon(G,{\boldsymbol{v}},\sigma^*,u,w)|
    &= \E^\plantedTag_{\beta,t,\varepsilon/2} \E_{u,w} |{\widetilde g}_\varepsilon(\sigma_{\partial B_R(uw)}, {\boldsymbol{v}}_{B_R(uw)})| \\
    \label{e:averaged-rerandomized-magnetization-step3}
    &= (\E^{\broadcastTag,R}_{\beta,t})^{\otimes 2} |{\widetilde g}_\varepsilon(\sigma_{{\mathbb{L}} \cup {\mathbb{L}}'}, {\boldsymbol{v}}_{\T \cup \T'})| + o_n(1)\,.
  \end{align}
  When the graph $B_R(uw)$ in \eqref{e:conditional-gibbs-msr-as-markov-field} is $\T \cup \T'$, the measure \eqref{e:conditional-gibbs-msr-as-markov-field} is precisely the posterior measure of (two independent copies of) the broadcast problem from ~\cref{d:antiferro-broadcast}, given observations $\sigma_{{\mathbb{L}} \cup {\mathbb{L}}'}, {\boldsymbol{v}}_{\T \cup \T'}$.
  It follows that
  \begin{equation}\label{e:averaged-rerandomized-magnetization-step4}
    (\E^{\broadcastTag,R}_{\beta,t})^{\otimes 2} |{\widetilde g}_\varepsilon(\sigma_{{\mathbb{L}} \cup {\mathbb{L}}'}, {\boldsymbol{v}}_{\T \cup \T'})|
    = \E^{\broadcastTag,R}_{\beta,t} \Big[
      \,\big|\E^{\broadcastTag,R}_{\beta,t}[\sigma_o | \sigma_{\mathbb{L}}, {\boldsymbol{v}}_\T]\big|\,
    \Big]^2\,.
  \end{equation}
  Finally, ~\cref{l:broadcast} provides $t_0, R_0$ such that for $t \in [0,t_0]$, $R \ge R_0$, 
  \begin{equation}\label{e:averaged-rerandomized-magnetization-step5}
    \E^{\broadcastTag,R}_{\beta,t} \Big[
      \,\big|\E^{\broadcastTag,R}_{\beta,t}[\sigma_o | \sigma_{\mathbb{L}}, {\boldsymbol{v}}_\T]\big|\,
    \Big] \le \varepsilon_2 \equiv \sqrt{\varepsilon_3 /2}\,.
  \end{equation}
  Combining \eqref{e:averaged-rerandomized-magnetization-step1}, \eqref{e:averaged-rerandomized-magnetization-step2}, \eqref{e:averaged-rerandomized-magnetization-step3}, \eqref{e:averaged-rerandomized-magnetization-step4}, and \eqref{e:averaged-rerandomized-magnetization-step5} yields
  \[
    \E^\plantedTag_{\beta,t,\varepsilon} \E_{u,w}
    |\mu_{G,\beta,t,<\varepsilon}^{{\boldsymbol{v}}}(\sigma_u \sigma_w)|
    \le o_n(1) + \varepsilon_3 / 2 \le \varepsilon_3\,.
  \]
  This proves \eqref{e:averaged-rerandomized-magnetization-one-replica}. 
  For the proof of \eqref{e:averaged-rerandomized-magnetization-two-replica}, we similarly define
  \begin{equation}\label{e:conditional-gibbs-msr-two-replica}
    \mu_{G,\beta,t,<\varepsilon}^{2,{\sigma'}^*,{\boldsymbol{v}}}(\sigma | \sigma_J = {\boldsymbol{z}}_J)
    \propto \allone\{\sigma \in \Sigma_n(\varepsilon), |R(\sigma,{\sigma'}^*)| < \varepsilon, \sigma_J = {\boldsymbol{z}}_J\} \exp(\isingHamiltonian_{G,\beta,t}(\sigma,{\boldsymbol{v}}))\,,
  \end{equation}
  and
  \[
    g_{\varepsilon,2}(G,{\sigma'}^*,{\boldsymbol{v}},{\boldsymbol{z}},u,w) = \mu_{G,\beta,t,<\varepsilon}^{2,{\sigma'}^*,{\boldsymbol{v}}}\Big(\sigma_u \sigma_w \Big| \sigma_{{\overline B}_R(u) \cap {\overline B}_R(w)} = {\boldsymbol{z}}_{{\overline B}_R(u) \cap {\overline B}_R(w)}\Big)\,.
  \]
  Then, identically to \eqref{e:averaged-rerandomized-magnetization-step1}, \eqref{e:averaged-rerandomized-magnetization-step2},
  \begin{align*}
    \E^{\plantedTag,2}_{\beta,t,\varepsilon} \E_{u,w}
    |\mu_{G,\beta,t,<\varepsilon}^{2,{\sigma'}^*,{\boldsymbol{v}}}(\sigma_u \sigma_w)|
    &= \E^{\plantedTag,2}_{\beta,t,\varepsilon} \E_{u,w}
    |\mu_{G,\beta,t,<\varepsilon}^{2,{\sigma'}^*,{\boldsymbol{v}}}(g_{\varepsilon,2}(G,{\sigma'}^*,{\boldsymbol{v}},\sigma,u,w))| \\
    &\le \E^{\plantedTag,2}_{\beta,t,\varepsilon} \E_{u,w}
    \mu_{G,\beta,t,<\varepsilon}^{2,{\sigma'}^*,{\boldsymbol{v}}}(|g_{\varepsilon,2}(G,{\sigma'}^*,{\boldsymbol{v}},\sigma,u,w)|) \\
    &= \E^{\plantedTag,2}_{\beta,t,\varepsilon} \E_{u,w} |g_{\varepsilon,2}(G,{\sigma'}^*,{\boldsymbol{v}},\sigma^*,u,w)| \\
    &= \E^{\plantedTag,2}_{\beta,t,\varepsilon/2} \E_{u,w} |g_{\varepsilon,2}(G,{\sigma'}^*,{\boldsymbol{v}},\sigma^*,u,w)| + o_n(1)\,.
  \end{align*}
  Any $(\sigma^*,{\sigma'}^*,G,{\boldsymbol{v}})$ in the support of $\nu^\plantedTag_{\beta,t,\varepsilon/2}$ satisfies $\sigma^* \in \Sigma_n(\varepsilon/2)$ and $|R(\sigma^*,{\sigma'}^*)| < \varepsilon/2$.
  Arguing similarly to above, any $\sigma$ such that $\sigma_{{\overline B}_R(u) \cap {\overline B}_R(w)} = \sigma^*_{{\overline B}_R(u) \cap {\overline B}_R(w)}$ must satisfy $\sigma \in \Sigma_n(\varepsilon)$ and $|R(\sigma,{\sigma'}^*)| < \varepsilon$.
  So these truncations in \eqref{e:conditional-gibbs-msr-two-replica} are both redundant, and $\mu_{G,\beta,t,<\varepsilon}^{2,{\sigma'}^*,{\boldsymbol{v}}}$ is the same Markov random field as \eqref{e:conditional-gibbs-msr-as-markov-field}.
  The rest of the proof of \eqref{e:averaged-rerandomized-magnetization-one-replica} proceeds identically to prove \eqref{e:averaged-rerandomized-magnetization-two-replica}, using Corollary~\cref{c:broadcast} in place of ~\cref{l:broadcast}.
\end{proof}
\end{proposition}

\begin{proof}[Proof of ~\cref{l:planted-augmented}]
  We provide the proof for part \cref{i:planted-augmented-one-replica}, as the proof for part \cref{i:planted-augmented-two-replica} is analogous.
  We estimate
  \begin{align*}
    \E^\plantedTag_{\beta,t,\varepsilon} \mu_{G,\beta,t,<\varepsilon}^{{\boldsymbol{v}}}(|R(\sigma,\sigma^*)| \ge \varepsilon_1)
    &\le \frac{1}{(\varepsilon_1)^2} \E^\plantedTag_{\beta,t,\varepsilon} \mu_{G,\beta,t,<\varepsilon}^{{\boldsymbol{v}}}(|R(\sigma,\sigma^*)|^2) \\
    &= \frac{1}{(\varepsilon_1)^2} \E^\plantedTag_{\beta,t,\varepsilon} \E_{u,w} \sigma^*_u \sigma^*_w \mu_{G,\beta,t,<\varepsilon}^{{\boldsymbol{v}}}(\sigma_u\sigma_w) \\
    &\le \frac{1}{(\varepsilon_1)^2} \E^\plantedTag_{\beta,t,\varepsilon} \E_{u,w} |\mu_{G,\beta,t,<\varepsilon}^{{\boldsymbol{v}}}(\sigma_u\sigma_w)|\,.
  \end{align*}
  By ~\cref{p:averaged-rerandomized-magnetization}, there exist $t_0,R_0$ such that for all $t\in [0,t_0]$ and $R \ge R_0$,
  \[
    \E^\plantedTag_{\beta,t,\varepsilon} \E_{u,w} |\mu_{G,\beta,t,<\varepsilon}^{{\boldsymbol{v}}}(\sigma_u\sigma_w)|
    \le \varepsilon_3 \equiv (\varepsilon_1)^3\,.
  \]
  This implies the result.
\end{proof}

\section{Lipschitzness of conditional means of restricted free energies}
\label{a:restricted-fe-lipschitz}

This appendix is devoted to the proof of ~\cref{p:fe-conc-conditional-mean-lipschitz}.
In \cref{aa:wass-stability} we will prove ~\cref{l:wass-stability-of-sbm}, which shows that the conditional laws of $G \sim \nu^\plantedTag_\beta(\cdot | \sigma^*)$ and $G \sim \nu^{\plantedTag,2}_\beta(\cdot | \sigma^*,{\sigma'}^*)$ are Wasserstein continuous in $\sigma^*$ and $(\sigma^*,{\sigma'}^*)$.
This will be the main input to ~\cref{p:fe-conc-conditional-mean-lipschitz}, whose proof is completed in \cref{aa:restricted-fe-lipschitz-completion}.

\subsection{Wasserstein stability of planted model in latent signal}
\label{aa:wass-stability}

The main result of this subsection is the following.
\begin{lemma}\label{l:wass-stability-of-sbm}
  There exists an absolute constant $C_0$ such that for any $\beta \in [2\beta_\lowerBoundTag, \betaKS - \delta]$, the following holds.
  \begin{enumerate}[label=(\alph*)]
    \item \label{i:wass-stability-of-sbm-one-replica} For any $\sigma^*,\widetilde{\sigma}^* \in \Sigma_n(\varepsilon)$ with Hamming distance $1$, it is possible to couple $G \sim \nu^\plantedTag_\beta(\cdot | \sigma^*)$ and ${\widetilde G} \sim \nu^\plantedTag_\beta(\cdot | \widetilde{\sigma}^*)$ such that $\E |E(G) \triangle E({\widetilde G})| \le C_0d$.
    \item \label{i:wass-stability-of-sbm-two-replica} For any $(\sigma^*,{\sigma'}^*), (\widetilde{\sigma}^*,\widetilde{\sigma'}^*) \in \Sigma_{n,2}(\varepsilon)$ (recall \eqref{e:SigmaN2}) with Hamming distance $1$, it is possible to couple $G \sim \nu^{\plantedTag,2}_\beta(\cdot | \sigma^*,{\sigma'}^*)$ and ${\widetilde G} \sim \nu^{\plantedTag,2}_\beta(\cdot | \widetilde{\sigma}^*,\widetilde{\sigma'}^*)$ such that $\E |E(G) \triangle E({\widetilde G})| \le C_0d$.
  \end{enumerate}
\end{lemma}
Recall that $\Sigma_1 = \{+1,-1\}$, $\Sigma_2 = \{(+1,+1), (+1,-1), (-1,+1), (-1,-1)\}$.
For $\sigma^*,{\sigma'}^* \in \Sigma_n$, let $m(\sigma^*) \in \Z_{\ge 0}^{\Sigma_1}$, $m(\sigma^*,{\sigma'}^*) \in \Z_{\ge 0}^{\Sigma_2}$ have entries
\begin{align*}
  m_\omega(\sigma^*) = d|\{v\in [n] : \sigma^*_v = \omega\}|\,, &
  m_\omega(\sigma^*,{\sigma'}^*) = d|\{v\in [n] : (\sigma^*_v,{\sigma'}^*_v) = \omega\}|\,,
\end{align*}
for $\omega \in \Sigma_1$ or $\omega \in \Sigma_2$, respectively.
For a graph $G$ on vertex set $[n]$, let $L(G;\sigma^*) \in \Z_{\ge 0}^{\Sigma_1 \times \Sigma_1}$ and $L(G;\sigma^*,{\sigma'}^*) \in \Z_{\ge 0}^{\Sigma_2 \times \Sigma_2}$ have entries
\begin{align}\nonumber
  L_{\omega,\rho}(G;\sigma^*) &= 
  \begin{cases}
    \#\{\text{edges $e=uv$ such that $\sigma^*_u = \omega$ and $\sigma^*_v = \rho$} \} & \omega \neq \rho\,, \\
    2\#\{\text{edges $e=uv$ such that $\sigma^*_u = \sigma^*_v = \omega$} \} & \omega = \rho\,,
  \end{cases} \\
  \label{e:L-two-replica}
  L_{\omega,\rho}(G;\sigma^*,{\sigma'}^*) &= 
  \begin{cases}
    \#\{\text{edges $e=uv$ such that $(\sigma^*_u,{\sigma'}^*_u) = \omega$ and $(\sigma^*_v,{\sigma'}^*_v) = \rho$} \} & \omega \neq \rho\,, \\
    2\#\{\text{edges $e=uv$ such that $(\sigma^*_u,{\sigma'}^*_u) = (\sigma^*_v,{\sigma'}^*_v) = \omega$} \} & \omega = \rho\,,
  \end{cases}
\end{align}
for $\omega,\rho \in \Sigma_1$ and $\omega,\rho \in \Sigma_2$, respectively.
We next work towards the following lemma.
\begin{lemma}\label{l:wass-stability-of-sbm-counts}
  There exists an absolute constant $C_1$ such that for any $\beta \in [2\beta_\lowerBoundTag, \betaKS - \delta]$, the following holds.
  \begin{enumerate}[label=(\alph*)]
    \item \label{i:wass-stability-of-sbm-counts-one-replica} For any $\sigma^*,\widetilde{\sigma}^* \in \Sigma_n(\varepsilon)$ with Hamming distance $1$, let $G \sim \nu^\plantedTag_\beta(\cdot | \sigma^*)$ and ${\widetilde G} \sim \nu^\plantedTag_\beta(\cdot | \widetilde{\sigma}^*)$.
    It is possible to couple $L = L(G;\sigma^*)$ and ${\widetilde L} = L({\widetilde G};\widetilde{\sigma}^*)$ such that $\E \|L - {\widetilde L}\|_1 \le C_1 d$.
    \item \label{i:wass-stability-of-sbm-counts-two-replica} For any $(\sigma^*,{\sigma'}^*), (\widetilde{\sigma}^*,\widetilde{\sigma'}^*) \in \Sigma_{n,2}(\varepsilon)$ with Hamming distance $1$, let $G \sim \nu^{\plantedTag,2}_\beta(\cdot | \sigma^*,{\sigma'}^*)$ and ${\widetilde G} \sim \nu^{\plantedTag,2}_\beta(\cdot | \widetilde{\sigma}^*,\widetilde{\sigma'}^*)$. It is possible to couple $L = L(G;\sigma^*,{\sigma'}^*)$ and ${\widetilde L} = L({\widetilde G};\widetilde{\sigma}^*,\widetilde{\sigma'}^*)$ such that $\E \|L - {\widetilde L}\|_1 \le C_1 d$.
  \end{enumerate}
  In either case, $\|M\|_1 = \sum_{\omega,\rho} |M_{\omega,\rho}|$, where $\omega,\rho$ belong to $\Sigma_1$ or $\Sigma_2$.
\end{lemma}
We will address both parts simultaneously.
For part~\cref{i:wass-stability-of-sbm-one-replica}, take $k=1$, and for part~\cref{i:wass-stability-of-sbm-two-replica}, take $k=2$; let $Q = \Sigma_k$.
Let $\Delta$ denote Hamming distance on $Q$, and for $\omega,\rho \in Q$ let
\[
  \lambda_{\omega,\rho} = \exp\{\beta\Delta(\omega,\rho)\}\,.
\]
For $m \in \Z_{\ge 0}^Q$ with even sum, let $S(m)$ denote the set of symmetric $L \in \Z_{\ge 0}^{Q \times Q}$ with even diagonal satisfying, for all $\omega \in Q$,
\[
  \sum_{\rho \in Q} L_{\omega,\rho} = m_\omega\,.
\]
Further, for $L \in S(m)$, define
\[
  W_k(L)
  =
  \prod_{\omega \in Q} \frac{1}{2^{L_{\omega,\omega}/2}(L_{\omega,\omega}/2)!}
  \prod_{\{\omega,\rho\} \in \binom{Q}{2}}
  \frac{\lambda_{\omega,\rho}^{L_{\omega,\rho}}}{L_{\omega,\rho}!}\,,
\]
and let $\mu_m^{(k)}$ be the probability measure on $S(m)$ given by
\[
  \mu_m^{(k)}(L) = \frac{W_k(L)}{Z_k(m)}\,,
\]
for normalizing constant $Z_k(m)$.
By the same configuration-model count as in \eqref{e:kmt-exact}, for $G \sim \nu^\plantedTag(\cdot | \sigma^*)$ we have $L(G;\sigma^*) \sim \mu_{m(\sigma^*)}^{(1)}$, and for $G \sim \nu^{\plantedTag,2}(\cdot | \sigma^*,{\sigma'}^*)$ we have $L(G;\sigma^*,{\sigma'}^*) \sim \mu_{m(\sigma^*,{\sigma'}^*)}^{(2)}$.
\begin{lemma}\label{l:adjacent-row-sum-count-coupling}
  Suppose $m,m' \in \Z_{\ge 0}^Q$ have even sum and
  \[
    m' = m + e_a - e_b\,,
  \]
  for some distinct $a,b \in Q$, where $e_a$ denotes the coordinate vector indexed by $a$.
  Then there is a coupling of $L\sim\mu_m^{(k)}$ and $L'\sim\mu_{m'}^{(k)}$ such that $\E\|L-L'\|_1 \le 4e^8$.
\begin{proof}
  For even $N\ge 0$, let
  \begin{align*}
    D_N &= \sup\left\{
      D(m,m') : m,m' \in \Z_{\ge 0}^Q, \|m\|_1 = N, m' = m + e_a - e_b\,\,\text{for distinct}\,\,a,b\in Q
    \right\}\,, \\
    D(m,m') &= \inf\left\{
      \E_\pi \|L-L'\|_1 : \pi\,\,\text{coupling of}\,\,L\sim\mu_m^{(k)},L'\sim\mu_{m'}^{(k)}
    \right\}\,.
  \end{align*}
  We will argue below that
  \begin{equation}\label{e:adjacent-row-sum-count-coupling}
    D_N \le 4 + (1-e^{-8})D_{N-2}\,.
  \end{equation}
  Since $D_0 = 0$, a standard induction then implies the result.
  For $\omega,\rho \in Q$, define $A^{\omega,\rho} \in \Z_{\ge 0}^{Q\times Q}$ by
  \[
    A^{\omega,\rho}_{\omega',\rho'}
    = \allone\{(\omega',\rho')=(\omega,\rho)\}
    + \allone\{(\omega',\rho')=(\rho,\omega)\}\,.
  \]
  Suppose $\omega \in Q$ is such that $m_\omega > 0$.
  We claim that $L\sim\mu_m^{(k)}$ can be sampled as follows. 
  First sample $\rho$ from the probability measure on $Q$ with
  \begin{equation}\label{e:adjacent-row-sum-count-coupling-one-edge}
    \Pr(\rho = \tau)
    =
    \frac{\lambda_{\omega,\tau}Z_k(m-e_\omega-e_\tau)}{m_\omega Z_k(m)}\,.
  \end{equation}
  Conditional on $\rho$, sample $R \sim \mu_{m-e_\omega-e_\rho}^{(k)}$.
  Then, let $L = A^{\omega,\rho} + R$.
  Indeed, in the configuration model, the above procedure corresponds to first matching an arbitrary half-edge of type $\omega$ (its partner's type $\rho$ has distribution \eqref{e:adjacent-row-sum-count-coupling-one-edge}) and then matching the remaining half-edges.

  Consider any $m,m' \in \Z_{\ge 0}^Q$ with $\|m\|_1 = N$, $m' = m + e_a - e_b$.
  Then $m_b > 0$ and $m'_a > 0$.
  For $\tau \in Q$, let
  \[
    Z_{k,\tau} = Z_k(m - e_b - e_\tau) = Z_k(m' - e_a - e_\tau)\,.
  \]
  Then we can sample $L\sim\mu_m^{(k)}$, $L'\sim\mu_{m'}^{(k)}$ by sampling $\rho,\rho'$ from
  \begin{align*}
    \Pr(\rho = \tau)
    &= \frac{\lambda_{b,\tau}Z_{k,\tau}}{m_b Z_k(m)}
    = \frac{\lambda_{b,\tau}Z_{k,\tau}}{\sum_{\tau'\in Q}\lambda_{b,\tau'}Z_{k,\tau'}}\,, &
    \Pr(\rho' = \tau)
    &= \frac{\lambda_{a,\tau}Z_{k,\tau}}{m'_a Z_k(m')}
    = \frac{\lambda_{a,\tau}Z_{k,\tau}}{\sum_{\tau'\in Q}\lambda_{a,\tau'}Z_{k,\tau'}}\,,
  \end{align*}
  then sampling $R \sim \mu_{m-e_b-e_\rho}^{(k)}$, $R \sim \mu_{m'-e_a-e_\rho'}^{(k)}$, and finally setting $L = A^{b,\rho} + R$, $L' = A^{a,\rho'} + R'$.
  Let
  \[
    p = \sup\left\{
      \Pr_\pi(\rho = \rho') : 
      \pi\,\,\text{coupling of $\rho,\rho'$}
    \right\}\,.
  \]
  We claim that
  \begin{equation}\label{e:adjacent-row-sum-count-coupling-p}
    D(m,m') \le 4 + (1-p)D_{N-2}\,.
  \end{equation}
  Note that pointwise,
  \[
    \|L - L'\|_1 \le \|A^{b,\rho} - A^{a,\rho}\|_1 + \|R - R'\|_1
    \le 4 + \|R - R'\|_1\,.
  \]
  Let us couple $\rho,\rho'$ such that $\Pr(\rho = \rho') = p$.
  On the event $\rho = \rho'$, the distributions of $R$ and $R'$ are the same so we can couple them perfectly.
  Otherwise, $m-e_b-e_\rho$ and $m'-e_a-e_\rho'$ are adjacent, so we can couple $R$ and $R'$ such that $\E\|R - R'\|_1 \le D_{N-2}$.
  This proves \eqref{e:adjacent-row-sum-count-coupling-p}.

  To show \eqref{e:adjacent-row-sum-count-coupling}, it remains to show $p \ge e^{-8}$.
  Note that $|\beta|<2$, as
  \begin{align*}
    \betaKS(d)
    &=
    2\arctanh\left(\frac{1}{\sqrt{d-1}}\right)
    < 2\,,&
    |2\beta_\lowerBoundTag| &= 2d^{-10} < 1\,.
  \end{align*}
  Also note that for all $\tau \in Q$,
  \[
    |\Delta(b,\tau)-\Delta(a,\tau)| \le k \le 2\,.
  \]
  This implies $e^{-4} \le \frac{\lambda_{b,\tau}}{\lambda_{a,\tau}} \le e^4$, and thus
  \[
    e^{-4}
    \le
    \frac{\sum_{\tau'\in Q}\lambda_{b,\tau'}Z_{k,\tau'}}{\sum_{\tau'\in Q}\lambda_{a,\tau'}Z_{k,\tau'}}
    \le
    e^4\,.
  \]
  Consequently,
  \[
    \frac{\Pr(\rho = \tau)}{\Pr(\rho' = \tau)}
    =
    \frac{\lambda_{b,\tau}/\lambda_{a,\tau}}
    {\left(\sum_{\tau'\in Q}\lambda_{b,\tau'}Z_{k,\tau'}\right)/\left(\sum_{\tau'\in Q}\lambda_{a,\tau'}Z_{k,\tau'}\right)}
    \in [e^{-8},e^8]\,.
  \]
  This implies that a maximal coupling of $\rho,\rho'$ has $\Pr(\rho = \rho') \ge e^{-8}$, as desired.
\end{proof}
\end{lemma}

\begin{proof}[Proof of ~\cref{l:wass-stability-of-sbm-counts}]
  For part \cref{i:wass-stability-of-sbm-counts-one-replica}, let $m = m(\sigma^*)$, ${\widetilde m} = m(\widetilde{\sigma}^*)$.
  For part \cref{i:wass-stability-of-sbm-counts-two-replica}, let $m = m(\sigma^*,{\sigma'}^*)$, ${\widetilde m} = m(\widetilde{\sigma}^*,\widetilde{\sigma'}^*)$.
  In either case, we have $L\sim\mu_m^{(k)}$ and ${\widetilde L}\sim\mu_{{\widetilde m}}^{(k)}$.
  Furthermore, we have 
  \[
    {\widetilde m} = m + d(e_a - e_b)
  \]
  for some distinct $a,b \in Q$.
  The result follows with $C_1 = 4e^8$ by iterating ~\cref{l:adjacent-row-sum-count-coupling} $d$ times.
\end{proof}

\begin{fact}\label{f:switching}
  Suppose $A$ and $B$ are finite sets with $|A|$, $|B|$ even, and $M_A$, $M_B$ are uniformly random perfect matchings of $A$ and $B$.
  Then, there is a coupling of $M_A$ and $M_B$ such that
  \begin{equation}\label{e:switching}
    |M_A \triangle M_B| \le \frac32|A \triangle B|
  \end{equation}
  almost surely.
  Here $A \triangle B = (A \setminus B) \cup (B\setminus A)$ denotes the symmetric difference of $A$ and $B$. 
\begin{proof}
  We first consider three simple cases $B = A \cup \{j\} \setminus \{i\}$, $B = A \setminus \{i,j\}$, and $B = A \cup \{i,j\}$.
  In each case, given a uniformly random perfect matching $M_A$ of $A$, we show how to construct a perfect matching $M_B$ of $B$, marginally distributed as a uniformly random perfect matching of $B$, such that \eqref{e:switching} holds.
  The general result follows by composing these elementary operations.

  \textit{Case 1: $B = A \cup \{j\} \setminus \{i\}$.}
  Let $(i,\ell)$ be the pair containing $i$ in $M_A$.
  Set $M_B = M_A \cup \{(j,\ell)\} \setminus \{(i,\ell)\}$.
  Then $|M_A \triangle M_B| = 2 = |A \triangle B|$, so \eqref{e:switching} holds.

  \textit{Case 2: $B = A \setminus \{i,j\}$.}
  If $(i,j) \in M_A$, set $M_B = M_A \setminus \{(i,j)\}$.
  Otherwise, suppose $(i,\ell), (j,m) \in M_A$.
  Set $M_B = M_A \setminus \{(i,\ell), (j,m)\} \cup \{(\ell,m)\}$.

  \textit{Case 3: $B = A \cup \{i,j\}$.}
  With probability $1/(|A|+1)$, set $M_B = M_A \cup \{(i,j)\}$ (as this is the probability that a uniformly random perfect matching of $B$ matches $i$ with $j$).
  Otherwise, select a uniformly random pair $(\ell,m)$ of $M_A$, and set $M_B$ to either $M_A \setminus \{(\ell,m)\} \cup \{(i,\ell), (j,m)\}$ or $M_A \setminus \{(\ell,m)\} \cup \{(i,m), (j,\ell)\}$, each with probability $1/2$.

  It is easy to check that $M_B$ is marginally uniformly random in all cases.
  In cases 2 and 3, $|M_A \triangle M_B| \le 3 = \frac32 |A \triangle B|$, so \eqref{e:switching} holds.
\end{proof}
\end{fact}

\begin{proof}[Proof of ~\cref{l:wass-stability-of-sbm}]
  We will only prove part \cref{i:wass-stability-of-sbm-two-replica}, as part \cref{i:wass-stability-of-sbm-one-replica} is strictly simpler.
  Assume without loss that $(\sigma^*.{\sigma'}^*),(\widetilde{\sigma}^*,\widetilde{\sigma'}^*)$ differ only in the coordinate $(\sigma^*_1,\widetilde{\sigma}^*_1) = (+1,-1)$.
  We will construct the coupling in several steps.

  \textit{Step 1: coupling the edge-type profile}.
  By ~\cref{l:wass-stability-of-sbm-counts}, we can couple $L = L(G;\sigma^*,{\sigma'}^*)$ and ${\widetilde L} = L({\widetilde G};\widetilde{\sigma}^*,\widetilde{\sigma'}^*)$ such that $\E\|L - {\widetilde L}\|_1 \le C_1 d$.  
  Let $L^- = L^-(G;\sigma^*,{\sigma'}^*)$ and ${\widetilde L}^- = {\widetilde L}^-({\widetilde G};\widetilde{\sigma}^*,\widetilde{\sigma'}^*)$ be defined as in \eqref{e:L-two-replica}, except where we do not count the edges $e=uv$ where $1 \in \{u,v\}$.
  Sample $L^-, {\widetilde L}^-$ conditional on $L,{\widetilde L}$.
  Since there are only $d$ edges incident to vertex $1$, we have
  \[
    \|L^- - {\widetilde L}^-\|_1 \le 4d + \|L - {\widetilde L}\|_1\,.
  \]
  Condition on $L,{\widetilde L},L^-,{\widetilde L}^-$ for the rest of the construction.

  \textit{Step 2: coupling non-monochromatic edges.}
  Call an edge $e = uv$ non-monochromatic in $G$ if $(\sigma^*_u,{\sigma'}^*_u) \neq (\sigma^*_v,{\sigma'}^*_v)$ (and in ${\widetilde G}$ if $(\widetilde{\sigma}^*_u,\widetilde{\sigma'}^*_u) \neq (\widetilde{\sigma}^*_v,\widetilde{\sigma'}^*_v)$).
  Let $E_{{\mathsf{nc}}}(G)$ and $E_{{\mathsf{nc}}}({\widetilde G})$ be the non-monochromatic edge sets of $G$ and ${\widetilde G}$.
  We next sample $E_{{\mathsf{nc}}}(G)$ and $E_{{\mathsf{nc}}}({\widetilde G})$ such that
  \begin{equation}\label{e:non-monochromatic-coupling}
    |E_{{\mathsf{nc}}}(G) \triangle E_{{\mathsf{nc}}}({\widetilde G})|
    \le 4d + \frac12 \|L - {\widetilde L}\|_1\,.
  \end{equation}
  Recall that vertices $2,\ldots,n$ have the same spins under $(\sigma^*,{\sigma'}^*)$ and $(\widetilde{\sigma}^*,\widetilde{\sigma'}^*)$.
  For $\omega \in \Sigma_2$ write
  \[
    V_\omega = \{
      v \in \{2,\ldots,n\} : (\sigma^*_v,{\sigma'}^*_v) = (\widetilde{\sigma}^*_v,\widetilde{\sigma'}^*_v) = \omega
    \}\,,
  \]
  and let $S_\omega$ be the set of half-edges incident to $V_\omega$.

  For distinct $\omega,\rho \in \Sigma_2$, let $E_{{\mathsf{nc}},\omega,\rho}(G) \subseteq E_{{\mathsf{nc}}}(G)$ be the set of edges with endpoints in $V_\omega$ and $V_\rho$.
  Note that $|E_{{\mathsf{nc}},\omega,\rho}(G)| = L^-_{\omega,\rho}$.
  Furthermore, the sets $E_{{\mathsf{nc}},\omega,\rho}(G)$ can be sampled as follows.
  For each $\omega \in \Sigma_2$, generate a uniformly random partition of $S_\omega$
  \[
    S_\omega = \bigcup_{\rho \in \Sigma_2} S_{\omega\rightarrow\rho}
  \]
  where for $\rho\neq\omega$, $|S_{\omega\rightarrow\rho}| = L^-_{\omega,\rho}$ (and this determines $|S_{\omega\rightarrow\omega}|$).
  Then, $E_{{\mathsf{nc}},\omega,\rho}(G)$ is a uniformly random perfect matching of $S_{\omega\rightarrow\rho}$ and $S_{\rho\rightarrow\omega}$.

  Define $E_{{\mathsf{nc}},\omega,\rho}({\widetilde G}) \subseteq E_{{\mathsf{nc}}}({\widetilde G})$ analogously to above.
  We will couple $E_{{\mathsf{nc}},\omega,\rho}(G)$ and $E_{{\mathsf{nc}},\omega,\rho}({\widetilde G})$ such that the smaller of these sets is contained in the larger one.
  For distinct $\omega,\rho \in \Sigma_2$, let $\ell_{\omega,\rho} = \min(L^-_{\omega,\rho}, {\widetilde L}^-_{\omega,\rho})$.
  For each $\omega \in \Sigma_2$, generate a uniformly random partition of $S_\omega$
  \[
    S_\omega = \bigcup_{\rho \in \Sigma_2} S^0_{\omega\rightarrow\rho}
  \]
  where for $\rho\neq\omega$, $|S^0_{\omega\rightarrow\rho}| = \ell_{\omega,\rho}$.

  Then, for distinct $\omega,\rho \in \Sigma_2$, generate $E^0_{{\mathsf{nc}},\omega,\rho}$ as a uniformly random perfect matching of $S^0_{\omega\rightarrow\rho}$ and $S^0_{\rho\rightarrow\omega}$.
  If $L^-_{\omega,\rho} \ge {\widetilde L}^-_{\omega,\rho}$, set $E_{{\mathsf{nc}},\omega,\rho}({\widetilde G}) = E^0_{{\mathsf{nc}},\omega,\rho}$, and let $E_{{\mathsf{nc}},\omega,\rho}(G)$ be the union of $E^0_{{\mathsf{nc}},\omega,\rho}$ and $L^-_{\omega,\rho} - {\widetilde L}^-_{\omega,\rho}$ edges, formed by matching this many uniformly random half-edges from $S^0_{\omega\rightarrow\omega}$ and $S^0_{\rho\rightarrow\rho}$.
  If $L^-_{\omega,\rho} < {\widetilde L}^-_{\omega,\rho}$ we reverse the roles of $G,{\widetilde G}$.

  It is clear that the sets $E_{{\mathsf{nc}},\omega,\rho}(G)$, $E_{{\mathsf{nc}},\omega,\rho}({\widetilde G})$ thus generated have the correct marginals, and
  \begin{align}\nonumber
    \sum_{(\omega,\rho) \in \binom{\Sigma_2}{2}}
    |E_{{\mathsf{nc}},\omega,\rho}(G) \triangle E_{{\mathsf{nc}},\omega,\rho}({\widetilde G})| 
    &= \sum_{(\omega,\rho) \in \binom{\Sigma_2}{2}} \|L^-_{\omega,\rho} - {\widetilde L}^-_{\omega,\rho}\| \\
    \label{e:non-monochromatic-non-1}
    &\le \frac12 \|L^- - {\widetilde L}^-\|_1
    \le 2d + \frac12 \|L - {\widetilde L}\|_1\,.
  \end{align}
  Conditional on the above, we sample the non-monochromatic edges incident to $1$ in $G,{\widetilde G}$.
  Since there are at most $d$ such edges in each of $G$ and ${\widetilde G}$,
  \[
    |E_{{\mathsf{nc}}}(G) \triangle E_{{\mathsf{nc}}}({\widetilde G})|
    \le 2d + \sum_{(\omega,\rho) \in \binom{\Sigma_2}{2}} 
    |E_{{\mathsf{nc}},\omega,\rho}(G) \triangle E_{{\mathsf{nc}},\omega,\rho}({\widetilde G})|\,,
  \]
  which proves \eqref{e:non-monochromatic-coupling}.

  \textit{Step 3: coupling monochromatic edges.}
  We finally generate the monochromatic edge sets $E_{{\mathsf{mc}}}(G)$ and $E_{{\mathsf{mc}}}({\widetilde G})$ of $G$ and ${\widetilde G}$.
  For $\omega \in \Sigma_2$, let $T_\omega(G)$ denote the sets of half-edges that are not yet matched in $G$, and that are incident to a vertex $v$ with $(\sigma^*_v,{\sigma'}^*_v) = \omega$.
  (Note that unlike the definition of $S_\omega$ above, we do not exclude half-edges incident to $1$).
  Analogously define $T_\omega({\widetilde G})$.

  The set $E_{{\mathsf{mc}}}(G)$ is distributed as the union of uniformly random perfect matchings of $T_\omega(G)$, for $\omega \in \Sigma_2$.
  The set $E_{{\mathsf{mc}}}({\widetilde G})$ is analogous.
  We claim that
  \begin{equation}\label{e:monochromatic-discrepancy-half-edges}
    \sum_{\omega\in\Sigma_2}
    |T_\omega(G) \triangle T_\omega({\widetilde G})| \le 6d + \|L - {\widetilde L}\|_1\,.
  \end{equation}
  Indeed, there are two possible sources of half-edges to the left-hand side of \eqref{e:monochromatic-discrepancy-half-edges}: 
  \begin{enumerate}[label=(\roman*)]
    \item For each edge in $E_{{\mathsf{nc}},\omega,\rho}(G) \triangle E_{{\mathsf{nc}},\omega,\rho}({\widetilde G})$, its constituent half-edges may belong to $T_\omega({\widetilde G}) \setminus T_\omega(G)$ and $T_\rho({\widetilde G}) \setminus T_\rho(G)$ if $E_{{\mathsf{nc}},\omega,\rho}({\widetilde G}) \subseteq E_{{\mathsf{nc}},\omega,\rho}(G)$, and to $T_\omega(G) \setminus T_\omega({\widetilde G})$ and $T_\rho(G) \setminus T_\rho({\widetilde G})$ if $E_{{\mathsf{nc}},\omega,\rho}(G) \subseteq E_{{\mathsf{nc}},\omega,\rho}({\widetilde G})$. By \eqref{e:non-monochromatic-non-1}, the number of half-edges of this form is at most
    \[
      2\sum_{(\omega,\rho) \in \binom{\Sigma_2}{2}}
      |E_{{\mathsf{nc}},\omega,\rho}(G) \triangle E_{{\mathsf{nc}},\omega,\rho}({\widetilde G})| 
      \le 4d + \|L - {\widetilde L}\|_1\,.
    \]
    \item There are at most $d$ half-edges incident to $1$ in $G$ that were not matched above, and similarly for ${\widetilde G}$.
    The total number of such half-edges is at most $2d$.
  \end{enumerate}
  This proves \eqref{e:monochromatic-discrepancy-half-edges}.
  Fact~\cref{f:switching} and \eqref{e:monochromatic-discrepancy-half-edges} imply that we can couple $E_{{\mathsf{mc}}}(G)$ and $E_{{\mathsf{mc}}}({\widetilde G})$ such that 
  \[
    |E_{{\mathsf{mc}}}(G) \triangle E_{{\mathsf{mc}}}({\widetilde G})|
    \le \frac32 \sum_{\omega\in\Sigma_2}
    |T_\omega(G) \triangle T_\omega({\widetilde G})|
    \le 9d + \frac32\|L - {\widetilde L}\|_1\,.
  \]
  Together with \eqref{e:non-monochromatic-coupling} this yields
  \[
    |E(G) \triangle E({\widetilde G})| \le 13d + 2\|L - {\widetilde L}\|_1\,.
  \]
  Finally, the guarantee $\E\|L - {\widetilde L}\|_1 \le C_1 d$ from ~\cref{l:wass-stability-of-sbm-counts} implies the conclusion.
\end{proof}

\subsection{Completion of the proof}
\label{aa:restricted-fe-lipschitz-completion}

In this subsection, we make the dependence of $\Sigma_n(\sigma^*,I)$ on $\varepsilon$ explicit, writing
\[
  \Sigma_n(\sigma^*,\varepsilon,I) = \Sigma_n(\sigma^*,I) = \left\{
    \sigma \in \Sigma_n(\varepsilon) : R(\sigma^*,\sigma) \in I
  \right\}\,.
\]
\begin{fact}\label{f:hamming-nbrs}
  Suppose $\sigma^* \in \Sigma_n(\varepsilon)$ and $I \in {\mathcal{I}}$.
  Write $\partial \Sigma_n(\sigma^*,\varepsilon,I) = \Sigma_n(\sigma^*,\varepsilon,I) \setminus \Sigma_n(\sigma^*,\varepsilon-\frac{2}{n},I)$.
  If $\sigma \in \partial \Sigma_n(\sigma^*,\varepsilon,I)$, then there are at least $n/50$ points $\widetilde{\sigma} \in \Sigma_n(\sigma^*,\varepsilon - \frac{2}{n},I)$ with Hamming distance $1$ from $\sigma$.
\begin{proof}
  Let $p = R(\sigma^*,\allone)$, $q = R(\sigma,\sigma^*)$, and $r = R(\sigma,\allone)$.
  The assumption $\sigma \in \partial \Sigma_n(\sigma^*,\varepsilon,I)$ implies $\varepsilon - \frac{2}{n} \le |r| < \varepsilon$ and $q\in I$.
  Assume without loss that $r>0$.
  Define
  \[
    T = \{i\in [n] : \sigma_i = 1\}\,,
  \]
  so that $|T| \ge n/2$.
  For $i\in T$, write $\sigma^{\oplus i}$ for the vector obtained from $\sigma$ by negating $\sigma_i$.
  Note that $R(\sigma^{\oplus i},\allone) = r - \frac{2}{n} < \varepsilon - \frac{2}{n}$.
  We will find $n/50$ values $i\in T$ such that $R(\sigma^{\oplus i}, \sigma^*) \in I$, which ensures $\sigma^{\oplus i} \in \Sigma_n(\sigma^*,\varepsilon - \frac{2}{n},I)$.
  Write
  \begin{align*}
    T_+ &= \{i\in T : \sigma^*_i = 1\}\,, & 
    T_- &= \{i\in T : \sigma^*_i = -1\}\,.
  \end{align*}
  For all $i\in T_+$, we have $R(\sigma^{\oplus i}, \sigma^*) = q - \frac{2}{n}$, and for all $i\in T_-$, we have $R(\sigma^{\oplus i}, \sigma^*) = q + \frac{2}{n}$.
  If $q - \frac{2}{n}$ and $q + \frac{2}{n}$ are both in $I$, then $R(\sigma^{\oplus i}, \sigma^*) \in I$ for all $i\in T$, as desired.
  
  Otherwise, exactly one of $q - \frac{2}{n}$ and $q + \frac{2}{n}$ is in $I$, and $q$ is within $\frac{2}{n}$ of an endpoint of $I$.
  The assumption $I\in {\mathcal{I}}$ ensures $|q| \le 0.9 + \frac{2}{n}$.
  Note that
  \begin{align*}
    \frac{|T_+|}{n} &= \frac{1+r+p+q}{4}\,, &
    \frac{|T_-|}{n} &= \frac{1+r-p-q}{4}\,.
  \end{align*}
  Finally, note that $r-|p| \ge (\varepsilon - \frac{2}{n}) - \varepsilon = -\frac{2}{n}$.
  So,
  \[
    \frac{\min(|T_+|,|T_-|)}{n}
    \ge \frac{1+r-|p|-|q|}{4}
    \ge \frac{1}{40} - \frac{1}{n}
    \ge \frac{1}{50}\,.
  \]
  The conclusion follows by taking all $i\in T_+$ if $q - \frac{2}{n} \in I$, and all $i\in T_-$ if $q + \frac{2}{n} \in I$.
\end{proof}
\end{fact}
\begin{proof}[Proof of ~\cref{p:fe-conc-conditional-mean-lipschitz}~\cref{i:fe-conc-conditional-mean-lipschitz-one-replica}]
  Recall that we abbreviate
  \[
    \widetilde{\isingHamiltonian}_{G,\beta,t}(\sigma;\sigma^*,{\boldsymbol{g}}) = \beta \isingHamiltonian_{G}(\sigma) + t \langle \sigma,\sigma^*\rangle + \sqrt{t} \langle \sigma - \sigma^*, {\boldsymbol{g}}\rangle\,.
  \]
  Define
  \[
    {\breve F}_{n,\beta,t}(\sigma^*,I) = \E^\plantedTag_{\beta,t}\left[
      \frac1n \log \sum_{\sigma \in \Sigma_n(\sigma^*,\varepsilon-\frac{2}{n},I)}
      \exp(\widetilde{\isingHamiltonian}_{G,\beta,t}(\sigma;\sigma^*,{\boldsymbol{g}}))
      \,\Big|\, \sigma^*
    \right]\,.
  \]
  This is ${\widehat F}_{n,\beta,t}(\sigma^*,I)$, except with the sum over $\sigma \in \Sigma_n(\sigma^*,\varepsilon,I)$ replaced by a sum over $\sigma \in \Sigma_n(\sigma^*,\varepsilon-\frac{2}{n},I)$.
  We will write
  \[
    {\widehat F}_{n,\beta,t}(\sigma^*,I) - {\widehat F}_{n,\beta,t}(\widetilde{\sigma}^*,I)
    = \left(
      {\widehat F}_{n,\beta,t}(\sigma^*,I) - {\breve F}_{n,\beta,t}(\sigma^*,I)
    \right) + \left(
      {\breve F}_{n,\beta,t}(\sigma^*,I) - {\widehat F}_{n,\beta,t}(\widetilde{\sigma}^*,I)
    \right)
  \]
  and separately bound these terms.
  Note that
  \begin{align}\nonumber
    n \Big({\widehat F}_{n,\beta,t}(\sigma^*,I) &- {\breve F}_{n,\beta,t}(\sigma^*,I)\Big)
    = \E^\plantedTag_\beta \E_{\boldsymbol{g}} \left[
      \log \frac{
        \sum_{\sigma \in \Sigma_n(\sigma^*,\varepsilon,I)}
        \exp(\widetilde{\isingHamiltonian}_{G,\beta,t}(\sigma;\sigma^*,{\boldsymbol{g}}))
      }{
        \sum_{\widetilde{\sigma} \in \Sigma_n(\sigma^*,\varepsilon-\frac{2}{n},I)}
        \exp(\widetilde{\isingHamiltonian}_{G,\beta,t}(\widetilde{\sigma};\sigma^*,{\boldsymbol{g}}))
      }
      \,\Big|\, \sigma^*
    \right] \\
    \label{e:fe-conc-conditional-mean-lipschitz-one-replica-step1}
    &= \E^\plantedTag_\beta \E_{\boldsymbol{g}} \left[
      \log \left(1 + \frac{
        \sum_{\sigma \in \partial \Sigma_n(\sigma^*,\varepsilon,I)}
        \exp(\widetilde{\isingHamiltonian}_{G,\beta,t}(\sigma;\sigma^*,{\boldsymbol{g}}))
      }{
        \sum_{\widetilde{\sigma} \in \Sigma_n(\sigma^*,\varepsilon-\frac{2}{n},I)}
        \exp(\widetilde{\isingHamiltonian}_{G,\beta,t}(\widetilde{\sigma};\sigma^*,{\boldsymbol{g}}))
      }
      \right)
      \,\Big|\, \sigma^*
    \right]
  \end{align}
  For each $\widetilde{\sigma} \in \Sigma_n(\sigma^*,\varepsilon-\frac{2}{n},I)$, let 
  \[
    S(\widetilde{\sigma}) = \left\{
      \sigma \in \partial \Sigma_n(\sigma^*,\varepsilon,I) : \text{$\sigma$ has Hamming distance $1$ from $\widetilde{\sigma}$}
    \right\}\,.
  \]
  Write
  \begin{align*}
    {\dot F}(\widetilde{\sigma}) 
    = \sum_{\sigma \in S(\widetilde{\sigma})}
    \exp(\widetilde{\isingHamiltonian}_{G,\beta,t}(\sigma;\sigma^*,{\boldsymbol{g}}))
    &\le \exp(\widetilde{\isingHamiltonian}_{G,\beta,t}(\widetilde{\sigma};\sigma^*,{\boldsymbol{g}})) \cdot e^{d|\beta| + 2t}
    \sum_{\sigma \in S(\widetilde{\sigma})}
    \exp(
      \sqrt{t} \langle \sigma - \widetilde{\sigma}, {\boldsymbol{g}}\rangle
    ) \\
    &\le \exp(\widetilde{\isingHamiltonian}_{G,\beta,t}(\widetilde{\sigma};\sigma^*,{\boldsymbol{g}})) \cdot e^{d|\beta| + 2t}
    \sum_{i=1}^n e^{2\sqrt{t} |g_i|}\,.
  \end{align*}
  By \cref{f:hamming-nbrs}, $|S^{-1}(\sigma)| \ge n/50$ for all $\sigma \in \partial \Sigma_n(\sigma^*,\varepsilon,I)$.
  Thus,
  \begin{align*}
    \sum_{\sigma \in \partial \Sigma_n(\sigma^*,\varepsilon,I)}
    &\exp(\widetilde{\isingHamiltonian}_{G,\beta,t}(\sigma;\sigma^*,{\boldsymbol{g}}))
    \le \frac{50}{n} \sum_{\widetilde{\sigma} \in \Sigma_n(\sigma^*,\varepsilon - \frac{2}{n}, I)}
    {\dot F}(\widetilde{\sigma}) \\
    &\le \frac{50}{n} e^{d|\beta| + 2t}
    \left(\sum_{i=1}^n e^{2\sqrt{t} |g_i|}\right)
    \sum_{\widetilde{\sigma} \in \Sigma_n(\sigma^*,\varepsilon - \frac{2}{n}, I)}
    \exp(\widetilde{\isingHamiltonian}_{G,\beta,t}(\widetilde{\sigma};\sigma^*,{\boldsymbol{g}}))\,.
  \end{align*}
  Together with \eqref{e:fe-conc-conditional-mean-lipschitz-one-replica-step1} this implies
  \begin{align}\nonumber
    n \Big({\widehat F}_{n,\beta,t}(\sigma^*,I) &- {\breve F}_{n,\beta,t}(\sigma^*,I)\Big)
    \le \E_{\boldsymbol{g}} \left[
      \log\left(
        1 + \frac{50}{n} e^{d|\beta| + 2t}
        \sum_{i=1}^n e^{2\sqrt{t} |g_i|}
      \right)
    \right] \\
    \label{e:fe-conc-conditional-mean-lipschitz-one-replica-step2}
    &\le \log\left(
      1 + \frac{50}{n} e^{d|\beta| + 2t}
      \sum_{i=1}^n \E_{\boldsymbol{g}} e^{2\sqrt{t} |g_i|}
    \right)
    \le \log\left(1 + 100e^{d|\beta| + 4t}\right)\,.
  \end{align}
  For the remaining term, assume without loss that $\sigma^*$ and $\widetilde{\sigma}^*$ differ in coordinate $1$.
  For ${\boldsymbol{v}} \in \R^n$, write ${\boldsymbol{v}}^{\oplus 1}$ for ${\boldsymbol{v}}$ with $v_1$ negated.
  Note that if $\sigma \in \Sigma_n(\sigma^*,\varepsilon - \frac{2}{n}, I)$, then $\sigma^{\oplus 1} \in \Sigma_n(\widetilde{\sigma}^*,\varepsilon, I)$.
  Then, if $\E$ denotes expectation over any coupling of $(G,{\boldsymbol{g}}) \sim \nu^\plantedTag_\beta(\cdot | \sigma^*) \otimes {\mathcal{N}}(0,{\boldsymbol{I}}_n)$ and $({\widetilde G},{\widetilde {\boldsymbol{g}}}) \sim \nu^\plantedTag_\beta(\cdot | \widetilde{\sigma}^*) \otimes {\mathcal{N}}(0,{\boldsymbol{I}}_n)$,
  \begin{align*}
    n\left({\breve F}_{n,\beta,t}(\sigma^*,I) - {\widehat F}_{n,\beta,t}(\widetilde{\sigma}^*,I)\right)
    &= \E \left[
      \log \frac{
        \sum_{\sigma \in \Sigma_n(\sigma^*,\varepsilon - \frac{2}{n}, I)}
        \exp(\widetilde{\isingHamiltonian}_{G,\beta,t}(\sigma;\sigma^*,{\boldsymbol{g}}))
      }{
        \sum_{\widetilde{\sigma} \in \Sigma_n(\widetilde{\sigma}^*,\varepsilon, I)}
        \exp(\widetilde{\isingHamiltonian}_{{\widetilde G},\beta,t}(\widetilde{\sigma};\widetilde{\sigma}^*,{\widetilde {\boldsymbol{g}}}))
      }
    \right] \\
    &\le \E \left[
      \log \frac{
        \sum_{\sigma \in \Sigma_n(\sigma^*,\varepsilon - \frac{2}{n}, I)}
        \exp(\widetilde{\isingHamiltonian}_{G,\beta,t}(\sigma;\sigma^*,{\boldsymbol{g}}))
      }{
        \sum_{\sigma \in \Sigma_n(\sigma^*,\varepsilon - \frac{2}{n}, I)}
        \exp(\widetilde{\isingHamiltonian}_{{\widetilde G},\beta,t}(\sigma^{\oplus 1};(\sigma^*)^{\oplus 1},{\widetilde {\boldsymbol{g}}}))
      }
    \right]\,.
  \end{align*}
  Let $\pi$ be the coupling of $G \sim \nu^\plantedTag_\beta(\cdot | \sigma^*)$ and ${\widetilde G} \sim \nu^\plantedTag_\beta(\cdot | \widetilde{\sigma}^*)$ provided by ~\cref{l:wass-stability-of-sbm}~\cref{i:wass-stability-of-sbm-one-replica}, so that
  \[
    \E_\pi |E(G) \triangle E({\widetilde G})| \le C_0d\,.
  \]
  Couple ${\boldsymbol{g}},{\widetilde {\boldsymbol{g}}}$ by ${\widetilde {\boldsymbol{g}}} = {\boldsymbol{g}}^{\oplus 1}$, independently of $(G,{\widetilde G})$.
  For each $\sigma \in \Sigma_n(\sigma^*,\varepsilon - \frac{2}{n}, I)$, we have
  \[
    |\widetilde{\isingHamiltonian}_{G,\beta,t}(\sigma;\sigma^*,{\boldsymbol{g}})
    - \widetilde{\isingHamiltonian}_{{\widetilde G},\beta,t}(\sigma^{\oplus 1};(\sigma^*)^{\oplus 1},{\widetilde {\boldsymbol{g}}})|
    = |\beta| |\isingHamiltonian_{G}(\sigma) - \isingHamiltonian_{{\widetilde G}}(\sigma^{\oplus 1})|
    \le |\beta| \left(|E(G) \triangle E({\widetilde G})| + d\right)\,,
  \]
  since changing $G$ to ${\widetilde G}$ changes at most $|E(G) \triangle E({\widetilde G})|$ edge contributions, and flipping coordinate $1$ changes only the edges incident to that coordinate.
  Applying this bound inside the preceding log-ratio gives
  \begin{equation}\label{e:fe-conc-conditional-mean-lipschitz-one-replica-step3}
    n\left({\breve F}_{n,\beta,t}(\sigma^*,I) - {\widehat F}_{n,\beta,t}(\widetilde{\sigma}^*,I)\right)
    \le |\beta| \E_{\pi,{\boldsymbol{g}}}\left[|E(G) \triangle E({\widetilde G})| + d\right]
    \le (C_0+1) |\beta| d\,.
  \end{equation}
  Combining \eqref{e:fe-conc-conditional-mean-lipschitz-one-replica-step2} and \eqref{e:fe-conc-conditional-mean-lipschitz-one-replica-step3} yields
  \[
    n\left({\widehat F}_{n,\beta,t}(\sigma^*,I)
    - {\widehat F}_{n,\beta,t}(\widetilde{\sigma}^*,I)\right)
    \le \log\left(1 + 100e^{d|\beta| + 4t}\right) + (C_0+1) |\beta| d
    \le C(d)(1+t)\,.
  \]
  The reverse inequality follows by symmetry.
\end{proof}
For the proof of ~\cref{p:fe-conc-conditional-mean-lipschitz}~\cref{i:fe-conc-conditional-mean-lipschitz-two-replica}, we introduce the notation
\begin{align*}
  \Sigma_n(\sigma^*,{\sigma'}^*,\varepsilon,I)
  &= \left\{
    \sigma \in \Sigma_n(\sigma^*,\varepsilon,I) : |R(\sigma,{\sigma'}^*)| < \varepsilon
  \right\} \\
  &= \left\{
    \sigma \in \Sigma_n : \max(|R(\sigma,\allone)|,|R(\sigma,{\sigma'}^*)|) < \varepsilon, 
    R(\sigma^*,\sigma) \in I
  \right\}\,.
\end{align*}
The following fact is analogous to Fact~\cref{f:hamming-nbrs}.
\begin{fact}\label{f:hamming-nbrs-two-replica}
  Suppose $(\sigma^*,{\sigma'}^*) \in \Sigma_{n,2}(\varepsilon)$ and $I \in {\mathcal{I}}$.
  Write $\partial \Sigma_n(\sigma^*,{\sigma'}^*,\varepsilon,I) = \Sigma_n(\sigma^*,{\sigma'}^*,\varepsilon,I) \setminus \Sigma_n(\sigma^*,{\sigma'}^*,\varepsilon-\frac{2}{n},I)$.
  If $\sigma \in \partial \Sigma_n(\sigma^*,{\sigma'}^*,\varepsilon,I)$, then there are at least $n^3/10^6$ points $\widetilde{\sigma} \in \Sigma_n(\sigma^*,{\sigma'}^*,\varepsilon - \frac{2}{n},I)$ with Hamming distance $3$ from $\sigma$.
\end{fact}
\begin{remark}
  The most direct adaptation of Fact~\cref{f:hamming-nbrs}, that there are at least $\Omega(n)$ points $\widetilde{\sigma} \in \Sigma_n(\sigma^*,{\sigma'}^*,\varepsilon - \frac{2}{n},I)$ with Hamming distance $1$ from $\sigma$, is not true.
  In fact, the number of such points can be $0$.
  To see this, consider $n$ a multiple of $4$ and a partition $[n] = S_{+++} \cup S_{+--} \cup S_{-+-} \cup S_{--+}$ of sizes
  \begin{align*}
    \frac{|S_{+++}|}{n} &= \frac14 + \frac{\varepsilon'}{2}\,, & 
    \frac{|S_{+--}|}{n} &= \frac14\,, & 
    \frac{|S_{-+-}|}{n} &= \frac14\,, & 
    \frac{|S_{--+}|}{n} &= \frac14 - \frac{\varepsilon'}{2}\,,
  \end{align*}
  where $\varepsilon' \in [\varepsilon - \frac{2}{n}, \varepsilon)$ is set so that $\frac{\varepsilon'}{2}$ is a multiple of $\frac{1}{n}$.
  Consider $I = [0,\frac12]$ and $\sigma,\sigma^*,{\sigma'}^* \in \Sigma_n$ such that for $i\in S_\omega$ we have $(\sigma_i,\sigma^*_i,{\sigma'}^*_i) = (\omega_1,\omega_2,\omega_3)$.
  We easily check that $(\sigma^*,{\sigma'}^*) \in \Sigma_{n,2}(\varepsilon)$ and $\sigma \in \partial \Sigma_n(\sigma^*,{\sigma'}^*,\varepsilon,I)$, with 
  \begin{align*}
    R(\sigma,\allone) = R(\sigma,{\sigma'}^*)
    &= \varepsilon' \in [\varepsilon - \textstyle\frac{2}{n}, \varepsilon)\,, & 
    R(\sigma,\sigma^*) &= 0 \in I\,.
  \end{align*}
  If $\widetilde{\sigma}$ is a Hamming distance $1$ neighbor of $\sigma$ satisfying $\max(|R(\sigma,\allone)|, |R(\sigma,{\sigma'}^*)|) < \varepsilon - \frac{2}{n}$, it must differ from $\sigma$ in a coordinate $i\in S_{+++}$.
  But then $R(\widetilde{\sigma},\sigma^*) = -\frac{2}{n} \not\in I$, so no such $\widetilde{\sigma}$ exist.
  This is why Fact~\cref{f:hamming-nbrs-two-replica} considers points with Hamming distance $3$.
\end{remark}
\begin{proof}[Proof of Fact~\cref{f:hamming-nbrs-two-replica}]
  Let $p = R(\sigma,\allone)$, $q = R(\sigma,\sigma^*)$, and $r = R(\sigma,{\sigma'}^*)$.
  For $\omega_1,\omega_2 \in \{\pm 1\}$, let
  \begin{align*}
    S_{p,q}(\omega_1,\omega_2) &= \{i \in [n] : \sigma_i = \omega_1, \sigma_i\sigma^*_i = \omega_2\}\,, \\
    S_{p,r}(\omega_1,\omega_2) &= \{i \in [n] : \sigma_i = \omega_1, \sigma_i{\sigma'}^*_i = \omega_2\}\,, \\
    S_{q,r}(\omega_1,\omega_2) &= \{i \in [n] : \sigma_i\sigma^*_i = \omega_1, \sigma_i{\sigma'}^*_i = \omega_2\}\,.
  \end{align*}
  Note that
  \begin{equation}\label{e:hamming-nbrs-two-replica}
    \frac{|S_{p,q}(\omega_1,\omega_2)|}{n} = \frac{1}{4} \left(1 + \omega_1 p + \omega_2 q + \omega_1 \omega_2 R(\sigma^*,\allone)\right)\,,
  \end{equation}
  and analogously for $S_{p,r}$ and $S_{q,r}$.
  Also define
  \begin{align*}
    \Sigma_p &= \{\omega \in \{\pm 1\} : |p - \textstyle \frac{6}{n} \omega| < \varepsilon\}\,, \\
    \Sigma_q &= \{\omega \in \{\pm 1\} : |q - \textstyle \frac{6}{n} \omega| < \varepsilon\}\,, \\
    \Sigma_r &= \{\omega \in \{\pm 1\} : r - \textstyle \frac{6}{n} \omega \in I \cup (-\infty,-1) \cup (1,+\infty)\}\,.
  \end{align*}
  That is, $\Sigma_q$ is the set of $\omega$ such that, if we obtain $\widetilde{\sigma}$ from $\sigma$ by negating between one and three coordinates $i$ with $\sigma_i\sigma^*_i = \omega$, the constraint $|R(\widetilde{\sigma},\sigma^*)| < \varepsilon$ holds.
  The analogous interpretation holds for $\Sigma_p$ and the constraint $|R(\widetilde{\sigma},\allone)| < \varepsilon$, and for $\Sigma_r$ and the constraint $R(\widetilde{\sigma},\sigma^*) \in I$.
  Note that in the last case, if $1$ is an endpoint of $I$ and $r \in (1-\frac{6}{n}, 1]$, we still include $-1 \in \Sigma_r$ because it is not possible to negate coordinates of $\sigma$ to obtain $\widetilde{\sigma}$ with $R(\widetilde{\sigma},\sigma^*) > 1$. Similarly, if $-1$ is an endpoint of $I$ and $r \in [-1,-1+\frac{6}{n})$, we still include $+1 \in \Sigma_r$.

  Note that $p,q$ can each be close to only one boundary of $(-\varepsilon,\varepsilon)$, and $r$ can be close to only one boundary of $I$, so each of $\Sigma_p,\Sigma_q,\Sigma_r$ is nonempty.
  Thus we may choose $\omega_p \in \Sigma_p$, $\omega_q \in \Sigma_q$, $\omega_r \in \Sigma_r$.
  We now perform casework on $|\Sigma_r|$.

  \textit{Case 1: $|\Sigma_r| = 1$}.
  By the above discussion, for any triples
  \[
    (i_1,i_2,i_3) \in 
    S_{p,q}(\omega_p,\omega_q) 
    \times S_{p,r}(\omega_p,\omega_r)
    \times S_{q,r}(\omega_q,\omega_r)
  \]
  where $i_1,i_2,i_3$ are distinct, the configuration $\widetilde{\sigma}$ obtained from $\sigma$ by negating coordinates $i_1,i_2,i_3$ is in $\Sigma_n(\sigma^*,{\sigma'}^*,\varepsilon - \frac{2}{n},I)$.
  We now count the number of $\widetilde{\sigma}$ of this form.
  Since $|\Sigma_r| = 1$, $r$ is within $\frac{6}{n}$ of an endpoint of $I$ other than $\pm 1$. 
  By the definition of $I$, we have $|r| \le 0.9 + \frac{6}{n} \le 0.91$.
  Together with the assumptions $|p|,|q| < \varepsilon$ and $(\sigma^*,{\sigma'}^*) \in \Sigma_{n,2}(\varepsilon)$, \eqref{e:hamming-nbrs-two-replica} this implies
  \begin{equation}\label{e:hamming-nbrs-two-replica-balanced}
    \frac{|S_{p,q}(\omega_p,\omega_q)|}{n} \ge \frac{1}{4} \left(1 - 3\varepsilon\right) \ge \frac15
  \end{equation}
  and the analogous formulas for $|S_{p,r}(\omega_1,\omega_2)|$ and $|S_{q,r}(\omega_1,\omega_2)|$ imply
  \[
    \min\left(
      \frac{|S_{p,r}(\omega_p,\omega_r)|}{n}, 
      \frac{|S_{q,r}(\omega_q,\omega_r)|}{n}
    \right) \ge \frac{1}{4} \left(1 - 2\varepsilon - 0.91\right)
    \ge \frac{1}{50}\,.
  \]
  This yields $\frac{n^3}{5\cdot 50^2}$ ordered triples $(i_1,i_2,i_3)$, where $i_1,i_2,i_3$ are not necessarily distinct.
  The number of $(i_1,i_2,i_3)$ with non-distinct entries is at most $O(n^2)$, and at most $3!=6$ triples $(i_1,i_2,i_3)$ can yield the same $\widetilde{\sigma}$.
  Thus the number of $\widetilde{\sigma}$ is at least
  \[
    \frac{1}{6} \left(\frac{n^3}{5\cdot 50^2} - O(n^2)\right) \ge \frac{n^3}{10^6}\,.
  \]
  \textit{Case 2: $|\Sigma_r| = 2$}.
  In this case, the constraint $R(\widetilde{\sigma},\sigma^*) \in I$ is satisfied for any $\widetilde{\sigma}$ with Hamming distance $3$ to $\sigma$.
  Then, for any distinct $i_1,i_2,i_3 \in S_{p,q}(\omega_p,\omega_q)$, the configuration $\widetilde{\sigma}$ obtained from $\sigma$ by negating coordinates $i_1,i_2,i_3$ is in $\Sigma_n(\sigma^*,{\sigma'}^*,\varepsilon - \frac{2}{n},I)$.
  The estimate \eqref{e:hamming-nbrs-two-replica-balanced} remains valid in this case.
  By a similar counting argument as above the number of such $\widetilde{\sigma}$ is at least
  \[
    \frac{1}{6} \left(\frac{n^3}{5^3} - O(n^2)\right) \ge \frac{n^3}{10^6}\,. \qedhere
  \]
\end{proof}

\begin{proof}[Proof of ~\cref{p:fe-conc-conditional-mean-lipschitz}~\cref{i:fe-conc-conditional-mean-lipschitz-two-replica}]
  Define
  \[
    {\breve F}^2_{n,\beta,t}(\sigma^*,{\sigma'}^*,I) = \E^{\plantedTag,2}_{\beta,t}\left[
      \frac1n \log \sum_{\sigma \in \Sigma_n(\sigma^*,{\sigma'}^*,\varepsilon-\frac{2}{n},I)}
      \exp(\widetilde{\isingHamiltonian}_{G,\beta,t}(\sigma;\sigma^*,{\boldsymbol{g}})) 
      \,\Big|\, \sigma^*,{\sigma'}^*
    \right]\,.
  \]
  This is ${\widehat F}_{n,\beta,t}^2(\sigma^*,{\sigma'}^*,I)$, except with the sum over $\sigma \in \Sigma_n(\sigma^*,{\sigma'}^*,\varepsilon,I)$ replaced by a sum over $\sigma \in \Sigma_n(\sigma^*,{\sigma'}^*,\varepsilon-\frac{2}{n},I)$.
  We write
  \begin{align*}
    {\widehat F}_{n,\beta,t}^2(\sigma^*,{\sigma'}^*,I) - {\widehat F}_{n,\beta,t}^2(\widetilde{\sigma}^*,\widetilde{\sigma'}^*,I)
    &= \left(
      {\widehat F}_{n,\beta,t}^2(\sigma^*,{\sigma'}^*,I) - {\breve F}^2_{n,\beta,t}(\sigma^*,{\sigma'}^*,I)
    \right) \\
    &\qquad + \left(
      {\breve F}^2_{n,\beta,t}(\sigma^*,{\sigma'}^*,I) - {\widehat F}_{n,\beta,t}^2(\widetilde{\sigma}^*,\widetilde{\sigma'}^*,I)
    \right)
  \end{align*}
  and separately bound these terms.
  For each $\widetilde{\sigma} \in \Sigma_n(\sigma^*,{\sigma'}^*,\varepsilon-\frac{2}{n},I)$, let 
  \[
    S(\widetilde{\sigma}) = \left\{
      \sigma \in \partial \Sigma_n(\sigma^*,{\sigma'}^*,\varepsilon,I) : \text{$\sigma$ has Hamming distance $3$ from $\widetilde{\sigma}$}
    \right\}\,.
  \]
  Similarly to above, write
  \[
    {\dot F}(\widetilde{\sigma}) 
    = \sum_{\sigma \in S(\widetilde{\sigma})}
    \exp(\widetilde{\isingHamiltonian}_{G,\beta,t}(\sigma;\sigma^*,{\boldsymbol{g}}))
    \le \exp(\widetilde{\isingHamiltonian}_{G,\beta,t}(\widetilde{\sigma};\sigma^*,{\boldsymbol{g}})) \cdot e^{3d|\beta| + 6t}
    \left(\sum_{i=1}^n e^{2\sqrt{t} |g_i|}\right)^3\,.
  \]
  By Fact~\cref{f:hamming-nbrs-two-replica}, $|S^{-1}(\sigma)| \ge n^3/10^6$ for all $\sigma \in \partial \Sigma_n(\sigma^*,{\sigma'}^*,\varepsilon,I)$.
  Thus
  \begin{align*}
    \sum_{\sigma \in \partial \Sigma_n(\sigma^*,{\sigma'}^*,\varepsilon,I)}
    &\exp(\widetilde{\isingHamiltonian}_{G,\beta,t}(\sigma;\sigma^*,{\boldsymbol{g}}))
    \le \frac{10^6}{n^3} \sum_{\widetilde{\sigma} \in \Sigma_n(\sigma^*,{\sigma'}^*,\varepsilon - \frac{2}{n}, I)}
    {\dot F}(\widetilde{\sigma}) \\
    &\le \frac{10^6}{n^3} e^{3d|\beta| + 6t}
    \left(\sum_{i=1}^n e^{2\sqrt{t} |g_i|}\right)^3
    \sum_{\widetilde{\sigma} \in \Sigma_n(\sigma^*,{\sigma'}^*,\varepsilon - \frac{2}{n}, I)}
    \exp(\widetilde{\isingHamiltonian}_{G,\beta,t}(\widetilde{\sigma};\sigma^*,{\boldsymbol{g}}))\,,
  \end{align*}
  which implies (similarly to above)
  \begin{align*}
    &n \Big({\widehat F}^2_{n,\beta,t}(\sigma^*,{\sigma'}^*,I) - {\breve F}^2_{n,\beta,t}(\sigma^*,{\sigma'}^*,I)\Big) \\
    &= \E^{\plantedTag,2}_\beta \E_{\boldsymbol{g}} \left[
      \log \left(1 + \frac{
        \sum_{\sigma \in \partial \Sigma_n(\sigma^*,{\sigma'}^*,\varepsilon,I)}
        \exp(\widetilde{\isingHamiltonian}_{G,\beta,t}(\sigma;\sigma^*,{\boldsymbol{g}}))
      }{
        \sum_{\widetilde{\sigma} \in \Sigma_n(\sigma^*,{\sigma'}^*,\varepsilon-\frac{2}{n},I)}
        \exp(\widetilde{\isingHamiltonian}_{G,\beta,t}(\widetilde{\sigma};\sigma^*,{\boldsymbol{g}}))
      }
      \right)
      \,\Big|\, \sigma^*,{\sigma'}^*
    \right] \\
    &\le \E_{\boldsymbol{g}} \left[\log\left(
      1 + \frac{10^6}{n^3} e^{3d|\beta| + 6t}
      \left(\sum_{i=1}^n e^{2\sqrt{t} |g_i|}\right)^3
    \right)\right] \\
    &\le 3\E_{\boldsymbol{g}} \left[\log\left(
      1 + \frac{100}{n} e^{d|\beta| + 2t}
      \sum_{i=1}^n e^{2\sqrt{t} |g_i|}
    \right)\right]
    \le 3\log\left(
      1 + 200 e^{d|\beta| + 4t}
    \right)\,.
  \end{align*}
  It remains to compare ${\breve F}^2_{n,\beta,t}(\sigma^*,{\sigma'}^*,I)$ and ${\widehat F}_{n,\beta,t}^2(\widetilde{\sigma}^*,\widetilde{\sigma'}^*,I)$.
  Let $\pi$ be the coupling of $G \sim \nu_\beta^{\plantedTag,2}(\cdot | \sigma^*,{\sigma'}^*)$ and ${\widetilde G} \sim \nu_\beta^{\plantedTag,2}(\cdot | \widetilde{\sigma}^*,\widetilde{\sigma'}^*)$ provided by ~\cref{l:wass-stability-of-sbm}~\cref{i:wass-stability-of-sbm-two-replica}, so that
  \[
    \E_\pi |E(G) \triangle E({\widetilde G})| \le C_0d\,.
  \]
  We consider two cases.
  First suppose $(\sigma^*,{\sigma'}^*)$ and $(\widetilde{\sigma}^*,\widetilde{\sigma'}^*)$ differ in $\sigma^*$, and assume without loss that $\sigma^*_1 \neq \widetilde{\sigma}^*_1$.
  Then $\widetilde{\sigma}^* = (\sigma^*)^{\oplus 1}$ and $\widetilde{\sigma'}^* = {\sigma'}^*$, and $\sigma \mapsto \sigma^{\oplus 1}$ is an injective map from $\Sigma_n(\sigma^*,{\sigma'}^*,\varepsilon - \frac{2}{n},I)$ to $\Sigma_n(\widetilde{\sigma}^*,\widetilde{\sigma'}^*,\varepsilon,I)$.
  Couple ${\boldsymbol{g}},{\widetilde {\boldsymbol{g}}}$ by ${\widetilde {\boldsymbol{g}}} = {\boldsymbol{g}}^{\oplus 1}$, independently of $(G,{\widetilde G})$.
  As in the proof of part \cref{i:fe-conc-conditional-mean-lipschitz-one-replica}, applying the Hamiltonian comparison inside the log-ratio gives
  \[
    n\left({\breve F}^2_{n,\beta,t}(\sigma^*,{\sigma'}^*,I) - {\widehat F}_{n,\beta,t}^2(\widetilde{\sigma}^*,\widetilde{\sigma'}^*,I)\right)
    \le |\beta| \E_{\pi,{\boldsymbol{g}}}\left[|E(G) \triangle E({\widetilde G})| + d\right]
    \le (C_0+1) |\beta| d\,.
  \]
  Otherwise, suppose $(\sigma^*,{\sigma'}^*)$ and $(\widetilde{\sigma}^*,\widetilde{\sigma'}^*)$ differ in ${\sigma'}^*$, and assume without loss that ${\sigma'}^*_1 \neq \widetilde{\sigma'}^*_1$.
  Then $\widetilde{\sigma}^* = \sigma^*$, and the identity map is an injection from $\Sigma_n(\sigma^*,{\sigma'}^*,\varepsilon - \frac{2}{n},I)$ to $\Sigma_n(\widetilde{\sigma}^*,\widetilde{\sigma'}^*,\varepsilon,I)$.
  The remaining argument is analogous to the first case, except that we couple ${\widetilde {\boldsymbol{g}}}={\boldsymbol{g}}$ instead of ${\widetilde {\boldsymbol{g}}} = {\boldsymbol{g}}^{\oplus 1}$.
  In either case this shows
  \[
    n\left({\breve F}^2_{n,\beta,t}(\sigma^*,{\sigma'}^*,I) - {\widehat F}_{n,\beta,t}^2(\widetilde{\sigma}^*,\widetilde{\sigma'}^*,I)\right)
    \le (C_0+1) |\beta| d\,.
  \]
  Combining this with the boundary estimate gives
  \[
    n\left({\widehat F}_{n,\beta,t}^2(\sigma^*,{\sigma'}^*,I) - {\widehat F}_{n,\beta,t}^2(\widetilde{\sigma}^*,\widetilde{\sigma'}^*,I)\right)
    \le 3\log\left(1 + 200 e^{d|\beta| + 4t}\right) + (C_0+1)|\beta|d
    \le C(d)(1+t)\,.
  \]
  The reverse inequality is symmetric.
\end{proof}

\section{Moments bounds for the reweighting factor}\label{sec:reweight-bounded-moments}
In this section, we sketch the proof of \cref{thm:bounded-moments}. The conceptual ideas are identical to those used to prove \cref{thm:reweighted-2ndmoment}. The main technical differences stem from the fact that here, we do not need to work with ``planted'' measure as in \cref{subsec:2nd-moment-complex} (which arose due to the presence of the truncated partition function inside the expectations), but we do need to bound higher moments. Throughout, we point out the necessary adjustments to establish \cref{thm:bounded-moments}.

Without loss of generality, we may consider only the even moments, since for any random variable $X$, we have $\E\wrapb{\abs{X}^{K}} \leq \sqrt{\E\wrapb{\abs{X}^{2K}}}$. Recalling the definitions of $\mathscr{G},\mathscr{B}$ from \cref{eq:good-bad-graphs} and the definition of the reweighting factor from \cref{eq:reweight-factor}, observe that
\begin{align*}
    \E\wrapb{\abs{A_{L,m}(\beta)}^{2K}} &= \exp\wrapp{K \sum_{k=1}^{L} \lambda_{k} \cdot 2\re \delta_{k}} \\
    &\qquad \cdot \underset{(\ast)}{\underbrace{\sum_{\Gamma_{1},\dots,\Gamma_{2K} \in \mathscr{G}(\leq m)} \Pr[\Gamma_{1} \cup \dotsb \cup \Gamma_{2K} \subseteq M] \cdot \prod_{j=1}^{K} \prod_{C \in \Gamma_{j}} b_{\abs{C}} \cdot \prod_{j=K+1}^{2K} \prod_{C \in \Gamma_{j}} \overline{b_{\abs{C}}}}}.
\end{align*}
The key technical result is the following.
\begin{proposition}\label{prop:ast-bound}
Under the assumptions of \cref{thm:bounded-moments},
\begin{align*}
    (\ast) &= (1 \pm o(1)) \cdot \prod_{k=1}^{L} \E_{\Lambda_{k} \sim \mathsf{Poi}(\lambda_{k})}\wrapb{\abs{1 + \delta_{k}}^{-2K \cdot \Lambda_{k}}} = (1 \pm o(1)) \cdot \exp\wrapp{\sum_{k=1}^{L} \lambda_{k} \cdot \wrapp{\abs{1 + \delta_{k}}^{-2K} - 1}}.
\end{align*}
\end{proposition}
We will also need the following simple analytic lemma, which follows from direct Taylor expansion; we omit the proof.
\begin{lemma}\label{lem:analytic-error-small}
For every $K,\delta > 0$, there exists a constant $C = C(K,\delta) > 0$ such that for all $z \in \C$ satisfying $\abs{1 + z} \geq \delta$, we have
\begin{align*}
    \abs{1 + z}^{-2K} - 1 + 2K \re z \leq C \cdot \abs{z}^{2}.
\end{align*}
\end{lemma}
Combining \cref{prop:ast-bound} with the first display in this section yields
\begin{align*}
    \E\wrapb{\abs{A_{L,m}(\beta)}^{2K}} &= (1 \pm o(1)) \cdot \exp\wrapp{\sum_{k=1}^{L} \lambda_{k} \cdot \wrapp{\abs{1 + \delta_{k}}^{-2K} - 1 + 2K \re \delta_{k}}} \tag{\cref{prop:ast-bound}} \\
    &= (1 \pm o(1)) \cdot \exp\wrapp{O_{K,d}(1) \cdot \sum_{k=1}^{L} \lambda_{k} \cdot \abs{\delta_{k}}^{2}} \tag{\cref{lem:analytic-error-small} and using $\abs{\delta_{k}} \leq \abs{\delta_{2}} < 1$ for $\beta$ in a small neighborhood of $[0,(1 - \eta)\beta_{\mathsf{KS}}]$} \\
    &\leq (1 \pm o(1)) \cdot \exp\wrapp{O_{K,d}(1) \cdot \sum_{k=0}^{\infty} \wrapp{(d-1) \cdot \abs{\delta_{2}}}^{k}} \\
    &\leq O_{K,d,\eta}(1) \tag{$\beta$ is sufficiently close to $[0,(1-\eta)\beta_{\mathsf{KS}}]$ so that $(d-1) \cdot \abs{\delta_{2}} < 1 - \eta/2$}.
\end{align*}
This proves \cref{thm:bounded-moments}, assuming \cref{prop:ast-bound}. We now prove \cref{prop:ast-bound}, which follows a similar argument to the proof of \cref{thm:reweighted-2ndmoment}. The only change is the fact that we are now looking at the $2K$th moment, i.e. we must sum over $2K$ collections of cycles $\Gamma_{1},\dots,\Gamma_{2K}$, and we do not need to invoke \cref{prop:2nd-moment}.

\subsection{Proof (Sketch) of \cref{prop:ast-bound}}
As in the proof of \cref{thm:reweighted-2ndmoment}, the dominant contributions to $(\ast)$ will come from $\Gamma_{1},\dots,\Gamma_{2K}$ for which all pairs of cycles $C,C' \in \Gamma_{1} \cup \dotsb \cup \Gamma_{2K}$ are either vertex-disjoint or identical. To capture this, for $\emptyset \subsetneq S \subseteq [2K]$, we will write $\bm{c}(S) = (\bm{c}_{1}(S),\dots,\bm{c}_{L}(S)) \in \N^{L}$ for the number of $k$-cycles which appear simultaneously in $\Gamma_{i}$ for all $i \in S$ and do not appear in $\Gamma_{i}$ for any $i \notin S$. We will collect these vectors into a function $\bm{c} : 2^{[2K]} \setminus \{\emptyset\} \to \N^{L}$. For any such $\bm{c}$, we write
\begin{align*}
    \varphi(\bm{c}) &\defeq \sum_{\substack{\Gamma^{(S)} \in \mathscr{G}(\bm{c}(S)) \\ \forall \emptyset \subsetneq S \subseteq [2K]}} \Pr\wrapb{\bigcup_{\emptyset \subsetneq S \subseteq [2K]} \Gamma^{(S)} \subseteq M} \\
    \psi(\bm{c}) &\defeq \prod_{k=1}^{L} \prod_{\emptyset \subsetneq S \subseteq [2K]} \frac{\lambda_{k}^{\bm{c}_{k}(S)}}{\bm{c}_{k}(S)!}.
\end{align*}
Then
\begin{align*}
    (\ast) = \sum_{\substack{\bm{c} : 2^{[2K]} \setminus \{\emptyset\} \to \N^{L} \\ \forall i \in [2K], \sum_{S \ni i} \abs{\bm{c}(S)} \leq m}} \varphi(\bm{c}) \cdot \prod_{k=1}^{L} \prod_{\emptyset \subsetneq S \subseteq [2K]} b_{k}^{\abs{S \cap \{1,\dots,K\}} \cdot \bm{c}_{k}(S)} \cdot \overline{b_{k}}^{\abs{S \cap \{K+1,\dots,2K\}} \cdot \bm{c}_{k}(S)}.
\end{align*}
Our goal is to show that $(\ast)$ is close to the following:
\begin{align*}
    &\sum_{\bm{c} : 2^{[2K]} \setminus \{\emptyset\} \to \N^{L}} \psi(\bm{c}) \cdot \prod_{k=1}^{L} \prod_{\emptyset \subsetneq S \subseteq [2K]} b_{k}^{\abs{S \cap \{1,\dots,K\}} \cdot \bm{c}_{k}(S)} \cdot \overline{b_{k}}^{\abs{S \cap \{K+1,\dots,2K\}} \cdot \bm{c}_{k}(S)} \\
    &= \exp\wrapp{\sum_{k=1}^{L} \lambda_{k} \sum_{\emptyset \subsetneq S \subseteq [2K]} b_{k}^{\abs{S \cap \{1,\dots,K\}}}\overline{b_{k}}^{\abs{S \cap \{K+1,\dots,2K\}}}} \\
    &= \exp\wrapp{\sum_{k=1}^{L} \lambda_{k} \cdot \wrapp{\abs{1 + b_{k}}^{2K} - 1}} \\
    &= \exp\wrapp{\sum_{k=1}^{L} \lambda_{k} \cdot \wrapp{\abs{1 + \delta_{k}}^{-2K} - 1}}.
\end{align*}
As before, we will achieve this via termwise comparison, i.e. we will show the following.
\begin{proposition}\label{prop:bddmom-termwise}
Fix $L,m \in \N$. Then for all fixed $K \in \N$ (independent of $N$) and $\bm{c} : 2^{[2K]} \setminus \{\emptyset\} \to \N^{L}$ satisfying $\sum_{S \ni i} \abs{\bm{c}(S)} \leq m$ for all $i \in [2K]$, we have
\begin{align*}
    \varphi(\bm{c}) = \wrapp{1 \pm \frac{\widetilde{\mathcal{E}}(\eta,L,m)}{n}} \cdot \psi(\bm{c})
\end{align*}
where
\begin{align*}
    \widetilde{\mathcal{E}}(\eta,L,m) \defeq O\wrapp{\exp\wrapp{e^{O_{d}(L)} + O(m\log m)}}.
\end{align*}
\end{proposition}
The proof that \cref{prop:bddmom-termwise} implies \cref{prop:ast-bound} is identical, mutatis mutandis, to the proof that \cref{prop:uniform-converge-small-Gamma} implies \cref{thm:reweighted-2ndmoment}. We omit this calculation for brevity. To prove \cref{prop:bddmom-termwise}, following what was executed in \cref{subsec:uniform-converge-small-Gamma}, let us introduce the intermediate quantity
\begin{align*}
    \varphi^{\mathscr{G}}(\bm{c}) &\defeq \sum_{\substack{\Gamma^{(S)} \in \mathscr{G}(\bm{c}(S)) \\ \forall \emptyset \subsetneq S \subseteq [2K] \\ \bigcup_{\emptyset \subsetneq S \subseteq [2K]} \Gamma^{(S)} \in \mathscr{G}}} \Pr\wrapb{\bigcup_{\emptyset \subsetneq S \subseteq [2K]} \Gamma^{(S)} \subseteq M}.
\end{align*}

\cref{prop:bddmom-termwise} is an immediate consequence of the following pair of lemmas. Note that since we require that $\sum_{S \ni i} \abs{\bm{c}(S)} \leq m$ for all $i \in [2K]$, the upper bound $m$ already depends implicitly on $K$. This explains why the estimates below only involve $m$, and $K$ does not explicitly appear anywhere.
\begin{lemma}\label{lem:good-K}
Under the assumptions of \cref{prop:bddmom-termwise},
\begin{align*}
    \varphi^{\mathscr{G}}(\bm{c}) = \wrapp{1 \pm O\wrapp{\frac{L^{3}m^{2}}{n}}} \cdot \psi(\bm{c}).
\end{align*}
\end{lemma}
\begin{lemma}\label{lem:bad-K}
Under the assumptions of \cref{prop:bddmom-termwise},
\begin{align*}
    \abs{\varphi(\bm{c}) - \varphi^{\mathscr{G}}(\bm{c})} \leq \frac{O\wrapp{\exp\wrapp{e^{O_{d}(L)} + O(m\log m)}}}{n} \cdot \psi(\bm{c}).
\end{align*}
\end{lemma}

\begin{proof}[Proof Sketch of \cref{lem:good-K}]
Reasoning as in the proof of \cref{lem:good}, we have
\begin{align*}
    \varphi^{\mathscr{G}}(\bm{c}) = \exp\wrapp{\pm O\wrapp{\frac{m^{2}}{n}}} \cdot \wrapp{\frac{1}{dn}}^{\sum_{\emptyset \subsetneq S \subseteq [2K]} \abs{\bm{c}(S)}} \cdot \abs{\widehat{\mathscr{G}}(\bm{c})},
\end{align*}
where
\begin{align*}
    \widehat{\mathscr{G}}(\bm{c}) \defeq \wrapc{(\Gamma^{(S)})_{\emptyset \subsetneq S \subseteq [2K]} : \substack{\bm{c}(\Gamma^{(S)}) = \bm{c}(S), \forall \emptyset \subsetneq S \subseteq [2K] \\ \bigcup_{\emptyset \subsetneq S \subseteq [2K]} \Gamma^{(S)} \in \mathscr{G}}}.
\end{align*}
Following the proof of \cref{claim:goodcount-approx}, we see that
\begin{align*}
    1 \geq \frac{\abs{\widehat{\mathscr{G}(\bm{c})}}}{(dn)^{\sum_{\emptyset \subsetneq S \subseteq [2K]} \abs{\bm{c}(S)}} \cdot \prod_{k=1}^{L} \prod_{\emptyset \subsetneq S \subseteq [2K]} \frac{\lambda_{k}^{\bm{c}_{k}(S)}}{\bm{c}_{k}(S)!}} \geq \exp\wrapp{-O\wrapp{\frac{L^{3}m^{2}}{n}}},
\end{align*}
crucially using the inequality in \cref{eq:Gc-size-estimate}, which was already proved there (in the context of \cref{lem:good}) in full generality.
\end{proof}
\begin{proof}[Proof Sketch of \cref{lem:bad-K}]
Writing $\bm{c}^{\mathsf{ALL}} = \sum_{\emptyset \subsetneq S \subseteq [2K]} \bm{c}(S) \in \N^{L}$, observe that
\begin{align*}
    \abs{\varphi(\bm{c}) - \varphi^{\mathscr{G}}(\bm{c})} &= \sum_{\substack{\Gamma^{(S)} \in \mathscr{G}(\bm{c}(S)) \\ \forall \emptyset \subsetneq S \subseteq [2K] \\ \bigcup_{\emptyset \subsetneq S \subseteq [2K]} \Gamma^{(S)} \notin \mathscr{G}}} \Pr\wrapb{\bigcup_{\emptyset \subsetneq S \subseteq [2K]} \Gamma^{(S)} \subseteq M} \\
    &\leq \prod_{k=1}^{L} \frac{\bm{c}_{k}^{\mathsf{ALL}}!}{\prod_{\emptyset \subsetneq S \subseteq [2K]} \bm{c}_{k}(S)!} \cdot \sum_{\Gamma \in \mathscr{B}(\bm{c}^{\mathsf{ALL}})} \Pr[\Gamma \subseteq M]. \tag{Setting $\Gamma = \bigcup_{\emptyset \subsetneq S \subseteq [2K]} \Gamma^{(S)}$}
\end{align*}
Applying the estimate from \cref{eq:sumbad-unif-bound} to $\bm{c}^{\mathsf{ALL}}$, which was already proved in full generality within the context of \cref{lem:bad-K}, we see that
\begin{align*}
    \sum_{\Gamma \in \mathscr{B}(\bm{c}^{\mathsf{ALL}})} \Pr[\Gamma \subseteq M] \leq \frac{\exp\wrapp{\exp(O_{d}(L)) + O(m\log m)}}{n} \cdot \prod_{k=1}^{L} \frac{\lambda_{k}^{\bm{c}_{k}^{\mathsf{ALL}}}}{\bm{c}_{k}^{\mathsf{ALL}}!}.
\end{align*}
Noting that $\lambda_{k}^{\bm{c}_{k}^{\mathsf{ALL}}} = \prod_{\emptyset \subsetneq S \subseteq [2K]} \lambda_{k}^{\bm{c}_{k}(S)}$, the combination of the preceding two displays proves the claim.
\end{proof}

\section{Poisson summation for shifted complex Gaussians}
Here we state and prove a consequence of the Poisson summation formula for shifted complex Gaussians that we use in the complex first and second moment calculations.

\begin{lemma}[Poisson summation over an affine lattice]
\label{lem:poisson-affine-lattice}
Fix $r\in\mathbb{N}$, and let $\mathcal{K}$ be a compact family of
complex symmetric $r\times r$ matrices such that
\[
    \Re(A)\succeq c_0 I
    \qquad\text{for every }A\in\mathcal{K}
\]
for some $c_0>0$. For $n>0$ and $a\in\mathbb{R}^r$,
define the affine lattice
\[
    \Lambda_{n,a}
    :=
    \frac{1}{\sqrt n}\bigl(\mathbb{Z}^r+a\bigr).
\]
Let $P\colon\mathbb{R}^r\to\mathbb{C}$ be a polynomial. Then
there exists $c>0$ such that, uniformly over $A\in\mathcal{K}$
and $a\in\mathbb{R}^r$,
\begin{align}
    \frac{1}{n^{r/2}}
    \sum_{t\in\Lambda_{n,a}}
    P(t)\exp\left(-\frac12t^{\mathsf T}At\right)
    &=
    \int_{\mathbb{R}^r}
    P(t)\exp\left(-\frac12t^{\mathsf T}At\right)\,dt
    +O_P(e^{-cn}).
    \label{eq:poisson-affine-lattice}
\end{align}
The implicit constant may depend on $P$ and $\mathcal{K}$, but
not on $A$, $a$, or $n$.

In particular, if $P$ is odd, then
\begin{equation}
    \frac{1}{n^{r/2}}
    \sum_{t\in\Lambda_{n,a}}
    P(t)\exp\left(-\frac12t^{\mathsf T}At\right)
    =
    O_P(e^{-cn}).
    \label{eq:poisson-affine-lattice-odd}
\end{equation}
For $P\equiv1$, we have
\begin{equation}
    \frac{1}{n^{r/2}}
    \sum_{t\in\Lambda_{n,a}}
    \exp\left(-\frac12t^{\mathsf T}At\right)
    =
    \frac{(2\pi)^{r/2}}{\sqrt{\det A}}
    \left(1+O(e^{-cn})\right),
    \label{eq:poisson-affine-lattice-gaussian}
\end{equation}
where the branch of $\sqrt{\det A}$ is obtained by continuous
extension from the real positive-definite matrices.
\end{lemma}

\begin{proof}
We use the Fourier-transform convention
\[
    \widehat f(\xi)
    =
    \int_{\mathbb{R}^r}
    f(t)e^{-2\pi i\langle t,\xi\rangle}\,dt.
\]
For $A\in\mathcal{K}$, define
\[
    f_{A,P}(t)
    :=
    P(t)\exp\left(-\frac12t^{\mathsf T}At\right).
\]
Since $\Re(A)\succeq c_0I$, the function $f_{A,P}$ belongs to
the Schwartz space.

Poisson summation
\cite[Chapter~5, Section~3, Theorem~3.1]{SteinShakarchi2003},
applied coordinatewise, gives
\begin{equation}
    \frac{1}{n^{r/2}}
    \sum_{t\in\Lambda_{n,a}} f_{A,P}(t)
    =
    \sum_{k\in\mathbb{Z}^r}
    e^{2\pi i\langle k,a\rangle}
    \widehat f_{A,P}(\sqrt n\,k).
    \label{eq:poisson-affine-proof}
\end{equation}
The sign in the phase corresponds to the convention
$\Lambda_{n,a}=n^{-1/2}(\mathbb{Z}^r+a)$; its precise sign will
not be important below.

The zero-frequency term on the right-hand side is
\[
    \widehat f_{A,P}(0)
    =
    \int_{\mathbb{R}^r}
    P(t)\exp\left(-\frac12t^{\mathsf T}At\right)\,dt.
\]
By completing the square, it follows that for $P\equiv 1$,  the complex Gaussian integral gives
\begin{equation}
    \widehat{f}_{A,1}(\xi)
    =
    \frac{(2\pi)^{r/2}}{\sqrt{\det A}}
    \exp\left(-2\pi^2\xi^{\mathsf T}A^{-1}\xi\right).
    \label{eq:fourier-complex-gaussian}
\end{equation}
The identity initially follows for real positive-definite $A$ and
extends to $\Re(A)\succ 0$ by analytic continuation. The same
continuation determines the branch of $\sqrt{\det A}$.

For a general polynomial $P$, differentiating both sides of~\cref{eq:fourier-complex-gaussian} implies that
\begin{equation}
    \widehat{f}_{A,P}(\xi)
    =
    Q_{A,P}(\xi)
    \exp\left(-2\pi^2\xi^{\mathsf T}A^{-1}\xi\right),
    \label{eq:fourier-polynomial-gaussian}
\end{equation}
where $Q_{A,P}$ is a polynomial of degree at most $\deg P$.
Its coefficients are uniformly bounded for $A\in\mathcal{K}$.

We also have
\[
    \Re(A^{-1})
    =
    (A^{-1})^*\Re(A)A^{-1}.
\]
If
\[
    M:=\sup_{A\in\mathcal{K}}\|A\|<\infty,
\]
then
\[
    \Re(A^{-1})\succeq \frac{c_0}{M^2}I
\]
uniformly over $A\in\mathcal{K}$. Hence, for $k\neq0$,
\[
    \left|
        \widehat f_{A,P}(\sqrt n\,k)
    \right|
    \leq
    C_P(1+\sqrt n\,\|k\|)^{\deg P}
    e^{-c_1n\|k\|^2}.
\]
Since the phase factor in \eqref{eq:poisson-affine-proof} has
absolute value one, summing over $k\neq0$ gives
\[
    \sum_{k\in\mathbb{Z}^r\setminus\{0\}}
    \left|
        \widehat f_{A,P}(\sqrt n\,k)
    \right|
    =
    O_P(e^{-cn}).
\]
This proves \eqref{eq:poisson-affine-lattice}.

If $P$ is odd, then the zero-frequency term vanishes because
\[
    t\longmapsto
    \exp\left(-\frac12t^{\mathsf T}At\right)
\]
is even. This proves \eqref{eq:poisson-affine-lattice-odd}.
Taking $P\equiv1$ and using the complex Gaussian integral
\[
    \int_{\mathbb{R}^r}
    \exp\left(-\frac12t^{\mathsf T}At\right)\,dt
    =
    \frac{(2\pi)^{r/2}}{\sqrt{\det A}}
\]
proves \eqref{eq:poisson-affine-lattice-gaussian}.
\end{proof}

\begin{corollary}[Truncated affine lattice sum]
\label{cor:poisson-affine-lattice-truncated}
Under the assumptions of
\cref{lem:poisson-affine-lattice}, fix $\gamma\in(0,1/2)$.
Then, uniformly over $A\in\mathcal{K}$ and $a\in\mathbb{R}^r$,
\begin{align}
    &\frac{1}{n^{r/2}}
    \sum_{\substack{t\in\Lambda_{n,a}\\
                    \|t\|_\infty\leq n^\gamma}}
    P(t)\exp\left(-\frac12t^{\mathsf T}At\right)
    \nonumber\\
    &\qquad =
    \int_{\mathbb{R}^r}
    P(t)\exp\left(-\frac12t^{\mathsf T}At\right)\,dt
    +
    O_P\left(e^{-c n^{2\gamma}}\right).
    \label{eq:poisson-affine-lattice-truncated}
\end{align}
In particular, if $P$ is odd, then
\begin{equation}
    \frac{1}{n^{r/2}}
    \sum_{\substack{t\in\Lambda_{n,a}\\
                    \|t\|_\infty\leq n^\gamma}}
    P(t)\exp\left(-\frac12t^{\mathsf T}At\right)
    =
    O_P\left(e^{-c n^{2\gamma}}\right).
    \label{eq:poisson-affine-lattice-truncated-odd}
\end{equation}
\end{corollary}

\begin{proof}
By \cref{lem:poisson-affine-lattice}, the corresponding
sum over the full affine lattice differs from the Gaussian
integral by $O_P(e^{-cn})$. On the other hand,
$\Re(A)\succeq c_0I$ implies
\[
    \left|
        P(t)e^{-\frac12t^{\mathsf T}At}
    \right|
    \leq
    C_P(1+\|t\|)^{\deg P}e^{-\frac{c_0}{2}\|t\|^2}.
\]
Consequently, the contribution of the lattice points satisfying
$\|t\|_\infty>n^\gamma$ is
$O_P(e^{-c n^{2\gamma}})$, after decreasing $c$ if necessary.
\end{proof}

\section{Technical Lemmata}\label{sec:technical}
\begin{proof}[Proof of \cref{lem:reweighting-analytic}]
Since $A_{L}$ and $A_{L,m}$ are analytic functions in $\delta_{k}$ and $b_{k} = \frac{1}{1 + \delta_{k}} - 1$, it suffices to prove that the functions $\delta_{k}(\beta)$ and $\frac{1}{1 + \delta_{k}(\beta)}$ are all analytic on the claimed set for all $k \in \N$.

The function $\delta_{k}(\beta)$ is analytic on the set of $\beta$ such that the denominator $(e^{\beta} + 1)^{k}$ is nonzero. The zeros of this function are the solutions to the equation $e^{\beta} = -1$, i.e. $\beta$ of the form $i(2\ell + 1)\pi$ for $\ell \in \Z$. Thus, for every $k$, the term $\exp(\lambda_{k} \cdot \delta_{k}(\beta))$ is analytic on the set $\wrapc{\beta \in \C : \abs{\imag \beta} < \pi}$, which contains the claimed strip.

The function $\frac{1}{1 + \delta_{k}(\beta)}$ is analytic for $\beta$ such that $1 + \delta_{k}(\beta) \neq 0$. The zeros of $1 + \delta_{k}(\beta)$ are such that for some $\ell \in \Z$,
\begin{align*}
    &\frac{e^{-\beta} - 1}{e^{-\beta} + 1} = e^{i\pi \cdot \frac{2\ell+1}{k}} \\
    &\iff e^{-\beta} = \frac{1 + e^{i\pi \cdot \frac{2\ell + 1}{k}}}{1 - e^{i\pi \cdot \frac{2\ell+1}{k}}} \\
    &\iff \beta = \log\wrapp{\frac{1 - e^{i\pi \cdot \frac{2\ell + 1}{k}}}{1 + e^{i\pi \cdot \frac{2\ell+1}{k}}}}.
\end{align*}
Letting $w = e^{i\pi \cdot \frac{2\ell+1}{k}} \in \partial \D$ and noting that the M\"{o}bius transformation $z \mapsto \frac{1-z}{1+z}$ sends the unit circle $\partial \D$ to the imaginary axis, we see that for any $k \in \N$, any solution to the equation $1 + \delta_{k}(\beta) = 0$ must satisfy $\imag \beta = \pi \cdot \wrapp{\frac{1}{2} + \ell}$ for some $\ell \in \Z$, i.e. $1 + \delta_{k}(\beta)$ is zero-free in the strip $\wrapc{\beta \in \C : \abs{\imag \beta} < \pi/2}$. This shows that $\frac{1}{1 + \delta_{k}(\beta)}$ is analytic on the claimed strip for any $k$.
\end{proof}

\begin{proof}[Proof of \cref{lem:pfn-disj-cycles}]
We use the standard ``even subgraphs expansion'' trick.
\begin{align*}
    Z_{F}(\beta) &= \sum_{\sigma \in \{\pm1\}^{k}} \exp\wrapp{\beta \sum_{uv \in E} \frac{1 - \sigma_{u}\sigma_{v}}{2}} \\
    &= \exp\wrapp{\beta k/2} \cdot \sum_{\sigma \in \{\pm1\}^{k}} \prod_{uv \in E} \exp\wrapp{-\frac{\beta}{2}\sigma_{u}\sigma_{v}} \\
    &= \exp\wrapp{\beta k/2} \cdot \sum_{\sigma \in \{\pm1\}^{k}} \prod_{uv \in E} \cosh\wrapp{-\frac{\beta}{2}\sigma_{u}\sigma_{v}}\cdot \wrapp{1 + \tanh\wrapp{-\frac{\beta}{2}\sigma_{u}\sigma_{v}}} \tag{Using $\exp(x) = \cosh(x) \cdot (1 + \tanh(x))$} \\
    &= \exp\wrapp{\beta k/2} \cdot \cosh(\beta/2)^{k} \cdot \sum_{\sigma \in \{\pm1\}^{k}} \prod_{uv \in E} \wrapp{1 - \sigma_{u}\sigma_{v}\tanh(\beta/2)} \tag{$\cosh(x)$ is even and $\tanh(x)$ is odd} \\
    &= \wrapp{\frac{1 + e^{\beta}}{2}}^{k} \cdot \sum_{S \subseteq E} (-1)^{\abs{S}}\tanh(\beta/2)^{\abs{S}} \sum_{\sigma \in \{\pm1\}^{k}} \prod_{uv \in S} \sigma_{u}\sigma_{v} \\
    &= 2^{k}\wrapp{\frac{1 + e^{\beta}}{2}}^{k} \cdot \sum_{S \subseteq E \text{ even}} (-1)^{\abs{S}}\tanh(\beta/2)^{\abs{S}} \tag{$S$ is even if all vertices have even degree in $S$} \\
    &= 2^{k}\wrapp{\frac{1 + e^{\beta}}{2}}^{k} \cdot \wrapp{1 + (-1)^{k}\tanh(\beta/2)^{k}} \tag{A cycle only has the two even subgraphs $S=\emptyset,E$} \\
    &= 2^{k}\wrapp{\frac{1 + e^{\beta}}{2}}^{k} \cdot (1 + \delta_{k}(\beta)).
\end{align*}
\end{proof}

\begin{proof}[Proof of \cref{lem:shortcyc-approxindep}]
The first claim is clear. For the second, note that if every edge appears at most once in $\Gamma$, then $m(\Gamma) = \abs{\Gamma}$. Hence,
\begin{align*}
    \Pr[\Gamma \subseteq M] = \prod_{j=0}^{\abs{\Gamma} - 1} \frac{1}{dn - (2j + 1)} \in \wrapb{\wrapp{\frac{1}{dn}}^{\abs{\Gamma}}, \wrapp{\frac{1}{dn - 2m}}^{\abs{\Gamma}}}
\end{align*}
Applying the above to each $C \in \Gamma$ also shows that $\prod_{C \in \Gamma} \Pr[C \subseteq M]$ lands in the same interval. Hence,
\begin{align*}
    \frac{\Pr[\Gamma \subseteq M]}{\prod_{C \in \Gamma} \Pr[C \subseteq M]} &\leq \wrapp{\frac{dn}{dn - 2m}}^{\abs{\Gamma}} \\
    &\leq \exp(O(m^{2}/n)). \tag{Using $\abs{\Gamma} \leq m$}
\end{align*}
An analogous argument shows that this ratio is at least $\exp\wrapp{-O(m^{2}/n)}$.
\end{proof}

\end{document}